\documentclass[10pt,reqno,hidelinks]{amsart}
\usepackage{amssymb}
\usepackage{amsmath, amssymb, verbatim, url, mathrsfs}
\usepackage{mathrsfs}
\usepackage{mathtools}
\usepackage{verbatim}
\usepackage{amsthm}
\usepackage{framed}
\usepackage{wasysym}
\usepackage{upgreek}
\usepackage{color}
\usepackage[dvipsnames]{xcolor}
\usepackage{tensor}
\usepackage{accents}
\usepackage{dsfont}
\usepackage{hyperref}
\usepackage{enumerate}
\usepackage[normalem]{ulem}
\usepackage{longtable}
\usepackage{tikz}
\usetikzlibrary{arrows.meta, shapes.geometric, arrows}

\tikzset{%
	>={Latex[width=2mm,length=2mm]},
	base/.style = {rectangle, rounded corners, draw=black,
		minimum width=2cm, minimum height=1cm,
		text centered, font=\sffamily},
	activityStarts/.style = {base, fill=blue!30, text width=2.5cm},
	startstop/.style = {base, fill=red!30, text width=9.5cm},
	activityRuns/.style = {base, fill=green!30, text width=9.5cm},
	process/.style = {base, minimum width=2.5cm, fill=orange!15,
		font=\ttfamily, text width=2.5cm},
}

\numberwithin{equation}{section}

\newtheorem{proposition}{Proposition}[section]
\newtheorem{lemma}[proposition]{Lemma}
\newtheorem{corollary}[proposition]{Corollary}
\newtheorem{theorem}[proposition]{Theorem}

\newtheorem{ass}[proposition]{Assumption}

\theoremstyle{definition}
\newtheorem{definition}[proposition]{Definition}
\newtheorem{remark}[proposition]{Remark}

\newcommand\merge[2]{%
	\ooalign{\hfil$\vcenter{\hbox{$#1$}}$\hfil\cr
		\hfil$\vcenter{\hbox{$\scriptstyle #2$}}$\hfil}}
\newcommand\innab[1]{\merge\nabla{#1}}

\newcommand{\vertiiii}[1]{{\left\vert\kern-0.25ex\left\vert\kern-0.25ex\left\vert\kern-0.25ex\left\vert #1 \right\vert\kern-0.25ex\right\vert\kern-0.25ex\right\vert\kern-0.25ex\right\vert}}
\newcommand{\vertiii}[1]{{\left\vert\kern-0.25ex\left\vert\kern-0.25ex\left\vert #1 \right\vert\kern-0.25ex\right\vert\kern-0.25ex\right\vert}}
\newcommand{\norm}[1]{\|#1\|}
\newcommand{\xv}{\mathbf{x}}
\newcommand{\yv}{\mathbf{y}}

\newcommand{\yve}{\vec{\mathbf{y}}}
\newcommand{\Bb}{\mathbf{B}}

\newcommand{\mfu}{\mathfrak{u}}
\newcommand{\smfu}{\breve{\mathfrak{u}}}
\newcommand{\bhU}{\mathbf{\hat{U}}}

\newcommand{\mrw}{\mathring{w}}

\newcommand{\mrH}{\mathring{H}}
\newcommand{\mrA}{\mathring{A}}

\newcommand{\ih}{\sqrt{\frac{3}{\Lambda}}}

\newcommand{\hmfu}{\hat{\mathfrak{u}}}

\newcommand{\Rbb}{\mathbb{R}}
\newcommand{\Zbb}{\mathbb{Z}}
\newcommand{\Tbb}{\mathbb{T}}
\newcommand{\Pbb}{\mathbb{P}}

\newcommand{\nnb}{\nonumber}
\newcommand{\del}[1]{{\partial_{#1}}}

\newcommand{\AND}{{\quad\text{and}\quad}}

\newcommand{\Hs}{H^{s-1}}
\newcommand{\Hss}{H^{s-1}}
\newcommand{\Hsss}{H^{s-2}}
\newcommand{\Rs}{R^{s-1}}
\newcommand{\Qs}{R^{s-1}}

\newcommand{\Li}{L^\infty}
\newcommand{\la}{\langle}
\newcommand{\ra}{\rangle}

\newcommand{\starcup}{$\sqcup$\kern-0.58em{$\star$}}
\newcommand{\udn}[1]{\bar{\innab{-}}_{#1}}

\newcommand{\bb}{\bar{\boxminus}}

\newcommand{\al}[2]{
\begin{align}\label{E:#1}
  #2
\end{align}
}

\newcommand{\ali}[1]{
\begin{align}
  #1
\end{align}
}
\newcommand{\gat}[1]{
	\begin{gather}
	#1
	\end{gather}
}
\newcommand{\als}[1]{
\begin{align*}
  #1
\end{align*}
}

\newcommand{\p}[1]{
\begin{pmatrix}
  #1
\end{pmatrix}
}

\allowdisplaybreaks

\DeclareMathOperator{\diag}{diag}

\begin{document}

\title{Cosmological Newtonian limits on large spacetime scales}

\address{School of Mathematical Sciences, Monash University, Melbourne, VIC 3800, Australia}
\email{chao.liu.math@foxmail.com}
\email{todd.oliynyk@monash.edu}
\author{Chao Liu \and Todd A. Oliynyk}

\begin{abstract}

We establish the existence of $1$-parameter families of $\epsilon$-dependent solutions to the Einstein-Euler equations
with a positive cosmological constant $\Lambda >0$ and a linear equation of state $p=\epsilon^2 K \rho$, $0<K\leq 1/3$, for the parameter
values $0<\epsilon < \epsilon_0$. These solutions
exist globally on the manifold $M=(0,1]\times \Rbb^3$, are future complete, and converge as $\epsilon \searrow 0$ to solutions of the cosmological Poisson-Euler equations.
They represent inhomogeneous, nonlinear perturbations of a  FLRW fluid solution where the inhomogeneities are driven by localized matter fluctuations that evolve to good approximation according to Newtonian gravity. 
\end{abstract}

\maketitle

\section{Introduction} \label{S:INTRO}

Galaxies and clusters of galaxies are prime examples of large scale structures in our universe. Their formation requires non-linear interactions and cannot be analyzed using perturbation theory alone. Currently, cosmological Newtonian N-body simulations \cite{Crocce2010,Evrard2001,HahnAngulo2016,Springel2005,Springel2005a,Teyssier2002} are the only well developed tool for studying structure formation. However, the Universe is fundamentally relativistic, and so the use of Newtonian simulations must be carefully justified. This leads naturally to the question: \emph{On what scales can Newtonian cosmological simulations be trusted to approximate realistic relativistic cosmologies?} The main aim of this article is to rigorously answer this question. Informally, we establish, under suitable assumptions, the existence of realistic inhomogeneous cosmological solutions that \textbf{(i)} admit a foliation by spacelike (i.e. constant time) hypersurfaces diffeomorphic to $\Rbb^3$, \textbf{(ii)} exist globally to the future, \textbf{(iii)} can be approximated to arbitrary precision by
a Newtonian solution, and \textbf{(iv)} represent a non-linear perturbation of a  Friedmann-Lema\^itre-Robertson-Walker  (FLRW) fluid solution; see Theorem \ref{T:MAINTHEOREM}
for the precise statement.

In this article, we treat all matter in the Universe as a perfect fluid with a linear equation of state, a widely used approximation in cosmological studies, and 
we assume a positive cosmological constant $\Lambda>0$ in concordance with observational evidence. The evolution of such fluids are governed by the  
Einstein-Euler equations given by
\begin{align}
    \tilde{G}^{\mu\nu}+\Lambda\tilde{g}^{\mu\nu}&=\tilde{T}^{\mu\nu}, \label{E:ORIGINALEESYSTEM.a}\\
    \tilde{\nabla}_\mu\tilde{T}^{\mu\nu}&=0, \label{E:ORIGINALEESYSTEM.b}
\end{align}
where $\tilde{G}^{\mu\nu}$ is the Einstein tensor of the metric $\tilde{g}=\tilde{g}_{\mu\nu}d\bar{x}^\mu d\bar{x}^\nu$,
\begin{align*}
  \tilde{T}^{\mu\nu}=(\bar{\rho}+ \bar{p})\tilde{v}^\mu \tilde{v}^\nu +\bar{p}\tilde{g}^{\mu\nu}
\end{align*}
is the perfect fluid stress-energy tensor, the pressure $\bar{p}$ is determined by the proper energy density $\bar{\rho}$
via the linear equation of state 
\begin{align*}
  \bar{p}=\epsilon^2 K\bar{\rho}, \qquad 0 < K \leq \frac{1}{3},
\end{align*}
the fluid four-velocity $\tilde{v}^\nu$ is normalized by
\begin{equation} \label{vnorm}
  \tilde{v}^\mu \tilde{v}_\mu=-1,
\end{equation}
and the dimensionless parameter $\epsilon$ can be identified with the ratio $\epsilon = \frac{v_T}{c}$,
where $c$ is the speed of light and  $v_T$ is a characteristic speed associated to the fluid.

The proof of our main result,  Theorem \ref{T:MAINTHEOREM}, is based on a rigorous \emph{Newtonian limit} argument, that is, taking the $\epsilon \searrow 0$
limit of solutions to the Einstein-Euler equations. The starting point for the Newtonian limit argument is the introduction of a suitable $1$-parameter 
family of background
solutions to the Einstein-Euler equations \eqref{E:ORIGINALEESYSTEM.a}-\eqref{E:ORIGINALEESYSTEM.b} that have a well defined Newtonian limit. For our argument, we use a $1$-parameter family of FLRW solutions that
represent a family of homogeneous, fluid filled universes undergoing accelerated expansion.  Letting $(\bar{x}^i)$, $i=1,2,3$, denote the standard coordinates on the $\mathbb{R}^3$ and
$t=\bar{x}^0$ a time coordinate on the interval $(0,1]$, the
FLRW family we employ is defined on the manifold
\begin{align*}
	M=(0,1]\times \mathbb{R}^3
\end{align*}
and the metric, four-velocity, and proper energy density are given by
\begin{align}
\tilde{h}(t) &= -\frac{3}{\Lambda t^2} dtdt + a(t)^2 \delta_{ij}d\bar{x}^i d\bar{x}^j, \label{FLRW.a}\\
\tilde{v}_{H}(t) &= -t\sqrt{\frac{\Lambda}{3}}\partial_{t}, \label{FLRW.b}
\intertext{and}
\mu(t) &= \frac{\mu(1)}{a(t)^{3(1+\epsilon^2 K)}} \label{FLRW.c},
\end{align}
respectively,
where the initial proper energy density $\mu(1)$ is freely specifiable and $a(t)$ satisfies
\begin{equation}\label{E:TPTA}
-t a'(t) =a(t)\ih \sqrt{\frac{\Lambda}{3}+\frac{\mu(t)}{3}},\qquad a(1)=1.
\end{equation}

\begin{remark} For simplicity,  we assume that the homogeneous initial density $\mu(1)$ is independent of
$\epsilon$. All of the results established in this article remain true if $\mu(1)$ is allowed to depend on $\epsilon$ in a $C^1$ manner,
that is, the map $[0,\epsilon_0]\ni \epsilon \longmapsto \mu^\epsilon(1) \in \Rbb_{>0}$ is $C^1$ for
some $\epsilon_0 > 0$.
\end{remark}

\begin{remark}
The representation \eqref{FLRW.a}-\eqref{FLRW.c} of the FLRW solutions is not the standard one due
to the choice of time coordinate that compactifies the time interval from $[0,\infty)$ in the standard presentation to $(0,1]$ in the coordinates
used here. Letting $\tau$ denote the standard time coordinate, the relationship between the two time coordinates is
\begin{align*}
  t=e^{-\sqrt{\frac{\Lambda}{3}}\tau}.
\end{align*}
Due to our choice of time coordinate, the future lies in the direction of \textit{decreasing} $t$ and
timelike infinity is located at $t=0$.
\end{remark}

\begin{remark} \label{FLRWlimrem}
As we show in \S\ref{FLRWanal}, the  FLRW solutions $\{a,\mu\}$ depend regularly on $\epsilon$ and have well
defined Newtonian limits.
Letting
\begin{equation} \label{arhoringdef}
\mathring{a} = \lim_{\epsilon\searrow 0} a \quad \text{and} \quad \mathring{\mu}  = \lim_{\epsilon\searrow 0} \mu
\end{equation}
denote the Newtonian limit of $a$ and $\mu$, respectively, it then follows from \eqref{FLRW.c} and
\eqref{E:TPTA} that $\{\mathring{a},\mathring{\mu}\}$ satisfy
\begin{equation*}\label{rhoringdef}
\mathring{\mu} =  \frac{\mathring{\mu}(1)}{\mathring{a}(t)^{3}}
\end{equation*}
and
\begin{equation*} \label{aringdef}
-t \mathring{a}'(t) =\mathring{a}(t)\ih \sqrt{\frac{\Lambda}{3}+\frac{\mathring{\mu}(t)}{3}},\qquad \mathring{a}(1)=1,
\end{equation*}
which define the Newtonian limit of the FLRW equations. We further note that 
\begin{equation*}
\mu(1)=\mathring{\mu}(1).
\end{equation*}
\end{remark}

Throughout this article, we will refer to the global coordinates $(\bar{x}^\mu)$ on the manifold $M$, defined above, as \textit{relativistic coordinates}.  In addition to the relativistic coordinates, 
we need to introduce the spatially rescaled coordinates
$(x^\mu)$ on $M$ defined by
\begin{align}\label{e:NRcoor}
t=\bar{x}^0=x^0 \quad \text{and} \quad \bar{x}^i=\epsilon x^i, \qquad \epsilon > 0,
\end{align}
which we will refer to as \textit{Newtonian coordinates}. These coordinates are necessary for the definition of the Newtonian limit
since they are used to define the sense in which solutions converge as $\epsilon \searrow 0$.

Before proceeding with our discussion of the Newtonian limit and the statement of Theorem \ref{T:MAINTHEOREM}, we need to first fix our notation and conventions, and introduce a number of new variables that will be needed to state our main result.

\subsection{Notation}

\subsubsection{Index of notation} An index containing frequently used definitions and non-standard notation can be found in Appendix \ref{index}.

\subsubsection{Indices and coordinates}\label{iandc} Unless stated otherwise, our indexing convention will be as follows: we use lower case Latin letters, e.g. $i, j,k$, for spatial indices that run from $1$ to $n$, and lower case Greek letters, e.g. $\alpha, \beta, \gamma$, for spacetime indices
that run from $0$ to $n$. When considering the Einstein-Euler equations, we will restrict our attention to the physical case $n=3$.

For scalar functions $\bar{f}(\bar{x}^0,\bar{x}^i)$
that are given in terms of the relativistic coordinates, we will use the notation
\begin{equation} \label{Neval}
\underline{\bar{f}}(t,x^i) := \bar{f}(t,\epsilon x^i)
\end{equation}
to denote the representation of $\bar{f}$ in Newtonian coordinates. More generally, we use this notation for components of
tensors. For example, given the representation $\bar{X}=\bar{X}^j(\bar{x}^0, x^i)\bar{\partial}_j$ of the vector field $\bar{X}$
in relativistic coordinates, then $\underline{\bar{X}^j}$ is defined by 
$\underline{\bar{X}^j}(t, x^i)=\bar{X}^j(t,\epsilon x^i)$.

\subsubsection{Derivatives}
Partial derivatives with respect to the  Newtonian coordinates $(x^\mu)=(t,x^i)$ and the relativistic coordinates $(\bar{x}^\mu)=(t,\bar{x}^i)$ will be denoted by $\partial_\mu = \partial/\partial x^\mu$ and
$\bar{\partial}_{\mu} = \partial/\partial \bar{x}^\mu$, respectively, and we use
$Du=(\partial_j u)$ and $\partial u = (\partial_\mu u)$ to denote the spatial and spacetime gradients, respectively, with respect to the Newtonian coordinates, and $\bar{\partial} u = (\bar{\partial}_\mu u)$ to denote the spacetime gradient with
respect to the relativistic coordinates.

Greek letters will also be used to denote multi-indices, e.g.
$\alpha = (\alpha_1,\alpha_2,\ldots,\alpha_n)\in \mathbb{Z}_{\geq 0}^n$, and we will employ the standard notation $D^\alpha = \partial_{1}^{\alpha_1} \partial_{2}^{\alpha_2}\cdots
\partial_{n}^{\alpha_n}$ for spatial partial derivatives. It will be clear from context whether a Greek letter stands for a spacetime coordinate index or a multi-index. Furthermore, we will use $D^k u = \{ D^\alpha u \,|\, |\alpha|=k\}$ to denote the collection of partial derivatives of order $k$, and we will have occasion to use the notation $\partial^i = \delta^{ij}\partial_j$
for spatial partial derivatives.

Given a vector-valued map $f(u)$, where $u$ is a vector, we use $D f$ and $D_u f$ interchangeably to denote the derivative with respect to the vector $u$, and use the standard notation
\begin{equation*}
  D f(u)\cdot \delta u := \left.\frac{d}{dt}\right|_{t=0} f(u+t\delta u)
\end{equation*}
for the action of the linear operator $D f$ on the vector $\delta u$. For vector-valued maps $f(u,v)$ of two (or more)
variables, we use the notation $D_1 f$ and $D_u f$ interchangeably for the partial
derivative with respect to the first variable, i.e.
\begin{equation*}
  D_u f(u,v)\cdot \delta u := \left.\frac{d}{dt}\right|_{t=0} f(u+t\delta u,v),
\end{equation*}
and a similar notation for the partial derivative with respect to the other variable.

\subsubsection{Function spaces}  \label{S:ulss}
Given a finite dimensional vector space $V$, we let
$H^s(\mathbb{R}^n,V)$, $s\in \mathbb{Z}_{\geq 0}$,
denote the space of maps from $\mathbb{R}^n$ to $V$ with $s$ derivatives in $L^2(\Rbb^n)$. When the
vector space $V$ is clear from context, we write $H^s(\mathbb{R}^n)$ instead of $H^s(\mathbb{R}^n,V)$.
Letting
\begin{equation*}
\langle{u,v\rangle} = \int_{\mathbb{R}^n} (u(x),v(x))\, d^n x,
\end{equation*}
where $(\cdot,\cdot)$
is a fixed inner product on $V$, denote the standard $L^2$ inner product, the $H^s$ norm is defined by
\begin{equation*}
\|u\|_{H^s}^2 = \sum_{0\leq |\alpha|\leq s} \langle D^\alpha u, D^\alpha u \rangle.
\end{equation*}

We let  $H^s_{\emph{ul}}(\Rbb^n,V)$ denote the \textit{uniformly local Sobolev spaces}, which we
recall are defined as follows:
 let $\theta\in C^\infty_0(\Rbb^n)$ be a function such that $\theta>0$ and
\begin{equation*}
\theta(x) =
\begin{cases}
1, &|x|\leq \frac{1}{2}  \\
0, &|x|>1
\end{cases}	
\end{equation*}	
and define $\theta_{d,y}(x)$ by $\theta_{d,y}=\theta((x-y)/d)$. Then $u$ belongs to $H^s_{\emph{ul}}(\Rbb^n,V)$ if there
exists a  $d>0$ such that
	\begin{align*}
		\|u\|_{H^s_{\emph{ul}}}:=\sup_{y\in \Rbb^n}\|\theta_{d,y}u\|_{H^s}<\infty.
	\end{align*}
We note that the norms corresponding  to different $d>0$ are equivalent, and in addition, they are equivalent to
the norm $\sqrt{\sup_{y\in \Rbb^3}\sum_{0\leq |\alpha|\leq s}\int_{\Rbb^n}\theta_{d,y}[D^\alpha u(x)]^2 d^n x}$.

For $s\in \Zbb_{\geq 1}$, we define the spaces
\begin{align*}
  R^s(\Rbb^3,V)= & \left\{u\in  L^6(\Rbb^3,V) \: | \: D u\in H^{s-1}(\Rbb^3,V)\right\} 
\intertext{and}
  K^s(\Rbb^n,V)= & \{\, u\in L^\infty(\Rbb^n,V) \: | \: Du \in H^{s-1}(\Rbb^n,V) \, \} 
\end{align*}
with norms 
\begin{align}
	\|u\|_{R^s}= \|D u\|_{H^{s-1}} +\|u\|_{L^6} 
\AND
    \norm{u}_{K^s} = \norm{u}_{L^\infty} + \norm{Du}_{H^{s-1}}, \label{e:ennorm}
\end{align}
respectively. On $\Rbb^3$ and for $s\in \Zbb_{\geq 2}$, the inequalities
\begin{equation} \label{E:NORMEQ1}
\norm{Du}_{H^{s-1}} + \norm{u}_{W^{s-1,6}} + \norm{u}_{W^{s-2,\infty}} \lesssim
\norm{u}_{R^s} \lesssim \norm{Du}_{H^{s-1}} + \norm{u}_{W^{s-1,6}} + \norm{u}_{W^{s-2,\infty}},
\end{equation}
\begin{equation}
  \|u\|_{R^s}\lesssim \|u\|_{H^s} \lesssim   \|u\|_{L^{\frac{6}{5}}}+\|u\|_{K^s}        \label{E:HQR}
\end{equation}
are a direct consequences of the Sobolev and interpolation inequalities, see Theorems \ref{T:HOLDER}\eqref{T:H2} and \ref{Sobolev}.

To handle the smoothness of coefficients that appear in various equations, we introduce the spaces
\begin{equation*}
E^{p}((0,\epsilon_0)\times (T_1,T_2)\times U,V),\quad p \in \Zbb_{\geq 0},
\end{equation*}
which are defined to be the set of $V$-valued maps $f(\epsilon,t,\xi)$ that
are smooth on the open set $(0,\epsilon_0)\times (T_1,T_2)\times U$, where $U$ $\subset$ $\Rbb^n \times \Rbb^N$
is open, and for which there exist constants $C_{k,\ell}>0$, $(k,\ell)\in \{0,1,\ldots,p\}\times \Zbb_{\geq 0}$,
such that
\begin{equation*}
|\del{t}^k  D_\xi^\ell f(\epsilon,t,\xi)| \leq C_{k,\ell}, \quad \forall \,
(\epsilon,t,\xi) \in  (0,\epsilon_0)\times (T_1,T_2)\times U.
\end{equation*}
If $V=\Rbb$ or $V$ clear from context, we will drop the $V$ and simply write $E^{p}((0,\epsilon_0)\times (T_1,T_2)\times U)$. Moreover,
we will use the notation $E^{p}((T_1,T_2)\times U,V)$ to denote the subspace of $\epsilon$-independent maps. By uniform continuity, the limit
$f_0(t,\xi) := \lim_{\epsilon \searrow 0}f(\epsilon,t,\xi)$ exists for each $f\in E^{p}((0,\epsilon_0)\times (T_1,T_2)\times U,V)$ and defines an element of $E^{p}((T_1,T_2)\times U,V)$.

We further define, for fixed $\epsilon_0 >0 $, the spaces
\begin{equation*}
X^s_{\epsilon_0}(\mathbb{R}^3) = (0,\epsilon_0)\times R^{s+1}(\mathbb{R}^3,\mathbb{S}_3) \times H^s(\mathbb{R}^3,\mathbb{S}_3) \times  \Bigl(L^\frac{6}{5} \cap K^s(\mathbb{R}^3)\Bigr) \times \Bigl(L^\frac{6}{5} \cap K^s(\mathbb{R}^3, \Rbb^3)\Bigr)
\end{equation*}
and
	\begin{align*}
X^s(\Rbb^3)=R^{s+1}(\Rbb^3,\mathbb{S}_4)\times R^{s+1}(\Rbb^3,\Rbb)\times R^s(\Rbb^3,\mathbb{S}_3)\times \bigl(R^s(\Rbb^3,\Rbb^3) \bigr)^2\times R^s(\Rbb^3,\Rbb)\times R^s(\Rbb^3,\Rbb^3)\times R^s(\Rbb^3,\Rbb),
\end{align*}
where $\mathbb{S}_N$ denotes the space of symmetric $N\times N$ matrices.

If $X$ and $Y$ are two Banach spaces with norms $\|\cdot\|_{X}$ and $\|\cdot\|_{Y}$, respectively, then we use
\begin{align*}
	\|f\|_{X\cap Y}:=\|f\|_{X}+\|f\|_{Y}, \quad f\in X\cap Y,
\end{align*}
to denote the intersection norm. We will also employ the notation  $B_r(X) = \{\, f \in X\, |, \|f\|_X< r\,\}$ to denote the open ball of radius $r$ in $X$ that is centered at $0$. 

\subsubsection{Constants}
We employ that standard notation
\begin{equation*}
a \lesssim b
\end{equation*}
for inequalities of the form
\begin{equation*}
a \leq C b
\end{equation*}
in situations where the precise value or dependence on
other quantities of the constant $C$ is not required. On the other hand, when the dependence of the constant
on other inequalities needs to be specified, for example if the constant depends on the norms $\|u\|_{L^\infty}$ and $\|v\|_{L^\infty}$, we use the notation
\begin{equation*}
C = C(\|u\|_{L^\infty},\|v\|_{L^\infty}).
\end{equation*}
Constants of this type will always be non-negative, non-decreasing, continuous functions of their arguments, and in general, $C$ will be used
to denote constants that may change from line to line. When we want to isolate
a particular constant for use later on, we will label the constant with a subscript, e.g. $C_1, C_2, C_3$, etc.

\subsubsection{Remainder terms\label{remainder}}
In order to simplify the handling of remainder terms whose exact form is not important, we will, unless otherwise stated, use upper case script letters, e.g.
  $\mathscr{S}(\epsilon,t,x,\xi)$ and $\mathscr{T}(\epsilon,t,x,\xi)$, and
 upper case script letters with a hat, e.g.
$\mathscr{\hat{S}}(\epsilon,t,x,\xi)$ and $\mathscr{\hat{T}}(\epsilon,t,x,\xi)$, 
 to denote vector valued maps that,
 for some $\epsilon_0,R >0$ and $N\in \mathbb{Z}_{\geq 1}$, are elements
of the spaces $E^0\bigl( (0,\epsilon_0)\times (0,2)\times \Rbb^n \times B_R\bigl(\mathbb{R}^N\bigr)\bigr)$ and 
$E^1\bigl( (0,\epsilon_0)\times (0,2)\times \Rbb^n \times B_R\bigl(\mathbb{R}^N\bigr)\bigr)$, respectively. In addition, we will use
upper case script letters with a breve, e.g. $\breve{\mathscr{Q}}(\xi)$ and $\breve{\mathscr{R}}(\xi)$, to denote analytic maps of
the variable $\xi$ whose exact form is
not important; for these maps, the domain of analyticity will be clear from context.

We will say that a function $f(x,y)$ \textit{vanishes to the $n^{\text{th}}$ order in $y$} if it satisfies $f(x,y)\sim \mathrm{O}(y^n)$ as $y\rightarrow 0$, that is, there exists a positive constant $C$ such that $|f(x,y)|\leq C|y|^n$ as $y \rightarrow 0$.

\subsection{Conformal Einstein-Euler equations} \label{conformalEinsteinEuler} Returning to the setup of the Newtonian limit, we follow \cite{Liu2017} and replace
the physical (inverse) metric $\tilde{g}{}^{\mu\nu}$ and fluid four-velocity
$\tilde{v}{}^\mu$ by the conformally rescaled versions defined by
\begin{align}
  \bar{g}^{\mu\nu}&=e^{2\Psi}\tilde{g}^{\mu\nu}\label{E:CONFORMALTRANSF}
\intertext{and}
  \bar{v}^\mu&=e^\Psi\tilde{v}^\mu, \label{E:CONFORMALVELOCITY}
\end{align}
respectively.
Recalling the well known identity
\begin{equation*}\label{e:riccidiff}
\tilde{R}_{\mu\nu}-\bar{R}_{\mu\nu}=-\bar{g}_{\mu\nu}\bar{\Box}\Psi-2\bar{\nabla}_\mu\bar{\nabla}_\nu \Psi+2(\bar{\nabla}_\mu\Psi\bar{\nabla}_\nu \Psi-|\bar{\nabla}\Psi|^2_{\bar{g}}\bar{g}_{\mu\nu}),
\end{equation*}
where $\bar{\nabla}_\mu$ and $\bar{R}_{\mu\nu}$ are the covariant derivative and Ricci tensor of $\bar{g}_{\mu\nu}$, respectively, $\bar{\Box}=\bar{\nabla}^\mu\bar{\nabla}_\mu$, and
 $|\bar{\nabla}\Psi|^2_{\bar{g}}=\bar{g}^{\mu\nu}\bar{\nabla}_\mu\Psi\bar{\nabla}_\nu\Psi$,
we find that under the change of variables \eqref{E:CONFORMALTRANSF}-\eqref{E:CONFORMALVELOCITY}
the Einstein equation \eqref{E:ORIGINALEESYSTEM.a} transforms as
\begin{equation}\label{E:CONFORMALEINSTEIN1}
  \bar{G}^{\mu\nu}=\bar{T}^{\mu\nu}:=e^{4\Psi}\tilde{T}^{\mu\nu}-e^{2\Psi}\Lambda\bar{g}^{\mu\nu}
  +2(\bar{\nabla}^\mu\bar{\nabla}^\nu\Psi-\bar{\nabla}^\mu\Psi\bar{\nabla}^\nu\Psi)
  -(2\bar{\Box}\Psi+|\bar{\nabla}\Psi|^2_{\bar{g}})
  \bar{g}^{\mu\nu},
\end{equation}
where here and in the following, unless otherwise specified, we raise and lower all coordinate tensor indices using the conformal metric $\bar{g}_{\mu\nu}$.
Contracting the free indices of \eqref{E:CONFORMALEINSTEIN1} gives $\bar{R}=4\Lambda-\bar{T}$,
where $\bar{T}=\bar{g}_{\mu\nu}\bar{T}^{\mu\nu}$ and $\bar{R}$ is
the Ricci scalar of the conformal metric. Using this and the definition $\bar{G}^{\mu\nu} =
\bar{R}^{\mu\nu}-\frac{1}{2}\bar{R}\bar{g}^{\mu\nu}$ of the Einstein tensor, we can write \eqref{E:CONFORMALEINSTEIN1} as
\begin{align}
  \bar{R}^{\mu\nu}=2\bar{\nabla}^\mu\bar{\nabla}^\nu\Psi-2\bar{\nabla}^\mu\Psi\bar{\nabla}^\nu\Psi
  +\left[\bar{\Box}\Psi
  +2|\bar{\nabla}\Psi|^2+\left(\frac{1-\epsilon^2K}{2}\bar{\rho}+\Lambda\right)e^{2\Psi}\right]\bar{g}
  ^{\mu\nu}+e^{2\Psi}(1+\epsilon^2K)\bar{\rho} \bar{v}^\mu \bar{v}^\nu, \label{E:EXPANSIONOFEIN}
\end{align}
which we will refer to as \textit{the conformal Einstein equations}.

Following \cite{Liu2017,Oliynyk2016a}, we fix the conformal factor by setting
\begin{align}\label{E:CONFORMALFACTOR}
\Psi=-\ln{t},
\end{align}
and we introduce the background metric
\begin{align}\label{E:CONFORMALFLRW}
\bar{h}= \bar{h}_{\mu\nu}d\bar{x}^\mu d\bar{x}^\nu := -\frac{3}{\Lambda}dtdt+E^2(t)\delta_{ij}d\bar{x}^id\bar{x}^j,
\end{align}
where
\begin{align} \label{E:DEFE}
E(t)=a(t)t.
\end{align}
We note that the background metric $\bar{h}_{\mu\nu}$ is conformally related to the FLRW metric \eqref{FLRW.a} according to \eqref{E:CONFORMALTRANSF} under
the replacement $\bar{g}^{\mu\nu} \mapsto \bar{h}^{\mu\nu}$ and  $\tilde{g}^{\mu\nu} \mapsto \tilde{h}^{\mu\nu}$,
where $\bar{h}^{\mu\nu}$ denotes the inverse of $\bar{h}_{\mu\nu}$.
By \eqref{E:TPTA}, we observe that $E(t)$ satisfies
\begin{align}\label{E:PTE}
\partial_t E(t)=\frac{1}{t}E(t)\Omega(t),
\end{align}
where $\Omega(t)$ is defined by
\begin{align} \label{e:ome}
\Omega(t)=1- \sqrt{1+\frac{\mu(t)}{\Lambda}}<0.
\end{align}
For use below, we observe that the relation
\begin{align}\label{e:muomeg}
\mu =\Omega(\Omega-2)
\Lambda
\end{align}
follow directly from the definition of $\Omega$.
Straightforward calculations show that the non-vanishing Christoffel symbols $\bar{\gamma}^\sigma_{\mu\nu}$, contracted Christoffel symbols $\bar{\gamma}^\sigma$, Riemannian tensor $\tensor{\bar{\mathcal{R}}}{_{\alpha\beta\sigma}^\mu}$, Ricci tensors $\bar{\mathcal{R}}_{\mu\nu}$ and $\bar{\mathcal{R}}^{\mu\nu}=\bar{h}^{\mu\alpha}\bar{h}^{\mu\beta}
\bar{\mathcal{R}}_{\alpha\beta}$, and Ricci scalar $\bar{\mathcal{R}}$ of the background metric \eqref{E:CONFORMALFLRW}  are given by
\begin{gather}
\bar{\gamma}^0_{ij}=\frac{\Lambda}{3t}E^2\Omega\delta_{ij} \qquad
\bar{\gamma}^i_{j0}=\frac{1}{t} \Omega \delta^i_j, \qquad \bar{\gamma}^\sigma:=\bar{h}^{\mu\nu}\bar{\gamma}^\sigma_{\mu\nu}=\frac{\Lambda}{t}\Omega
\delta^\sigma_0, \label{E:HOMCHRIS} \\
\tensor{\bar{\mathcal{R}}}{_{0i0}^j}=-\tensor{\bar{\mathcal{R}}}{_{i00}^j}= \frac{1}{t^2}(\Omega-\Omega^2-t\del{t}\Omega) \delta^j_i,  \label{e:Hriem1}\\
\tensor{\bar{\mathcal{R}}}{_{0ij}^0}=-\tensor{\bar{\mathcal{R}}}{_{i0j}^0}=\frac{\Lambda}{3t^2}E^2(\Omega-\Omega^2-t\del{t}\Omega)\delta_{ij} , \qquad
\tensor{\bar{\mathcal{R}}}{_{ijk}^l}=\frac{\Lambda}{3t^2}E^2\Omega^2(\delta_{ik}\delta^l_j-\delta_{jk}\delta^l_i), \label{e:Hriem} \\
\bar{\mathcal{R}}_{00}=\tensor{\bar{\mathcal{R}}}{_{0i0}^i}=\frac{3}{t^2}(\Omega-\Omega^2-t\del{t}\Omega), \qquad 
\bar{\mathcal{R}}_{kl}
=  -\frac{\Lambda}{3t^2}E^2(\Omega-3\Omega^2-t\del{t}\Omega)\delta_{kl}, \label{e:R1}\\
\bar{\mathcal{R}}^{00}
=\frac{\Lambda^2}{3t^2} (\Omega-\Omega^2-t\del{t}\Omega), \qquad 
\bar{\mathcal{R}}^{ij}
=-\frac{\Lambda}{3t^2}E^{-2}(\Omega-3\Omega^2-t\del{t}\Omega) \delta^{ij} \\
\intertext{and}
\bar{\mathcal{R}}=-\frac{2\Lambda}{t^2}(\Omega-2\Omega^2-t\del{t}\Omega).  \label{e:R2}
\end{gather}
Since the FLRW solution \eqref{FLRW.a}-\eqref{FLRW.c} satisfies the Einstein equations, we deduce from
\eqref{E:EXPANSIONOFEIN} that  $\bar{h}^{\mu\nu}$ and $\mu$ satisfy
\begin{align}\label{e:homconfein}
\bar{\mathcal{R}}^{\mu\nu}=2\udn{}^\nu\udn{}^\nu \Psi -2\udn{}^\mu\Psi\udn{}^\nu\Psi+\biggl[\bb\Psi+2|\udn{}\Psi|^2_{\bar{h}}+\Bigl(\frac{1-\epsilon^2K}{2}\bar{\mu}+\Lambda\Bigr)e^{2\Psi}\biggr]\bar{h}^{\mu\nu}+e^{2\Psi}(1+\epsilon^2K)\mu \frac{\Lambda}{3}\delta^\mu_0\delta^\nu_0,
\end{align}
where $\udn{\mu}$ is the covariant derivative with respect to the background metric $\bar{h}_{\alpha\beta}$,
$\bb=\bar{h}^{\mu\nu}\udn{\mu}\udn{\nu}$ and $|\udn{}\Psi|^2_{\bar{h}}=\bar{h}^{\mu\nu}\udn{\mu}\Psi\udn{\nu}\Psi$.

Routine calculations show the Ricci tensor can be expressed as
\begin{align}\label{e:ricci1}
\bar{R}^{\mu\nu}=\frac{1}{2} \bar{g}^{\lambda\sigma}\udn{\lambda}\udn{\sigma}\bar{g}^{\mu\nu}+\bar{\nabla}^{(\mu}\bar{X}^{\nu)}+\mathcal{\bar{R}}^{\mu\nu}+\bar{P}^{\mu\nu}
+\bar{Q}^{\mu\nu}
\end{align}
where
\al{X}{
\bar{X}^\alpha=\bar{g}^{\beta \gamma} \tensor{\bar{X}}{^\alpha_{\beta \gamma}} =-\udn{\lambda}\bar{g}^{\alpha \lambda}+\frac{1}{2}\bar{g}^{\alpha \lambda}\bar{g}_{\sigma \delta}\udn {\lambda}\bar{g}^{\sigma \delta} \AND \tensor{\bar{X}}{^\alpha_{\beta \gamma}}
=-\frac{1}{2}\bigl(\bar{g}_{\lambda \gamma }\udn{\beta}\bar{g}^{\alpha \lambda}+\bar{g}_{\beta \lambda}\udn{\gamma}\bar{g}^{\alpha \lambda}-\bar{g}^{\alpha \lambda}\bar{g}_{\beta \sigma}\bar{g}_{\gamma \delta}\udn{\lambda}\bar{g}^{\sigma \delta}\bigr),
}
\al{P}{
	\bar{P}^{\mu\nu}
	=& -\frac{1}{2} (\bar{g}^{\mu \lambda}-\bar{h}^{\mu \lambda})\bar{h}^{\alpha\beta}\tensor{\mathcal{\bar{R}}}{_{\lambda \alpha\beta}^\nu} -\frac{1}{2} \bar{h}^{\mu \lambda}(\bar{g}^{\alpha\beta}-\bar{h}^{\alpha\beta})\tensor{\mathcal{\bar{R}}}{_{\lambda \alpha\beta}^\nu} -\frac{1}{2} (\bar{g}^{\mu \lambda}-\bar{h}^{\mu \lambda})(\bar{g}^{\alpha\beta}-\bar{h}^{\alpha\beta})\tensor{\mathcal{\bar{R}}}{_{\lambda \alpha\beta}^\nu} \nnb \\
	&-\frac{1}{2}(\bar{g}^{\nu \lambda}-\bar{h}^{\nu \lambda})\bar{h}^{\alpha\beta}\tensor{\mathcal{\bar{R}}}{_{\lambda\alpha\beta}^\mu}-\frac{1}{2}h^{\nu \lambda}(\bar{g}^{\alpha\beta}-\bar{h}^{\alpha\beta}) \tensor{\mathcal{\bar{R}}}{_{\lambda\alpha\beta}^\mu}-\frac{1}{2}(\bar{g}^{\nu \lambda}-\bar{h}^{\nu \lambda})(\bar{g}^{\alpha\beta}-\bar{h}^{\alpha\beta})\tensor{\mathcal{\bar{R}}} {_{\lambda \alpha\beta}^\mu},
}
and
\al{Q}{
	\bar{Q}^{\mu\nu}
= & -\frac{1}{4}\bigl(\bar{g}^{\mu \sigma}\bar{g}^{\nu \beta}\udn{\sigma}\bar{g}_{\lambda \alpha}\udn{\beta}\bar{g}^{\lambda \alpha}+\bar{g}^{\mu \sigma}\bar{g}^{\nu \beta}\udn{\beta}\bar{g}_{\lambda\alpha}\udn{\sigma}\bar{g}^{\lambda\alpha}+\bar{g}_{\alpha\beta}\bar{g}^{\mu \lambda}\udn{\lambda}\bar{g}^{\nu \sigma}\udn{\sigma}\bar{g}^{\alpha\beta}  +\bar{g}_{\alpha\beta}\bar{g}^{\nu \sigma}\udn{\sigma}\bar{g}^{\mu \lambda}\udn{\lambda}\bar{g}^{\alpha\beta}\bigr)\nnb \\
	& +\frac{1}{2}\big(\bar{g}^{\mu \alpha}\bar{g}_{\beta\lambda}\udn{\alpha} \bar{g}^{\nu \sigma} \udn{\sigma}\bar{g}^{\beta\lambda}+\bar{g}^{\mu \alpha} \bar{g}^{\nu \sigma} \udn{\alpha} \bar{g}_{\beta\lambda} \udn{\sigma} \bar{g}^{\beta\lambda}-\udn{\lambda}\bar{g}^{\mu \alpha}\udn{\alpha}\bar{g}^{\lambda\nu}-\udn{\alpha}\bar{g}^{\nu \lambda}\udn{\lambda} \bar{g}^{\alpha\mu}\big) -\bar{g}^{\mu \lambda}\tensor{\bar{X}}{^\nu _{\lambda\sigma}}\bar{X}^\sigma  \nnb\\
	& +\bar{g}^{\mu \beta}\bar{g}^{\nu \sigma}\tensor{\bar{X}}{^{\lambda}_{\beta\sigma}}\tensor{\bar{X}}{^\alpha _{\lambda \alpha}}-\bar{g}^{\mu \beta}\bar{g}^{\nu \sigma}\tensor{\bar{X}}{^{\lambda}_{\alpha\sigma}}\tensor{\bar{X}}{^\alpha _{\lambda\beta}}+\tensor{\bar{X}}{^\alpha_{\alpha\sigma}} \bar{g}^{\mu \beta} \udn{\beta} \bar{g}^{\nu \sigma} -\tensor{\bar{X}}{^\alpha_{\beta\sigma}}\bar{g}^{\mu \beta}\udn{\alpha} \bar{g}^{\nu \sigma}-\tensor{\bar{X}}{^\alpha_{\beta\sigma}}g^{\nu \sigma}\udn{\alpha} \bar{g}^{\mu \beta}.
}
Employing \eqref{e:ricci1}, we can write the conformal Einstein equations \eqref{E:EXPANSIONOFEIN} as
\begin{align*}
&-\bar{g}^{\alpha\beta}\udn{\alpha}\udn{\beta}\bar{g}^{\mu\nu}-2\bar{\nabla}^{(\mu}\bar{X}^{\nu)}-2\mathcal{\bar{R}}^{\mu\nu}-2\bar{P}^{\mu\nu}
-2\bar{Q}^{\mu\nu}
=-4\bar{\nabla}^\mu\bar{\nabla}^\nu\Psi+4\bar{\nabla}^\mu\Psi\bar{\nabla}^\nu\Psi  \nnb \\ &\hspace{4.5cm}
	-2\left[\bar{\Box}\Psi
	+2|\bar{\nabla}\Psi|^2+\left(\frac{1-\epsilon^2K}{2}\bar{\rho}+\Lambda\right)e^{2\Psi}\right]\bar{g}
	^{\mu\nu} -2e^{2\Psi}(1+\epsilon^2K)\bar{\rho} \bar{v}^\mu \bar{v}^\nu.
\end{align*}

Letting $\tilde{\Gamma}^\gamma_{\mu\nu}$ and $\bar{\Gamma}^\gamma_{\mu\nu}$ denote the Christoffel symbols of the metrics $\tilde{g}_{\mu\nu}$
and $\bar{g}_{\mu\nu}$, respectively, the difference $\tilde{\Gamma}^\gamma_{\mu\nu}-
\bar{\Gamma}^\gamma_{\mu\nu}$ is readily calculated to be
$\tilde{\Gamma}^\gamma_{\mu\nu}-\bar{\Gamma}^\gamma_{\mu\nu} =
\bar{g}^{\gamma \alpha}\bigl(\bar{g}_{\mu \alpha}\bar{\nabla}_\nu\Psi
+ \bar{g}_{\nu\alpha}\bar{\nabla}_\mu\Psi - \bar{g}_{\mu\nu} \bar{\nabla}_\alpha\Psi \bigr)$.
Using this, we can express the Euler equations \eqref{E:ORIGINALEESYSTEM.b} as
\begin{equation}\label{Confeul}
\bar{\nabla}_\mu \tilde{T}^{\mu \nu} = -6\tilde{T}^{\mu\nu}\bar{\nabla}_\mu\Psi +\bar{g}_{\alpha\beta}\tilde{T}^{\alpha\beta}
\bar{g}^{\mu\nu}\bar{\nabla}_\mu\Psi,
\end{equation}
which we refer to as the \textit{conformal Euler equations}.

\begin{remark}
Due to our choice of time orientation, the conformal fluid four-velocity $\bar{v}^\mu$, which we assume is future oriented, satisfies
$\bar{v}^0 < 0.$
Furthermore, it follows directly from \eqref{vnorm}, \eqref{E:CONFORMALTRANSF} and \eqref{E:CONFORMALVELOCITY} that
$\bar{v}^\mu$ is normalized, that is, 
\begin{equation} \label{normal}
\bar{v}^\mu\bar{v}_\mu = -1.
\end{equation}
\end{remark}

\subsubsection{Wave gauge}\label{Wavegauge}
In order to obtain a hyperbolic reduction of the conformal Einstein equations that is useful for analyzing the Newtonian limit over long time scales,
we need to choose a gauge that is well defined on long time scales and in the limit $\epsilon \searrow 0$.
For this, we follow  \cite{Liu2017} and employ the \textit{wave gauge} defined by
\begin{align}\label{E:WAVEGAUGE}
   \bar{Z}^\mu := \bar{X}^\mu+\bar{Y}^\mu= 0
\end{align}
where
\begin{align}
  \bar{X}^\mu&  =-\udn{\nu}\bar{g}^{\mu\nu}
  +\frac{1}{2}\bar{g}^{\mu\sigma}\bar{g}_{\alpha\beta}\udn{\sigma}\bar{g}^{\alpha\beta}
\label{e:X}
\intertext{and}
  \bar{Y}^\mu& 
=-2(\bar{g}^{\mu\nu}-\bar{h}^{\mu\nu})\bar{\nabla}_\nu\Psi
=\frac{2}{t}\left( \bar{g}^{\mu 0}+\frac{\Lambda}{3}\delta^\mu_0\right). \label{e:Y}
\end{align}

\subsubsection{Field variables\label{vardefs}}
The gravitational and matter field variables $\{\bar{g}^{\mu\nu}(\bar{x}), \bar{\rho}(\bar{x}), \bar{v}^\mu(\bar{x})\}$ in relativistic coordinates, as they stand, are not suitable
for establishing the global existence of solutions or taking the Newtonian limit $\epsilon \searrow 0$. 
In order to obtain suitable variables, we switch to Newtonian coordinates $(x^\mu)$, $t=x^0$, and employ the following
field variables, which are closely related to the ones used in \cite{Liu2017}:
\begin{align}
  u^{0\mu}&=\frac{1}{\epsilon}\frac{\underline{\bar{g}^{0\mu}}-\bar{h}^{0\mu}}{2t}, \label{E:u.a} \\
  u^{0\mu}_0&=\frac{1}{\epsilon}\left(\delta^0_\nu \underline{\udn{0}\bar{g}^{\mu\nu}}-\frac{3(\underline{\bar{g}^{0\mu}}
  -\bar{h}^{0\mu})}{2t}\right) \label{E:u.b},\\
  u^{0\mu}_i &=\frac{1}{\epsilon}\underline{\delta^0_\nu\udn{i}\bar{g}^{\mu\nu}}, \label{E:u.c} \\
  u^{ij}  &=\frac{1}{\epsilon}\bigl( \underline{\bar{\mathfrak{g}}^{ij}}- \bar{h}^{ij}\bigr), \label{E:u.d} \\
  u^{ij}_\mu & 
  =\frac{1}{\epsilon}\delta^i_\sigma\delta^j_\nu \underline{\udn{\mu}(\alpha^{-1} \bar{g}^{\sigma\nu}-   \bar{h}^{\sigma\nu}) }
  , \label{E:u.e} \\
  u&=\frac{1}{\epsilon}\underline{\bar{\mathfrak{q}}}, \label{E:u.f} \\
  u_\mu&=\frac{1}{\epsilon}\biggl(\underline{\delta^0_\sigma\delta^0_\nu\udn{\mu}(\bar{g}^{\sigma\nu}-\bar{h}^{\sigma\nu})}-\frac{\Lambda}{3} \underline{\udn
  \mu \ln \alpha}  \biggr), \label{E:u.g} \\
 z_i&=\frac{1}{\epsilon}\underline{\bar{v}_i},\label{E:z.b}\\
 \zeta&=\frac{1}{1+\epsilon^2 K}\ln\bigl(t^{-3(1+\epsilon^2 K)}\underline{\bar{\rho}}\bigr), \label{E:ZETA}\\
 \intertext{and}
  \delta\zeta&=\zeta-\zeta_H \label{E:DELZETA}
\end{align}
where
\begin{gather}
\bar{\mathfrak{g}}^{ij}= \alpha^{-1}\bar{g}^{ij},\qquad \alpha:=(\det{\bar{g}^{kl}})^{\frac{1}{3}}/(\det{\bar{h}^{kl}})^{\frac{1}{3}}=
E^2 (\det{\check{g}_{ij}})^{-\frac{1}{3}}=E^2 (\det{\bar{g}^{kl}})^{\frac{1}{3}}, \qquad \check{g}_{ij}=(\bar{g}^{ij})^{-1}, \label{E:GAMMA}\\
\bar{\mathfrak{q}}
=\bar{g}^{00}-\bar{h}^{00}-\frac{\Lambda}{3}\ln{\alpha} \label{E:q}, \\
  \zeta_H(t)=\frac{1}{1+\epsilon^2 K}\ln\bigl(t^{-3(1+\epsilon^2 K)}\mu(t)\bigr)
 \label{E:ZETAH1}
\end{gather}
and we are freely using the notation \eqref{Neval}.
As we show below in \S\ref{FLRWanal}, $\zeta_H$ is given by the explicit formula
\begin{equation} \label{E:ZETAH3}
  \zeta_H(t)=\zeta_H(1)-\frac{2}{1+\epsilon^2K}\ln{\left(\frac{C_0-t^{3(1+\epsilon^2K)}}{C_0-1}\right)}
\end{equation}
where the constants $C_0$ and $\zeta_H(1)$ are defined by
\begin{equation} \label{C0def}
 C_0=\frac{\sqrt{\Lambda+\mu(1)}+\sqrt{\Lambda}}{\sqrt{\Lambda+\mu(1)}-\sqrt{\Lambda}}>1
\end{equation}
and
\begin{equation*}\label{zetaH1}
\zeta_H(1)=\frac{1}{1+\epsilon^2K}\ln{ \mu(1)},
\end{equation*}
respectively. Letting
\begin{equation} \label{zetaHringdefa}
\mathring{\zeta}_H = \lim_{\epsilon\searrow 0} \zeta_H
\end{equation}
denote the Newtonian limit of $\zeta_H$, it is clear from  \eqref{E:ZETAH3} that
\begin{equation} \label{zetaHringform}
\mathring{\zeta}_H (t) = \ln{ \mu(1)}- 2\ln{\left(\frac{C_0-t^{3}}{C_0-1}\right)}.
\end{equation}
For later use, we also define
\begin{align}
 z^i &= \frac{1}{\epsilon}\underline{\bar{v}^i}.\label{E:z.a}
\end{align}

\begin{remark}
It is important to emphasize that in the above variables we are treating components of
the geometric quantities with respect to the relativistic coordinates as scalars when transforming to Newtonian coordinates.
This procedure is necessary in order to obtain variables that have a well defined Newtonian limit.
We further emphasize that, for any fixed $\epsilon >0$, the gravitational and matter fields $\{\bar{g}^{\mu\nu}(\bar{x}),\bar{v}^\mu(\bar{x}),\bar{\rho}(\bar{x})\}$ in relativistic coordinates are completely equivalent to
the fields $\{u^{0\mu}(x), u^{ij}(x),u(x),z_i(x),\zeta(x)\}$ defined in the Newtonian coordinates.
\end{remark}

\subsection{Conformal Poisson-Euler equations}\label{CPEequations}
The $\epsilon \searrow 0$ limit of the conformal Einstein-Euler equations  define the
\textit{conformal cosmological Poisson-Euler equations} and are given by
\begin{align}
    \partial_t \mathring{\rho}+\sqrt{\frac{3}{\Lambda}}\partial_j\left(\mathring{\rho}\mathring{z}^j\right)
    &=\frac{3(1-\mathring{\Omega})}{t}\mathring{\rho}, \label{E:COSEULERPOISSONEQ.a}\\
    \sqrt{\frac{\Lambda}{3}}\mathring{\rho}\partial_t\mathring{z}^j+K
    \frac{\delta^{ji}}{\mathring{E}^2} \partial_i\mathring{\rho}+\mathring{\rho}\mathring{z}^i\partial_i\mathring{z}^j
    &=\sqrt{\frac{\Lambda}{3}}\frac{1}{t}\mathring{\rho}\mathring{z}^j-\frac{1}{2}
    \frac{3}{\Lambda}\frac{t}{\mathring{E}}\mathring{\rho}\frac{\delta^{ji}}{\mathring{E}^2} \partial_i\mathring{\Phi}
    , \label{E:COSEULERPOISSONEQ.b}\\
    \Delta\mathring{\Phi}&=\frac{\Lambda}{3}\frac{\mathring{E}^3}{t^3} \delta \mathring
    {\rho} \label{E:COSEULERPOISSONEQ.c}
\end{align}
where $\Delta:=\delta^{ij}\partial_i\partial_j$ is the Euclidean Laplacian on $\Rbb^3$,
\begin{equation}\label{Eringform}
\mathring{E}(t) =
  \left(\frac{C_0-t^{3}}{C_0-1}\right)^{\frac{2}{3}},
\end{equation}
\begin{equation}\label{e:delrr}
	\delta\mathring{\rho} =\mathring{\rho}-\mathring{\mu},
\end{equation}
and
\begin{equation} \label{Oringdef}
\mathring{\Omega}(t) = \frac{2t^3}{t^3-C_0},
\end{equation}
with $C_0$ as defined above by \eqref{C0def} and $\mathring{\mu}$ defined by \eqref{FLRW.c} and \eqref{arhoringdef}.
It will be important for our analysis to introduce the modified density variable $\mathring{\zeta}$ defined by
\begin{equation*} \label{zetaringdef}
 \mathring{\zeta} = \ln(t^{-3}\mathring{\rho}),
\end{equation*}
which is the non-relativistic version of the variable $\zeta$ defined above by \eqref{E:ZETA}. A short calculation
shows that the conformal cosmological Poisson-Euler equations can be expressed in terms of this modified density as follows:
\begin{align}
    \partial_t \mathring{\zeta}+\sqrt{\frac{3}{\Lambda}}\bigl( \mathring{z}^j\partial_j \mathring{\zeta} + \partial_j\mathring{z}^j\bigr)
    &=-\frac{3\mathring{\Omega}}{t}, \label{CPeqn1}\\
    \sqrt{\frac{\Lambda}{3}}\partial_t\mathring{z}^j+ \mathring{z}^i\partial_i\mathring{z}^j+ K
    \frac{\delta^{ji}}{\mathring{E}^2} \partial_i \mathring{\zeta}
    &=\sqrt{\frac{\Lambda}{3}}\frac{1}{t}\mathring{z}^j-\frac{1}{2}
    \frac{3}{\Lambda}\frac{t}{\mathring{E}}\frac{\delta^{ji}}{\mathring{E}^2} \partial_i\mathring{\Phi}, \label{CPeqn2}\\
    \Delta\mathring{\Phi}&=\frac{\Lambda}{3} \mathring{E}^3( e^{\mathring{\zeta}}-e^{\mathring{\zeta}_H} ).  \label{CPeqn3}
\end{align}

\subsection{Initial Data}
Thus far, the set up for the Newtonian limit closely mirrors that from \cite{Liu2017} with the essential difference being that in this article, we are concerned with a fixed
spacetime of the form $M=[0,1)\times \Rbb^3$ as opposed to $\epsilon$-dependent spacetimes of the form\footnote{In \cite{Liu2017}, unlike the current article, the manifold changes
according to the coordinate system used. In relativistic coordinates, the manifold is $[0,1)\times \mathbb{T}^3_\epsilon$, where $\mathbb{T}^3_\epsilon$
is the 3-torus obtain from identifying the sides of the box $[0,\epsilon]^3$, while in Newtonian coordinates, the relevant manifold is $[0,1)\times\mathbb{T}^3_1$.}
$[0,1)\times \mathbb{T}^3_\epsilon$. 
The change in the spatial hypersurfaces from $\mathbb{T}^3_\epsilon$ to $\mathbb{R}^3$ is important because it will allow us to consider initial data that is physically relevant in the cosmological setting. To understand this improvement, we first recall that the $1$-parameter families of $\epsilon$-dependent solutions to the Einstein-Euler equations that were shown in \cite{Liu2017} to exist globally to the future on  $[0,1)\times \mathbb{T}^3_\epsilon$  and converge as $\epsilon\searrow 0$ to solutions of the cosmological Poisson-Euler equations were interpreted as the cosmological analogues of isolated systems. This interpretation, first discussed in \cite{Green2012}, comes from lifting these solutions to the covering
space where they become periodic solutions on $[0,1)\times \Rbb^3$ with period $\sim \epsilon$. Since the period determines the spatial size of
the universe, the matter in these solutions have a characteristic size $\sim \epsilon$ that shrinks to zero in the Newtonian limit, which is analogous to the behavior
of matter in an isolated system under the Newtonian limit. This leads to the conclusion that the solutions from   \cite{Liu2017} do not represent gravitating systems on spatial cosmological scales. Instead, they represent gravitating systems on spatial scales comparable to isolated systems. However, they do exist globally to the future which certainly includes time scales that are cosmologically relevant. We further note that the isolated system interpretation also applies to the local-in-time cosmological Newtonian limits from \cite{Oliynyk2009a,Oliynyk2009b}, which while only being shown to exist locally in (cosmological) time, do not require a small initial data assumption.

To obtain solutions that are relevant for cosmology, 
the initial data must be chosen
correctly. In particular, the inhomogeneous component of the fluid
density should be composed of localized fluctuations, each one of which behaves like an isolated Newtonian system. Furthermore,
 the fluctuations should be separated from one another by
light travel times that remain bounded away from zero
as $\epsilon \searrow 0$.  Thus, we need to specify, in relativistic coordinates, 1-parameter families
of $\epsilon$-dependent families of initial data that can be separated into homogeneous and inhomogeneous components
where the homogeneous component has a regular limit as $\epsilon \searrow 0$ while the inhomogeneous component consists of a finite number of
spikes with characteristic width $\sim \epsilon$ that can be centered at arbitrarily chosen, $\epsilon$-independent spatial points. 
Initial data of this type represents cosmological
initial data that deviates from homogeneity due to the presence of a finite number
of matter fluctuations that remain casually separated and behave as isolated systems in the limit $\epsilon \searrow 0$.

The starting point for constructing this type of initial data is to first select initial data for the matter, 
which we will specify on the initial hypersurface 
\begin{equation*}
\Sigma := \{1\} \times \Rbb^3 \cong \Rbb^3.
\end{equation*} 
In relativistic coordinates, we choose $1$-parameter families of initial data
for the proper energy density $\bar{\rho}$ and the spatial components $\bar{v}^I$ of the conformal $3$ velocity by setting
\begin{equation*}
\bar{\rho}(1,\mathbf{\bar{x}}) = \mu(1)+\delta\breve{\rho}_{\epsilon,\yve}\Bigl(\frac{\mathbf{\bar{x}}}{\epsilon}\Bigr)
\AND \bar{v}^j(1,\mathbf{\bar{x}}) =\epsilon \breve{z}{}^j_{\epsilon,\yve}\Bigl(\frac{\mathbf{\bar{x}}}{\epsilon}\Bigr),
\end{equation*}
where
\begin{align}
\delta\breve{\rho}_{\epsilon,\yve}(\mathbf{x}) &= \sum_{\lambda=1}^N  \delta\breve{\rho}_\lambda\Bigl(\mathbf{x}-\frac{\mathbf{y}_\lambda}{\epsilon}\Bigr),
\label{matidata1}\\
\breve{v}{}^j_{\epsilon,\yve}(\mathbf{x}) &= \sum_{\lambda=1}^N  \breve{z}{}^j_\lambda\Bigl(\mathbf{x}-\frac{\mathbf{y}_\lambda}{\epsilon}\Bigr),
\label{matidata2}
\end{align}
and the profile functions $\delta\breve{\rho}_\lambda$ and $\breve{z}{}^j_\lambda$ are elements of $L^{\frac{6}{5}}\cap K^s$ for some $s\in \mathbb{Z}_{\geq 3}$. It is
clear from these formulas that this initial data represents a perturbation of the FLRW initial data $(\bar{\rho}|_{\Sigma},\bar{v}^j|_{\Sigma})=(\mu(1),0)$ by $N$ fluctuation of width $\sim \epsilon$
that are centered at the fixed spatial points $\mathbf{y}_\lambda\in \Rbb^3$, $\lambda=1,2,\ldots,N$, on the initial hypersurface $\Sigma$. 
As is well known, initial data for the Einstein equations cannot be chosen freely due to presence of constraints that must be satisfied on the initial hypersurface. This has the effect that the description of the initial data for the gravitational field is much more complicated, and consequently, we defer further discussion of the initial data to \S \ref{S:INITIALIZATION}.

\subsection{Main Theorem}
With the set up complete, we are now able to state the main result of this article. The proof is given in \S \ref{S:MAINPROOF}.
\begin{theorem}\label{T:MAINTHEOREM}
	Suppose $s\in \mathbb{Z}_{\geq 3}$, $0<K\leq \frac{1}{3}$, $\Lambda >0$, $\mu(1)>0$, $r>0$, $\yve=(\yv_1,\cdots,\yv_N)\in \Rbb^{3N}$, $\Sigma  = \{1\}\times \Rbb^3$ is the initial hypersurface, and the free initial data $\{\smfu^{ij}_\epsilon,\smfu^{ij}_{0,\epsilon},\delta\breve{\rho}_\lambda,\breve{z}^i_\lambda\}$ is chosen on $\Sigma$ so that: $\smfu^{ij}_{\epsilon}\in R^{s+1}(\Rbb^3,\mathbb{S}_3)$, $\smfu^{ij}_{0,\epsilon}\in H^s(\Rbb^3,\mathbb{S}_3)$, $\delta\breve{\rho}_{\lambda}\in L^{\frac{6}{5}}\cap K^s(\Rbb^3,\Rbb)$ and $\breve{z}^j_{\lambda} \in L^{\frac{6}{5}}\cap K^s(\Rbb^3,\Rbb^3)$ for $\lambda=1,\cdots, N$, and  $\delta\breve{\rho}_{\epsilon,\yve}$ and $\breve{z}^j_{\epsilon,\yve}$ are as defined above by
\eqref{matidata1} and \eqref{matidata2},
respectively.
Then there exists a constant $r>0$ such that if the free initial data is chosen to satisfy
\begin{align*}
\|\breve{\xi}_\epsilon\|_s:=\|\smfu^{ij}_\epsilon\|_{R^{s+1}}+\|\smfu^{ij}_{0, \epsilon}\|_{H^s}+\|\delta\breve{\rho}_\lambda\|_{L^{\frac{6}{5}}\cap K^s} +\|\breve{z}^j_\lambda\|_{L^{\frac{6}{5}}\cap K^s} \leq r,
\end{align*}	
then there exists a constant  $\epsilon_0=\epsilon_0(r)>0$ and maps
$\breve{u}^{\mu\nu}_{\epsilon, \yve}: X^s_{\epsilon_0}(\mathbb{R}^3)\rightarrow R^{s+1}(\mathbb{R}^3,\mathbb{S}_{4}) $,
$\breve{u}_{\epsilon, \yve} : X^s_{\epsilon_0}(\mathbb{R}^3)\rightarrow R^{s+1}(\mathbb{R}^3) $,
$\breve{u}^{\mu\nu}_{0,\epsilon, \yve} : X^s_{\epsilon_0}(\mathbb{R}^3)\rightarrow R^{s}(\mathbb{R}^3,\mathbb{S}_{4}) $,   $\breve{u}_{0,\epsilon, \yve} : X^s_{\epsilon_0}(\mathbb{R}^3)\rightarrow R^{s}(\mathbb{R}^3) $,
$ \breve{z}_{i,\epsilon, \yve}: X^s_{\epsilon_0}(\mathbb{R}^3)\rightarrow R^{s}(\mathbb{R}^3,\mathbb{R}^3) $, and
$\delta\breve{\zeta}_{\epsilon, \yve} :(0,\epsilon_0)\times    \bigl(L^\frac{6}{5} \cap K^s(\mathbb{R}^3)\bigr)\rightarrow 
R^s(\mathbb{R}^3)$,
such that 
\begin{align*}
u^{0\mu}_{\epsilon, \yve}|_{\Sigma} =&\breve{u}_{\epsilon, \yve}^{0\mu}(\epsilon,\smfu^{kl}_\epsilon,\smfu^{kl}_{0,\epsilon}, \delta\breve{\rho}_{\epsilon, \yve}, \breve{z}^l_{\epsilon, \yve}) 
=\textrm{\em O}(\epsilon),
\\
u_{\epsilon, \yve}|_{\Sigma}=&\breve{u}_{\epsilon, \yve}(\epsilon,\smfu^{kl}_\epsilon,\smfu^{kl}_{0,\epsilon}, \delta\breve{\rho}_{\epsilon, \yve}, \breve{z}^l_{\epsilon, \yve})= \epsilon^2\frac{2\Lambda}{9} \smfu^{ij}_\epsilon \delta_{ij}+\textrm{\em O}(\epsilon^3),  
\\
u^{ij}_{\epsilon, \yve}|_{\Sigma}=&\breve{u}_{\epsilon, \yve}^{ij}(\epsilon,\smfu^{kl}_\epsilon,\smfu^{kl}_{0,\epsilon}, \delta\breve{\rho}_{\epsilon, \yve}, \breve{z}^l_{\epsilon, \yve}) = \epsilon^2 \left(\smfu^{ij}_\epsilon
-\frac{1}{3}\smfu^{kl}_\epsilon\delta_{kl}\delta^{ij}\right)+\textrm{\em O}(\epsilon^3), 
\\
z_{j,\epsilon, \yve}|_{\Sigma}=& \breve{z}_{j,\epsilon, \yve}(\epsilon,\smfu^{kl}_\epsilon,\smfu^{kl}_{0,\epsilon}, \delta\breve{\rho}_{\epsilon, \yve}, \breve{z}^l_{\epsilon, \yve})=  \delta_{kl}\breve{z}^k_{\epsilon, \yve} +\textrm{\em O}(\epsilon),
\\
\delta\zeta_{\epsilon, \yve}|_{\Sigma}=&\delta\breve{\zeta}_{\epsilon, \yve}(\epsilon, \delta\breve{\rho}_{\epsilon, \yve} )=\frac{1}{1+\epsilon^2 K}\ln{\left(1+\frac{\delta\breve{\rho}_{\epsilon, \yve}}{\mu(1)} \right)}, 
\\
u^{\mu\nu}_{0,\epsilon,\yve}|_{\Sigma} =&\breve{u}^{\mu\nu}_{0,\epsilon,\yve}(\epsilon,\smfu^{kl}_\epsilon,\smfu^{kl}_{0,\epsilon}, \delta\breve{\rho}_{\epsilon, \yve}, \breve{z}^l_{\epsilon, \yve})= \textrm{\em O}(\epsilon), 
\\
\intertext{and}
u_{0,\epsilon,\yve}|_{\Sigma} =&\breve{u}_{0,\epsilon,\yve}(\epsilon,\smfu^{kl}_\epsilon,\smfu^{kl}_{0,\epsilon}, \delta\breve{\rho}_{\epsilon, \yve}, \breve{z}^l_{\epsilon, \yve})= \textrm{\em O}(\epsilon) 
\end{align*} 
determine, via the formulas \eqref{E:u.a}, \eqref{E:u.b}, \eqref{E:u.d}, \eqref{E:u.f}, \eqref{E:z.b}, \eqref{E:ZETA}, and
	\eqref{E:DELZETA}, a
	solution of the gravitational and gauge constraint equations, see \eqref{E:CONSTRAINT}-\eqref{E:WAVECONSTRAINT}.
	Furthermore, there exists a constant $\sigma \in (0,r]$, such that if the free initial data is chosen to satisfy
	\begin{align*}
	\|\smfu^{ij}_\epsilon\|_{R^{s+1}}+\|\smfu^{ij}_{0, \epsilon}\|_{H^s} +\|\delta\breve{\rho}_\lambda\|_{L^{\frac{6}{5}}\cap K^s} +\|\breve{z}^j_\lambda\|_{L^{\frac{6}{5}}\cap K^s}  \leq \sigma,
	\end{align*}
	then there exist maps
	\begin{align*}
	u^{\mu\nu}_{\epsilon, \yve} & \in C^0((0,1], R^{s}(\mathbb{R}^3,\mathbb{S}_4))\cap C^1((0,1], R^{s-1}(\mathbb{R}^3,\mathbb{S}_4)),\\
	u^{\mu\nu}_{\gamma,\epsilon,\yve} & \in C^0((0,1], R^{s}(\mathbb{R}^3,\mathbb{S}_4))\cap C^1((0,1], R^{s-1}(\mathbb{R}^3,\mathbb{S}_4)),\\
	u_{\epsilon,\yve} & \in C^0((0,1], R^{s}(\mathbb{R}^3))\cap C^1((0,1], R^{s-1}((\mathbb{R}^3)),\\
	u_{\gamma,\epsilon,\yve} & \in C^0((0,1], R^{s}(\mathbb{R}^3))\cap C^1((0,1], R^{s-1}((\mathbb{R}^3)),\\
	\delta\zeta_{\epsilon,\yve}& \in C^0((0,1], R^{s}(\mathbb{R}^3))\cap C^1((0,1], R^{s-1}(\mathbb{R}^3)),\\
	z_{i,\epsilon,\yve} & \in C^0((0,1], R^{s}(\mathbb{R}^3,\mathbb{R}^3))\cap C^1((0,1], R^{s-1}(\mathbb{R}^3,\mathbb{R}^3)),
	\end{align*}
	for $\epsilon \in (0,\epsilon_0)$, and
	\begin{align*}
	\mathring{\Phi}_{\epsilon,\yve} & \in C^0((0,1], R^{s+2}(\mathbb{R}^3))\cap C^1((0,1], R^{s+1}(\mathbb{R}^3)),\\
	\delta \mathring{\zeta}_{\epsilon,\yve} & \in C^0((0,1], H^{s}(\mathbb{R}^3))\cap C^1((0,1], H^{s-1}(\mathbb{R}^3)),\\
	\mathring{z}_{i,\epsilon,\yve} & \in C^0((0,1], H^{s}(\mathbb{R}^3,\mathbb{R}^3))\cap C^1((0,1], H^{s-1}(\mathbb{R}^3,\mathbb{R}^3)),\\
	\end{align*}
	such that
	\begin{enumerate}[(i)]
		\item \label{I} $\{u^{\mu\nu}_{\epsilon, \yve}(t, x), u_{\epsilon, \yve}(t, x), \delta\zeta_{\epsilon, \yve}(t, x), z_{i,\epsilon, \yve} (t, x)\}$ determines, via \eqref{E:CONFORMALTRANSF}, \eqref{E:CONFORMALVELOCITY}, \eqref{normal},  \eqref{E:u.a}, \eqref{E:u.d}, \eqref{E:u.f}, \eqref{E:z.b}
		and \eqref{E:ZETA}-\eqref{E:q}, a $1$-parameter family of solutions to the Einstein-Euler equations \eqref{E:ORIGINALEESYSTEM.a}-\eqref{E:ORIGINALEESYSTEM.b} in the wave gauge \eqref{E:WAVEGAUGE} that exists globally to the future on $M=(0,1]\times \mathbb{R}^3$,
		\item \label{II} $\{\mathring{\Phi}_{\epsilon,\yve}(t,x), \mathring{\zeta}_{\epsilon,\yve}(t,x):=\delta\mathring{\zeta}_{\epsilon,\yve}+\mathring{\zeta}_H, \mathring{z}^i_{\epsilon, \yve}(t,x):=
		\mathring{E}(t)^{-2}\delta^{ij}\mathring{z}_{j,\epsilon,\yve}(t,x)\}$, with $\mathring{\zeta}_H$ and $\mathring{E}$ given
		by \eqref{zetaHringform} and \eqref{Eringform}, respectively, solves the conformal cosmological Poisson-Euler equations \eqref{CPeqn1}-\eqref{CPeqn3} on  that exists globally to the future on $M$ and satisfy the initial conditions
		\begin{align*}
		\mathring{\zeta}_{\epsilon,\yve}|_{\Sigma}= \ln\biggl(\frac{4C_0\Lambda}{(C_0-1)^2}+ \delta\breve{\rho}_{\epsilon,\yve} \biggr)  \AND \mathring{z}^i_{\epsilon,\yve}|_{\Sigma}=\breve{z}^i_{\epsilon,\yve},
		\end{align*}
		\item  \label{III} the uniform bounds
		\begin{gather*}
		\|\delta\mathring{\zeta}_{\epsilon,\yve}\|_{L^\infty((0,1],H^{s})} +  \|\mathring{\Phi}_{\epsilon,\yve}\|_{L^\infty((0,1],H^{s+2})} +  \|\mathring{z}_{j,\epsilon,\yve}\|_{L^\infty((1,0]\times H^{s})}+
		\|\delta\zeta_{\epsilon,\yve}\|_{L^\infty((0,1],R^{s})} +  \|z_{j, \epsilon,\yve}\|_{L^\infty((0,1], R^{s})} \lesssim 1
		\intertext{and}
		\|u^{\mu\nu}_{\epsilon,\yve}\|_{L^\infty((1,0],R^{s})}+
		\|u^{\mu\nu}_{\gamma,\epsilon,\yve}\|_{L^\infty((0,1],R^{s})} +  \|u_{\epsilon,\yve}\|_{L^\infty((0,1],R^{s})} +\|u_{\gamma,\epsilon,\yve}\|_{L^\infty((0,1],R^{s})} \lesssim 1,
		\end{gather*}
		hold for $\epsilon \in (0, \epsilon_0)$,
		\item \label{IV} and the uniform error estimates
		\begin{gather*}
		\|\delta\zeta_{\epsilon, \yve}-\delta\mathring{\zeta}_{\epsilon,\yve}\|_{L^\infty((0,1],R^{s-1})} +
		\| z_{j,\epsilon, \yve}-\mathring{z}_{j,\epsilon,\yve}\|_{L^\infty((1,0]\times R^{s-1})}\lesssim \epsilon, \\
		\|u^{\mu\nu}_{0, \epsilon, \yve}\|_{L^\infty((1,0],R^{s-1})}+
		\|u^{\mu\nu}_{k,\epsilon, \yve}-\delta^\mu_0\delta^\nu_0\partial_k\mathring{\Phi}_{\epsilon, \yve}\|_{L^\infty((0,1],R^{s-1})}+
		\|u^{\mu\nu}_{\epsilon, \yve}\|_{L^\infty((0,1],R^{s-1})}\lesssim \epsilon
		\intertext{and}
		\|u_{\gamma,\epsilon, \yve}\|_{L^\infty((0,1],R^{s-1})}+
		\|u_{\epsilon, \yve}\|_{L^\infty((0,1],R^{s-1})}\lesssim \epsilon
		\end{gather*}
		hold for $\epsilon \in (0, \epsilon_0)$.
	\end{enumerate}
\end{theorem}

\subsection{Future directions}
While Theorem \ref{T:MAINTHEOREM} establishes the existence of a large class of inhomogeneous
cosmological solutions that are approximated on large cosmological scales by solutions of Newtonian gravity and gives a positive answer to the question
at the beginning of the introduction, many questions remain to be answered. For example, the small initial data assumption needed
to establish Theorem \ref{T:MAINTHEOREM} suggests the problem of understanding when to expect a similar result to hold without
a small initial data assumption. Due to phenomena such as black hole formation, such a result will not hold for all choices of initial data. However, it could hold
for carefully chosen large initial data. 

In a separate direction, there are relativistic effects that are important for precision cosmology that are not captured by the Newtonian solutions. To understand these effects would require that the Theorem \ref{T:MAINTHEOREM} is generalized to account for higher order post-Newtonian (PN) corrections starting  with the $1/2$-PN expansion, which is, by definition, the $\epsilon$ order correction
to the Newtonian gravity. Preliminary work in this direction is currently in preparation \cite{OliRob}. An interesting result of this work is that it characterizes
the subset of the $1$-parameter families of solutions from Theorem \ref{T:MAINTHEOREM} that can be interpreted on large scales as a linear perturbation of an
FLRW solution as those that admit a 1/2-PN expansion. Thus a generalization of Theorem \ref{T:MAINTHEOREM} to include the existence of 
$1$-parameter families of solutions to the Einstein-Euler equations that admit a 1/2-PN expansion would provide a mathematically rigorous resolution to the following perplexing question: \textit{How can the Universe be accurately modelled on small scales using Newtonian gravity,
yet, at the same time, be accurately modelled on large scales as a fully relativistic perturbation
of an FLRW spacetime?}

\subsection{Prior and related work} The future non-linear stability of the FLRW fluid solutions for a linear equation of state $p=K\rho$
was first established  under the condition $0<K<1/3$ and the assumption of zero fluid vorticity
by Rodnianski and Speck in \cite{RodnianskiSpeck:2013} using a generalization of a wave based method developed by Ringstr\"{o}m in \cite{Ringstroem2008}.
Subsequently, it has been shown \cite{Friedrich2017,Hadzic2015,Luebbe2013,Speck2012}
that this future non-linear stability result remains true for fluids with non-zero vorticity and
also for the equation of state parameter values $K=0$ and $K=1/3$, which correspond to dust and pure radiation, respectively.
It is worth noting that the stability results established in \cite{Luebbe2013} and \cite{Friedrich2017} for $K=1/3$ and
$K=0$, respectively,
rely on Friedrich's conformal method \cite{Friedrich1986,Friedrich1991},
which is completely different from the techniques used in \cite{Hadzic2015,RodnianskiSpeck:2013,Speck2012} for the parameter values $0\leq K<1/3$.

In the Newtonian setting, the global existence to the future of solutions to the cosmological Poisson-Euler equations was established in
\cite{Brauer1994} under a small initial data assumption and for a class of polytropic equations of state.

A new method was introduced in \cite{Oliynyk2016a} to prove the future non-linear stability of the FLRW fluid solutions that was based on
a wave formulation of a conformal version of the Einstein-Euler equations. The global existence results in this article are established using
this approach. We also note that this method was recently used to establish the non-linear stability of the FLRW fluid solutions that satisfy the generalized Chaplygin equation of state \cite{LeFloch2015a}.

\subsection{Overview}
In \S \ref{S:hypEE},  we employ the variables \eqref{E:u.a}-\eqref{E:DELZETA} and the wave gauge \eqref{E:WAVEGAUGE} to write the
conformal Einstein-Euler system, given by \eqref{E:EXPANSIONOFEIN} and \eqref{Confeul}, as a symmetric hyperbolic system that is
jointly singular in $\epsilon$ and $t$. By the results of \cite{Oliynyk2016a}, this system is suitable for obtaining the existence of solutions
that exist globally to the future; however, it not suitable for obtaining such solutions in the limit  $\epsilon \searrow 0$ due to the singular dependence
of the solutions on the parameter $\epsilon$. This defect is remedied in \S \ref{S:NloEE}, where we introduce a non-local transformation that
brings the system into a form, given by \eqref{E:REALEQ}, that is suitable for establishing  global existence with uniform control as $\epsilon \searrow 0$.

In \S  \ref{S:INITIALIZATION}, we use a fixed point method to construct $\epsilon$-dependent $1$-parameter families of initial data for the reduced conformal Einstein-Euler equations that satisfy the constraint equations on
the initial hypersurface $\Sigma$. The fixed point method is similar to one employed in \cite{Liu2017}. However, due to the non-compact nature of the initial hypersurface and the translation invariance of the the norms, the proof is technically more difficult and relies crucially on potential theory, in particular, the Riesz and Yukawa potential operators. 

In \S \ref{E:ulex}, we state and prove a local-in-time existence and uniqueness result  for solutions of the reduced
conformal Einstein-Euler equations along with a continuation principle. We establish this local-in-time result by first working in uniformly local Sobolev spaces where we can apply standard theorems. We then show that these results continue to hold in the functions spaces used to obtain global existence in \S \ref{S:MAINPROOF}. 
Similarly, in \S \ref{S:locPE},
we state and prove a local-in-time existence and uniqueness result and continuation principle for solutions of the conformal cosmological Poisson-Euler equations. 

We generalize, in \S \ref{S:MODEL}, the uniform a priori  and error estimates established in \cite{Liu2017} to hold for a closely related 
class of symmetric hyperbolic equations on $[T_0,T_1)\times \Rbb^3$ that are
jointly singular in $\epsilon$ and $t$. This class includes both the formulation \eqref{E:REALEQ} of the conformal Einstein-Euler equations and the $\epsilon \searrow 0$
limit of these equations.

Finally, in \S\ref{S:MAINPROOF}, we complete the proof of Theorem \ref{T:MAINTHEOREM} by using the results from \S \ref{S:INITIALIZATION} to \S \ref{S:MODEL} to verify that all the assumptions from \S \ref{S:MODEL} hold for the non-local formulation  \eqref{E:REALEQ} of the conformal Einstein-Euler equations. This allows us to apply Theorem \ref{T:MAINMODELTHEOREM} to get the desired conclusion.

\section{A singular symmetric hyperbolic formulation of the conformal Einstein--Euler system}  \label{S:hypEE}
In this section, we employ the variables \eqref{E:u.a}-\eqref{E:DELZETA} and the wave gauge \eqref{E:WAVEGAUGE} to transform the
conformal Einstein-Euler system, given by \eqref{E:EXPANSIONOFEIN} and \eqref{Confeul}, into the form of a symmetric hyperbolic system that has a particular singular dependence
on $\epsilon$ and $t$;  more specifically, the $\epsilon$-dependent singular terms are of a form that has been well-studied beginning with the pioneering work of Browning, Klainerman, Kreiss and Majda \cite{Browning1982,Klainerman1981,Klainerman1982,Kreiss1980}, while
the $t$-dependent singular terms are of
the type analyzed in \cite{Oliynyk2016a}. This type of system has been investigated on spacetime regions of the form $(0,1]\times \Tbb^3$ in our previous work \cite{Liu2017}. 
In this section, we derive a modified version of this system that is adapted to spacetime regions of the form
$(0,1]\times\Rbb^3$.

\subsection{Analysis of the FLRW solutions\label{FLRWanal}}
In \cite{Liu2017}, we derived some explicit formulas for the functions
$\Omega(t)$, $\mu(t)$ and $E(t)$ that will be needed again in this article.
For reader's convenience, we reproduce them here beginning with
\begin{align}\label{E:OMEGAREP}
\Omega(t)=\frac{2t^{3(1+\epsilon^2K)}}{t^{3(1+\epsilon^2K)}-C_0} \AND \mu(t)=\frac{4C_0\Lambda t^{3(1+\epsilon^2K)}}{(C_0-t^{3(1+\epsilon^2K)})^2},
\end{align}
where $C_0$ is as defined above by \eqref{C0def}.
From the formula for $\mu(t)$, it is then clear that the formula \eqref{E:ZETAH3} for  $\zeta_H(t)$ follows immediately
from the definition \eqref{E:ZETAH1}. Furthermore, it clear from the above formulas that  $\Omega$, $\mu$ and $\zeta_H$, as functions of $(t,\epsilon)$, lie
in
$C^2([0,1]\times [0,\epsilon_0])\cap W^{3,\infty}([0,1]\times [-\epsilon_0,\epsilon_0])$ for any fixed $\epsilon_0 > 0$, from which it follows
that we can represent $t^{-1}\Omega$ and $\del{t}\Omega$ as
\begin{equation*}\label{E:OMEGATIMEIDENTITY}
\frac{1}{t}\Omega= E^{-1}\partial_tE=t^{2+3\epsilon^2K}\mathscr{\hat{Q}}_1(t) \qquad \text{and} \qquad \partial_t\Omega=  t^{2+3\epsilon^2K}\mathscr{\hat{Q}}_2(t),
\end{equation*}
respectively, where here, we are employing the notation from \S \ref{remainder} for the remainder terms
$\mathscr{\hat{Q}}_1$ and $\mathscr{\hat{Q}}_2$.

Using \eqref{E:OMEGAREP},  we can integrate \eqref{E:PTE} to obtain
\begin{align}\label{E:EREP}
E(t)=  \exp{\left(\int_1^t \frac{2s^{2+3\epsilon^2K}}{s^{3(1+\epsilon^2K)}-C_0}ds\right)}=
\left(\frac{C_0-t^{3(1+\epsilon^2K)}}{C_0-1}\right)^{\frac{2}{3(1+\epsilon^2K)}} \geq 1
\end{align}
for $t\in [0,1]$. From this formula, it is clear that $E\in C^2([0,1]\times  [-\epsilon_0,\epsilon_0])\cap W^{3,\infty}([0,1]\times  [-\epsilon_0,\epsilon_0])$, and the Newtonian limit of $E$, denoted $\mathring{E}$ and defined by
\begin{equation*} \label{Eringdef}
\mathring{E}(t) = \lim_{\epsilon\searrow 0} E(t),
\end{equation*}
is given by the formula \eqref{Eringform}. Similarly, we denote the Newtonian limit of $\Omega$ by
\begin{equation*} \label{Oringlim}
\mathring{\Omega}(t) = \lim_{\epsilon\searrow 0} \Omega(t),
\end{equation*}
which we see from \eqref{E:OMEGAREP} is given by the formula \eqref{Oringdef}.
For latter use, we observe that $E$, $\Omega$ and $\zeta_H$ satisfy
\begin{align*}
\partial_t \zeta_H =-\frac{3}{t}\Omega=-3E^{-1}\partial_tE=-\bar{\gamma}^i_{i0}=-\bar{\gamma}^i_{0i}
=t^{2+3\epsilon^2K}\mathscr{\hat{Q}}_3(t)
\end{align*}
as can be verified by a straightforward calculation using the formulas \eqref{E:ZETAH3} and \eqref{E:OMEGAREP}-\eqref{E:EREP}. By \eqref{zetaHringform} and \eqref{Oringdef}, it is also easy to verify
\al{PTZETAH2}{
	\del{t}\mathring{\zeta}_H=-\frac{3}{t}\mathring{\Omega}=\frac{6t^2}{C_0-t^3}.
}
We record the following useful expansions of $t^{1+3\epsilon^2 K}$, $E(\epsilon, t)$ and $\Omega(\epsilon, t)$:
\al{CHI1}{
	t^{1+3\epsilon^2 K}=t+ \epsilon^2 \mathscr{X}(\epsilon,t), \quad  \mathscr{X}(\epsilon,t)=\frac{6K}{\epsilon^2}\int^\epsilon_0 \lambda t^{1+3\lambda^2 K}\ln t d\lambda,
}
\al{EOEXP}{
	E(\epsilon, t)=\mathring{E}(t)+\epsilon\mathscr{\hat{E}}(\epsilon, t) \AND \Omega(\epsilon, t)=\mathring{\Omega}(t)+\epsilon \mathscr{\hat{A}}(\epsilon, t)
}
for $(\epsilon,t)\in (0,\epsilon_0)\times (0,1]$, where $\mathscr{X}$, $\mathscr{\hat{E}}$ and $\mathscr{\hat{A}}$ are all remainder terms
as defined in \S \ref{remainder}.

\subsection{$\epsilon$-expansions and remainder terms}\label{epexpansions} In order to transform the reduced conformal Einstein-Euler equations into the desired form, we need to understand the lowest order $\epsilon$-expansion for a number of quantities.  We compute and collect together these
expansions in this section. Throughout this section, we work in Newtonian coordinates, and we frequently employ the notation
\eqref{Neval} for evaluation in Newtonian coordinates and the notation from
\S\ref{remainder} for remainder terms.

First, we observe, using \eqref{E:u.a}, \eqref{E:u.f} and \eqref{E:q}, that we can write $\alpha$
as
\begin{align*}
\underline{\alpha} = \exp{\biggl(\epsilon \frac{3}{\Lambda}\left(2tu^{00}-u\right)\biggr)}=1+\epsilon \frac{3}{\Lambda}\left(2tu^{00}-u\right)+\epsilon^2\mathscr{\hat{Z}}(\epsilon,t, u^{\mu\nu}, u),
\end{align*}
where $\mathscr{\hat{Z}}(\epsilon,t, 0, 0)=0$.
Using this, we can write the conformal metric as
\begin{equation}\label{E:GIJ}
\underline{ \bar{g}^{ij}}
= E^{-2}\delta^{ij}+\epsilon \Theta^{ij},
\end{equation}
where
\begin{equation*} \label{Thetaijdef}
\Theta^{ij} = \Theta^{ij}(\epsilon,t,u,u^{\mu\nu}):= \frac{1}{\epsilon}\left(\underline{\alpha}
-1\right)E^{-2}\delta^{ij} + \underline{\alpha}u^{ij},
\end{equation*}
and $\Theta^{ij}$ satisfies  $\Theta^{ij}(\epsilon,t,0,0)=0$
and the $E^1$-regularity properties of a remainder term, see \S\ref{remainder}.
From the definition of $u^{0\mu}$, see \eqref{E:u.a}, we have that
\begin{align}\label{E:G0MU}
\underline{\bar{g}^{0\mu}}=\bar{h}^{0\mu}+2\epsilon t u^{0\mu},
\end{align}
and by \eqref{E:u.b} and \eqref{E:u.c},
\begin{align}\label{E:PIG}
\delta^0_\nu \underline{\udn{0}\bar{g}^{\mu\nu}}=\epsilon(u^{0\mu}_0+3u^{0\mu})
\quad \text{and}\quad \delta^0_\nu \underline{\udn{i} \bar{g}^{\mu\nu}}
=\epsilon u^{0\mu}_i
\end{align}
for the derivatives. 
Additionally, with the help of \eqref{E:PTE}, \eqref{E:u.a}-\eqref{E:u.c}
and \eqref{E:u.f}-\eqref{E:u.g},  we have also that
\begin{align}\label{e:alpha}
	\underline{\udn{\beta}\alpha}
	=\epsilon \underline{\alpha} \frac{3}{\Lambda}\bigl(
	3u^{00}\delta^0_\beta +u^{00}_\beta -u_\beta \bigr).
\end{align}
Then differentiating \eqref{E:GIJ}, we find, using the above expression and \eqref{E:u.d}-\eqref{E:u.e}, that
\begin{align}\label{E:PGIJ}
	\delta^i_\mu \delta^j_\nu \underline{\udn{\sigma}\bar{g}^{\mu \nu}}
	= & \epsilon \underline{\alpha} u^{ij}_\sigma+\epsilon \frac{3}{\Lambda}\underline{\alpha}(\bar{h}^{ij}+\epsilon u^{ij})\bigl( 3u^{00}\delta^0_\sigma+u^{00}_\sigma-u_\sigma\bigr).
\end{align}

Since $\check{g}_{ij}$ is, by definition, the inverse of $\bar{g}^{ij}$,  it follows from \eqref{E:GIJ} and Lemma \ref{t:expinv}
that we can express $\check{g}_{ij}$ as
\begin{align}\label{E:CHECKGIJ}
\underline{\check{g}_{ij}}
=E^2\delta_{ij}+\epsilon\mathscr{\hat{Z}}_{ij}(\epsilon,t,u,u^{\mu\nu}),
\end{align}
where $\mathscr{\hat{Z}}_{ij}(\epsilon,t,0,0)=0$. 
From \eqref{E:GIJ}, \eqref{E:G0MU} and  Lemma \ref{t:expinv}, we then see that
\begin{align}\label{E:G_MUNU}
\underline{\bar{g}_{\mu\nu}}=\bar{h}_{\mu\nu}+\epsilon \Xi_{\mu\nu}(\epsilon,t, u^{\sigma\gamma}, u),
\end{align}
where $\Xi_{\mu\nu}$ satisfies $\Xi_{\mu\nu}(\epsilon,t, 0, 0)=0$ and the $E^1$-regularity properties of a remainder term.
Due to the identity
$\udn{\lambda}\bar{g}_{\mu\nu}=-\bar{g}_{\mu\sigma}\udn{\lambda}
\bar{g}^{\sigma\gamma}
\bar{g}_{\gamma\nu}$,
we can easily derive from \eqref{E:PIG}, \eqref{E:PGIJ} and \eqref{E:G_MUNU} that
\begin{align*}
\underline{\udn{\sigma} \bar{g}_{\mu\nu}}=
\epsilon \mathscr{\hat{Z}}_{\mu\nu\sigma}(\epsilon,t, \mathbf{u}),
\end{align*}
where
\begin{equation*}
\mathbf{u} = (u^{\alpha\beta},u,u^{\alpha\beta}_\sigma, u_\sigma),
\end{equation*}
and
$\mathscr{\hat{Z}}_{\mu\nu\sigma}(\epsilon,t, 0)=0$, 
which in turn, implies that
\begin{align}\label{E:CHRISTOFFEL}
\underline{X^\sigma_{\mu\nu}}= \underline{\bar{\Gamma}^\sigma_{\mu\nu}}-\bar{\gamma}^\sigma_{\mu\nu}&=\epsilon
\mathscr{\hat{I}}^\sigma_{\mu\nu}(\epsilon,t,\mathbf{u}),
\end{align}
where $\mathscr{\hat{I}}^\sigma_{\mu\nu}( \epsilon,t,0)=0$. 
Later, we will also need the explicit form of the next order term in the $\epsilon$-expansion for $\bar{\Gamma}^i_{00}$.
By a straightforward calculation, it is then not difficult to verify that
\begin{align} \label{E:GAMMAI00}
\underline{\bar{\Gamma}^i_{00}}=&\epsilon \frac{3}{\Lambda}(u^{0i}_0+3u^{0i})-\epsilon\frac{1}{2}\left(\frac{3}{\Lambda}\right)^2E^{-2}\delta^{ik}u^{00}_k+\epsilon^2
\mathscr{\hat{I}}^i_{00}(\epsilon,t, \mathbf{u}),  \\
\underline{\bar{\Gamma}^i_{k0}}-\underline{\bar{\gamma}^i_{k0}}=
& -\frac{3}{2\Lambda}\epsilon(\delta^{ij}\delta_{kl} u^{0l}_j-u^{0i}_k)-\epsilon \frac{1}{2} E^2\delta_{kj}\bigl[ \bigl(u^{ij}_0+\frac{3}{\Lambda}E^{-2}(3u^{00} + u^{00}_0-u_0 )\delta^{ij}\bigr)\bigr]+\epsilon^2\mathscr{\hat{I}}^i_{i0}(\epsilon,t, \mathbf{u}), \label{E:GSUBSTG}
\end{align}
where $\mathscr{\hat{I}}^i_{00}(\epsilon,t,0)=0$ and $\mathscr{\hat{I}}^i_{i0}(\epsilon,t,0)=0$. 

Continuing on, we observe from \eqref{E:ZETA} that we can express the proper energy density
in terms of $\zeta$  by
\begin{align}\label{E:ZETA2}
\rho:= \underline{\bar{\rho}}= t^{3(1+\epsilon^2 K)}e^{(1+\epsilon^2 K) \zeta},
\end{align}
and correspondingly,  by \eqref{E:ZETAH1},
\begin{align}\label{E:ZETAH2}
\mu= t^{3(1+\epsilon^2 K)}e^{(1+\epsilon^2 K) \zeta_H }
\end{align}
for the FLRW proper energy density.
From \eqref{E:DELZETA}, \eqref{E:ZETA2} and \eqref{E:ZETAH2}, it is then clear that we can express the difference between $\rho$ and $\mu$ in terms of $\delta\zeta$ by
\begin{align}\label{E:DELRHO} \delta \rho:=\rho-\mu=t^{3(1+\epsilon^2K)}e^{(1+\epsilon^2K)\zeta_H}
\Bigl(e^{(1+\epsilon^2K)\delta\zeta}-1\Bigr).
\end{align}

Due to the normalization $\bar{v}^\mu\bar{v}_\mu=-1$, only three components of $\bar{v}_\mu$ are independent. Solving
$\bar{v}^\mu\bar{v}_\mu=-1$
for $\bar{v}_0$ in terms of the components $\bar{v}_{i}$, we obtain
\begin{align}\label{e:vup0}
\bar{v}_0=\frac{-\bar{g}^{0i}\bar{v}_i+\sqrt{(\bar{g}^{0i}\bar{v}_i)^2-\bar{g}^{00}(\bar{g}^{ij}\bar{v}_i\bar{v}_j+1)}}{\bar{g}^{00}},
\end{align}
which, in turn, using the definitions \eqref{E:u.a}, \eqref{E:u.d}, \eqref{E:u.f}, \eqref{E:z.b}, we can write as
\begin{equation} \label{E:V_0}
\underline{\bar{v}_0}=-\frac{1}{\sqrt{-\underline{\bar{g}^{00}}}}+ \epsilon^2 \mathscr{\hat{V}}_2(\epsilon,t, u, u^{\mu\nu}, z _j),
\end{equation}
where  $\mathscr{\hat{V}}_2(\epsilon,t, u, u^{\mu\nu},0)=0$.
From this and the definition
$\bar{v}^0 =  \bar{g}^{0\mu}\bar{v}_\mu$,  we get
\begin{equation} \label{E:V^0}
\underline{\bar{v}^0}=\sqrt{-\underline{\bar{g}^{00}}}+  \epsilon^2
\mathscr{\hat{W}}_2(\epsilon,t, u, u^{\mu\nu}, z _j),
\end{equation}
where  $\mathscr{\hat{W}}_2(\epsilon,t, u, u^{\mu\nu},0)=0$.
We also observe that
\begin{align}\label{E:VELOCITY}
\underline{\bar{v}^k}=\epsilon (2tu^{0k}\underline{\bar{v}_0}+
\underline{\bar{g}^{ik}} z _i)\quad \text{and} \quad z ^k=2tu^{0k}\underline{\bar{v}_0}+\underline{\bar{g}^{ik}} z _i
\end{align}
follow immediately from the definitions \eqref{E:z.b} and \eqref{E:z.a}.
For later use,
note that $z^k$ can also be written in terms of $(\underline{\bar{g}^{\mu\nu}}, z_j)$ as
\begin{equation}\label{E:Z_IANDZ^I}
z ^i=\underline{\bar{g}^{ij} }z _j+\frac{\underline{\bar{g}^{i0}}}{\underline{\bar{g}^{00}}}
\left[-\underline{\bar{g}^{0j}} z _j
+\frac{1}{\epsilon}\sqrt{-\underline{\bar{g}^{00}}}\sqrt{1-\frac{1}{\underline{\bar{g}^{00}}}\epsilon^2(\underline{\bar{g}^{0j}}
	z _j)^2+\epsilon^2\underline{\bar{g}^{jk}} z _j z _k}\right].
\end{equation}

\subsection{The reduced conformal Einstein equations}
The next step in transforming the conformal Einstein-Euler system is
to replace the conformal Einstein equations  \eqref{E:EXPANSIONOFEIN} with the gauge
reduced version given by
\begin{align}
-\bar{g}^{\alpha\beta}\udn{\alpha}\udn{\beta}&\bar{g}^{\mu\nu}-2\bar{\nabla}^{(\mu}\bar{X}^{\nu)}-2\mathcal{\bar{R}}^{\mu\nu}-2\bar{P}^{\mu\nu}-2\bar{Q}^{\mu\nu} +2\bar{\nabla}^{(\mu}\bar{Z}^{\nu)}+\bar{A}_\sigma^{\mu\nu}\bar{Z}^\sigma
   =  -4\bar{\nabla}^\mu
\bar{\nabla}^\nu\Psi \nnb \\
&+4\bar{\nabla}^\mu\Psi\bar{\nabla}^\nu\Psi-2\biggl[\bar{\Box}\Psi
 +2|\bar{\nabla}\Psi|^2
+\biggl(\frac{1-\epsilon^2K}{2}\bar{\rho}+\Lambda\biggr)e^{2\Psi}\biggr]\bar{g}
^{\mu\nu}-2e^{2\Psi}(1+\epsilon^2K)\bar{\rho} \bar{v}^\mu \bar{v}^\nu,   \label{E:REDUCEDEINSTEIN}
\end{align}
where
\begin{align*}
\bar{A}_\sigma^{\mu\nu}:=-\bar{X}^{(\mu}\delta^{\nu)}_\sigma+\bar{Y}^{(\mu}\delta^{\nu)}_\sigma.
\end{align*}
We will refer to these equations as the \textit{reduced conformal Einstein equations}.
From \eqref{E:P} and \eqref{E:Q}, it is not difficult to verify, using \eqref{E:u.a}-\eqref{E:u.g}, that $\bar{P}^{\mu\nu}$ and $\bar{Q}^{\mu\nu}$, when expressed
in Newtonian coordinates, can be expanded as  
\begin{align*}
\underline{\bar{P}^{\mu\nu}} = & \epsilon \mathscr{L}^{\mu\nu} (\epsilon, t, u^{\alpha\beta},u)+\epsilon^2\mathscr{K}^{\mu\nu} (\epsilon,t,u^{\alpha\beta},u), \\
\underline{\bar{Q}^{\mu\nu} }= & \epsilon^2\mathscr{W}^{\mu\nu} (\epsilon,t,\mathbf{u}),
\end{align*}
where
$\mathscr{K}^{\mu\nu}$ is quadratic in $(u^{\alpha\beta},u)$, $\mathscr{W}^{\mu\nu}$ vanishes to second order in $\mathbf{u}$, and $\mathscr{L}^{\mu\nu}$ is linear in $(u^{\alpha\beta},u)$.

\begin{remark}
	In the above formula and for the rest of this section, the remainder terms satisfy the following properties:
$\mathscr{W}^{\mu\nu}(\epsilon,t,\mathbf{u})$ vanishes to second order in $\mathbf{u}$,  $\mathscr{L} (\epsilon, t, \mathbf{u})$, $\mathscr{L}^\mu(\epsilon,t, \mathbf{u})$, $\mathscr{L}^\mu_l(\epsilon,t, \mathbf{u})$, $\mathscr{L}^{\mu\nu\lambda}(\epsilon,t, \mathbf{u})$ and $\mathscr{L}^{\mu\nu} (\epsilon, t, \mathbf{u})$ are linear in $\mathbf{u}$, while
	$\mathscr{J} (\epsilon,t,\mathbf{u})$, $\mathscr{J}^{\mu } (\epsilon,t,\mathbf{u})$, $\mathscr{J}^{\mu\nu} (\epsilon,t,\mathbf{u})$, and $\mathscr{J}^{\mu\nu\lambda} (\epsilon,t,\mathbf{u})$  vanish for $\mathbf{u}=0$.
\end{remark}
Using \eqref{e:homconfein} and \eqref{E:WAVEGAUGE},  we observe that the reduced conformal
Einstein equations \eqref{E:REDUCEDEINSTEIN} can be written as
\begin{align}
-\bar{g}^{\alpha\beta}\udn{\alpha}\udn{\beta}&\bar{g}^{\mu\nu}+2\bar{\nabla}^{(\mu}\bar{Y}^{\nu)} -2\bar{P}^{\mu\nu}-
2\bar{Q}^{\mu\nu} +\bar{Y}^\mu \bar{Y}^\nu-\bar{X}^\mu \bar{X}^\nu
=  -4(\bar{\nabla}^\mu
\bar{\nabla}^\nu\Psi-\udn{}^\mu\udn{}^\nu \Psi) \nnb \\
&+4(\bar{\nabla}^\mu\Psi\bar{\nabla}^\nu\Psi- \udn{}^\mu\Psi\udn{}^\nu\Psi)-2e^{2\Psi}(1+\epsilon^2K)\left(\bar{\rho} \bar{v}^\mu \bar{v}^\nu-  \mu \frac{\Lambda}{3}\delta^\mu_0\delta^\nu_0 \right) \nnb \\
&-2\biggl[\bar{\Box}\Psi
+2|\bar{\nabla}\Psi|^2
+\biggl(\frac{1-\epsilon^2K}{2}\bar{\rho}+\Lambda\biggr)e^{2\Psi}\biggr]\bar{g}
^{\mu\nu} +2\bigl[\bb\Psi+2|\udn{}\Psi|^2_{\bar{h}}+\bigl(\frac{1-\epsilon^2K}{2} \mu +\Lambda\bigr)e^{2\Psi}\bigr]\bar{h}^{\mu\nu}.
\label{E:REDUCEDEINSTEIN2} 
\end{align}
In the following proposition, we list various formulas for crucial terms in  \eqref{E:REDUCEDEINSTEIN2}. These formulas can be established via direct computation; we omit the details.
\begin{proposition} \label{wgprop}
	If the wave gauge \eqref{E:WAVEGAUGE} is satisfied, and $\Psi$ and $\bar{\gamma}^\nu$ are as given by \eqref{E:CONFORMALFACTOR} and \eqref{E:HOMCHRIS}, respectively, then the following relations hold:
	\begin{gather*}
	\bar{\nabla}^\mu\Psi= -
	\bar{g}^{\mu 0}\frac{1}{t},\qquad \udn{}^\mu\Psi= -
	\bar{h}^{\mu 0}\frac{1}{t}=
	\frac{\Lambda}{3t}\delta^\mu_0 ,\qquad \bar{\nabla}_\mu\Psi=\udn{\mu} \Psi= -
	\delta^0_\mu\frac{1}{t}, \qquad
	\bar{\Box}\Psi=  \frac{1}{t^2}\bar{g}^{00}-\frac{1}{t} \bar{Y}^0+ \frac{1}{t} \bar{\gamma}^0, \\
	\bb\Psi=  \frac{1}{t^2}\bar{h}^{00} + \frac{1}{t} \bar{\gamma}^0=  -\frac{\Lambda}{3t^2} + \frac{1}{t} \bar{\gamma}^0, \qquad
	|\bar{\nabla}\Psi|_{\bar{g}}^2 =  \frac{1}{t^2}\bar{g}^{00},\qquad
	|\udn{}\Psi|_{\bar{h}}^2 =  \frac{1}{t^2}\bar{h}^{00} =  -\frac{\Lambda}{3t^2} ,\\
	\bar{Y}^\mu \bar{Y}^\nu= 4\bar{\nabla}^\mu \Psi \bar{\nabla}^\nu \Psi +\frac{8\Lambda}{3t^2} \delta^{(\mu}_0\bar{g}^{\nu)0} + \frac{4\Lambda^2}{9t^2} \delta^\mu_0 \delta^\nu_0 , \qquad
	4\udn{}^\mu\udn{}^\nu \Psi
	=  \frac{4\Lambda^2}{9t^2}\delta^ \mu_0 \delta^\nu_0+\frac{4\Lambda}{3t^2}\Omega \bar{h}^{ij} \delta^\mu_j \delta^\nu_i \\
	\intertext{and}
	\bar{\nabla}^{(\mu}\bar{Y}^{\nu)}=-2\bar{\nabla}^\mu \bar{\nabla}^\nu \Psi -\frac{2\Lambda}{3t^2} \bar{g}^{0(\mu}\delta^{\nu)}_0 -\frac{ \Lambda}{3t} \delta^\sigma_0 \udn{\sigma}\bar{g}^{\mu\nu}+\frac{2 \Lambda}{3t^2} \Omega \bar{g}^{i(\mu}\delta^{\nu)}_i.
	\end{gather*}
\end{proposition}

Using Proposition \ref{wgprop}, we see that \eqref{E:REDUCEDEINSTEIN2} can be written as
\begin{align}\label{E:REDUCEDCONFEINSTEINEQ1}
&-\bar{g}^{\alpha\beta}\bar{ \partial}_{\alpha}\udn{\beta}(\bar{g}^{\mu\nu} -\bar{h}^{\mu\nu} ) +
\bar{\mathcal{E}}^{\mu\nu} -2\bar{P}^{\mu\nu} -\bar{\mathcal{Q}}^{\mu\nu}
=    \frac{2 \Lambda}{3t} \delta^\sigma_0 \udn{\sigma} (\bar{g}^{\mu\nu} -\bar{h}^{\mu\nu} ) - \frac{4\Lambda}{3t^2}\left(\bar{g}^{0\lambda}+\frac{\Lambda}{3}\delta^\lambda_0\right)\delta^{(\mu}_\lambda\delta^{\nu)}_0  \nnb \\ &\hspace{1.5cm} -\frac{2}{t^2}\biggl(\bar{g}^{00}+\frac{\Lambda}{3} \biggr)\bar{g}
^{\mu\nu} - \frac{2}{t^2}(1+\epsilon^2K)\left((\bar{\rho}- \mu ) \bar{v}^\mu \bar{v}^\nu + \mu  \bigl( \bar{v}^\mu \bar{v}^\nu-  \frac{\Lambda}{3}\delta^\mu_0\delta^\nu_0\bigr) \right)   - \frac{2\Omega}{t^2}\Lambda (\bar{g}
^{\mu\nu}   -    \bar{h}^{\mu\nu})
  \nnb \\ & \hspace{2.5cm}
-\frac{1-\epsilon^2K}{ t^2}(\bar{\rho}- \mu )\bar{g}^{\mu\nu}-\frac{1-\epsilon^2K}{ t^2} \mu (\bar{g}^{\mu\nu}-\bar{h}^{\mu\nu})-\frac{4 \Lambda}{3t^2} \Omega \left[ (\bar{g}^{ij}-\bar{h}^{ij})\delta^{(\mu}_j\delta^{\nu)}_i+\bar{g}^{i0}\delta^{(\mu}_0\delta^{\nu)}_i\right]
\end{align}
where 
\begin{align*}
	\underline{\bar{\mathcal{Q}}^{\mu\nu}} :=\underline{2\bar{Q}^{\mu\nu}+\bar{X}^\mu \bar{X}^\nu}
= \epsilon^2\mathscr{W}^{\mu\nu} (\epsilon,t,\mathbf{u})
\end{align*}
and
\begin{align*}
	\bar{\mathcal{E}}^{\mu\nu}= \bar{g}^{\alpha\beta}\bar{ \gamma}^\lambda_{\alpha\beta}\udn{\lambda}(\bar{g}^{\mu\nu} -\bar{h}^{\mu\nu} )-\bar{g}^{\alpha\beta}\bar{ \gamma}^\mu_{\alpha\lambda}\udn{\beta}(\bar{g}^{\lambda\nu} -\bar{h}^{\lambda\nu} )-\bar{g}^{\alpha\beta}\bar{ \gamma}^\nu_{\alpha\lambda}\udn{\beta}(\bar{g}^{\mu\lambda} -\bar{h}^{\mu\lambda} ).
\end{align*}
Contracting both sides of  \eqref{E:REDUCEDCONFEINSTEINEQ1} with $\delta^0_\nu$, we find that
\begin{align}
&-\bar{g}^{\alpha\beta}\bar{ \partial}_{\alpha}\delta^0_\nu \udn{\beta}(\bar{g}^{\mu\nu} -\bar{h}^{\mu\nu} )+\bar{\mathcal{E}}^{\mu 0}  -2\bar{P}^{\mu0}-\bar{\mathcal{Q}}^{\mu 0}
=  \frac{2 \Lambda}{3t} \delta^\sigma_0 \delta^0_\nu \udn{\sigma} (\bar{g}^{\mu\nu} -\bar{h}^{\mu\nu} ) - \frac{4\Lambda}{3t^2}\left(\bar{g}^{00}+\frac{\Lambda}{3} \right)\delta^{ \mu}_0 \nnb \\ &
\hspace{1.5cm} - \frac{4\Lambda}{3t^2} \bar{g}^{0k} \delta^{(\mu}_k\delta^{0)}_0  -\frac{2}{t^2}\biggl(\bar{g}^{00}+\frac{\Lambda}{3} \biggr)\bar{g}
^{\mu0}  - \frac{2}{t^2}(1+\epsilon^2K)\left((\bar{\rho}- \mu ) \bar{v}^\mu \bar{v}^0 + \mu  \bigl( \bar{v}^\mu \bar{v}^0-  \frac{\Lambda}{3}\delta^\mu_0 \bigr) \right) \nnb \\ & \hspace{2.5cm}  - \frac{2\Omega}{t^2}\Lambda (\bar{g}
^{\mu0}   -    \bar{h}^{\mu0})
-\frac{1-\epsilon^2K}{ t^2}(\bar{\rho}- \mu )\bar{g}^{\mu 0}-\frac{1-\epsilon^2K}{ t^2} \mu (\bar{g}^{\mu0}-\bar{h}^{\mu0}) -\frac{2 \Lambda}{3t^2} \Omega   \bar{g}^{i0}\delta^{ \mu}_i  \label{e:ein1.a}
\end{align}
where
\begin{align*}
\underline{\bar{\mathcal{E}}^{\mu 0}}:=\delta^0_\nu\underline{\bar{\mathcal{E}}^{\mu \nu}}
=& -2 \epsilon E^{-2}\frac{\Omega}{t}\delta^{km}  \delta^\mu_k  u^{00}_m+\epsilon \mathscr{L}^{\mu 0} (\epsilon, t, \mathbf{u})+\epsilon^2\mathscr{J}^{\mu 0} (\epsilon,t,\mathbf{u}).
\end{align*}
To proceed, we need to express \eqref{e:ein1.a} in terms of evolution variables \eqref{E:u.a}-\eqref{E:DELZETA}. In order to achieve this, we first derive, using\eqref{E:u.a}, \eqref{E:u.b}, \eqref{E:u.c} and \eqref{E:PIG}, the identities
\begin{align}
u^{0\mu}_0= & \frac{1}{\epsilon}(\delta_\nu^0\underline{\udn{0}\bar{g}^{\mu\nu}}-3\epsilon u^{0\mu})
=  -u^{0\mu}+2t  \partial _0 u^{0\mu}+ 2 \Omega\delta^\mu_{k} u^{0 k}, \label{e:d0u0} \\
u^{0\mu}_i= & \frac{1}{\epsilon}\delta^\nu_0\underline{\udn{i}\bar{g}^{\mu\nu}}
= 2 t\frac{1}{\epsilon}\partial _i u^{0\mu} + \frac{4\Lambda}{3 } E^2\Omega\delta_{ik}   u^{k0}\delta^\mu_0 + \frac{\Lambda}{3t} E^2\Omega\delta_{ik}\delta^\mu_l \Theta^{kl}   + 2 \Omega\delta^\mu_i  u^{00}.   \label{e:diuo}
\end{align}
From these identities, we get
\begin{align}
-3\epsilon \underline{\bar{g}^{00}}  \partial _{0} u^{0\mu} -12\epsilon  t u^{0k} \partial _{k} u^{0\mu} 
= &  - \epsilon \frac{3}{2t} \underline{\bar{g}^{00}} (u^{0\mu}_0+u^{0\mu} )+ \epsilon \mathscr{L}^{0\mu} (\epsilon, t, \mathbf{u})+\epsilon^2\mathscr{J}^{0\mu} (\epsilon,t,\mathbf{u}).\label{e:com}
\end{align}
Using this, we can write \eqref{e:ein1.a} as
\begin{align}
&- \underline{\bar{g}^{00}}  \partial_{0} u^{0\mu}_0  - 4  t u^{0k}  \partial_{k}u^{0\mu}_0 - \frac{1}{\epsilon} \underline{\bar{g}^{kl}} \partial_{k} u^{0\mu}_l+  \mathscr{L}^{0\mu} (\epsilon, t, \mathbf{u}) +\epsilon \mathscr{J}^{0\mu} (\epsilon,t,\mathbf{u}) \nnb  \\
= & - \frac{1}{2t} \underline{\bar{g}^{00}} (u^{0\mu}_0+ u^{0\mu}) + 2  E^{-2}\frac{\Omega}{t}\delta^{km}  \delta^\mu_k  u^{00}_m + 4 \epsilon u^{00}  u^{0\mu}_0  - 4 \epsilon u^{00}  u^{0\mu} - \frac{4 \Lambda}{3t } \Omega   \delta^{ \mu}_i u^{i0} -2 \frac{1-\epsilon^2K}{ t } \mu u^{0\mu} \nnb  \\
&   - \frac{2}{t^2}\frac{1}{\epsilon}(1+\epsilon^2K)\left(( \rho - \mu ) \underline{\bar{v}^\mu} \underline{\bar{v}^0} + \mu \bigl( \underline{\bar{v}^\mu} \underline{\bar{v}^0} -  \frac{\Lambda}{3}\delta^\mu_0 \bigr) \right)  -   \frac{4\Omega}{t }\Lambda
u^{\mu 0}
-\frac{1-\epsilon^2K}{ t^2}\frac{1}{\epsilon} ( \rho - \mu )\underline{\bar{g}^{\mu 0} } ,  \label{e:Ein1.a}
\end{align}
while we see that
\begin{align}\label{e:Ein1.c}
-\underline{\bar{g}^{00}}\partial _0 u^{0\mu}  =-\frac{1}{2t}\underline{\bar{g}^{00}} u^{0\mu}_0-\frac{1}{2t}\underline{\bar{g}^{00}} u^{0\mu}+ \underline{\bar{g}^{00}} \frac{1}{t} \Omega\delta^\mu_{k} u^{0 k}
\end{align}
follows from  \eqref{e:d0u0}.
We see also that
\begin{align*}
	\del{l} u^{0\mu}_0=&\frac{1}{\epsilon}\biggl(\delta^0_\nu \del{l} \underline{\udn{0} \bar{g}^{\mu\nu}} -\frac{3}{2t}\del{l} \underline{\bar{g}^{0\mu}} \biggr)
	=  \epsilon \del{0} u^{0\mu}_l-\frac{3}{2t} \epsilon u^{0\mu}_l+\epsilon \mathscr{L}^\mu_l(\epsilon,t, \mathbf{u}),   
\end{align*}
follows from \eqref{E:PIG}, from which we deduce that
\begin{align}
\underline{\bar{g}^{kl}}\partial_0 u^{0\mu}_k   - \frac{1}{\epsilon} \underline{\bar{g}^{kl}} \partial_k u^{0\mu}_0=  \frac{3}{2t}\underline{\bar{g}^{kl} } u^{0\mu}_k  + \mathscr{L}^{\mu l} (\epsilon,t,   \mathbf{u})+\epsilon \mathscr{J}^{\mu l} (\epsilon,t, \mathbf{u}). \label{e:Ein1.b}
\end{align}
Together, \eqref{e:Ein1.a}, \eqref{e:Ein1.c}, \eqref{e:Ein1.b} for a system of evolution equations whose principal part involves the metric variables $\{u^{0\mu}_0,u^{0\mu}_l,u^{0\mu}\}$.

Next, we contract both sides of \eqref{E:REDUCEDCONFEINSTEINEQ1} with $\delta^k_\mu\delta^l_\nu$ to get
\begin{align}\label{E:REDUCEDCONFEINSTEINEQ1.a}
&-\bar{g}^{\alpha\beta}\delta^k_\mu\delta^l_\nu\bar{ \partial}_{\alpha} \udn{\beta}(\bar{g}^{\mu\nu} -\bar{h}^{\mu\nu} ) + \bar{\mathcal{E}}^{kl}  -2 \bar{P}^{kl} -\bar{\mathcal{Q}}^{kl}
=    \frac{2 \Lambda}{3t} \delta^k_\mu\delta^l_\nu \udn{0} (\bar{g}^{\mu\nu} -\bar{h}^{\mu\nu} )  -\frac{2}{t^2}\biggl(\bar{g}^{00}+\frac{\Lambda}{3} \biggr)\bar{g}
^{kl}   \nnb \\ &\hspace{2.5cm} - \frac{2}{t^2}(1+\epsilon^2K) \bar{\rho}  \bar{v}^k \bar{v}^l
-\frac{1-\epsilon^2K}{ t^2}(\bar{\rho}- \mu )\bar{g}^{kl}-\frac{1-\epsilon^2K}{ t^2} \mu (\bar{g}^{kl}-\bar{h}^{kl})-\frac{10 \Lambda}{3t^2} \Omega  (\bar{g}^{kl}-\bar{h}^{kl}),
\end{align}
where
\begin{align*}
\underline{\bar{\mathcal{E}}^{kl}}:=\delta^k_\mu \delta^l_\nu\underline{\bar{\mathcal{E}}^{\mu \nu}}
=&  
\epsilon \mathscr{L}^{kl} (\epsilon, t, \mathbf{u})+\epsilon^2\mathscr{J}^{kl} (\epsilon,t,\mathbf{u}).
\end{align*}
Using the identity
\begin{align}
	\check{g}_{kl}\delta^k_\sigma\delta^l_\nu \udn{\mu}\bar{g}^{\sigma\nu}
	=3 \udn{\mu}\ln  \alpha  +2\check{g}_{kl}\delta^{ k}_{\mu}\bar{g}^{l 0}\frac{\Omega}{t}, \label{e:222}
\end{align}
where we
recall that $(\check{g}_{kl}) = (\bar{g}^{kl})^{-1}$,  and
the definitions  \eqref{E:u.a} and \eqref{E:u.g}, a direct calculation, using the identity
\begin{align}\label{e:111}
	\bar{ \partial}_{\alpha}\check{g}_{kl}=- \check{g}_{ki} (\bar{ \partial}_{\alpha} \bar{g}^{ij}) \check{g}_{jl},
\end{align}
shows that
\begin{gather*}
	 \frac{\Lambda}{3} \frac{2}{3}\underline{\bar{g}^{\alpha\beta}\bar{\partial}_{\alpha}\biggl( \check{g}_{kl}\delta^{k}_{\beta}\bar{g}^{l0}\frac{\Omega}{t} \biggr) }
	 =  \epsilon \mathscr{L} (\epsilon, t, \mathbf{u})+\epsilon^2\mathscr{J} (\epsilon,t,\mathbf{u})  \\
	 \intertext{and}
	 \delta^k_\mu \delta^l_\nu\underline{\bar{g}^{\alpha\beta} (\bar{\partial}_\alpha \check{g}_{kl})\udn{\beta}\bigl(\bar{g}^{\mu\nu}-\bar{h}^{\mu\nu}\bigr)}=\epsilon\mathscr{L}(\epsilon,t,\mathbf{u})+\epsilon^2\mathscr{J}(\epsilon,t,\mathbf{u}).
\end{gather*}
We then observe that these two expressions can be used to write \eqref{E:REDUCEDCONFEINSTEINEQ1.a} as
\begin{align}
&- \underline{\bar{g}^{00}}\partial_{0} u_0-4 t u^{0k} \partial_{k} u_0- \frac{1}{\epsilon} \underline{\bar{g}^{kl}} \partial_{k} u_l  + \mathscr{L} (\epsilon, t, \mathbf{u})+\epsilon \mathscr{J} (\epsilon,t,\mathbf{u})
=  - \frac{2  }{ t} \underline{\bar{g}^{00}} u_0 + 4 \epsilon u^{00} u_0 \nnb \\
& \hspace{0.5cm} - 8 \epsilon (u^{00})^2+ \frac{\Lambda}{3} \frac{1}{\epsilon} \frac{2}{3t^2}(1+\epsilon^2K)  \rho   \underline{\bar{v}^i} \underline{\bar{v}^j}   \underline{\check{g}_{ij}}  - \frac{4\Omega}{t }\Lambda u^{00}
-\frac{1}{\epsilon} \frac{1-\epsilon^2K}{ t^2}( \rho - \mu )\left(\underline{\bar{g}^{00}}-\frac{\Lambda}{3}\right) + \frac{\Lambda}{3} \frac{10 \Lambda}{9t^2} \Omega  \Theta^{ij} \underline{\check{g}_{ij}} \nnb \\
&\hspace{1cm}- \frac{1}{\epsilon} \frac{2}{t^2}(1+\epsilon^2K)\left(( \rho- \mu ) \underline{\bar{v}^0} \underline{\bar{v}^0} + \mu  \bigl( \underline{\bar{v}^0} \underline{\bar{v}^0}-  \frac{\Lambda}{3}  \bigr) \right)  -2 \frac{1-\epsilon^2K}{ t } \mu  u^{00} +  \frac{\Lambda}{3} \frac{1-\epsilon^2K}{ 3 t^2} \mu  \Theta^{ij} \underline{\check{g}_{ij}}.   \label{e:Ein3.a}
\end{align}
We also observe that the equations
\begin{align}\label{e:Ein3.b}
	\underline{\bar{g}^{kl}}\del{0}u_l-\frac{1}{\epsilon}\underline{\bar{g}^{kl}}\del{l} u_0=\frac{2\Lambda}{3t} E^2 \Omega \delta_{lj} u^{0j}_0 \underline{\bar{g}^{kl}}+\frac{2\Lambda}{3t}(E^2\Omega+2E^2 \Omega^2 +2Et \del{t}\Omega)\delta_{lj}u^{0j} \underline{\bar{g}^{kl} }
\end{align}
and
\begin{align}\label{e:Ein3.c}
	-\underline{\bar{g}^{00}}\del{0}u =-\underline{\bar{g}^{00}} u_0
\end{align}
follow easily from differentiating $u_l$ and $u$, see \eqref{E:u.f}-\eqref{E:u.g}, with respect to $t$.
Together \eqref{e:Ein3.a}, \eqref{e:Ein3.b} and \eqref{e:Ein3.c} form a system of evolution equations whose principal
part involves the metric variables $\{u_0,u_l, u\}$.

To obtain evolution equations for the remaining metric variables $\{u^{ij}_0,u^{ij}_l, u^{ij}\}$,
we define
\begin{align*}
\mathfrak{L}^{ij}_{kl}=\delta^i_k\delta^j_l-\frac{1}{3}\check{g}_{kl}\bar{g}^{ij},
\end{align*}
and contract $\frac{1}{\alpha}\mathfrak{L}^{ij}_{lm}
$ 
on both sides of \eqref{E:REDUCEDCONFEINSTEINEQ1.a}. A calculation using the identities (the first identity can be derived with the help of \eqref{e:222})
\begin{align*}
 \alpha^{-1}\mathfrak{L}^{ij}_{lm}\delta^l_\mu\delta^m_\nu \udn{\sigma}\bar{g}^{\mu\nu}=\delta^i_\mu\delta^j_\nu\udn{\sigma}(\alpha^{-1}
\bar{g}^{\mu\nu})-\frac{2}{3}\bar{\mathfrak{g}}^{ij}\check{g}_{kl}\delta^{ k}_\sigma \bar{g}^{l 0}\frac{\Omega}{t} \qquad \text{and} \qquad \mathfrak{L}^{ij}_{lm}\bar{g}^{lm}=0,
\end{align*}
where we recall that $\mathfrak{\bar{g}}^{ij}$ is defined by \eqref{E:GAMMA},
then shows that 
the following equation holds:
\begin{align}
&-\bar{g}^{\alpha\beta}\bar{ \partial}_{\alpha} \bigl( \delta^i_\mu\delta^j_\nu \udn{\beta}(\alpha^{-1}
\bar{g}^{\mu\nu})\bigr)  +\frac{2}{3}\bar{g}^{\alpha\beta}\bar{ \partial}_{\alpha}\Bigl(\bar{\mathfrak{g}}^{ij}\check{g}_{kl}\delta^{k}_\beta \bar{g}^{l0}\frac{\Omega}{t} \Bigr)+\bar{g}^{\alpha\beta}\bar{ \partial}_{\alpha} (\alpha^{-1}\mathfrak{L}^{ij}_{lm}) \delta^i_\mu\delta^m_\nu\udn{\beta} \bar{g}^{\mu\nu}  +\alpha^{-1}\mathfrak{L}^{ij}_{lm}\mathcal{E}^{lm} \nnb  \\
& -\frac{2}{\alpha}\mathfrak{L}^{ij}_{lm}\bar{P}^{lm}(\bar{g}^{-1}) -\frac{1}{\alpha}\mathfrak{L}^{ij}_{lm}\mathcal{Q}^{lm}(\bar{g},\udn{}\bar{g}^{-1})
\nnb \\
= &  \frac{2 \Lambda}{3t}  \delta^i_\mu\delta^j_\nu \udn{0}(\alpha^{-1}
\bar{g}^{\mu\nu})  - \frac{2}{t^2}(1+\epsilon^2K) \bar{\rho} \frac{1}{\alpha}\mathfrak{L}^{ij}_{lm}\bar{v}^l \bar{v}^m       +\frac{1-\epsilon^2K}{ t^2}\bar{\mu}\frac{1}{\alpha}\mathfrak{L}^{ij}_{lm} \bar{h}^{lm} +\frac{10 \Lambda}{3t^2} \Omega  \frac{1}{\alpha}\mathfrak{L}^{ij}_{lm} \bar{h}^{lm}.   \label{e:ein3}
\end{align}
Furthermore, using identity \eqref{e:111},
it is not difficult via a straightforward calculation to verify that
\begin{align}
\underline{\mathfrak{L}^{ij}_{lm}}\bar{h}^{lm}=&\epsilon \mathscr{L}^{ij}(\epsilon, t, \mathbf{u})+\epsilon^2\mathscr{J}^{ij}(\epsilon,t,\mathbf{u}), \label{e:tri1} \\
\frac{2}{3}\underline{\bar{g}^{\alpha\beta} \bar{ \partial}_{\alpha} \bigl(\bar{\mathfrak{g}}^{ij}\check{g}_{kl}\delta^{k}_\beta \bar{g}^{l0}\frac{\Omega}{t} \bigr)}
	= &  \epsilon \mathscr{L}^{ij}(\epsilon, t, \mathbf{u})+\epsilon^2\mathscr{J}^{ij}(\epsilon,t,u^{\alpha\beta},\mathbf{u})  \\
\intertext{and}
\underline{\bar{g}^{\alpha\beta}\bar{ \partial}_{\alpha}(\alpha^{-1}\mathfrak{L}^{ij}_{lm})\delta^l_\mu\delta^m_\nu \udn{\beta} \bar{g}^{\mu\nu}} = & \epsilon^2\mathscr{J}^{ij}(\epsilon,t,\mathbf{u}) \label{e:tri2}.
\end{align}
Using \eqref{e:tri1}-\eqref{e:tri2}, we can then rewrite \eqref{e:ein3} as
\begin{align}
& - \underline{\bar{g}^{00} }\partial_{0} u^{ij}_0- 4 t u^{0k} \partial_{k} u^{ij}_0-  \frac{1}{\epsilon} \underline{\bar{g}^{kl}} \partial_{k} u^{ij}_l + \mathscr{L}^{ij}(\epsilon, t, \mathbf{u})+\epsilon \mathscr{J}^{ij}(\epsilon,t,\mathbf{u})  \nnb  \\
& \hspace{6.5cm} =  - \frac{2 }{t}\underline{\bar{g}^{00}}  u^{ij}_0  +4 \epsilon u^{00}  u^{ij}_0 - \frac{1}{\epsilon} \frac{2}{t^2}(1+\epsilon^2K)  \rho  \underline{\frac{1}{\alpha}\mathfrak{L}^{ij}_{lm}\bar{v}^l \bar{v}^m },  \label{e:Ein2.a}
\end{align}
while differentiating $u^{ij}_l$ and $u^{ij}$,  \eqref{E:u.d}-\eqref{E:u.e}, with respect to $t$ yields
\begin{align}\label{e:Ein2.b}
\underline{\bar{g}^{kl}}\del{0}u^{ij}_l-\frac{1}{\epsilon}\underline{\bar{g}^{kl}}\del{l} u^{ij}_0=\mathscr{L}^{kij}(\epsilon,t,\mathbf{u})+\epsilon \mathscr{J}^{kij}(\epsilon,t, \mathbf{u})
\end{align}
and
\begin{align}\label{e:Ein2.c}
-\underline{\bar{g}^{00}}\del{0}u^{ij} =-\underline{\bar{g}^{00}} u^{ij}_0+  \underline{\bar{g}^{00}}\frac{2}{t}\Omega u^{ij}.
\end{align}
Together \eqref{e:Ein2.a}, \eqref{e:Ein2.b} and \eqref{e:Ein2.c} form a system of evolution equations
whose principal term involves the remaining metric variables $\{u^{ij}_0,u^{ij}_l, u^{ij}\}$.

Gathering \eqref{e:Ein1.a}-\eqref{e:Ein1.b},
\eqref{e:Ein3.a}-\eqref{e:Ein3.c} and \eqref{e:Ein2.a}-\eqref{e:Ein2.c} together, we
arrive at the following formulation of the reduced conformal Einstein equations:
\begin{align}
\tilde{B}^0\partial_0\begin{pmatrix}
u^{0\mu}_0\\u^{0\mu}_k\\u^{0\mu}
\end{pmatrix}+\tilde{B}^k\partial_k\begin{pmatrix}
u^{0\mu}_0\\u^{0\mu}_l\\u^{0\mu}
\end{pmatrix}+\frac{1}{\epsilon}\tilde{C}^k\partial_k\begin{pmatrix}
u^{0\mu}_0\\u^{0\mu}_l\\u^{0\mu}
\end{pmatrix}&=\frac{1}{t}\mathfrak{\tilde{B}}\mathbb{P}_2\begin{pmatrix}
u^{0\mu}_0\\u^{0\mu}_l\\u^{0\mu}
\end{pmatrix}+\hat{S}_1, \label{E:EIN1}\\
\tilde{B}^0\partial_0\begin{pmatrix}
u^{ij}_0\\u^{ij}_k\\u^{ij}
\end{pmatrix}+\tilde{B}^k\partial_k\begin{pmatrix}
u^{ij}_0\\u^{ij}_l\\u^{ij}
\end{pmatrix}+\frac{1}{\epsilon}\tilde{C}^k\partial_k\begin{pmatrix}
u^{ij}_0\\u^{ij}_l\\u^{ij}
\end{pmatrix}&=-\frac{2 E^2\underline{\bar{g}^{00}}}{t}\breve{\mathbb{P}}_2 \begin{pmatrix}
u^{ij}_0\\u^{ij}_l\\u^{ij}
\end{pmatrix}+ \tilde{S}_2+ \tilde{G}_2, \label{E:EIN2}\\
\tilde{B}^0\partial_0\begin{pmatrix}
u_0\\u_k\\u
\end{pmatrix}+\tilde{B}^k\partial_k\begin{pmatrix}
u_0\\u_l\\u
\end{pmatrix}+\frac{1}{\epsilon}\tilde{C}^k\partial_k\begin{pmatrix}
u_0\\u_l\\u
\end{pmatrix}&=-\frac{2 E^2 \underline{\bar{g}^{00}}}{t}\breve{\mathbb{P}}_2 \begin{pmatrix}
u_0\\u_l\\u
\end{pmatrix}+\tilde{S}_3+\tilde{G}_3, \label{E:EIN3}
\end{align}
where
\begin{align}\label{E:EINBk}
\tilde{B}^0=E^2\begin{pmatrix}
-\underline{\bar{g}^{00}}& 0 & 0\\
0 & \underline{\bar{g}^{kl}} & 0 \\
0 & 0 & -\underline{\bar{g}^{00}}
\end{pmatrix},
\qquad
\tilde{B}^k=E^2\begin{pmatrix}
-4tu^{0k} & -\Theta^{kl} & 0\\
-\Theta^{kl} & 0 & 0 \\
0 & 0 & 0
\end{pmatrix},
\end{align}
\begin{align} \label{E:EINCk}
\tilde{C}^k=\begin{pmatrix}
0 & - \delta^{kl} & 0\\
- \delta^{kl} & 0 & 0 \\
0 & 0 & 0
\end{pmatrix},
\qquad
\mathfrak{\tilde{B}}=E^2\begin{pmatrix}
-\underline{\bar{g}^{00}} & 0 & 0\\
0 & \frac{3}{2}\underline{\bar{g}^{ki}} & 0 \\
0 & 0 & -\underline{\bar{g}^{00}}
\end{pmatrix},
\end{align}
\begin{align} \label{E:EINP2}
\mathbb{P}_2=\begin{pmatrix}
\frac{1}{2} & 0 & \frac{1}{2}\\
0 & \delta^l_i & 0 \\
\frac{1}{2} & 0 & \frac{1}{2}
\end{pmatrix},
\qquad
\breve{\mathbb{P}}_2=\begin{pmatrix}
1 & 0 & 0\\
0 & 0 & 0 \\
0 & 0 & 0
\end{pmatrix},
\end{align}
\begin{align} \label{E:EINS1}
\hat{S}_1=E^2\begin{pmatrix}
2  E^{-2}\frac{\Omega}{t}\delta^{km}  \delta^\mu_k  u^{00}_m  -\frac{2(1-\epsilon^2K)}{t}  \delta \rho u^{0\mu} + \hat{f}^{0\mu}+  \mathscr{L}^{0\mu} +\epsilon \mathscr{J}^{0\mu}  \\
\mathscr{L}^{0\mu l} +\epsilon \mathscr{J}^{0\mu l}   \\
\mathscr{L}^{00\mu} +\epsilon \mathscr{J}^{00\mu}
\end{pmatrix},
\end{align}
\begin{align*}
	\hat{f}^{0\mu}
	= & -\frac{\Lambda}{3} \frac{1}{t^2}\frac{1}{\epsilon}\delta\rho\delta^\mu_0- \frac{4}{t^2}\frac{1}{\epsilon}\sqrt{\frac{\Lambda}{3}}\delta\rho\delta^\mu_0(\underline{\bar{v}^0}-\sqrt{\frac{\Lambda}{3}})- \frac{2}{t^2}\frac{1}{\epsilon}\delta\rho\delta^\mu_0(\underline{\bar{v}^0}-\sqrt{\frac{\Lambda}{3}})^2- \frac{2}{t^2} \delta\rho\delta^\mu_i z^i \underline{\bar{v}^0}  - \frac{2}{t^2} \mu \delta^\mu_i   z^i \underline{\bar{v}^0}    \nnb  \\
	&  - \frac{2}{t^2}\frac{1}{\epsilon} \mu \delta^\mu_0 \bigl( \underline{\bar{v}^0}-\sqrt{\frac{\Lambda}{3}} \bigr)\bigl( \underline{\bar{v}^0}+ \sqrt{\frac{\Lambda}{3}} \bigr)   - \frac{2}{t^2}  \epsilon K \left(\delta\rho \underline{\bar{v}^\mu} \underline{\bar{v}^0} + \mu \bigl( \underline{\bar{v}^\mu} \underline{\bar{v}^0}-  \frac{\Lambda}{3}\delta^\mu_0 \bigr) \right)  -\frac{\Lambda}{3}\frac{\epsilon K}{ t^2}  \delta\rho \delta^\mu_0,  
\end{align*}
\begin{align} \label{E:EING1}
\tilde{G}_2=E^2\begin{pmatrix}
- \epsilon \frac{2}{t^2}(1+\epsilon^2K)  \rho  \underline{\alpha}^{-1}\underline{\mathfrak{L}^{ij}_{lm}}z^l z^m +  \mathscr{L}^{ij} \\
\mathscr{L}^{ijl}  \\
\mathscr{L}^{0ij}
\end{pmatrix},
\end{align}
\begin{align} \label{E:EING3}
\tilde{G}_3=E^2\begin{pmatrix}
\epsilon \frac{2(1+\epsilon^2K)\Lambda}{9t^2}  \rho   z^i z^j   \underline{\check{g}_{ij} }
-\frac{(1+\epsilon^2 K)}{t^2} \rho\bigl(\underline{v^0}+\sqrt{\frac{\Lambda}{3}}\bigr)\frac{ \underline{\bar{v}^0}-\sqrt{\frac{\Lambda}{3}} }{\epsilon}-\epsilon \frac{4\Lambda}{3t^2} K \delta \rho
-  \frac{2(1-\epsilon^2K)}{ t}\delta\rho u^{00}  +  \mathscr{L}  \\
\mathscr{L}^l  \\
\mathscr{L}^0
\end{pmatrix},
\end{align}
\begin{align} \label{E:EINS23}
\tilde{S}_2=\epsilon (\mathscr{J}^{ij},\mathscr{J}^{ijl},\mathscr{J}^{0ij})^T \AND \tilde{S}_3=\epsilon (\mathscr{J},\mathscr{J}^l,\mathscr{J}^0)^T.
\end{align}

\subsection{The conformal Euler equations\label{conformalEul}}
In this section, we turn to the problem of transforming
the conformal Euler equations. We begin by noting that it follows from the computation in \cite{Liu2017, Oliynyk2015} that conformal Euler equations \eqref{Confeul},  when expressed in Newtonian coordinates, are given by
\begin{align}\label{E:FINALEULER1}
\bar{B}^0\partial_0\begin{pmatrix}
\zeta\\
z ^i
\end{pmatrix}+
\bar{B}^k\partial_k\begin{pmatrix}
\zeta\\
z ^i
\end{pmatrix}=
\frac{1}{t}\mathcal{\bar{B}}\hat{\mathbb{P}}_2\begin{pmatrix}
\zeta\\
z ^i
\end{pmatrix}+\bar{S},
\end{align}
where
\begin{align*}
\bar{B}^0 &=\begin{pmatrix}
1 & \epsilon\frac{L^0_i}{\underline{\bar{v}^0}}\\
\epsilon \frac{L^0_j}{\underline{\bar{v}^0}} & K^{-1}M_{ij}
\end{pmatrix} 
,\\
\bar{B}^k &=\begin{pmatrix}
\frac{1}{\epsilon}\frac{\underline{\bar{v}^k}}{\underline{\bar{v}^0}} & \frac{L^k_i}{\underline{\bar{v}^0}}\\
\frac{L^k_j}{\underline{\bar{v}^0}} & K^{-1}M_{ij}\frac{1}{\epsilon}\frac{\underline{\bar{v}^k}}{\underline{\bar{v}^0}}
\end{pmatrix}
=\begin{pmatrix}
\frac{1}{\underline{\bar{v}^0}} z ^k & \frac{1}{\underline{\bar{v}^0}}\delta^k_i\\
\frac{1}{\underline{\bar{v}^0}}\delta^k_j & K^{-1}\frac{1}{\underline{\bar{v}^0}}M_{ij} z ^k
\end{pmatrix} 
,\\
\mathfrak{\bar{B}} &=\begin{pmatrix}
1 & 0\\
0 & -K^{-1}(1-3\epsilon^2K)\frac{\underline{\bar{g}_{ik}}}{\underline{\bar{v}_0}\underline{\bar{v}^0}}
\end{pmatrix}, \label{barBfr}  \\
\hat{\mathbb{P}}_2&=\begin{pmatrix}
0 & 0\\
0 & \delta^k_{j}
\end{pmatrix}, 
\\
\bar{S}&=\begin{pmatrix}
-L^\mu_i\underline{\bar{\Gamma}^i_{\mu\nu}}\,\underline{\bar{v}^\nu}\frac{1}{\underline{\bar{v}^0}} \\
-K^{-1}(1-3\epsilon^2K)\frac{1}{\underline{\bar{v}_0}}\underline{\bar{g}_{0j}}
-K^{-1}M_{ij}\underline{\bar{v}^\mu}\frac{1}{\epsilon}\underline{\bar{\Gamma}^i_{\mu\nu}} \,
\underline{\bar{v}^\nu}\frac{1}{\underline{\bar{v}^0}}
\end{pmatrix}, 
\\
L^\mu_i&=\delta^\mu_i-\frac{\underline{\bar{v}_i}}{\underline{\bar{v}_0}}\delta^\mu_0 
\intertext{and}
M_{ij}&=\underline{\bar{g}_{ij}}-\frac{\underline{\bar{v}_i}}{\underline{\bar{v}_0}}
\underline{\bar{g}_{0j}}-\frac{\underline{\bar{v}_j}}{\underline{\bar{v}_0}}
\underline{\bar{g}_{0i}}
+\frac{\underline{\bar{g}_{00}}}{(\underline{\bar{v}_0})^2}\underline{\bar{v}_i}\,\underline{\bar{v}_j}.
\end{align*}
In order to bring \eqref{E:FINALEULER1} into the required form, we perform a change of variables from $z^i$ to $z_j$,
which are related via  the map
$z^i=z^i(z_j,\underline{\bar{g}^{\mu\nu}})$ given by \eqref{E:Z_IANDZ^I}.
Denoting the Jacobian of the transformation by
\begin{equation*}
J^{im}:=\frac{\partial z ^i}{\partial z _m},
\end{equation*}
we observe that
\begin{equation*}
\partial_\sigma z ^i=J^{im}\partial_\sigma z _m+\delta_\sigma^0 \frac{\partial z ^i}{\partial\underline{\bar{g}^{\mu\nu}}}
\underline{\bar{\partial}_0\bar{g}^{\mu\nu}} +\epsilon \delta_\sigma^j \frac{\partial z ^i}{\partial\underline{\bar{g}^{\mu\nu}}}
\underline{\bar{\partial}_j\bar{g}^{\mu\nu}}.
\end{equation*}
Multiplying \eqref{E:FINALEULER1} by the block matrix $\diag{(1, J^{jl})}$ and changing variables from $(\zeta,z^i)$ to
$(\delta \zeta,z_j)$, where we recall from \eqref{E:DELZETA} that $\delta \zeta = \zeta - \zeta_H$, it is not difficult to
verify that we can write
\eqref{E:FINALEULER1} as
\begin{align}\label{E:FINALEULEREQUATIONS}
B^0\partial_0\begin{pmatrix}
\delta\zeta\\
z _m
\end{pmatrix}+
B^k\partial_k\begin{pmatrix}
\delta\zeta\\
z _m
\end{pmatrix}=
\frac{1}{t}\mathfrak{B}\hat{\mathbb{P}}_2\begin{pmatrix}
\delta\zeta\\
z _m
\end{pmatrix}+\hat{S}
\end{align}
where
\begin{align*}
B^0&=\begin{pmatrix}
1 & \epsilon\frac{L^0_i}{\underline{\bar{v}^0}}J^{im}\\
\epsilon \frac{L^0_j}{\underline{\bar{v}^0}}J^{jl} & K^{-1}M_{ij}J^{jl}J^{im}
\end{pmatrix}, 
\\
B^k&=\begin{pmatrix}
\frac{1}{\underline{\bar{v}^0}} z ^k & \frac{1}{\underline{\bar{v}^0}}J^{km}\\
\frac{1}{\underline{\bar{v}^0}}J^{kl}& K^{-1}\frac{1}{\underline{\bar{v}^0}}M_{ij}J^{jl}J^{im} z ^k
\end{pmatrix}, 
\\
\mathfrak{B}&=\begin{pmatrix}
1 & 0\\
0 & -K^{-1}(1-3\epsilon^2K)\frac{1}{\underline{\bar{v}_0}\underline{\bar{v}^0}}J^{ml}
\end{pmatrix} 
\end{align*}
and
\begin{align}
\hat{S}=&\begin{pmatrix}
-L^0_i\underline{\bar{\Gamma}^i_{00}}-L^\mu_i\underline{\bar{\Gamma}^i_{\mu j}}\,\underline{\bar{v}^j} \frac{1}{\underline{\bar{v}^0}}+(\bar{\gamma}^i_{i0}-\underline{\bar{\Gamma}^i_{i0}}) \\
-K^{-1}J^{jl}M_{ij}\underline{\bar{v}^\mu}\frac{1}{\epsilon}\underline{\bar{\Gamma}^i_{\mu\nu}}
\,\underline{\bar{v}^\nu}\frac{1}{\underline{\bar{v}^0}}+\epsilon\frac{L^0_j}{\underline{\bar{v}^0}}J^{jl}\bar{\gamma}^i_{i0}
\end{pmatrix}-\begin{pmatrix}
\epsilon\frac{L^0_i}{\underline{\bar{v}^0}}\frac{\partial z ^i}{\partial\underline{\bar{g}^{\mu\nu}}}
\underline{\bar{\partial_0}\bar{g}^{\mu\nu}}+\epsilon\frac{\delta^k_i}{\underline{\bar{v}^0}}\frac{\partial z ^i}
{\partial\underline{\bar{g}^{\mu\nu}}}
\underline{\bar{\partial}_k\bar{g}^{\mu\nu}}\\
K^{-1}M_{ij}J^{jl}\frac{\partial z ^i}{\partial\underline{\bar{g}^{\mu\nu}}}
\underline{\bar{\partial}_0\bar{g}^{\mu\nu}}
+\epsilon K^{-1}\bar{M}_{ij}\frac{ z ^k}{\underline{\bar{v}^0}}J^{jl}\frac{\partial z ^i}{\partial\underline{\bar{g}^{\mu\nu}}}\underline{\bar{\partial}_k\bar{g}^{\mu\nu}}
\end{pmatrix}. \nonumber
\end{align}
A direct calculation employing \eqref{E:Z_IANDZ^I} and the expansions \eqref{E:GIJ} and \eqref{E:G0MU} shows that
\begin{equation}\label{E:JACOBI}
J^{ik}
=E^{-2}\delta^{ik}+\epsilon \Theta^{ik}+\epsilon^2 \mathscr{\hat{S}}{}^{ik}(\epsilon,t,u,u^{\mu\nu},z_j),
\end{equation}
where  $\mathscr{\hat{S}}{}^{ik}(\epsilon,t,0,0,0)=0$.
Similarly, it is not difficult to see from \eqref{E:Z_IANDZ^I} and
the expansions  \eqref{E:GIJ} and \eqref{E:G0MU}-\eqref{E:PGIJ} that
\begin{gather}
	 \frac{\partial z ^i}{\partial\underline{\bar{g}^{\mu\nu}}}
	\underline{\bar{\partial}_\sigma\bar{g}^{\mu\nu}}
	= -2\biggl(\delta^0_\sigma E^{-2}\frac{\Omega}{t}z_j\delta^{ij} +\sqrt{\frac{3}{\Lambda}}\bigl(\delta^0_\sigma(u^{0i}_0 + 3u^{0i})+\delta^j_\sigma u^{0i}_j-2\delta^0_\sigma \Omega u^{0i}\bigr)\biggr)
	+ \epsilon \mathscr{S}^i_\sigma(\epsilon,t,\mathbf{u},z_j)
	\label{E:JACOBI2}
	\intertext{and}
	\epsilon  \frac{\delta^k_i}{\underline{\bar{v}^0}}\frac{\partial z^i}{\partial \bar{g}^{\mu\nu}}\bar{\partial}_k\bar{g}^{\mu\nu}=- \epsilon \frac{6}{\Lambda} u^{0i}_k\delta^k_i+\epsilon^2 \mathscr{S} (\epsilon,t,\mathbf{u},z_j),
\end{gather}
where $\mathscr{S}^i_\sigma(\epsilon,t,0,0)=\mathscr{S}(\epsilon,t,0,0)=0$. 
We further note that  the term $-K^{-1}J^{jl}M_{ij}\underline{\bar{v}^\mu}\frac{1}{\epsilon}\underline{\bar{\Gamma}^j_{\mu\nu}}
\underline{\bar{v}^\nu}\frac{1}{\underline{\bar{v}^0}}$ found in $\hat{S}$
above is not singular in $\epsilon$. This can be seen from the
expansions \eqref{E:GIJ}, \eqref{E:G0MU}, \eqref{E:G_MUNU} and \eqref{E:GAMMAI00}, which can be used to derive
\begin{align}
\frac{1}{\epsilon}\underline{\bar{\Gamma}^j_{\mu\nu}} \underline{v^\mu}\underline{v^\nu}= & 2\underline{\bar{\Gamma}^j_{0i}} \underline{v^0} z^i+ \epsilon\underline{\bar{\Gamma}^j_{ik}} z^i z^k+\frac{1}{\epsilon}\underline{\bar{\Gamma}^j_{00}} \underline{v^0} \underline{v^0} \nonumber\\
=& \sqrt{\frac{\Lambda}{3}}\frac{2\Omega}{t}E^{-2} z_i\delta^{ij}
+ u^{0j}_0+3u^{0j}-\frac{1}{2}\left(\frac{3}{\Lambda}\right)  E^{-2}\delta^{jk}u^{00}_k  + \epsilon \mathscr{S}^j (\epsilon,t,\mathbf{u},z_j), \label{Ssingterm}
\end{align}
where  $\mathscr{S}^j (\epsilon,t,0,0)=0$. 
Moreover, using the expansions \eqref{E:JACOBI}, \eqref{E:JACOBI2} and \eqref{Ssingterm} in conjunction
with \eqref{E:GIJ}, \eqref{E:G0MU}, \eqref{E:G_MUNU}, \eqref{E:CHRISTOFFEL}, \eqref{E:V_0}, \eqref{E:V^0} and \eqref{E:VELOCITY},
we observe that the matrices $\{B^0, B^k,\mathfrak{B}\}$ and source term $\hat{S}$ can be expanded as
\begin{align}
B^0
=&\p{
	1 & 0 \\
	0 & K^{-1}E^{-2}\delta^{lm}
}+\epsilon \p{0 & 0 \\ 0 & K^{-1} \Theta^{lm}}+\epsilon^2 \mathscr{\hat{S}}^0(\epsilon,t,\mathbf{u},z_j), \label{E:B0REMAINDER}\\
B^k
=&\sqrt{\frac{3}{\Lambda}}\p{
	z^k & E^{-2} \delta^{km}\\
	E^{-2}\delta^{kl} & K^{-1} E^{-2}\delta^{lm}z^k
}+\epsilon \sqrt{\frac{3}{\Lambda}} \p{\frac{3}{\Lambda}t u^{00} z^k & \Theta^{km}+\frac{3}{\Lambda} t u^{00}E^{-2} \delta^{km} \\ \Theta^{kl}+\frac{3}{\Lambda} t u^{00}E^{-2} \delta^{kl} & K^{-1}\bigl(\Theta^{lm}+\frac{3}{\Lambda} t u^{00} E^{-2} \delta^{lm}\bigr) z^k } \nnb \\&+\epsilon^2 \mathscr{\hat{S}}^k(\epsilon,t,\mathbf{u},z_j), \label{E:BkREMAINDER}  \\
\mathfrak{B}
=& \p{
	1 & 0 \\
	0 & K^{-1} (1-3\epsilon^2 K) E^{-2}\delta^{lm}
}+\epsilon \p{0 & 0 \\ 0 &  K^{-1}\Theta^{lm}} +\epsilon^2 \mathscr{\hat{S}}(\epsilon,t,\mathbf{u},z_j)  \label{E:BCALREMAINDER}
\intertext{and}
\hat{S}
=&\p{0 \\
	-K^{-1}\left[\sqrt{\frac{3}{\Lambda}}\bigl(-u^{0l}_0+(-3+4\Omega)u^{0l}\bigr)+\frac{1}{2}\left(\frac{3}{\Lambda}\right)^{\frac{3}{2}}E^{-2}\delta^{lk}u^{00}_k\right]	
}  \nnb  \\
&  +\epsilon \p{ \frac{1}{2} E^2\delta_{ij} \bigl(u^{ij}_0+\frac{3}{\Lambda}E^{-2}(3u^{00} + u^{00}_0-u_0 )\delta^{ij}\bigr) + \frac{6}{\Lambda} u^{0i}_k\delta^k_i \\  \mathscr{S}_1(\epsilon,t,\mathbf{u},z_j)}    +\epsilon^2 \mathscr{S} (\epsilon,t,\mathbf{u},z_j), \label{E:SREMAINDER}
\end{align}
where all the remainder terms $\mathscr{\hat{S}}^\mu$, $\mathscr{\hat{S}}$, $\mathscr{S}_1$ and $\mathscr{S}$ vanish
for $(\mathbf{u},z_j)=(0,0)$. 

\subsection{The reduced conformal Einstein-Euler equations\label{rcEEeqns}}
Collecting \eqref{E:EIN1}, \eqref{E:EIN2}, \eqref{E:EIN3} and  \eqref{E:FINALEULEREQUATIONS}  together and
setting
\begin{align}\label{E:REALVAR0}
\mathbf{\hat{U}}=(\mathbf{\hat{U}}_1, \mathbf{U}_2)^T, 
\end{align}
where
\begin{align}\label{E:REALVAR1}
\hat{\mathbf{U}}_1= (u^{0\mu}_0, u^{0\mu}_k, u^{0\mu}, u^{ij}_0,u^{ij}_k,u^{ij},u_0,u_k,u)^T \AND
\mathbf{U}_2=(\delta\zeta, z _i)^T,
\end{align}
we obtain the following symmetric hyperbolic formulation of the reduced conformal Einstein-Euler equations:
\begin{equation}\label{E:LOCEQ}
\begin{aligned}
\mathbf{B}^0\partial_t \bhU+\mathbf{B}^i\partial_i \bhU+\frac{1}{\epsilon}\mathbf{C}^i\partial_i\bhU=\frac{1}{t}\mathbf{B}\mathbf{P}
\bhU+\mathbf{\hat{H}}
\end{aligned}
\end{equation}
where
\begin{align}
\mathbf{B}^0=\begin{pmatrix}
\tilde{B}^0 & 0 & 0 & 0   \\
0 & \tilde{B}^0 & 0 & 0   \\
0 & 0 & \tilde{B}^0 & 0   \\
0 & 0 & 0 & B^0
\end{pmatrix},
\quad
\mathbf{B}^i=\begin{pmatrix}
\tilde{B}^i & 0 & 0 & 0  \\
0 & \tilde{B}^i & 0 & 0   \\
0 & 0 & \tilde{B}^i & 0   \\
0 & 0 & 0 & B^i
\end{pmatrix},
\quad
\mathbf{C}^i=\begin{pmatrix}
\tilde{C}^i & 0 & 0 & 0   \\
0 & \tilde{C}^i & 0 & 0  \\
0 & 0 & \tilde{C}^i & 0  \\
0 & 0 & 0 & 0
\end{pmatrix}, \label{E:REALEQa}
\end{align}

\begin{align}
\mathbf{B}=\begin{pmatrix}
\mathfrak{\tilde{B}} & 0 & 0 & 0   \\
0 & -2E^2\underline{\bar{g}^{00}}I & 0 & 0   \\
0 & 0 & -2E^2\underline{\bar{g}^{00}}I & 0  \\
0 & 0 & 0 & \mathfrak{B}
\end{pmatrix},
\quad
\mathbf{P}=\begin{pmatrix}
\mathbb{P}_2 & 0 & 0 & 0   \\
0 & \mathbb{\breve{P}}_2 & 0 & 0   \\
0 & 0 & \mathbb{\breve{P}}_2 & 0   \\
0 & 0 & 0 & \hat{\mathbb{P}}_2
\end{pmatrix}, \label{E:REALEQb}
\end{align}
\begin{align*}
\mathbf{\hat{H}}
=(\hat{S}_1,\tilde{S}_2+\tilde{G}_2,\tilde{S}_3+\tilde{G}_3,\hat{S})^T,
\end{align*}
and $\mathbf{\hat{H}}$ vanishes for $\mathbf{\hat{U}}=0$.
The importance of this formulation of the reduced conformal Einstein-Euler equations is that it is now of the form
analyzed in \cite{Oliynyk2016a}. As a consequence, we could, for fixed $\epsilon>0$, use the results of \cite{Oliynyk2016a} to obtain
the global existence to the future under a suitable small initial data assumption.  What the above formulation is not yet suitable for is analyzing the limit $\epsilon \searrow 0$. To bring the system into a form that is suitable
requires a further non-local transformation, which is carried out in \S \ref{S:NloEE}.

\section{Initial data}\label{S:INITIALIZATION}
Before continuing on with the analysis of the evolution equations, we will, in this section, turn to the problem
of selecting initial data. It is well known that the initial data for the reduced conformal Einstein-Euler equations cannot be chosen freely on the initial hypersurface
\begin{equation*}
\Sigma  = \{1\}\times \Rbb^3 \subset M=(0,1]\times \Rbb^3
\end{equation*}
due to constraint equations that must satisfied on $\Sigma$. To solve these constraints, we employ a variation of Lottermoser's method \cite{Lottermoser1992} (also see \cite{Liu2017,Oliynyk2009a,Oliynyk2009b,Oliynyk2014}), which we use to construct $1$-parameter families of $\epsilon$-dependent solutions to the constraint equations that behave appropriately in the limit $\epsilon\searrow 0$. In order to use Lottermoser's method,  we represent the gravitation field in terms of the variables $\hmfu^{\mu\nu}$
and $\hmfu^{\mu\nu}_{\sigma}$ that are defined via the formulas
\begin{align}\label{E:DEFHATG}
\hat{g}^{\mu\nu}:=\theta \bar{g}^{\mu\nu}= \bar{h} ^{\mu\nu}+\epsilon^2 \hmfu^{\mu\nu} \qquad \text{and} \qquad
\hmfu^{\mu\nu}_{\sigma}:=\bar{\partial}_{\sigma} \hmfu^{\mu\nu},
\end{align}
respectively, where
\begin{align}\label{E:SMALLTHETA}
\theta= \frac{\sqrt{| \bar{g}|}}{\sqrt{|\bar{h}|}}=E^{-3}\sqrt{\frac{\Lambda}{3}|\bar{g}|}, \quad | \bar{g}|=-\det{ \bar{g}_{\mu\nu}}   \quad  \text{and} \quad |\bar{h}|=-\det{\bar{h}_{\mu\nu}}=\frac{3}{\Lambda}E^6.
\end{align}
The complete set of constraints that we must solve on $\Sigma$ are:
\begin{align}
( \bar{G}^{0\mu}  - \bar{T}^{0\mu} )|_{\Sigma}&=0 \quad (\text{gravitational constraints}),\label{E:CONSTRAINT}\\
\bar{Z}^\mu|_{\Sigma}&=0 \quad (\text{gauge constraints}) \label{E:WAVECONSTRAINT}
\intertext{and}
( \bar{v}^\mu\bar{v}_\mu +1)|_{\Sigma}&=0 \quad (\text{velocity normalization}).\label{E:NORMALIZATION}
\end{align}

\subsection{Transformation formulas}
Before proceeding, we first establish some transformation formulas that will be used repeatedly in our analysis of the constraint equations.
In the following, we will freely use the notation set out in \S \ref{remainder} for analytic remainder terms.

\begin{lemma}\label{L:IDENTITY}
	\begin{align*}
	\theta=E^3\sqrt{-\frac{3}{\Lambda}\det{(\bar{h}^{\mu\nu}+\epsilon^2\hmfu^{\mu\nu}})}
	=1+\frac{1}{2}\epsilon^2   \left(-\frac{3}{\Lambda}\hmfu^{00}+ E^2 \hmfu^{ij}\delta_{ij}
	\right)+\epsilon^4 \breve{\mathscr{S}}(\epsilon,t, E, \hmfu^{\mu\nu}),   
	\end{align*}
	where 
	$\breve{\mathscr{S}}$ vanishes to second order in $\hmfu^{\mu\nu}$.
\end{lemma}
\begin{proof}
The proof follows from a direct calculation.
\end{proof}
Using the above lemma, we obtain the related formulas
\begin{align}\label{e:the1}
\frac{1}{\theta}-1=-\frac{1}{2}\epsilon^2\left(-\frac{3}{\Lambda}\hmfu^{00}+E^2\hmfu^{ij}\delta_{ij}\right)+\epsilon^4\breve{\mathscr{S}}(\epsilon,t, E, \hmfu^{\mu\nu}) =  \frac{3}{2\Lambda}\epsilon^2 \hmfu^{00}-\frac{1}{2}\epsilon^2E^2\hmfu^{ij}\delta_{ij} +\epsilon^4\breve{\mathscr{S}}(\epsilon,t, E, \hmfu^{\mu\nu})
\end{align}
and
\begin{align}\label{e:the2}
\frac{\theta-1}{\epsilon^2} =\frac{1}{2}\left(-\frac{3}{\Lambda}\hmfu^{00}+E^2\hmfu^{ij}\delta_{ij}\right)+\epsilon^2\breve{\mathscr{S}}(\epsilon,t, E, \hmfu^{\mu\nu}) = -\frac{3}{2\Lambda}\hmfu^{00}+\frac{1}{2}E^2\hmfu^{ij}\delta_{ij} +\epsilon^2\breve{\mathscr{S}}(\epsilon,t, E, \hmfu^{\mu\nu}),
\end{align}
where as above the remainder terms $\breve{\mathscr{S}}$ vanish to second order in $\hmfu^{\mu\nu}$.

\begin{lemma}\label{L:RELATION1}
	The metric variables $u^{0\mu}$, $u^{ij}$ and $u$ can be expressed in terms of the
	$\hmfu^{\mu\nu}$ via the transformation formulas
	\begin{align}
	u^{0\mu}=&\frac{\epsilon}{2t}\left(\frac{1}{2}\underline{\hmfu^{00}}\delta^\mu_0+ \underline{\hmfu^{0k}}\delta^\mu_k+\frac{\Lambda}{6} E^2 \underline{\hmfu^{ij}}\delta_{ij}\delta^\mu_0 \right)+\epsilon^3 \breve{\mathscr{S}} (\epsilon,t,E, \Omega/t ,\underline{\hmfu^{\alpha\beta}}),\label{E:U0MUANDINI}\\
	u = & \epsilon \frac{2\Lambda}{9} E^2\underline{\hmfu^{ij}}\delta_{ij} +\epsilon^3
	\breve{\mathscr{S}}(\epsilon,t,E, \Omega/t ,\underline{\hmfu^{\alpha\beta}}),
	\label{E:UANDINI}\\
	u^{ij}= & \epsilon  \left(\underline{\hmfu^{ij}}-\frac{1}{3}\underline{\hmfu^{kl}}\delta_{kl}\delta^{ij}\right)+
	\epsilon^3\breve{\mathscr{S}} ( \epsilon,t,E, \Omega/t ,\underline{\hmfu^{\alpha\beta}}),\label{E:UIJANDINI}
	\end{align}
	where all of the remainder terms $\breve{\mathscr{S}}$ vanish to second order in $\underline{\hmfu^{\mu\nu}}$. Moreover, the $0$-component of the conformal fluid four-velocity $\bar{v}^\mu$ can be written as
	\begin{align*}
	\underline{\bar{v}^0}=\sqrt{\frac{\Lambda}{3}}+\epsilon^2\breve{\mathscr{T}}(\epsilon,t,E, \Omega/t, \underline{\hmfu^{\alpha\beta}},z_j)
	\end{align*}
	where $\breve{\mathscr{T}}$ vanishes to first order in $(\underline{\hmfu^{\alpha\beta}},z_j)$.
\end{lemma}
\begin{proof}
	First, we observe that the first formula in the statement of the lemma follows directly from \eqref{E:G0MU} and Lemma \ref{L:IDENTITY}.
	Next, using \eqref{E:DEFHATG}, it is not difficult to verify that
	\begin{align*} 
	\det{(  \bar{g}^{kl} )}=\theta^{-3}(E^{-6}+\epsilon^2E^{-4}\hmfu^{ij}\delta_{ij})+\epsilon^4\breve{\mathscr{S}}  = E^{-6}+\frac{1}{2}\epsilon^2 E^{-6} \left(\frac{9}{\Lambda}\hmfu^{00}-E^2\hmfu^{ij}\delta_{ij}\right) +\epsilon^4\breve{\mathscr{S}}(\epsilon,t, E, \Omega/t, \hmfu^{\alpha\beta}),
	\end{align*}
	from which, with the help of \eqref{E:GAMMA}, we get
	\begin{align} \label{idalpha}
	 \alpha  =1+\frac{1}{6}\epsilon^2 \left(\frac{9}{\Lambda}\hmfu^{00}-E^2\hmfu^{ij}\delta_{ij}\right)+\epsilon^4 \breve{\mathscr{S}}(\epsilon,t,E, \Omega/t, \hmfu^{\alpha\beta}).
	\end{align}
	Then by \eqref{E:u.f}, \eqref{E:q} and \eqref{idalpha}, we obtain
	\begin{align*}  
	u=2tu^{00}-\frac{1}{\epsilon}\frac{\Lambda}{3} \ln[1+(\underline{\alpha}-1)]= \epsilon \frac{2\Lambda}{9} E^2\underline{\hmfu^{ij}}\delta_{ij} +\epsilon^3\breve{\mathscr{S}}(\epsilon,t,E, \Omega/t, \underline{\hmfu^{\alpha\beta}}),
	\end{align*}
	while we see that
	\begin{align*}
	u^{ij}=\frac{1}{\epsilon}\left((\underline{\alpha\theta})^{-1}\underline{\hat{g}^{ij}}- \bar{h}^{ij}\right)=\epsilon  \left(\underline{\hmfu^{ij}}-\frac{1}{3}\underline{\hmfu^{kl}}\delta_{kl}\delta^{ij}\right)+\epsilon^3\breve{\mathscr{S}}( \epsilon,t,E, \Omega/t, \underline{\hmfu^{\alpha\beta}})
	\end{align*}
	follows from \eqref{E:GIJ}, \eqref{E:DEFHATG}, \eqref{idalpha} and
$( \alpha \theta)^{-1}=1-\epsilon^2 \frac{1}{3}E^2\hmfu^{ij}\delta_{ij}+\epsilon^4\breve{\mathscr{S}}(\epsilon,t,E, \Omega/t, \hmfu^{\mu\nu})$; this establishes the second and third formulas from the statement of the lemma. Finally, we observe that last formula is a consequence of \eqref{E:G0MU}, \eqref{E:V^0} and the first three formulas.
\end{proof}

\subsection{Reformulation of the constraint equations}
We start the process of expressing the constraint equations \eqref{E:CONSTRAINT}-\eqref{E:NORMALIZATION}
in terms of the variables  \eqref{E:DEFHATG} by
noting that
\begin{align}\label{e:gtherel}
\bar{g}_{\lambda\sigma}\udn{\nu}\bar{g}^{\lambda\sigma}=-\frac{2}{\theta}\udn{\nu } \theta \AND
\hat{g}_{\lambda\sigma}\udn{\nu}\hat{g}^{\lambda\sigma}=\frac{2}{\theta}\udn{\nu}\theta
\end{align}
where $(\hat{g}_{\lambda\sigma})=(\hat{g}^{\alpha\beta})^{-1}$. Using this, we can express the
vector fields $\bar{X}^\mu$ and $\bar{Y}^\mu$, recall $\bar{Z}^\mu = \bar{X}^\mu+\bar{Y}^\mu$ by
\eqref{E:WAVEGAUGE}, in terms of the variables \eqref{E:DEFHATG} by
\begin{align*}
\bar{X}^\mu&  =-\udn{\nu}\bar{g}^{\mu\nu}
+\frac{1}{2}\bar{g}^{\mu\nu}\bar{g}_{\alpha\beta}\udn{\nu}\bar{g}^{\alpha\beta} = -\udn{\nu}\bar{g}^{\mu\nu}
-\bar{g}^{\mu\nu}\frac{1}{\sqrt{|\bar{g}|}} \bar{\partial}_\nu \sqrt{|\bar{g}|}+\bar{g}^{\mu\nu}\bar{\gamma}^\alpha_{\nu\alpha}=-\frac{1}{\theta}\udn{\nu}\hat{g}^{\mu\nu} =-\epsilon^2\frac{1}{\theta}\udn{\nu}\hat{\mfu}^{\mu\nu}
\intertext{and}
\bar{Y}^\mu&=-2\bar{\nabla}^\mu\Psi+ \frac{2\Lambda}{3t} \delta^\mu_0 =-2(\bar{g}^{\mu\nu}-\bar{h}^{\mu\nu})\bar{\nabla}_\nu\Psi=-2  \bar{\nabla}^\mu\Psi+2\udn{}^\mu\Psi =\frac{2}{t}\left( \bar{g}^{\mu 0}+\frac{\Lambda}{3}\delta^\mu_0\right),
\end{align*}
respectively, which in turn, allows us to express the gauge constraint equations \eqref{E:WAVECONSTRAINT} as
\begin{equation}
\udn{0}\hat{\mfu}^{\mu0} =-\udn{i}\hat{\mfu}^{\mu i}+\frac{2}{t}\left(\hmfu^{0\mu}+\frac{\Lambda}{3}\frac{\theta-1}{\epsilon^2}\delta^\mu_0\right)=-  \bar{\partial}_i\hat{\mfu}^{\mu i}- \bar{\gamma}^\mu_{i\lambda}\hat{\mfu}^{\lambda i}- \bar{\gamma}^i_{i \lambda }\hat{\mfu}^{\lambda\mu}+\frac{2}{t}\left(\hmfu^{0\mu}+\frac{\Lambda}{3}\frac{\theta-1}{\epsilon^2}\delta^\mu_0\right). \label{e:Du0mu}
\end{equation}
 Using \eqref{e:the2}, it not difficult to verify that \eqref{e:Du0mu} is equivalent to the pair of equations
\begin{align}
\del{0} \underline{\hmfu^{00}}= & -\frac{1}{\epsilon}\del{i}\underline{\hmfu^{0i}}+\frac{1}{t}(1-3\Omega)\underline{\hmfu^{00}}+ \frac{\Lambda}{3t} E^2(1-\Omega) \delta_{ij} \underline{\hat{\mfu}^{j i}}  +\epsilon^2 \breve{\mathscr{S}}(\epsilon,t, E, \underline{\hmfu^{\mu\nu}}),  \label{e:D0u00}   \\
\del{0} \underline{\hmfu^{k0}}=& - \frac{1}{\epsilon} \partial _i \underline{\hat{\mfu}^{ k i} } +\frac{1}{t}(2-5\Omega) \underline{\hmfu^{0 k} },  \label{e:D0uk0}
\end{align}
where $\breve{\mathscr{S}}$ vanishes
to second order in $ \underline{\hmfu^{\mu\nu}} $.

The importance of the equations \eqref{e:D0u00}-\eqref{e:D0uk0} is that they allow us to determine the time derivatives $\del{0} \underline{\hmfu^{\mu 0}}$ from metric variables $\underline{\hmfu^{\mu\nu}}$ and their spatial
derivatives on the initial hypersurface $\Sigma$.  As an application, we see after taking the time derivative of \eqref{e:the1}
and then using \eqref{e:D0u00} to replace $\del{0} \underline{\hmfu^{00}}$
with the right hand side of \eqref{e:D0u00} that
\begin{align*}
\del{t}\biggl(\frac{\underline{\theta}-1}{\epsilon^2}\biggr)=& \frac{3}{2\Lambda}\frac{1}{\epsilon}\del{i} \underline{\hmfu^{0i}}-\frac{3}{2\Lambda t}(1-3\Omega)\underline{\hmfu^{00}}+\breve{\mathscr{L}}(\epsilon,t,E,\Omega/t,\underline{\hmfu^{kl}},\underline{\hmfu^{ij}_0})+\epsilon \breve{\mathscr{B}}(\epsilon,t, E, \Omega/t, \underline{\hmfu^{\mu\nu}},D\underline{\hmfu^{\lambda\sigma}}) \nnb  \\
&\hspace{6cm}+ \epsilon^2 \breve{\mathscr{R}} (\epsilon, t, E, \Omega/t, \underline{\hmfu^{\mu\nu}},\underline{\hmfu^{ij}_0})+ \epsilon^2\breve{\mathscr{S}}(\epsilon,t, E, \underline{\hmfu^{\mu\nu}}), 
\end{align*}
where $\breve{\mathscr{L}}$ is linear in $(\underline{\hmfu^{kl}},\underline{\hmfu^{ij}_0})$, $\breve{\mathscr{S}}$ vanishes
to second order in $\underline{\hmfu^{\mu\nu}}$, and  $\breve{\mathscr{R}}$ and  $\breve{\mathscr{B}}$ both
vanish to first order in  $\underline{\hmfu^{\mu\nu}}$ and are linear in $\underline{\hmfu^{ij}_0}$ and $D\underline{\hmfu^{\lambda\sigma}}$, respectively. Furthermore, differentiating \eqref{e:D0u00}-\eqref{e:D0uk0}
with respect to $t$, we find, after using \eqref{e:D0u00}-\eqref{e:D0uk0} to replace the time
derivatives $\del{0} \underline{\hmfu^{\mu 0}}$, that the second time derivatives $\del{}_0^2 \underline{\hmfu^{\mu 0}}$
can be, on the initial hypersurface, expressed as
\begin{align}
\partial_0^2\underline{\hmfu^{00}}=& \frac{1}{\epsilon^2}\del{i}\del{j}\underline{\hmfu^{ij}}+\frac{1}{\epsilon}\frac{1}{t}(8\Omega-3)\del{i}\underline{\hmfu^{0i}}+\frac{1}{t^2}(9\Omega^2-6\Omega-3t\del{t}\Omega+1)\underline{\hmfu^{00}}  +  \breve{\mathscr{L}}(\epsilon,t,E,\Omega/t,\underline{\hmfu^{kl}},\underline{\hmfu^{kl}_0})\nnb  \\
& \hspace{3cm}+\epsilon \breve{\mathscr{B}}(\epsilon,t, E, \Omega/t, \underline{\hmfu^{\mu\nu}},D\underline{\hmfu^{\lambda\sigma} }) + \epsilon^2 \breve{\mathscr{R}} (\epsilon, t, E, \Omega/t, \underline{\hmfu^{\mu\nu}},\underline{\hmfu^{kl}_0})+ \epsilon^2\breve{\mathscr{S}}(\epsilon,t, E, \underline{\hmfu^{\mu\nu}}), \label{e:dt2u00} \\
\partial_0^2\underline{\hmfu^{j0}} = & -\frac{1}{\epsilon}\del{i}\underline{\hmfu^{ij}_0}-\frac{1}{\epsilon}\frac{1}{t}(2-5\Omega)\del{i}\underline{\hmfu^{ij}}+\frac{1}{t^2}(25\Omega^2-15\Omega-5t\del{t}\Omega+2)\underline{\hmfu^{0j}},  \label{e:dt2u0j}
\end{align}
where $\breve{\mathscr{L}}$, $\breve{\mathscr{B}}$, $\breve{\mathscr{R}}$ and $\breve{\mathscr{S}}$ are defined as above.

With the reformulation of the gauge constraints complete, we turn our attention to the gravitational constraint equations \eqref{E:CONSTRAINT}. We begin the reformulation process by observing the Ricci scalar $\bar{R}$ is given by
\begin{align}\label{e:scalarR2}
\bar{R}=\bar{g}_{\mu\nu}\bar{R}^{\mu\nu}\overset{\eqref{e:ricci1}}{=}\frac{1}{2}\bar{g}_{\mu\nu}\bar{g}^{\alpha\beta}\udn{\alpha}\udn{\beta}\bar{g}^{\mu\nu}+\bar{\nabla}_\lambda\bar{X}^\lambda+\bar{g}_{\mu\nu}\bar{\mathcal{R}}^{\mu\nu}+\bar{g}_{\mu\nu}\bar{P}^{\mu\nu}+\bar{g}_{\mu\nu}\bar{Q}^{\mu\nu}.
\end{align}
Using \eqref{e:ricci1}, \eqref{e:gtherel} and \eqref{e:scalarR2} in conjunction with the identities
\begin{align}
\udn{\lambda}\bar{g}^{\alpha\beta}= & \frac{1}{\theta}\udn{\lambda}\hat{g}^{\alpha\beta}-\frac{1}{2\theta}\hat{g}^{\alpha\beta}\hat{g}_{\mu\sigma}\udn{\lambda}\hat{g}^{\mu\sigma},  \label{e:udngup}
\intertext{and}
\udn{\lambda}\bar{g}_{\alpha\beta}= &  \theta \udn{\lambda}\hat{g}_{\alpha\beta}+\frac{1}{2}\theta\hat{g}_{\alpha\beta}\hat{g}_{\mu\sigma}\udn{\lambda}\hat{g}^{\mu\sigma}= -\theta \hat{g}_{\alpha\mu}\hat{g}_{\beta\nu}\udn{\lambda}\hat{g}^{\mu\nu}+\frac{1}{2}\theta\hat{g}_{\alpha\beta}\hat{g}_{\mu\sigma}\udn{\lambda}\hat{g}^{\mu\sigma}, \label{e:udngdown}
\end{align}
which follow from \eqref{e:gtherel} and relation $-\udn{\lambda}\hat{g}_{\alpha\beta}=\hat{g}_{\alpha\mu}\hat{g}_{\beta\nu}\udn{\lambda}\hat{g}^{\mu\nu}$,
we see that the Einstein tensor is given by
\begin{align}\label{e:eintens2}
\bar{G}^{\mu\nu}=& 
 \frac{1}{2\theta^2}\hat{g}^{\alpha\beta}\udn{\alpha}\udn{\beta}\hat{g}^{\mu\nu}  +\bar{\nabla}^{(\mu}\bar{X}^{\nu)}-\frac{1}{2\theta}\hat{g}^{\mu\nu}
\bar{\nabla}_\lambda\bar{X}^\lambda+\bar{\mathcal{R}}^{\mu\nu}-\frac{1}{2}\bar{h}^{\mu\nu} \bar{\mathcal{R}} +\tilde{\mathcal{P}}^{\mu\nu}+ \tilde{\mathcal{Q}}^{\mu\nu}-\frac{1}{2}  \bar{X}^\mu\bar{X}^\nu
\end{align}
where
\begin{align}
\tilde{\mathcal{P}}^{\mu\nu}
=&\frac{1}{2}\bar{h}^{\mu\nu}\bar{h}_{\alpha\beta}\bar{\mathcal{R}}^{\alpha\beta}-\frac{1}{2}\hat{g}^{\mu\nu}\hat{g}_{\alpha\beta}\bar{\mathcal{R}}^{\alpha\beta}+ \bar{P}^{\mu\nu}-\frac{1}{2}\hat{g}^{\mu\nu}\hat{g}_{\alpha\beta}\bar{P}^{\alpha\beta}, \label{e:curlP1}\\
\tilde{\mathcal{Q}}^{\mu\nu}
=&\frac{1}{8}\bar{g}^{\alpha\beta}\bar{g}_{\lambda\sigma}\bar{g}^{\mu\nu}\bar{g}_{\gamma\delta}\udn{\beta}\bar{g}^{\lambda\sigma}\udn{\alpha}\bar{g}^{\gamma\delta}-\frac{1}{4}\bar{g}^{\alpha\beta}\bar{g}^{\mu\nu}\udn{\alpha}\bar{g}_{\lambda\sigma}\udn{\beta}\bar{g}^{\lambda\sigma}+\bar{Q}^{\mu\nu}-\frac{1}{2}\hat{g}^{\mu\nu}\hat{g}_{\alpha\beta}
\bar{Q}^{\alpha\beta}+\frac{1}{2}  \bar{X}^\mu\bar{X}^\nu,  \label{e:curlQ1}
\end{align}
and $\bar{P}^{\mu\nu}$ and $\bar{Q}^{\mu\nu}$ are defined previously by \eqref{E:P} and \eqref{E:Q}, respectively.

To proceed, we use \eqref{e:udngup} and \eqref{e:udngdown} to express $\udn{\lambda}\bar{g}^{\alpha\beta}$ and $\udn{\lambda}\bar{g}_{\alpha\beta}$ in \eqref{e:curlQ1} in terms of $\udn{\lambda}\hat{g}^{\mu\nu}$ followed by replacing $\hat{g}^{\mu\nu}$ with $\hmfu^{\mu\nu}$ using  \eqref{E:DEFHATG}. This allows us to write $\tilde{\mathcal{Q}}^{\mu\nu}$ as
\begin{equation*}
\tilde{\mathcal{Q}}^{\mu\nu}=\epsilon^2\breve{\mathscr{W}}^{\mu\nu}(\epsilon,t,E,\Omega/t, \hmfu^{\lambda\sigma}, D \hmfu^{\alpha\beta}, \hmfu^{ij}_0)
\end{equation*}
where
\begin{align}
\breve{\mathscr{W}}^{\mu\nu}(\epsilon,t,E,\Omega/t, \hmfu^{\mu\nu}, D \hmfu^{\alpha\beta}, \hmfu^{ij}_0)=&\epsilon^2 \breve{\mathscr{S}}^{\mu\nu} (\epsilon,t,E,\Omega/t,   \hmfu^{\alpha\beta})+ \epsilon  \breve{\mathscr{R}}^{\mu\nu} (\epsilon,t,E,\Omega/t,  \epsilon\hmfu^{\lambda\sigma},  \epsilon\hmfu^{ij}_0, D\hmfu^{\alpha\beta},\hmfu^{ij}_0) \nnb  \\
&\hspace{0.8cm} + \breve{\mathscr{Q}}^{\mu\nu} (\epsilon,t,E,\Omega/t, \hmfu^{\lambda\sigma},  D\hmfu^{\alpha\beta}) +\epsilon
\breve{\mathscr{B}}^{\mu\nu}(\epsilon,t,E,\Omega/t,  \hmfu^{\alpha\beta},  D\hmfu^{\lambda\sigma}),  \label{e:Qdecp}
\end{align}
and in this expression,   $\breve{\mathscr{S}}^{\mu\nu}$ vanishes
to second order in $ \hmfu^{\alpha\beta} $, $\breve{\mathscr{Q}}^{\mu\nu}$ vanishes to second order in $D\hmfu^{\lambda\sigma}$, $\breve{\mathscr{R}}^{\mu\nu}$vanishes to first order in $(\epsilon\hmfu^{\lambda\sigma},  \epsilon\hmfu^{ij}_0, D\hmfu^{\mu\nu})$ and is linear in $\hmfu^{ij}_0$, and $\breve{\mathscr{B}}^{\mu\nu}$
vanishes to first order in $\hmfu^{\alpha\beta}$ and is linear in $D\hmfu^{\lambda\sigma}$.

\begin{remark}
For the remainder of this section, we will use the following notation unless otherwise stated:  $\breve{\mathscr{L}}(\epsilon,t,E,\Omega/t,   \hmfu^{kl})$ and $\breve{\mathscr{L}}^j(\epsilon,t,E,\Omega/t,   \hmfu^{kl})$ will denote remainder terms that are linear in $ \hmfu^{kl} $, while $\breve{\mathscr{S}}(\epsilon,t,E,\Omega/t,   \hmfu^{\alpha\beta})$, $\breve{\mathscr{S}}^j(\epsilon,t,E,\Omega/t,   \hmfu^{\alpha\beta})$, $\breve{\mathscr{S}}^{\mu\nu}(\epsilon,t,E,\Omega/t,   \hmfu^{\alpha\beta})$ and $\breve{\mathscr{S}}_{\alpha\beta}(\epsilon,t,E,\Omega/t,   \hmfu^{\alpha\beta})$ will denote remainder
terms that vanish to second order in $ \hmfu^{\alpha\beta} $.
\end{remark}

Next, we express $\bar{P}^{\mu\nu}$, see \eqref{E:P}, in terms of $\hmfu^{\mu\nu}$ by using the expansion
\begin{align*}
\bar{g}^{\mu\lambda}-\bar{h}^{\mu\lambda}=&\hat{g}^{\mu\lambda}-\bar{h}^{\mu\lambda}+\bar{h}^{\mu\lambda}\left(\frac{1}{\theta}-1\right)+(\hat{g}^{\mu\lambda}-\bar{h}^{\mu\lambda})\left(\frac{1}{\theta}-1\right)  \\
=&\left[\epsilon^2\hmfu^{\mu\lambda}+\epsilon^2\frac{3}{2\Lambda}\bar{h}^{\mu\lambda} \hmfu^{00}\right] +\left[ \bar{h}^{\mu\lambda}\left(\frac{1}{\theta}-1- \frac{3}{2\Lambda}\epsilon^2 \hmfu^{00}\right)+(\hat{g}^{\mu\lambda}-\bar{h}^{\mu\lambda})\left(\frac{1}{\theta}-1\right)\right],
\end{align*}
together with \eqref{e:the1} to get
\begin{align}
\bar{P}^{\mu\nu}
=& -\frac{1}{2}  \left[\epsilon^2\hmfu^{\mu\lambda}+\epsilon^2\frac{3}{2\Lambda}\bar{h}^{\mu\lambda} \hmfu^{00}\right] \bar{h}^{\alpha\beta}\tensor{\mathcal{\bar{R}}}{_{\lambda \alpha\beta}^\nu} -\frac{1}{2} \left[\epsilon^2\hmfu^{\alpha\beta}+\epsilon^2\frac{3}{2\Lambda}\bar{h}^{\alpha\beta} \hmfu^{00}\right]\bar{h}^{\mu \lambda}\tensor{\mathcal{\bar{R}}}{_{\lambda \alpha\beta}^\nu}   \nnb \\
&-\frac{1}{2}\left[\epsilon^2\hmfu^{\nu\lambda}+\epsilon^2\frac{3}{2\Lambda}\bar{h}^{\nu\lambda} \hmfu^{00}\right]\bar{h}^{\alpha\beta}\tensor{\mathcal{\bar{R}}}{_{\lambda\alpha\beta}^\mu}-\frac{1}{2}\left[\epsilon^2\hmfu^{\alpha\beta}+\epsilon^2\frac{3}{2\Lambda}\bar{h}^{\alpha\beta} \hmfu^{00}\right]h^{\nu \lambda}  \tensor{\mathcal{\bar{R}}}{_{\lambda\alpha\beta}^\mu} \nnb  \\
& +\epsilon^2 \breve{\mathscr{L}}^{\mu\nu}(\epsilon,t,E,\Omega/t,   \hmfu^{kl})+
\epsilon^4\breve{\mathscr{S}}^{\mu\nu}(\epsilon,t,E,\Omega/t,   \hmfu^{\alpha\beta}),\nnb
\end{align}
which, with the help of \eqref{e:Hriem1}-\eqref{e:Hriem}, we can write as
\begin{align}
\bar{P}^{ 0 0}
=&  \epsilon^2 \breve{\mathscr{L}} (\epsilon,t,E,\Omega/t,   \hmfu^{ij})+\epsilon^4\breve{\mathscr{S}}(\epsilon,t,E,\Omega/t,   \hmfu^{\mu\nu}),   \label{e:P3} \\
\bar{P}^{ j 0}
=& - \epsilon^2   \frac{\Lambda}{3t^2} (\Omega-2 \Omega^2-t\del{t}\Omega) \hmfu^{ j0}  +\epsilon^2 \breve{\mathscr{L}}^j (\epsilon,t,E,\Omega/t,   \hmfu^{kl})+\epsilon^4\breve{\mathscr{S}}^j(\epsilon,t,E,\Omega/t,   \hmfu^{\mu\nu})  \label{e:P4} \\
\intertext{and}
\bar{P}^{ i j}
=& \epsilon^2 E^{-2}  \frac{2}{t^2} \Omega^2 \delta^{ i j} \hmfu^{00}   +\epsilon^2 \breve{\mathscr{L}}^{ij}(\epsilon,t,E,\Omega/t,   \hmfu^{kl})+\epsilon^4\breve{\mathscr{S}}^{ij}(\epsilon,t,E,\Omega/t,   \hmfu^{\mu\nu}).  \label{e:P5}
\end{align}

By Lemma \ref{t:expinv} in Appendix \ref{s:matr}, we see that
\begin{align*}
\hat{g}_{\alpha\beta}-\bar{h}_{\alpha\beta}=-\epsilon^2\bar{h}_{\alpha\lambda}\hmfu^{\lambda\sigma}\bar{h}_{\sigma\beta}+\epsilon^4\breve{\mathscr{S}}_{\alpha\beta}(\epsilon,t,E,\Omega/t,   \hmfu^{\mu\nu}),
\end{align*}
which together with \eqref{E:DEFHATG} gives
\begin{align}
-\frac{1}{2}\hat{g}^{\mu 0}\hat{g}_{\alpha\beta}\bar{P}^{\alpha\beta}=&-\frac{1}{2}(\hat{g}^{\mu 0}-\bar{h}^{\mu 0})(\hat{g}_{\alpha\beta}-\bar{h}_{\alpha\beta})\bar{P}^{\alpha\beta}-\frac{1}{2}(\hat{g}^{\mu 0}-\bar{h}^{\mu 0}) \bar{h}_{\alpha\beta} \bar{P}^{\alpha\beta}-\frac{1}{2} \bar{h}^{\mu 0}( \hat{g}_{\alpha\beta}-\bar{h}_{\alpha\beta})\bar{P}^{\alpha\beta}-\frac{1}{2} \bar{h}^{\mu 0} \bar{h}_{\alpha\beta}  \bar{P}^{\alpha\beta} \nnb \\=& -\frac{1}{2} \bar{h}^{\mu 0} \bar{h}_{\alpha\beta}  \bar{P}^{\alpha\beta} +\epsilon^4 \breve{\mathscr{S}}^{0\mu}(\epsilon,t,E,\Omega/t,   \hmfu^{\mu\nu}) .\label{e:trP1}
\end{align}
Further, we observe that
\begin{align}
\bar{h}_{\alpha\beta}\bar{P}^{\alpha\beta} =  \bar{h}_{00}\bar{P}^{00}+ \bar{h}_{ij}\bar{P}^{ij}
= \epsilon^2 \frac{6}{t^2}\Omega\hmfu^{00}+\epsilon^2 \breve{\mathscr{L}} (\epsilon,t,E,\Omega/t,   \hmfu^{ij})+\epsilon^4\breve{\mathscr{S}}(\epsilon,t,E,\Omega/t,   \hmfu^{\mu\nu}) \label{e:trP2}
\end{align}
by \eqref{e:P3}, \eqref{e:P5}. Then recalling the definition \eqref{e:curlP1} of $\tilde{\mathcal{P}}^{\mu\nu}$, we can, with the help of \eqref{e:trP1},
write $\tilde{\mathcal{P}}^{\mu 0}$ as
\begin{align*}
\tilde{\mathcal{P}}^{\mu 0}
= & -\frac{1}{2}\mathcal{\bar{R}}^{\alpha\beta}[(\hat{g}^{\mu 0}-\bar{h}^{\mu 0})\bar{h}_{\alpha\beta}+(\hat{g}^{\mu 0}-\bar{h}^{\mu 0})(\hat{g}_{\alpha\beta}-\bar{h}_{\alpha\beta})+\bar{h}^{\mu 0}(\hat{g}_{\alpha\beta}-\bar{h}_{\alpha\beta})] + \bar{P}^{\mu 0}-\frac{1}{2}\hat{g}^{\mu 0}\hat{g}_{\alpha\beta}\bar{P}^{\alpha\beta} \nnb  \\
= & -\frac{1}{2}\mathcal{\bar{R}}^{\alpha\beta}[\epsilon^2 \hmfu^{\mu 0} \bar{h}_{\alpha\beta} -\epsilon^2\bar{h}^{\mu 0} \bar{h}_{\alpha\lambda}\hmfu^{\lambda\sigma}\bar{h}_{\sigma\beta} ] + \bar{P}^{\mu 0} -\frac{1}{2}\bar{h}^{\mu 0}\bar{h}_{\alpha\beta}\bar{P}^{\alpha\beta}+\epsilon^4\breve{\mathscr{S}}^{\mu 0} (\epsilon,t,E,\Omega/t,   \hmfu^{\mu\nu}) 
\end{align*}
from which we see, by \eqref{e:P3} and \eqref{e:trP2}, that the $\mu=0$ and $\mu=j$ components of $\mathcal{P}^{\mu0}$ can be expressed as
\begin{align}
\tilde{\mathcal{P}}^{ 0 0}
= & \epsilon^2  \frac{ \Lambda}{2 t^2}(3\Omega-3\Omega^2-t\del{t}\Omega) \hmfu^{0 0}  +\epsilon^2 \breve{\mathscr{L}} (\epsilon,t,E,\Omega/t,   \hmfu^{ij})+\epsilon^4\breve{\mathscr{S}}(\epsilon,t,E,\Omega/t,   \hmfu^{\mu\nu}) \label{e:curP0}
\intertext{and}
\tilde{\mathcal{P}}^{j 0}
= & \epsilon^2    \frac{2 \Lambda}{3t^2}(\Omega-2\Omega^2-t\del{t}\Omega)  \hmfu^{j 0}  +\epsilon^2 \breve{\mathscr{L}}^j (\epsilon,t,E,\Omega/t,   \hmfu^{kl})+\epsilon^4\breve{\mathscr{S}}^j(\epsilon,t,E,\Omega/t,   \hmfu^{\alpha\beta}),  \label{e:curPj}
\end{align}
respectively.

On the initial hypersurface $\Sigma$, we know from the above arguments that we can satisfy the constraint equations $\bar{Z}^\mu=0$
by choosing $\del{0}\underline{\hmfu^{j0}}$ according to \eqref{e:D0u00} and \eqref{e:D0uk0}. Doing so,
we find using \eqref{e:eintens2} that we can write the conformal Einstein equations \eqref{E:CONFORMALEINSTEIN1} as
\begin{align}
&\frac{1}{2\theta^2}\hat{g}^{\alpha\beta}\udn{\alpha}\udn{\beta}\hat{g}^{\mu\nu}  -\bar{\nabla}^{(\mu}\bar{Y}^{\nu)}+\frac{1}{2\theta}\hat{g}^{\mu\nu}
\bar{\nabla}_\lambda\bar{Y}^\lambda+\bar{\mathcal{R}}^{\mu\nu}-\frac{1}{2}\bar{h}^{\mu\nu} \bar{\mathcal{R}} +\tilde{\mathcal{P}}^{\mu\nu}+ \tilde{\mathcal{Q}}^{\mu\nu}   -\frac{1}{2} \bar{Y}^\mu\bar{Y}^\nu  \nnb \\
& \hspace{3cm} =e^{4\Psi}\tilde{T}^{\mu\nu}-\frac{1}{\theta} e^{2\Psi}\Lambda\hat{g}^{\mu\nu}
+2(\bar{\nabla}^\mu\bar{\nabla}^\nu\Psi-\bar{\nabla}^\mu\Psi\bar{\nabla}^\nu\Psi)
-\frac{1}{\theta}(2\bar{\Box}\Psi+|\bar{\nabla}\Psi|^2_{\bar{g}})
\hat{g}^{\mu\nu}. \label{e:ein1}
\end{align}
We also note that conformal Einstein equations for the conformal FLRW metric \eqref{E:CONFORMALFLRW} are given by
\begin{align}
&\bar{\mathcal{R}}^{\mu\nu}-\frac{1}{2} \bar{\mathcal{R}}\bar{h}^{\mu\nu}=e^{4\Psi}\tilde{\mathcal{T}}^{\mu\nu}- e^{2\Psi}\Lambda\bar{h}^{\mu\nu}
+2(\udn{}^\mu\udn{}^\nu\Psi-\udn{}^\mu\Psi\udn{}^\nu\Psi)
- (2\bb\Psi+|\udn{}\Psi|^2_{\bar{h}})
\bar{h}^{\mu\nu} \nnb \\
=& (1+\epsilon^2K)\mu\frac{\Lambda}{3}\delta^\mu_0\delta^\nu_0 e^{2\Psi}+\epsilon^2K \mu\bar{h}^{\mu\nu}e^{2\Psi}- e^{2\Psi}\Lambda\bar{h}^{\mu\nu}
+2(\udn{}^\mu\udn{}^\nu\Psi-\udn{}^\mu\Psi\udn{}^\nu\Psi)
- (2\bb\Psi+|\udn{}\Psi|^2_{\bar{h}})
\bar{h}^{\mu\nu}.  \label{e:homeni1}
\end{align}

In order to expand \eqref{e:ein1} further, we list some key calculations below. First, with the help of \eqref{E:X}, \eqref{e:Y} and Proposition \ref{wgprop},
we see from a direct calculation that
\begin{align}\label{e:NabY}
\bar{\nabla}_\lambda\bar{Y}^\lambda=&-2\bar{\Box}\Psi+2\bb\Psi +2\tensor{\bar{X}}{^\lambda_{ \lambda\sigma}} \udn{}^\sigma\Psi 
=\frac{2}{t^2}\left(\bar{g}^{00}+\frac{\Lambda}{3}\right)+\frac{2\Lambda}{3t}\tensor{\bar{X}}{^\lambda_{\lambda0}},
\end{align}
where, using \eqref{e:gtherel}, we note $\tensor{\bar{X}}{^\lambda_{\lambda 0}}$ can be expressed as
\begin{align*}
\tensor{\bar{X}}{^\lambda_{\lambda 0}}
=&-\frac{1}{2}\bigl(\bar{g}_{\sigma 0}\udn{\lambda}\bar{g}^{\lambda \sigma}+\bar{g}_{\lambda \sigma}\udn{0}\bar{g}^{\lambda \sigma}-\bar{g}^{\lambda \sigma}\bar{g}_{\lambda \delta}\bar{g}_{0 \gamma}\udn{\sigma}\bar{g}^{\delta \gamma}\bigr)
= \frac{1}{2}\hat{g}_{\lambda\sigma}\udn{0}\hat{g}^{\lambda\sigma}=\frac{1}{\theta}\udn{0}\theta.
\end{align*}
Using Proposition \ref{wgprop} again, we see that
\begin{align}\label{e:Boxphi}
2\bar{\Box}\Psi+|\bar{\nabla}\Psi|^2_{\bar{g}} =  
\frac{3}{t^2}\bar{g}^{00}-\frac{4}{t^2}  \left( \bar{g}^{0 0}+\frac{\Lambda}{3} \right) + \frac{2\Lambda\Omega}{t^2}  \AND
2\bb\Psi+|\udn{}\Psi|^2_{\bar{h}} =  
-\frac{ \Lambda}{ t^2} + \frac{2\Lambda\Omega}{t^2},
\end{align}
which together can be used to show that
\begin{align}\label{e:Boxphi2}
&-2\theta  (2\bar{\Box}\Psi+|\bar{\nabla}\Psi|^2_{\bar{g}})
\hat{g}^{\mu\nu}
+ 2\theta^2 (2\bb\Psi+|\udn{}\Psi|^2_{\bar{h}})
\bar{h}^{\mu\nu} =   \frac{2\Lambda}{t^2}(\hat{g}^{\mu\nu}-\bar{h}^{\mu\nu}) +\frac{2\Lambda}{t^2}(1-\theta^2)\bar{h}^{\mu\nu}  \nnb \\
& \hspace{2.5cm} +\frac{2}{t^2}\left(\hat{g}^{00}+\frac{\Lambda}{3}\right)\hat{g}^{\mu\nu}-\frac{8\Lambda}{3t^2}(1-\theta)\hat{g}^{\mu\nu} -\frac{4\Lambda\Omega}{t^2}\theta(\hat{g}^{\mu\nu}-\bar{h}^{\mu\nu})-\frac{4\Lambda\Omega}{t^2}\theta(1-\theta)\bar{h}^{\mu\nu}. 
\end{align}
Furthermore, by direct calculation, it is not difficult to verify that
\begin{align} \label{e:rhsubr}
\bar{\rho}\bar{g}^{\mu\nu}-\mu\bar{h}^{\mu\nu}=(\bar{\rho}-\mu)\frac{1}{\theta}\hat{g}^{\mu\nu}+\mu\bar{h}^{\mu\nu}\biggl(\frac{1}{\theta}-1\biggr)(1-\theta)+\mu\frac{1}{\theta}(1-\theta)(\hat{g}^{\mu\nu}-\bar{h}^{\mu\nu})+\mu(\hat{g}^{\mu\nu}-\bar{h}^{\mu\nu})+\mu(1-\theta)\bar{h}^{\mu\nu}.
\end{align}

Inserting \eqref{e:homeni1}--\eqref{e:rhsubr} into \eqref{e:ein1} yields the following representation of conformal Einstein equations:
\begin{align}
&\hat{g}^{\alpha\beta}\udn{\alpha}\udn{\beta}(\hat{g}^{\mu\nu}-\bar{h}^{\mu\nu})  +\theta \hat{g}^{\mu\nu}
\frac{2}{t^2}\left(\bar{g}^{00}+\frac{\Lambda}{3}\right)  +2\theta^2\tilde{\mathcal{P}}^{\mu\nu}+ 2\theta^2\tilde{\mathcal{Q}}^{\mu\nu}    \nnb \\
= & - \theta \frac{2 \Lambda}{3t} \delta^\sigma_0 \udn{\sigma}(\hat{g}^{\mu\nu}-\bar{h}^{\mu\nu})  +\theta \frac{4\Lambda}{3t^2} \left(\hat{g}^{0\lambda}+\frac{\Lambda}{3}\delta^\lambda_0\right)\delta^{(\mu}_\lambda\delta^{\nu)}_0  + \theta(\theta-1) \frac{4\Lambda^2}{9t^2} \delta^\lambda_0 \delta^{(\mu}_\lambda \delta^{\nu)}_0 \nnb  \\
& -  \theta  \frac{4\Lambda}{3t^2}\Omega(\theta \bar{h}^{ij} \delta^\mu_j \delta^\nu_i   -   \hat{g}^{ij}\delta_j^{(\mu}\delta^{\nu)}_i) + \theta \frac{4 \Lambda}{3t^2} \Omega \hat{g}^{i0}\delta_0^{(\mu}\delta^{\nu)}_i
+2\theta^2(1+\epsilon^2K)\frac{1}{t^2}\biggl[(\bar{\rho}-\mu)\bar{v}^\mu\bar{v}^\nu + \mu(\bar{v}^\mu\bar{v}^\nu-\frac{\Lambda}{3}\delta^\mu_0\delta^\nu_0 ) \biggl] \nnb  \\
& +2\theta^2\epsilon^2 K \frac{1}{t^2} \biggl( (\bar{\rho}-\mu)\frac{1}{\theta}\hat{g}^{\mu\nu}+\mu\bar{h}^{\mu\nu}(\frac{1}{\theta}-1)(1-\theta)+\mu\frac{1}{\theta}(1-\theta)(\hat{g}^{\mu\nu}-\bar{h}^{\mu\nu})+\mu(\hat{g}^{\mu\nu}-\bar{h}^{\mu\nu})+\mu(1-\theta)\bar{h}^{\mu\nu} \biggr)  \nnb \\
& - \frac{2 \Lambda}{t^2}(\theta-1) (\hat{g}^{\mu\nu}-\bar{h}^{\mu\nu})  -\frac{2 \Lambda}{t^2} (1 -   \theta) \bar{h}^{\mu\nu}      +\frac{2}{t^2}\left(\hat{g}^{00}+\frac{\Lambda}{3}\right)\hat{g}^{\mu\nu}-\frac{8\Lambda}{3t^2}(1-\theta)\hat{g}^{\mu\nu} \nnb \\
&-\frac{4\Lambda\Omega}{t^2}\theta(\hat{g}^{\mu\nu}-\bar{h}^{\mu\nu})-\frac{4\Lambda\Omega}{t^2}\theta(1-\theta)\bar{h}^{\mu\nu}. \label{e:ein1.1}
\end{align}
Next, with the help of \eqref{e:D0uk0}-\eqref{e:dt2u0j}, we observe that the $(\mu,\nu)=(0,0)$ and $(\mu,\nu)= (j,0)$ components of the principal term
$\hat{g}^{\alpha\beta}\udn{\alpha}\udn{\beta}(\hat{g}^{\mu\nu}-\bar{h}^{\mu\nu})$
of \eqref{e:ein1.1} can be expressed as
\begin{align}
& \delta^0_\mu\delta^0_\nu\hat{g}^{\alpha\beta}\udn{\alpha}\udn{\beta}(\hat{g}^{\mu\nu}-\bar{h}^{\mu\nu}) =  \epsilon^2\delta^0_\mu\delta^0_\nu\hat{g}^{\alpha\beta}\udn{\alpha}\udn{\beta}\hmfu^{\mu\nu} \nnb  \\
= & E^{-2}\Delta\hmfu^{00}+ \epsilon^2 \hmfu^{ij}\del{i}\del{j}\hmfu^{00}+\epsilon^2\hmfu^{00}\del{i}\del{j}\hmfu^{ij}-2\epsilon^2\hmfu^{0i}\del{i}\del{j}\hmfu^{0j}-\frac{\Lambda}{3}\del{i}\del{j}\hmfu^{ij}+\epsilon\frac{\Lambda}{3t}(3-\Omega)\del{i}\hmfu^{0i} \nnb  \\
& +\epsilon^2 \frac{\Lambda}{3t^2}(6\Omega^2+3\Omega+3t\del{t}\Omega-1)\hmfu^{00}+\epsilon^2 \breve{\mathscr{L}}(\epsilon,t,E,\Omega/t, t\del{t}\Omega, \hmfu^{kl}, \hmfu^{kl}_0)+\epsilon^4\breve{\mathscr{S}} (\epsilon,t,E,\Omega/t, t\del{t}\Omega,  \hmfu^{\alpha\beta})\nnb  \\
& + \epsilon^3  \breve{\mathscr{R}} (\epsilon,t,E,\Omega/t, t\del{t}\Omega, \epsilon\hmfu^{\alpha\beta} ,\hmfu^{kl}_0)   +\epsilon^3 \breve{\mathscr{B}} (\epsilon,t,E,\Omega/t,
t\del{t}\Omega, \hmfu^{\alpha\beta},  D\hmfu^{\lambda\sigma})  \label{e:ddu00}
\end{align}
and
\begin{align}
& \delta^j_\mu\delta^0_\nu\hat{g}^{\alpha\beta}\udn{\alpha}\udn{\beta}(\hat{g}^{\mu\nu}-\bar{h}^{\mu\nu}) =  \epsilon^2\delta^j_\mu\delta^0_\nu\hat{g}^{\alpha\beta}\udn{\alpha}\udn{\beta}\hmfu^{\mu\nu} \nnb  \\
=& E^{-2}\Delta \hmfu^{j0}+ \epsilon \frac{\Lambda}{3}\del{i}\hmfu^{ij}_0+\epsilon^2\hmfu^{kl}\del{k}\del{l}\hmfu^{j0}-\epsilon^3 \hmfu^{00}\del{i}\hmfu^{ij}_0 -2\epsilon^2\hmfu^{0i}\del{i}\del{k} \hmfu^{jk} + 2 \epsilon E^{-2} \delta^{kl}\frac{\Omega}{t}\del{k}\hmfu^{00}  \nnb \\
& -\epsilon^2 \frac{\Lambda}{3t^2}(-2\Omega^2-6\Omega-4t\del{t}\Omega+2)\hmfu^{0j}+\epsilon^2 \breve{\mathscr{L}}^j(\epsilon,t,E,\Omega/t, t\del{t}\Omega, \hmfu^{kl},
\partial_k\hmfu^{kl})+\epsilon^4\breve{\mathscr{S}}^j(\epsilon,t,E,\Omega/t, t\del{t}\Omega,  \hmfu^{\alpha\beta})\nnb  \\
& + \epsilon^3  \breve{\mathscr{R}}^j(\epsilon,t,E,\Omega/t, t\del{t}\Omega, \epsilon\hmfu^{\alpha\beta},\hmfu^{kl}_0)  +\epsilon^3 \breve{\mathscr{B}}^j(\epsilon,t,E,\Omega/t, t\del{t}\Omega, \hmfu^{\alpha\beta}, D\hmfu^{\lambda\sigma}).    \label{e:ddu0j}
\end{align}
respectively, where $\breve{\mathscr{L}}$ and $\breve{\mathscr{L}}^j$ are linear in $(\hmfu^{kl}, \hmfu^{kl}_0)$ and $(\hmfu^{kl}, \del{k} \hmfu^{kl})$ respectively, $\breve{\mathscr{S}}$ and $\breve{\mathscr{S}}^j$ vanish
to second order in $ \hmfu^{\mu\nu} $,  $\breve{\mathscr{R}}$ and $\breve{\mathscr{R}}^j$ vanish to first order in $\epsilon\hmfu^{\alpha\beta}$ and are linear in $\hmfu^{ij}_0$, and $\breve{\mathscr{B}}$ and $\breve{\mathscr{B}}^j$
vanish to first order in $\hmfu^{\alpha\beta}$ and are linear in $D\hmfu^{\lambda\sigma}$.
Furthermore, we observe, using \eqref{E:HOMCHRIS}, \eqref{E:DEFHATG} and \eqref{e:D0u00}--\eqref{e:D0uk0} that
\begin{align}
-\delta^0_\mu\delta^0_\nu \frac{2 \Lambda}{3t} \delta^\sigma_0 \udn{\sigma}(\hat{g}^{\mu\nu}-\bar{h}^{\mu\nu}) =&
\epsilon\frac{2\Lambda}{3t} \del{i}\hmfu^{0i}-\epsilon^2\frac{2\Lambda}{3t^2}(1-3\Omega)\hmfu^{00}+\epsilon^2 \breve{\mathscr{L}}(\epsilon,t,E,\Omega/t,   \hmfu^{kl})+\epsilon^4\breve{\mathscr{S}}(\epsilon,t,E,\Omega/t,   \hmfu^{\mu\nu})   \label{e:du00}  \\
\intertext{and}
-\delta^j_\mu\delta^0_\nu \frac{2 \Lambda}{3t} \delta^\sigma_0 \udn{\sigma}(\hat{g}^{\mu\nu}-\bar{h}^{\mu\nu}) =&
 \epsilon \frac{2 \Lambda}{3t} \del{i}\hmfu^{ji}- \epsilon^2 \frac{2 \Lambda}{3t^2} (2-4\Omega)\hmfu^{0j}.  \label{e:du0j}
\end{align}

Since the gravitational constraint equations \eqref{E:CONSTRAINT} only involve the $(\mu,\nu)=(0,0)$ and $(\mu,\nu)=(j,0)$ components
of the conformal Einstein equations, we separate these out from
\eqref{e:ein1.1} to get
\begin{align}
&\delta^0_\mu\delta^0_\nu\hat{g}^{\alpha\beta}\udn{\alpha}\udn{\beta}(\hat{g}^{\mu\nu}-\bar{h}^{\mu\nu})  +2\theta^2\tilde{\mathcal{P}}^{ 0 0}+ 2\theta^2\tilde{\mathcal{Q}}^{ 0 0}
=  - \theta \delta^0_\mu\delta^0_\nu \frac{2 \Lambda}{3t} \delta^\sigma_0 \udn{\sigma}(\hat{g}^{\mu\nu}-\bar{h}^{\mu\nu})  +\theta \frac{4\Lambda}{3t^2} \left(\hat{g}^{00}+\frac{\Lambda}{3}\right)  + \theta(\theta-1) \frac{4\Lambda^2}{9t^2}   \nnb \\
&\hspace{0.5cm}
+2\theta^2(1+\epsilon^2K)\frac{1}{t^2}\biggl[(\bar{\rho}-\bar{ \mu})\bar{v}^ 0\bar{v}^ 0 + \bar{ \mu}(\bar{v}^ 0\bar{v}^ 0-\frac{\Lambda}{3} ) \biggl]  +2\theta \epsilon^2 K \frac{1}{t^2} (\bar{\rho}-\bar{ \mu}) \hat{g}^{ 0 0}+2\theta \epsilon^2 K \frac{1}{t^2} \mu\bar{h}^{ 0 0} (1-\theta)^2\nnb  \\
&
\hspace{1cm} +2\theta \epsilon^2 K \frac{1}{t^2} \mu (1-\theta)(\hat{g}^{ 0 0}-\bar{h}^{ 0 0})  +2\theta^2\epsilon^2 K \frac{1}{t^2} \mu(\hat{g}^{ 0 0}-\bar{h}^{ 0 0}) +2\theta^2\epsilon^2 K \frac{1}{t^2} \mu(1-\theta)\bar{h}^{ 0 0}   +\frac{4\Lambda}{ t^2}(\theta-1) \bar{h}^{ 0 0} \nnb \\
& \hspace{5cm} -\frac{4\Lambda\Omega}{t^2}\theta(\hat{g}^{ 0 0}-\bar{h}^{ 0 0})-\frac{4\Lambda\Omega}{t^2}\theta(1-\theta)\bar{h}^{ 0 0}  \label{e:Ein00}
\end{align}
and
\begin{align}
&\delta^j_\mu\delta^0_\nu\hat{g}^{\alpha\beta}\udn{\alpha}\udn{\beta}(\hat{g}^{ \mu\nu}-\bar{h}^{\mu\nu})  +2\theta^2\tilde{\mathcal{P}}^{ j 0}+ 2\theta^2\tilde{\mathcal{Q}}^{ j 0}
= - \theta \delta^j_\mu\delta^0_\nu \frac{2 \Lambda}{3t} \delta^\sigma_0 \udn{\sigma}(\hat{g}^{\mu\nu}-\bar{h}^{\mu\nu})  +\theta \frac{2\Lambda}{3t^2}(1-5\Omega)\hat{g}^{0j}  +2\theta \epsilon^2 K \frac{1}{t^2}   (\bar{\rho}-\mu) \hat{g}^{ j 0}   \nnb \\
& \hspace{0.5cm} +2\theta^2(1+\epsilon^2K)\frac{1}{t^2}\biggl[(\bar{\rho}-\mu)\bar{v}^ j\bar{v}^ 0 + \mu \bar{v}^ j\bar{v}^ 0  \biggl]   +2\epsilon^2 K \frac{1}{t^2} \mu \theta  (1-\theta)(\hat{g}^{ j 0}-\bar{h}^{ j 0})+2\theta^2\epsilon^2 K \frac{1}{t^2} \mu(\hat{g}^{ j 0}-\bar{h}^{ j 0}).     \label{e:Ein0j}
\end{align}
To continue, we introduce the following notation for the initial data:
\begin{gather}
\smfu^{ij}(\xv) = \frac{1}{\epsilon}\underline{\hmfu^{ij}}(1,\xv), \quad \smfu^{ij}_0(\xv)=\underline{\hmfu^{ij}_0}(1,\xv), \quad \smfu^{0\mu}(\xv)=\underline{\hmfu^{0\mu}}(1,\xv), \quad \smfu^{0\mu}_0(\xv)=\underline{\hmfu^{0\mu}_0}(1,\xv) \label{initvarsA}
\intertext{and}
\delta\breve{\rho}(\xv)= \delta\rho(1,\xv), \quad \breve{z}^j(\xv)=z^j(1,\xv), \label{initvarsB}
\end{gather}
where $\xv=(x^i)$.
We then find after a long, but straightforward calculation using \eqref{e:muomeg}, \eqref{E:DEFHATG}, \eqref{e:Qdecp},
\eqref{e:curP0}-\eqref{e:curPj},    \eqref{e:ddu00}-\eqref{e:du0j}, Lemma \ref{L:IDENTITY} and the expansion (which is
a  consequence of \eqref{e:vup0} and $\bar{v}^0=\bar{g}^{00}\bar{v}_0+\bar{g}^{0i}\bar{v}_i$)
\begin{align*}
\underline{\bar{v}^0}\underline{\bar{v}^0} =\frac{\Lambda}{3}- \epsilon^2 \frac{1}{2}\smfu^{00}+\epsilon^2 \breve{\mathscr{S}}(\epsilon, \smfu^{\mu\nu})+\epsilon^2\breve{\mathscr{F}}_1(\epsilon^2, \smfu^{\mu\nu}, \breve{z}_k)+\epsilon^2\breve{\mathscr{F}}_2 (\epsilon^2, \smfu^{\mu\nu}, \breve{z}_k)+\epsilon^2 \breve{\mathscr{L}}( \epsilon \smfu^{ij}),
\end{align*}
where the remainder terms $\breve{\mathscr{F}}_1$ and $\breve{\mathscr{F}}_2$ vanish to first and second order
in $\breve{z}_k$, respectively,
that \eqref{e:Ein00}-\eqref{e:Ein0j}, when written in terms of Newtonian coordinates, take
the following form on the initial hypersurface $\Sigma$:
\begin{align}
\Delta \smfu^{00}= & 
\epsilon^2 E^2(1)    \frac{\Lambda}{3} \bigl(7-6\Omega(1)\bigr) \smfu^{00} +\epsilon \frac{\Lambda}{3 }\bigl(\Omega(1)-1\bigr) E^2(1)  \del{i}\smfu^{0i}   \nnb \\ &\hspace{4cm}   + \epsilon E^2(1) \frac{\Lambda}{3} \del{i}\partial _j \smfu^{ i j} +   \frac{2\Lambda}{3} E^2(1) \delta\breve{\rho} +\epsilon^2 \texttt{A}^{0} (\epsilon, \smfu^{00},\smfu^{0k}, \breve{\xi}),  \label{E:constru1} \\
\Delta \smfu^{0j}
= & 2\Lambda \epsilon^2  E^2(1)  (1+\epsilon^2K)\bigl(\Omega(1) -2\bigr)\Omega(1) \smfu^{0j} - 2\epsilon  \Omega(1) \delta^{jl} \del{l}\smfu^{00}   \nnb  \\
& \hspace{4cm}
- \epsilon  E^2(1) \frac{\Lambda}{3}  \del{i} \smfu^{j i}_0  + 2 \epsilon E^2(1)   \sqrt{\frac{\Lambda}{3}}  \breve{\rho} \breve{z}^ j  +\epsilon^2 \texttt{A}^{j} (\epsilon, \smfu^{00},\smfu^{0k}, \breve{\xi})  \label{E:constru2}
\end{align}
where
\begin{align}
\texttt{A}^{0}(\epsilon, \smfu^{00},\smfu^{0j}, \breve{\xi})= & \epsilon E^2(1)\smfu^{00}\del{i}\del{j}\smfu^{ij}+\epsilon E^2(1) \smfu^{ij}\del{i}\del{j}\smfu^{00}+E^2(1) \del{i}\del{j}\bigl(\smfu^{i0}\smfu^{0j}\bigr)+ \breve{\mathscr{L}}_1(\epsilon,\smfu^{kl}_0, \epsilon\smfu^{kl})+ \epsilon^2 \breve{\mathscr{S}}_1 (\epsilon,  \smfu^{\alpha\beta})\nnb \\
& + \epsilon  \breve{\mathscr{R}}_1(\epsilon, \epsilon\smfu^{\mu\nu},  \epsilon\smfu^{kl}_0, D\smfu^{\alpha\beta},\smfu^{kl}_0) + \breve{\mathscr{Q}}_1(\epsilon,\smfu^{\alpha\beta}, D\smfu^{\alpha\beta}) +\epsilon
\breve{\mathscr{B}}_1(\epsilon, \smfu^{\alpha\beta} , D\smfu^{\alpha\beta})
+ \breve{\mathscr{F}}_1(\epsilon,  \smfu^{\alpha\beta}, \breve{z}^k,\delta\breve{\rho})  \nnb  \\
& + \epsilon^2
\breve{\mathscr{G}}_1(\epsilon, \smfu^{\alpha\beta}, \delta\breve{\rho}, \breve{z}^k),  \label{e:A00}
\end{align}
\begin{align}
\texttt{A}^{j}(\epsilon, \smfu^{00},\smfu^{0j}, \breve{\xi})= & 2\epsilon  E^2(1)\smfu^{0i}\del{i}\del{k}\smfu^{kj}-\epsilon E^2(1) \smfu^{kl}\del{k}\del{l}\smfu^{0j}+\epsilon E^2(1)\smfu^{00} \del{i}\smfu^{ij}_0 + \breve{\mathscr{L}}^j_2(\epsilon, \epsilon \smfu^{kl}, D\smfu^{kl})+  \epsilon^2 \breve{\mathscr{S}}_2^j (\epsilon, \smfu^{\alpha\beta})  \nnb  \\
&  + \epsilon  \breve{\mathscr{R}}_2^j(\epsilon, \epsilon\smfu^{\alpha\beta},  \epsilon\smfu^{kl}_0,
D\smfu^{\alpha\beta} ,\smfu^{kl}_0) + \breve{\mathscr{Q}}_2^j(\epsilon,\smfu^{\alpha\beta},  D\smfu^{\alpha\beta}) +\epsilon \breve{\mathscr{B}}^j_2(\epsilon, \smfu^{\alpha\beta}, D\smfu^{\alpha\beta})+\epsilon^2 \breve{\mathscr{F}}^j_2(\epsilon, \smfu^{\alpha\beta},\breve{z}^k,\delta\breve{\rho}) \nnb  \\
& + \epsilon \breve{\mathscr{G}}^j_2(\epsilon, \smfu^{\mu\nu},
\delta\breve{\rho}, \breve{z}^k),  \label{e:A0j}
\end{align}
and
\begin{equation}\label{e:freedata}
\breve{\xi}=(\smfu^{ij},\smfu^{ij}_0, \breve{z}^k, \delta\breve{\rho})
\end{equation}
denotes collectively the free initial data.
Here, the remainder terms $\breve{\mathscr{L}}_1$ and $\breve{\mathscr{L}}^j_2$ are linear in $(\smfu^{kl}_0,\epsilon \smfu^{kl})$ and $(\epsilon\smfu^{kl}, D\smfu^{kl})$ respectively,   $\breve{\mathscr{S}}_1$ and $\breve{\mathscr{S}}^j_2$ vanish
to second order in $ \smfu^{\alpha\beta} $,  $\breve{\mathscr{R}}_1$ and $\breve{\mathscr{R}}^j_2$ vanish to first order in $(\epsilon\smfu^{\alpha\beta},  \epsilon\smfu^{kl}_0, D\smfu^{\alpha\beta})$ and are linear in $\smfu^{kl}_0$,  $\breve{\mathscr{B}}_1$ and $\breve{\mathscr{B}}^j_2$
vanish to first order in $\smfu^{\alpha\beta}$ and are linear in $D\smfu^{\lambda\sigma}$, $\breve{\mathscr{Q}}_1$ and $\breve{\mathscr{Q}}^j_2$ vanish to second order in $D\smfu^{\alpha\beta}$, $\breve{\mathscr{F}}_1$ and $\breve{\mathscr{F}}_2^j$ are linear in $\delta\breve{\rho}$, and $\breve{\mathscr{G}}_1$ and $\breve{\mathscr{G}}_2^j$ vanish to first order in $\breve{z}^k$.

As a final observation, we note that $\Omega(1)<0$ is a consequence of the definition \eqref{e:ome} of $\Omega(t)$ from which it follows that $7-6\Omega(1)>0$, $(\Omega(1)-2)\Omega(1)>0$, $\Omega(1)-1<0$ and $-\Omega(1) >0$;
this observation will be important for the analysis carried out in \S \ref{S:InidFixpt}.

\subsection{Yukawa potentials} \label{S:Yuk}
The \textit{Yukawa potential operator of order} $s$, denoted $(\kappa^2-\Delta)^{-\frac{s}{2}}$, is one of the main technical tools
we employ for
solving the constraints. It is defined for $0<s<\infty$ and $\kappa  \geq 0$, and it acts on function $f$ via the formula
\begin{equation*}
(\kappa^2-\Delta)^{-\frac{s}{2}}(f)=(\widehat{\mathcal{Y}}_{s,\kappa}\widehat{f})^{\vee}=\mathcal{Y}_{s,\kappa}\ast f
\end{equation*}
where
\begin{equation*}
\mathcal{Y}_{s,\kappa}(x)=\bigl((\kappa^2+4\pi^2|\xi|^2)^{-\frac{s}{2}}\bigr)^{\vee}(x).
\end{equation*}
In the special case $n=3$ and $s=2$,  see \cite[\S $3.2$]{Oliynyk2006}, we have the closed form convolution
representation
\begin{align*}
(\kappa^2-\Delta)^{-1}(f)(x) = \frac{1}{4\pi} \int_{\Rbb^3} \frac{e^{-\kappa|x-y|}}{|x-y|}f(y) d^n y.
\end{align*}
Note also that the Yukawa potential operator coincides with the Riesz potential operator and Bessel potential
operator when $\kappa=0$ and $\kappa=1$, respectively; see Appendices \ref{S:Riesz} and \ref{S:Bessel}.

Before moving forward, we recall the following well known fact concerning convolution operators
\begin{lemma}\cite[Exercise $1.2.9$]{Grafakos2014}\label{T:opnorm}
Let $T(f)=f\ast K$, where $K$ is a positive $L^1(\Rbb^n)$ function and $f$ is in $L^p(\Rbb^n)$, $1\leq p \leq \infty$. Then the operator norm of $T:L^p(\Rbb^n)\rightarrow L^p(\Rbb^n)$ is equal to $\|K\|_{L^1}$.
\end{lemma}

An important property of the Yukawa potential operator $(\kappa^2-\Delta)^{-\frac{s}{2}}$ is
that it maps $L^p$ to itself  whenever $\kappa>0$. The following proposition gives a precise statement of
this mapping property and it should be viewed as a generalization of the mapping property for the Bessel potential
operator from Theorem \ref{T:Bessel}.\eqref{T:B1}.

\begin{proposition}\label{T:genYu}
	For $0<s<\infty$, $\kappa>0$ and  $1\leq p\leq \infty$, the operator $\kappa^s(-\Delta+\kappa^2)^{-\frac{s}{2}}$ maps $L^p(\Rbb^n)$ to itself with norm $1$, that is,
	\begin{align*}
	\|\kappa^s(-\Delta+\kappa^2)^{-\frac{s}{2}}f\|_{L^p } \leq \|f\|_{L^p }
	\end{align*}
	 for all $f\in L^p(\Rbb^n)$ and
	\begin{equation*}
	\|\kappa^s(\kappa^2-\Delta )^{-\frac{s}{2}}\|_{\emph{op}}=\|\kappa^s\mathcal{Y}_s\|_{L^1} =1. 
	\end{equation*}
\end{proposition}
\begin{proof}
	Let\footnote{In the following, we will use the well known identities
		$\widehat{S_\lambda(f)}=\lambda^{-n}S_{\frac{1}{\lambda}}(\hat{f})$ and
		$\|S_{\lambda}(f)\|_{L^p}= \lambda^{-\frac{n}{p}}\|f\|_{L^p}$.} $S_\lambda (f)(x)=f(\lambda x)$ denote the scaling operator.
	Then from the identity $S_{\frac{1}{\kappa}}(\hat{G}_s)=\kappa^s \hat{\mathcal{Y}}_{s,\kappa}$, see Appendix \ref{S:Bessel}
for the definition of $\mathcal{G}_s$,
	we find that $S_{\frac{1}{\kappa}}(\hat{\mathcal{G}}_s)\hat{f}=S_{\frac{1}{\kappa}}(\hat{\mathcal{G}}_s S_\kappa\hat{f})=\kappa^s \hat{\mathcal{Y}}_{s,\kappa} \hat{f}$. Taking the inverse Fourier transform of this expression gives
	\begin{align*}
	\kappa^s(\kappa^2-\Delta)^{-\frac{s}{2}}f=\kappa^s\mathcal{Y}_{s,\kappa} \ast f = \kappa^n S_\kappa(\mathcal{G}_s)\ast f.
	\end{align*}
	The proof now follows from Lemma \ref{T:opnorm} since
	$\|\kappa^s\mathcal{Y}_{s,\kappa}\|_{L^1}=\|\kappa^nS_{\kappa}(\mathcal{G}_s)\|_{L^1}=\|\mathcal{G}_s\|_{L^1}=1$
	by Theorem \ref{T:Bessel}.(1).
\end{proof}

For applications in this articles, we single out the inequalities from Proposition \ref{T:genYu} on $\Rbb^3$ corresponding
to $s=1$ and $s=2$, which are given by
	\begin{align}
	\|\kappa (-\Delta+\kappa^2)^{-\frac{1}{2}}f\|_{L^p } &\leq \|f\|_{L^p} \label{e:halfYu}
	\intertext{and}
	\|\kappa^2 (\kappa^2-\Delta)^{-1}f\|_{L^p } &\leq \|f\|_{L^p }, \label{E:yuest0}
	\end{align}
respectively.

Next, we obtain estimates for the operators $\del{j}(\kappa^2-\Delta)^{-1} f $ and $\del{j}\del{k} (\kappa^2-\Delta)^{-1} f$ on
$\Rbb^3$ that are uniform in $\kappa$.
\begin{proposition}\label{t:ddel}
	Suppose $s\in \Zbb_{\geq 0}$, $k\in \Zbb_{\geq 1}$, $\kappa \geq 0$ and  $1<p<q<\infty$. Then there exists a constant $C>0$, independent of
$\kappa$,
such that:
	\begin{enumerate}
		\item \label{ddel1}
		\begin{align*}
		\|\del{j}(\kappa^2-\Delta)^{-1} f\|_{L^q}
		\leq & C\|f\|_{L^p} 
		\end{align*}
 for all $f\in L^p(\Rbb^n)$ provided that  $p$ and $q$ also satisfy $\frac{1}{p}-\frac{1}{q}=\frac{1}{n}$,
		\item \label{ddel2}
		\begin{align*}
		\|\del{j}\del{k}(\kappa^2-\Delta)^{-1} f\|_{W^{s,p}} \leq & C\|f\|_{W^{s,p}}
		\end{align*}
for all $f\in W^{s,p}(\Rbb^n)$, and
		\item \label{ddel3}
		\begin{align*}
		\|\del{k}(\kappa^2-\Delta)^{-1}f\|_{R^k }\leq C\|f\|_{H^{k-1}}
		\end{align*}
for all $f\in H^{s-1}(\Rbb^3)$.
	\end{enumerate}
\end{proposition}
\begin{proof}
	First, for $\kappa=0$, we note that the above estimates are a direct consequence of Proposition \ref{E:ddelin1}.
Therefore, we assume that $\kappa>0$. Then differentiating the identity
\begin{align*}
	(\kappa^2-\Delta)^{-1}f
=&(\kappa^2-\Delta)^{-1}(\kappa^2-\Delta-\kappa^2)(-\Delta)^{-1}f 
	= (-\Delta)^{-1}f- \kappa^2 (\kappa^2-\Delta)^{-1} (-\Delta)^{-1}f  
	\end{align*}
 gives
	\begin{align} \label{diffE:yuf1}
	\del{j}(\kappa^2-\Delta)^{-1} f 
= -\mathfrak{R}_j(-\Delta)^{-\frac{1}{2}}  f+\kappa^2\mathfrak{R}_j(-\Delta)^{-\frac{1}{2}} (\kappa^2-\Delta)^{-1} f,
	\end{align}
where $\mathfrak{R}_j$ is the Riesz transform, see Appendix \ref{S:Riesz}. Taking the $L^q$ norm of both sides, where $q$ and $p$ 
are related by $\frac{1}{p}-\frac{1}{q}=\frac{1}{n}$, we find
that
\begin{align}
\|\del{j}(\kappa^2-\Delta)^{-1} f\|_{L^q} 
	\lesssim & \|\mathfrak{R}_j(-\Delta)^{-\frac{1}{2}}  f\|_{L^q}
+\kappa^2\|\mathfrak{R}_j(-\Delta)^{-\frac{1}{2}} (\kappa^2-\Delta)^{-1} f \|_{L^q} \nnb  \\
	\lesssim &  \|f\|_{L^p} +\| \kappa^2 (\kappa^2-\Delta)^{-1} f \|_{L^p}  \nnb \\
	\lesssim &\|f\|_{L^p} \label{diffE:yuf2}
	\end{align}
by Theorems \ref{T:riepot} and  \ref{T:rietran}, and \eqref{E:yuest0}. This proves the first inequality.
To prove the second inequality, we differentiate \eqref{diffE:yuf1} again
to get
\begin{align*}
	\del{j}\del{k}(\kappa^2-\Delta)^{-1} f =
	\mathfrak{R}_j\mathfrak{R}_k   f-\kappa^2\mathfrak{R}_j\mathfrak{R}_k (\kappa^2-\Delta)^{-1} f.
	\end{align*}
From this it then follows that
\begin{align}
	\|\del{j}\del{k}(\kappa^2-\Delta)^{-1} f\|_{W^{s,p}}
	\lesssim & \|\mathfrak{R}_j\mathfrak{R}_k   f\|_{W^{s,p}}+\|\kappa^2\mathfrak{R}_j\mathfrak{R}_k (\kappa^2-\Delta)^{-1} f\|_{W^{s,p}}
 \nnb  \\
	\lesssim & \|f\|_{W^{s,p}}+\|\kappa^2 (\kappa^2-\Delta)^{-1} f\|_{W^{s,p}}\nnb \\
	\lesssim & \|f\|_{W^{s,p}} \label{diffE:yuf3}
	\end{align}
by Propositions \ref{T:genYu} and  \ref{E:ddelin1}, and the fact that $[D^\alpha, (\kappa^2-\Delta)^{-1}]=0$. Finally, we
observe that the last inequality follows from \eqref{diffE:yuf1} and the inequalities \eqref{diffE:yuf2}, with $(n,q,p)=(3,6,2)$, and
\eqref{diffE:yuf3}:
\begin{align*}
	\|\del{k}(\kappa^2-\Delta)^{-1}f\|_{R^s}=\|\del{k}(\kappa^2-\Delta)^{-1}f\|_{L^6 }+\|D\del{k}(\kappa^2-\Delta)^{-1}f\|_{\Hs}\lesssim \|f\|_{H^{s-1}}.
	\end{align*}
\end{proof}

\subsection{Relation between the Riesz and Yukawa potential operators} \label{S:Yuk2}
In the following proposition,  we establish an estimate, uniform in $\kappa$, that quantifies the relation between the Riesz and Yukawa potential operators, which is a variation on Lemma $2$ from \cite[\S$3.2$]{Elias1970}.
\begin{proposition}\label{T:genYurel}
	Suppose $0<s<\infty$, $\kappa>0$,  and  $1\leq p\leq \infty$. Then there exists a constant $C>0$, independent of $\kappa$, such
that
	\begin{align*}
	\| (-\Delta+\kappa^2)^{-\frac{s}{2}}(-\Delta)^{\frac{s}{2}}(f)\|_{L^p } \leq C\|f\|_{L^p }
	\end{align*}
for all $f\in L^p(\Rbb^n)$.
\end{proposition}
\begin{proof}
	By Young's inequality for convolutions, see Proposition \ref{T:Youngineq}, we get
	\begin{align*}
	&\|(-\Delta+\kappa^2)^{-\frac{s}{2}}(-\Delta)^{\frac{s}{2}}f \|_{L^p}=  \|\mathcal{Y}_{s,\kappa}\ast \bigl( (4\pi^2|\xi|^2)^{\frac{s}{2}}\bigr)^{\vee}\ast f\|_{L^p}  \nnb \\
	&\hspace{2cm}\leq \|\mathcal{Y}_{s,\kappa}\ast \bigl( (4\pi^2|\xi|^2)^{\frac{s}{2}}\bigr)^{\vee}\|_{L^1}\|f\|_{L^p}=\|\bigl(\widehat{\mathcal{Y}}_{s,\kappa}\cdot \bigl( (4\pi^2|\xi|^2)^{\frac{s}{2}}\bigr) \bigr)^{\vee}\|_{L^1}\|f\|_{L^p}.
	\end{align*}
By Lemma \ref{T:opnorm}, the proof would then follow from a bound on  $\|\bigl(\widehat{\mathcal{Y}}_{s,\kappa}\cdot\bigl( (4\pi^2|\xi|^2)^{\frac{s}{2}}\bigr) \bigr)^{\vee}\|_{L^1}$. To see that this bound holds, we first observe that
	\begin{align*}
	\|\bigl(\widehat{\mathcal{Y}}_s\cdot \bigl( (4\pi^2|\xi|^2)^{\frac{s}{2}}\bigr) \bigr)^{\vee}\|_{L^1}=&\biggl\|\biggl(\frac{(4\pi^2|\xi|^2)^{\frac{s}{2}}}{(\kappa^2+4\pi^2|\xi|^2)^{\frac{s}{2}}}\biggr)^{\vee}\biggr\|_{L^1}
	=   
	\biggl\| \biggl[\biggl(\mathds{1}-\frac{\kappa^2}{ \kappa^2+4\pi^2|\xi |^2}\biggr)^{\frac{s}{2}}\biggr]^{\vee}\biggr\|_{L^1}.
	\end{align*}
	Since $s>0$, we see from
	$\widehat{\mathcal{Y}}_{2m,\kappa}=(\kappa^2+4\pi^2|\xi|^2)^{-\frac{2m}{2}}$ and expanding
	$(1-y)^{\frac{s}{2}}$ in a power series that
	\begin{align*}
	\biggl(\mathds{1}-\frac{\kappa^2}{ \kappa^2+4\pi^2|\xi |^2}\biggr)^{\frac{s}{2}}=\mathds{1}+\sum^{\infty}_{m=1}A_{m,s} \kappa^{2m} (\kappa^2+4\pi^2|\xi|^2)^{-\frac{2m}{2}}=\mathds{1}+\sum^{\infty}_{m=1}A_{m,s}\kappa^{2m} \widehat{\mathcal{Y}}_{2m,\kappa}
	\end{align*}
	holds for $|\xi| > 0$. But since $\sum_{m=1}^{\infty}|A_{m,s}| < \infty$
by Raabe's test\footnote{To apply Raabe's test,  consider the series  $(1-y)^s=1+\sum_{m=1}^{\infty}\p{s\\m}(-1)^m y^m$, where $\p{s\\m}:=\frac{s!}{m!(s-m)!}$ and $s>0$. Since for $m$ is large enough,
$m\left(\frac{\left|\p{s\\m}\right|}{\left| \p{s\\m+1}\right|}-1\right)= \frac{m(1+s)}{m-s}  \rightarrow 1+s>1$ as
$m\rightarrow \infty$, the sum $\sum_{m=1}^{\infty} \p{s\\m} (-1)^m$ is absolutely convergent.}, see \cite{Arfken1985},
and
	\begin{align*}
	\biggl[\biggl(\mathds{1}-\frac{\kappa^2}{ \kappa^2+4\pi^2|\xi |^2}\biggr)^{\frac{s}{2}}\biggr]^{\vee} =\bigl[\mathds{1}+\sum^{\infty}_{m=1}A_{m,s}\kappa^{2m} \widehat{\mathcal{Y}}_{2m,\kappa}\bigr]^{\vee}=\delta_0 +\sum^{\infty}_{m=1}A_{m,s}\kappa^{2m} \mathcal{Y}_{2m,\kappa},
	\end{align*}
we deduce that
	\begin{align*}
	& \|\bigl(\widehat{\mathcal{Y}}_{s,\kappa}\cdot\bigl( (4\pi^2|\xi|^2)^{\frac{s}{2}}\bigr)  \bigr)^{\vee}\|_{L^1}=\biggl\|\biggl[\biggl(\mathds{1}-\frac{\kappa^2}{ \kappa^2+4\pi^2|\xi |^2}\biggr)^{\frac{s}{2}}\biggr]^{\vee}\biggr\|_{L^1} \nnb  \\
	& \hspace{3cm} \leq 1 + \sum^{\infty}_{m=1}|A_{m,s}| \|\kappa^{2m}\mathcal{Y}_{2m,\kappa}\|_{L^1}=1 +  \sum^{\infty}_{m=1}|A_{m,s}| \leq \infty,
	\end{align*}
where in deriving the final equality we used $\|\kappa^{2m} \mathcal{Y} _{2m,\kappa}\|_{L^1}= 1$, which is a consequence of Proposition \ref{T:genYu}.
\end{proof}

We proceed by using Proposition \ref{T:genYurel} to obtain estimates for the operators  $(-\Delta+\kappa^2)^{-\frac{1}{2}}\del{j}$
and $\kappa (-\Delta+\kappa^2)^{-1} \del{j}$ that are uniform in $\kappa$.
\begin{proposition}\label{T:Yu3est}
	Suppose $s\in \Zbb_{\geq 0}$, $\kappa>0$ and $1\leq p\leq \infty$. Then there exists a constant $C>0$, independent of $\kappa$,
such that
\begin{equation*}
\|(-\Delta+\kappa^2)^{-\frac{1}{2}}\del{j} f\|_{W^{s,p} }
+	\|\kappa (-\Delta+\kappa^2)^{-1} \del{j} f\|_{W^{s,p} } \leq  C\|f\|_{W^{s,p} }
	\end{equation*}
for all $f\in W^{s,p}(\Rbb^n)$.
\end{proposition}
\begin{proof}
The proof follows directly from the identities
	\begin{gather*}
	\kappa (-\Delta+\kappa^2)^{-1} \del{j}f = -\kappa (-\Delta+\kappa^2)^{-\frac{1}{2}} (-\Delta+\kappa^2)^{-\frac{1}{2}} \del{j} f
\intertext{and}
	(-\Delta+\kappa^2)^{-\frac{1}{2}}\del{j}f=-(-\Delta+\kappa^2)^{-\frac{1}{2}}(-\Delta)^{\frac{1}{2}} \mathfrak{R}_j f,
	\end{gather*}
and an application of Propositions \ref{T:genYu} and \ref{T:genYurel}, and Theorem \ref{T:rietran}.
\end{proof}

We will also need the following generalization of Theorems \ref{T:riepot} and
\ref{T:Bessel}.\eqref{T:B2}  for the operator $(\kappa^s-\Delta)^{-\frac{s}{2}}$ with estimates that are uniform in $\kappa$.
\begin{proposition}\label{cor39}
	Suppose $0<s<n$  and $1 < p<q<\infty$ satisfy $\frac{1}{p}-\frac{1}{q}=\frac{s}{n}$, and $\kappa>0$.
	Then there exists a constant $C>0$, independent of $\kappa$, such that
	\begin{align*}
	\|(\kappa^2-\Delta)^{-\frac{s}{2}}f\|_{L^q}\leq C\|f\|_{L^p }
	\end{align*}
for all $f\in L^p (\Rbb^n)$.
\end{proposition}
\begin{proof}
Noting the identity $(\kappa^2-\Delta)^{-\frac{s}{2}} f=(\kappa^2-\Delta)^{-\frac{s}{2}} (-\Delta)^{\frac{s}{2}}(-\Delta)^{-\frac{s}{2}} f$,
it is clear that the proof follows directly from an application of Proposition \ref{T:genYurel} and Theorem \ref{T:riepot}.
\end{proof}
For convenience, we list the following special cases of the estimate from Proposition \ref{cor39} on $\Rbb^3$, which we will use frequently in subsequence
sections:
\begin{align}
\|(\kappa^2-\Delta)^{-\frac{1}{2}}f\|_{L^6 }&\lesssim \|f\|_{L^2 }, \label{e:sYupq} \\
\|(\kappa^2-\Delta)^{-1}f\|_{L^6 }&\lesssim \|f\|_{L^{6/5} }  \label{e:sYupq2}
\intertext{and}
\|(\kappa^2-\Delta)^{-\frac{1}{2}}f\|_{L^2 }&\lesssim \|f\|_{L^{6/5} }. \label{e:sYupq3}
\end{align}

\subsection{Solving the constraint equations}\label{S:InidFixpt}
Having established the necessary estimates for the Yukawa potential operator, we now turn to solving the constraint equations \eqref{E:constru1}-\eqref{E:constru2}. For this, we will use a fixed point method by
first specifying the \textit{free data} $(\smfu^{ij}, \smfu^{ij}_0, \delta\breve{\rho}, \breve{z}^l)$ and then rewriting the constraint
equations \eqref{E:constru1}-\eqref{E:constru2} as a contraction mapping whose fixed point yields a solution $(\smfu^{00}, \smfu^{0j})$
to the constraint equations.

\subsubsection{The contraction mapping}
To streamline the set up of the contraction mapping, we
set
\begin{equation*}
\phi= \smfu^{00} \AND \psi^j= \smfu^{0j},
\end{equation*}
and  we define the constants
\begin{gather*}
\texttt{a}= \frac{\Lambda}{3 } E^2(1) \bigl( 7-6\Omega(1)\bigr)>0, \quad \texttt{b}=  \frac{\Lambda}{3}\bigl(\Omega(1)-1\bigr) E^2(1)<0,\label{e:a} \\
\texttt{c}= 2\Lambda E^2(1)   (1+\epsilon^2K)\bigl(\Omega(1)-2\bigr)\Omega(1)>0 \AND \texttt{d}=-2 \Omega(1) >0,  \label{e:c}
\end{gather*}
where $\Omega(t)$ is as defined previously by \eqref{E:OMEGAREP}.
Further, we fix the free data according to
\begin{align}
\smfu^{ij}\in R^{s+1}(\Rbb^3, \mathbb{S}_3), \quad \smfu^{ij}_0 \in H^s  (\Rbb^3, \mathbb{S}_3),  
\quad
\breve{z}^j\in L^{6/5}\cap K^s(\Rbb^3, \Rbb^3) \AND \delta \breve{\rho} \in L^{6/5}\cap K^s(\Rbb^3, \Rbb). \label{e:inisetup}
\end{align}
With the above definition, the constraint equations \eqref{E:constru1}-\eqref{E:constru2} become
\begin{align} \label{e:const0}
\p{\Delta-\epsilon^2\texttt{a} & -\epsilon \texttt{b} \del{j} \\
	-\epsilon \texttt{d} \del{ }^j & \Delta-\epsilon^2 \texttt{c}}\p{\phi \\ \psi^j}= \p{f(\epsilon, \phi,\psi^j,\breve{\xi}) \\ g^j(\epsilon, \phi,\psi^k,\breve{\xi}) }
\end{align}
where
\begin{align}
f(\epsilon, \phi,\psi^k,\breve{\xi})= &  \epsilon E^2(1) \frac{\Lambda}{3} \del{i}\partial _j \smfu^{ i j}  +   \frac{2 \Lambda}{3} E^2 (1) \delta\breve{\rho}    +\epsilon^2  \texttt{A}^{0}(\epsilon, \phi,\psi^k,\breve{\xi}),  \label{e:f}
\\
g^j(\epsilon, \phi,\psi^k,\breve{\xi})= &  - \epsilon  E^2(1) \frac{\Lambda}{3}  \del{i} \hat{\mfu}^{j i}_0 +2 \epsilon E^2 (1)  \sqrt{\frac{\Lambda}{3}}  (\breve{\rho} \breve{z}^ j ) +\epsilon^2  \texttt{A}^{j}(\epsilon, \phi,\psi^k,\breve{\xi}).   \label{e:gj}
\end{align}

Apart from the regularity requirements  \eqref{e:inisetup} on the free initial data , we also require the following smallness condition
on the initial density, see \eqref{FLRW.c},  of the background FRLW solution given by 
\begin{align}\label{inisetup2}
	0< \mu(1)<\frac{1}{8}(19 +5\sqrt{29})\Lambda,
\end{align}
which, by \eqref{e:ome}, implies that
\begin{align}\label{e:addcon}
	\frac{1}{4}(-1-\sqrt{29})<\Omega(1)<0 \AND -\frac{\texttt{a}}{\texttt{bd}}>2.
\end{align}

\begin{remark}
It is probable that the condition \eqref{inisetup2} is not necessary for establishing the existence of $1$-parameter families of $\epsilon$-dependent initial data that satisfy the constraint equations. Indeed, in the article \cite{Oliynyk2015}, a similar method was used to establish the existence of  $1$-parameter families
of $\epsilon$-dependent  initial data satisfying the constraint equations without a similar smallness condition on the background FLRW solution. However, the gauge condition used in this article, which is suited to the long time evolution problem, is different from the gauge used in \cite{Oliynyk2015}, and the analysis of the constraint equations in \cite{Oliynyk2015} employed a more complicated conformal decomposition. Consequently, it is not clear if the choice of gauge or the particular representation of the constraint equations used in this article is responsible for the requirement \eqref{inisetup2}. In any case, we stress that  \eqref{inisetup2} is only needed to establish the existence of suitable  $1$-parameter families
of $\epsilon$-dependent initial data; this restriction is not needed for the evolution of the initial data where a smallness condition is only needed on the perturbed part of the initial data, see Theorem \ref{T:MAINTHEOREM} for details, and not on the background FLRW solution.
\end{remark}
Replacing $\psi^j$ with $\vartheta^j$ defined by
\begin{equation} \label{e:varphi}
\vartheta^j =\psi^j -\epsilon \texttt{d}\del{}^j(\Delta-\epsilon^2 \texttt{c})^{-1} \phi
\end{equation}
allows us to write \eqref{e:const0} as 
\begin{align} \label{e:const1}
\p{\Delta-\epsilon^2(\texttt{a}+\texttt{bd}) & 	-\epsilon \texttt{b} \del{j}  \\
0 & \Delta-\epsilon^2 \texttt{c}}\p{\phi \\ \vartheta^j}= \p{\tilde{f}(\epsilon, \phi,\vartheta^k,\breve{\xi}) \\ \tilde{g}^j(\epsilon, \phi,\vartheta^k,\breve{\xi}) }
\end{align}
where
\begin{align}
\tilde{f}(\epsilon, \phi,\vartheta^k,\breve{\xi})=&f(\epsilon, \phi,\vartheta^k+\epsilon \texttt{d} \del{}^k (\Delta-\epsilon^2 \texttt{c})^{-1} \phi,\breve{\xi})+\epsilon^4 \texttt{bcd}  (\Delta-\epsilon^2 \texttt{c})^{-1} \phi \label{e:tilfg1} \\
\intertext{and}
\tilde{g}^j(\epsilon, \phi,\vartheta^k,\breve{\xi})= & g^j(\epsilon, \phi,\vartheta^k+\epsilon \texttt{d} \del{}^k (\Delta-\epsilon^2 \texttt{c})^{-1} \phi,\breve{\xi}).  \label{e:tilfg2}
\end{align}
The advantage of \eqref{e:const1} over \eqref{e:const0} is that the linear operator on the left hand side
of \eqref{e:const1} is  invertible with the inverse given by\footnote{Note that $a+bd>-bd >0$ by \eqref{e:addcon}.}
\begin{align*}
\p{\Delta-\epsilon^2(\texttt{a}+\texttt{bd}) & \epsilon \texttt{b} \del{j}  \\
0 & \Delta-\epsilon^2 \texttt{c}}^{-1}= \p{\bigl(\Delta-\epsilon^2 (\texttt{a}+\texttt{bd})\bigr)^{-1} & \epsilon \texttt{b} \del{j} (\Delta-\epsilon^2 \texttt{c})^{-1}\bigl(\Delta-\epsilon^2 (\texttt{a}+\texttt{bd})\bigr)^{-1} \\
	0 &  (\Delta-\epsilon^2 \texttt{c})^{-1}   }.
\end{align*}
Acting on both sides of \eqref{e:const1} with the inverse operator yields the equations
\begin{align}
\p{\phi \\ \vartheta^j}= \p{\bigl(\Delta-\epsilon^2 (\texttt{a}+\texttt{bd})\bigr)^{-1} & \epsilon \texttt{b} \del{j} (\Delta-\epsilon^2 \texttt{c})^{-1}\bigl(\Delta-\epsilon^2 (\texttt{a}+\texttt{bd})\bigr)^{-1} \\
	0 &  (\Delta-\epsilon^2 \texttt{c})^{-1}   } \p{\tilde{f}(\epsilon, \phi,\vartheta^k,\breve{\xi}) \\ \tilde{g}^j(\epsilon, \phi,\vartheta^k,\breve{\xi}) }, \label{e:phsexp1}
\end{align}
which, we stress, are completely equivalent to the constraint equations \eqref{E:constru1}-\eqref{E:constru2}.

In order to solve the constraint equations \eqref{e:phsexp1} using a fixed point method, we let the right hand side of \eqref{e:phsexp1} define
a map $\mathbf{G}$ which takes elements $\acute{\iota}=(\acute{\phi},  \acute{\vartheta}^k)$ and maps
them to $\grave{\iota}=\mathbf{G}(\acute{\iota})$, where $\grave{\iota}=(\grave{\phi},  \grave{\vartheta}^k)$, according to
\begin{align}
\grave{\phi}= & \bigl(\Delta-\epsilon^2 (\texttt{a}+\texttt{bd})\bigr)^{-1} \tilde{f}(\epsilon, \acute{\phi},\acute{\vartheta}^k,\breve{\xi})+\epsilon \texttt{b} \del{j} (\Delta-\epsilon^2 \texttt{c})^{-1}\bigl(\Delta-\epsilon^2 (\texttt{a}+\texttt{bd})\bigr)^{-1} \tilde{g}^j(\epsilon, \acute{\phi},\acute{\vartheta}^k,\breve{\xi}), \label{e:phi.1}\\
\grave{\vartheta}^k= &  (\Delta-\epsilon^2 \texttt{c})^{-1} \tilde{g}^k(\epsilon, \acute{\phi},\acute{\vartheta}^j,\breve{\xi}). \label{e:psi.1}
 \end{align}

Before considering the map  $\mathbf{G}$ further, we first establish the following technical lemma that will allows us to estimate terms of the form
$ \|(\Delta-\epsilon^2 \lambda)^{-1}  \tilde{g}^j(\epsilon, \acute{\phi},\acute{\vartheta}^k,\breve{\xi})\|_{R^{s+1}}$ and $\| (\Delta -\epsilon^2 \lambda)^{-1} \tilde{f}(\epsilon, \acute{\phi},\acute{\vartheta}^j,\breve{\xi})\|_{R^{s+1}}$, where $\lambda>0$ is some constant. These
estimates will be needed below in the proof of Proposition \ref{t:contraction} where we show that $\mathbf{G}$ defines a contraction map.
\begin{lemma}\label{T:Yuestdec2}
	Suppose $s \in \Zbb_{\geq 3}$, $0<\epsilon<\epsilon_0$, $\lambda\in \Rbb_{>0}$, and $F$ is defined by
	\begin{align*}
	F=&\epsilon^4 \texttt{H}_1(\epsilon,f_1,f_2)   
	+\epsilon \del{i}\del{j} f_3+ f_4 +\epsilon^3  \texttt{H}_5(\epsilon,f_5, \del{i}\del{j}f_6)+\epsilon^3\texttt{H}_7(\epsilon, f_7, \del{i} f_8 ) \nnb  \\
	& \hspace{5cm}  +\epsilon^3\texttt{H}_0(\epsilon, f_0, f_8)  +\epsilon^3  f_9 +\epsilon^2   f_{10} 
	+\epsilon  \del{i} f_{11} +\epsilon f_{12}  
	\end{align*}
	where  $f_1, f_2, f_3, f_5, f_6, f_7, f_9 \in R^{s+1}(\Rbb^3)$,
	$f_4, f_{12} \in L^{\frac{6}{5}}\cap K^s(\Rbb^3)$, $f_0\in R^s(\Rbb^3)$,  $f_8,f_{10},f_{11}\in H^s (\Rbb^3)$, and the maps
	 $\texttt{H}_\ell(\epsilon,u,v)$, $\ell=0,1,5,7$, are smooth, vanish to first order in $u$, and are linear in $v$.
	 Then $(\epsilon^2\lambda-\Delta)^{-1} F \in R^{s+1}$ and
	\begin{align}
	\|(\epsilon^2\lambda-\Delta)^{-1} F &\|_{R^{s+1}}\leq  C_0\Bigl[ \epsilon^2 \|f_1\|_{R^{s+1}}\|f_2\|_{R^{s+1}} + \epsilon \|f_3\|_{R^{s+1}}+\|f_4\|_{L^\frac{6}{5}\cap K^s} + \epsilon  \|f_5\|_{R^{s+1}}\|f_6\|_{R^{s+1}}   \nnb  \\
&+\epsilon\|f_9\|_{R^{s+1}}
	 +\epsilon ( \|f_7\|_{R^{s+1}}+  \|f_0\|_{R^s} )  	\|f_8\|_{H^s }  +\epsilon \|f_{10}\|_{H^s}+\epsilon\|f_{11}\|_{H^s} +\epsilon\|f_{12}\|_{L^\frac{6}{5}\cap K^s}\Bigr] \label{e:inimod}
	\end{align}
where $C_0=C_0\bigl(\norm{f_0}_{R^{s}},\norm{f_1}_{R^{s+1}},\norm{f_2}_{R^{s+1}},\norm{f_5}_{R^{s+1}},\norm{f_6}_{R^{s+1}},\norm{f_7}_{R^{s+1}},
\norm{f_8}_{H^{s}}\bigr)$.

\bigskip

\noindent Furthermore, if $f_{10}=\texttt{G}(f)g$, where $f\in K^s(\Rbb^3)$, $g\in H^s(\Rbb^3)$ and $\texttt{G}(u)$ is
smooth, then
\begin{align*}
		\|f_{10}\|_{H^s} \leq C(\|f\|_{K^s})\|g\|_{H^s},
\end{align*}
and, in the case $\texttt{G}(u)$ also vanishes to first order in $u$,
\begin{align*}
		\|f_{10}\|_{H^s }\leq C(\|f\|_{K^s})\|f\|_{K^s} \|g\|_{H^s}.
\end{align*}
\end{lemma}
\begin{proof}
Since
\begin{equation*}
\|(\epsilon^2\lambda-\Delta)^{-1} F\|_{R^{s+1}} =\|(\epsilon^2\lambda-\Delta)^{-1} F\|_{L^6}+\|(\epsilon^2\lambda-\Delta)^{-1} D F\|_{H^s},
\end{equation*}
we proceed by estimating each term separately starting with $\|(\epsilon^2\lambda-\Delta)^{-1} F\|_{L^6}$.
Using \eqref{E:yuest0}, \eqref{e:sYupq2}, Proposition \ref{T:Yu3est} and Theorem \ref{calcpropC}, we find that
	\begin{align}
	& \|(\epsilon^2\lambda-\Delta)^{-1}\bigl(\epsilon^4 \texttt{H}_1(\epsilon,f_1,f_2)   +f_4+\epsilon^3\texttt{H}_5(\epsilon,f_5, \del{i}\del{j}f_6) +\epsilon^3\texttt{H}_7(\epsilon, f_7, \del{i} f_8 )  +\epsilon^3\texttt{H}_0(\epsilon, f_0, f_8)  +\epsilon^3 f_9  \bigr)\|_{L^6}  \nnb  \\
	 & \leq C_0 \Bigl[\epsilon^2\|f_1\|_{L^\infty}\|f_2\|_{L^6}+ \|f_4\|_{L^{\frac{6}{5}}}+\epsilon\|f_5\|_{L^\infty}\|\del{i}\del{j}f_6\|_{L^6}+\epsilon \|f_7\|_{L^\infty}\|\del{i} f_8\|_{L^6}+\epsilon \|f_0\|_{L^\infty}\|f_8\|_{L^6}+\epsilon \|f_9\|_{L^6}\Bigr], \label{lemest1}
	\end{align}
where here and for the rest of the proof, we take $C_0$ to be a constant of the form
\begin{equation*}
C_0=C_0\bigl(\norm{f_0}_{R^{s}},\norm{f_1}_{R^{s+1}},\norm{f_2}_{R^{s+1}},\norm{f_5}_{R^{s+1}},\norm{f_6}_{R^{s+1}},\norm{f_7}_{R^{s+1}},
\norm{f_8}_{H^{s}}\bigr).
\end{equation*}

Next, we have that
\begin{equation} \label{lemest2}
\|\epsilon(\epsilon^2\lambda-\Delta)^{-1}\del{i}\del{j}f_3\|_{L^6}\lesssim \epsilon\|f_3\|_{L^6}
\end{equation}
by Proposition \ref{t:ddel}.\eqref{ddel2},
while
\begin{equation} \label{lemest3}
\|\epsilon(\epsilon^2\lambda-\Delta)^{-1}\del{i} f_{11}\|_{L^6}\lesssim \epsilon\|f_{11}\|_{L^2}
\end{equation}
follows from Proposition \ref{t:ddel}.\eqref{ddel1}. Furthermore, we see that
\begin{equation} \label{lemest4}
	\|\epsilon^2(\epsilon^2\lambda-\Delta)^{-1}  f_{10} \|_{L^6}\lesssim  \|\epsilon (\epsilon^2\lambda-\Delta)^{-\frac{1}{2}}  f_{10} \|_{L^6}  \lesssim \epsilon \|f_{10}\|_{L^2}
\end{equation}
by \eqref{e:halfYu} and \eqref{e:sYupq}, as well as
\begin{equation*}
\|\epsilon(\epsilon^2\lambda-\Delta)^{-1}f_{12}\|_{L^6}\lesssim \epsilon \|f_{12}\|_{L^{6/5}}
\end{equation*}
by \eqref{e:sYupq2}. Combining  \eqref{lemest1}-\eqref{lemest4} yields the estimate
	\begin{align}
	 \|(\epsilon^2\lambda-\Delta)^{-1}F\|_{L^6}
	\leq & C_0\Bigl[\epsilon^2\|f_1\|_{L^\infty}\|f_2\|_{L^6}+ \epsilon\|f_3\|_{L^6}+ \|f_4\|_{L^{\frac{6}{5}}}+\epsilon\|f_5\|_{L^\infty}\|\del{i}\del{j}f_6\|_{L^6} +\epsilon \|f_7\|_{L^\infty}\|\del{i} f_8\|_{L^6}\nnb  \\
	& +\epsilon \|f_0\|_{L^\infty}\|f_8\|_{L^6}+\epsilon \|f_9\|_{L^6}+\epsilon \|f_{11}\|_{L^2}+\epsilon\|f_{10}\|_{L^2}+\epsilon \|f_{12}\|_{L^{6/5}} \Bigr]\label{e:1est}
	\end{align}
for the term $\|(\epsilon^2\lambda-\Delta)^{-1}F\|_{L^6}$.

Next, we turn to estimating $ \|(\epsilon^2\lambda-\Delta)^{-1}D F\|_{H^s}$.
First, using the Leibniz's rule, we see, with the help of Theorem \ref{moserlemB} and Proposition \ref{T:Moser2}, that the estimate
\begin{align}\label{e:id1}
	 \| D  (\texttt{G} (f)g ) \|_{H^s}  \lesssim   \|\texttt{G} (f )\|_{W^{1,\infty}}\|Dg\|_{H^s}+  \|D\texttt{G} ( f )\|_{H^s}\|g\|_{W^{1,\infty}}\lesssim \|\texttt{G} ( f )\|_{R^{s+1}}\|g\|_{R^{s+1}}\leq C(\|f\|_{R^{s+1}}) \|f \|_{R^{s+1}}\|g\|_{R^{s+1}}
\end{align}
holds for smooth functions $\texttt{G}(u)$ that vanish to first order in $u$. Then from this estimate and \eqref{E:yuest0},
it follows that
	\begin{align}
	& \|(\epsilon^2\lambda-\Delta)^{-1}D( \epsilon^4 \texttt{H}_1(\epsilon,f_1,f_2) ) \|_{H^s} \lesssim \epsilon^2
\| D (\texttt{H}_1(\epsilon,f_1,f_2) ) \|_{H^s}  
	\leq C_0 \epsilon^2 \|f_1\|_{R^{s+1}}\|f_2\|_{R^{s+1}}, \label{e:te1}
	\end{align}
	and
\begin{equation}
	\|(\epsilon^2\lambda-\Delta)^{-1}D ( \epsilon^3 f_9) \|_{H^s} \lesssim \epsilon
\|D f_9\|_{H^s}\lesssim \epsilon  \|  f_9\|_{R^{s+1}}. \label{e:te1a}
\end{equation}
Continuing on, we note by Proposition \ref{t:ddel} and Theorem \ref{Sobolev} that
	\begin{align} \label{e:te3}
	\|D(\Delta-\epsilon^2 \lambda )^{-1}f_4\|_{H^s}\lesssim & \|D(\Delta-\epsilon^2 \lambda )^{-1} f_4\|_{L^2}+\|D^2(\Delta-\epsilon^2 \lambda )^{-1} f_4\|_{H^{s-1}} 
\lesssim  \|f_4\|_{L^\frac{6}{5}\cap K^s}
	\end{align}
and
\begin{align} \label{e:te3a}
\|D(\Delta-\epsilon^2 \lambda )^{-1}\epsilon f_{12}\|_{H^s}
\lesssim  \epsilon \|f_{12}\|_{L^\frac{6}{5}\cap K^s}.
\end{align}
We further observe by \eqref{E:NORMEQ1}, Proposition \ref{t:ddel}.\eqref{ddel2}, Proposition \ref{T:Yu3est} and Theorems  \ref{moserlemB} and \ref{T:Moser2} that the inequality
	\begin{align*}
	\|D(\Delta-\epsilon^2\lambda )^{-1} & \epsilon  (\texttt{G} ( f )D g)\|_{H^s}=   \|D(\Delta-\epsilon^2\lambda )^{-1}\epsilon D(\texttt{G} ( f ) g)-D(\Delta-\epsilon^2\lambda )^{-1}\epsilon (D\texttt{G} ( f ) g)\|_{H^s} \nnb  \\
	\lesssim &\epsilon  \|\texttt{G} ( f )g\|_{H^s}+ \|D\texttt{G} ( f )g\|_{H^s} \nnb \\
	\lesssim & \epsilon  \|\texttt{G} ( f )\|_{L^\infty}\|g\|_{H^s}+\epsilon  \|D\texttt{G} ( f )\|_{H^{s-1}}\|g\|_{L^\infty}+ \|D\texttt{G} ( f )\|_{H^s}\|g\|_{L^\infty}+\|D\texttt{G} ( f )\|_{L^\infty}\|Dg\|_{H^{s-1}}   \nnb \\
	\lesssim  & \|\texttt{G} (f)\|_{R^{s+1}}\| g\|_{H^s} \nnb \\
     \leq &C(\|f \|_{R^{s+1}})\|f \|_{R^{s+1}}\| g\|_{H^s}
	\end{align*}
holds for smooth functions $\texttt{G}(u)$ that vanish to first order in $u$.
Making use of this inequality, we find that
	\begin{align}\label{e:te4}
	\|D(\Delta-\epsilon^2\lambda )^{-1}\epsilon^3 (\texttt{H}_5(\epsilon,f_5, \del{i}\del{j} f_6))\|_{H^s}
	\leq C_0  \epsilon^2 \|f_5\|_{R^{s+1}}\|f_6\|_{R^{s+1}}
	\end{align}
and
	\begin{align}
	&\|D(\Delta-\epsilon^2\lambda )^{-1} \epsilon^3  \bigl(\texttt{H}_7(\epsilon, f_7,\del{i} f_8) \bigr)\|_{H^s}
\lesssim \epsilon^2 \|f_7\|_{R^{s+1}} \|f_8\|_{H^s} .   \label{e:te7}
	\end{align}
By Proposition \ref{t:ddel}.\eqref{ddel2}, we deduce that
	\begin{align}
	\|D(\Delta-\epsilon^2\lambda )^{-1}\epsilon\del{i}\del{j} f_3\|_{H^s}\lesssim   \epsilon \|Df_3\|_{H^s} \AND
	\|D(\Delta-\epsilon^2\lambda )^{-1}\epsilon\del{i} f_{11}\|_{H^s}\lesssim    \epsilon \|f_{11}\|_{H^s} 
	\label{e:te6}
	\end{align}
while it is clear that
\begin{equation} \label{e:te6a}
\|\epsilon^2(\epsilon^2 \lambda-\Delta)^{-1}D  f_{10} \|_{H^s}\lesssim \epsilon \|  f_{10} \|_{H^s}
\end{equation}
is a direct consequence of  Proposition \ref{T:Yu3est}.

With the help of Proposition \ref{T:Yu3est}, Theorem \ref{moserlemB} and \ref{T:Moser2}, we see that
	\begin{align*}
		\|(\epsilon^2\lambda-\Delta)^{-1}D(\epsilon \texttt{G}(f)g)\|_{H^s}\lesssim \|\texttt{G}(f)g\|_{H^s}\lesssim \|\texttt{G}(f)\|_{\Li}\|g\|_{H^s}+\|D\texttt{F}(f)\|_{\Hs}\|g\|_{\Li}\leq C(\|f\|_{R^s})\|f\|_{R^s}\|g\|_{H^s}
	\end{align*}
holds for smooth functions $\texttt{G}(u)$ that vanish to first order in $u$.
Using this inequality, we obtain the estimate
	\begin{align}\label{e:hf8}
	\| (\epsilon^2\lambda-\Delta)^{-1}D ( \epsilon^3 \texttt{H}_0(\epsilon,f_0, f_8))\|_{H^s}
	\leq  C_0\epsilon^2 \|f_0\|_{R^s} \|f_8\|_{H^s}.
	\end{align}
	Gathering the estimates \eqref{e:te1}-\eqref{e:hf8}, we arrive at the
estimate
	\begin{align}
	\|(\epsilon^2\lambda-\Delta)^{-1}\del{m} F\|_{H^s}& \lesssim C_0\Bigl[  \epsilon^2 \|f_1\|_{R^{s+1}}\|f_2\|_{R^{s+1}} + \epsilon \|Df_3\|_{H^s}+\|f_4\|_{L^\frac{6}{5}\cap K^s} + \epsilon^2 \|f_5\|_{R^{s+1}}\|f_6\|_{R^{s+1}}
\nnb  \\ &+\epsilon\|  f_9\|_{R^{s+1}}
	+\epsilon^2 ( \|f_7\|_{R^{s+1}}+  \|f_0\|_{R^s} ) \|f_8\|_{H^s} +\epsilon \|f_{11}\|_{H^s}  +\epsilon \|f_{10}\|_{H^s}+\epsilon \|f_{12}\|_{L^\frac{6}{5}\cap K^s}\Bigr]. \label{e:2est}
	\end{align}	
The estimate \eqref{e:inimod} is then a direct consequence
of \eqref{e:1est}, \eqref{e:2est}, the triangle inequality, and the inequality $\|f_8\|_{W^{1,6}}
\lesssim\|f_8\|_{R^s}\lesssim \|f_8\|_{H^s}$, which follows from \eqref{E:NORMEQ1}-\eqref{E:HQR}.

Finally, if $\texttt{G}(u)$ is a smooth function, it follows directly from Theorems \ref{Sobolev}, \ref{calcpropC} and \ref{moserlemB}
that
	\begin{align*}
	\|f_{10}\|_{H^s}= \|\texttt{G}(f)g\|_{H^s} \lesssim \|\texttt{G}(f)\|_{\Li}\|g\|_{H^s}+\|D\texttt{G}(f)\|_{\Hs}\|g\|_{\Li}\leq C(\|f\|_{K^s})\|g\|_{H^s} ,
	\end{align*}
which can be improved to
	 \begin{align*}
		\|f_{10}\|_{H^s}= \|\texttt{F}(f)g\|_{H^s} \lesssim \|\texttt{G}(f)\|_{\Li}\|g\|_{H^s}+\|D\texttt{G}(f)\|_{\Hs}\|g\|_{\Li}\leq C(\|f\|_{K^s})\|f\|_{K^s} \|g\|_{H^s}
	\end{align*}
if $\texttt{G}(u)$ also vanishes to first order in $u$.
\end{proof}

Next, we use the above lemma to prove that $\mathbf{G}$ is a contraction mapping in the following proposition.
\begin{proposition}\label{t:contraction}
	Suppose $s \in \Zbb_{\geq 3}$, $r>0$, $\mathbf{G}$ is a map defined by \eqref{e:phi.1}-\eqref{e:psi.1},
	$\mu(1)$ satisfies \eqref{inisetup2} and the free data $\breve{\xi}=(\smfu^{ij},\smfu^{ij}_0, \breve{z}^k, \delta\breve{\rho})$ is bounded by
	\begin{align}\label{e:normfreedata}
	\|\xi\|_s := \|\smfu^{ij} \|_{R^{s+1}}+\|\smfu^{ij}_{0 }\|_{H^s} +\|\delta\breve{\rho} \|_{L^{\frac{6}{5}}\cap K^s} +\|\breve{z}^j \|_{L^{\frac{6}{5}}\cap K^s}\leq r.
	\end{align}
	Then for
\begin{equation*}
l> \frac{2\Lambda}{3}\frac{7-4\mathring{\Omega}(1)-2\mathring{\Omega}^2(1)}{7-2\mathring{\Omega}(1)-4\mathring{\Omega}^2(1)}\mathring{E}^2(1)r,
\end{equation*} 
 there exists constants  $\epsilon_0>0$ and $k\in (0,1)$
 such that for every $\epsilon \in (0,\epsilon_0)$, $\mathbf{G}$ maps the closed ball
 $\overline{B_l\bigl(R^{s+1}(\Rbb^3,\Rbb^4)\bigr)}$ to itself and satisfies
	\begin{equation*}\label{e:contr}
	\|\mathbf{G}(\acute{\iota}_1)-\mathbf{G}(\acute{\iota}_2)\|_{R^{s+1}}\leq k\|\acute{\iota}_1-\acute{\iota}_2\|_{R^{s+1}}
	\end{equation*}
	for all $\acute{\iota}_1$, $\acute{\iota}_2\in\overline{B_l\bigl(R^{s+1}(\Rbb^3,\Rbb^4)\bigr)}$.
\end{proposition}
\begin{proof}
Suppose $\acute{\iota} = (\acute{\phi},\acute{\vartheta}^j)\in \overline{B_l\bigl(R^{s+1}(\Rbb^3,\Rbb^4)\bigr)}$ with the radius $l>0$
to be chosen later. Then by definition, $\grave{\iota}= \mathbf{G}(\acute{\iota})$, where
$\grave{\iota}=(\grave{\phi},\grave{\vartheta}^k)$ with $\grave{\phi}$ and $\grave{\vartheta}^k$
given by \eqref{e:phi.1} and \eqref{e:psi.1}, respectively. Differentiating $\grave{\phi}$ and $\grave{\vartheta}^k$
yields
	\begin{align}
	\del{m}\grave{\phi}= & \del{m}\bigl(\Delta-\epsilon^2 (\texttt{a}+\texttt{bd})\bigr)^{-1} \tilde{f}(\epsilon,\acute{\phi},\acute{\vartheta}^k,\breve{\xi})+\epsilon \texttt{b} \del{m} \del{j} (\Delta-\epsilon^2 \texttt{c})^{-1}\bigl(\Delta-\epsilon^2 (\texttt{a}+\texttt{bd})\bigr)^{-1} \tilde{g}^j(\epsilon,\acute{\phi},\acute{\vartheta}^k,\breve{\xi}), 
	 \label{e:phieq6} \\
	\del{m}\grave{\vartheta}^j= &
	\del{m}(\Delta-\epsilon^2 \texttt{c})^{-1} \tilde{g}^j(\epsilon, \acute{\phi},\acute{\vartheta}^k,\breve{\xi}). \label{e:psieq6}
	\end{align}
Then taking $L^6$ norm of $\grave{\phi}$ and $\grave{\vartheta}^j$ and the $H^s$ norm of $\del{m}\grave{\phi}$ and
$\del{m}\grave{\vartheta}^j$, we obtain, with the help of Proposition \ref{T:Yu3est}, the estimates
	\begin{align}
	\|\grave{\phi}\|_{R^{s+1}}
	\leq & 	\frac{-\texttt{bd}}{\texttt{a}+\texttt{bd}}\|\acute{\phi}\|_{R^{s+1}}+\|\bigl(\Delta -\epsilon^2 (\texttt{a}+\texttt{bd})\bigr)^{-1}f \bigl(\epsilon, \acute{\phi},\acute{\vartheta}^k+\epsilon \texttt{d} \del{}^k (\Delta-\epsilon^2 \texttt{c})^{-1} \acute{\phi},\breve{\xi}\bigr)\|_{R^{s+1}} \nnb  \\
	&
+ C\|(\Delta-\epsilon^2 \texttt{c})^{-1}  \tilde{g}^j(\epsilon, \acute{\phi},\acute{\vartheta}^j,\breve{\xi})\|_{R^{s+1}} ,  \label{e:est1R} \\
	\|\grave{\vartheta^j}\|_{R^{s+1}}
	\lesssim & \|(\Delta-\epsilon^2 \texttt{c})^{-1}  \tilde{g}^k(\epsilon, \acute{\phi},\acute{\vartheta}^j,\breve{\xi})\|_{R^{s+1}},
	\label{e:est2R}
	\end{align}
	where, by \eqref{e:addcon}, $-\texttt{bd}/(\texttt{a}+\texttt{bd})<1$. By looking at the expressions \eqref{e:A00}-\eqref{e:A0j},
\eqref{e:f}-\eqref{e:gj} and \eqref{e:tilfg1}-\eqref{e:tilfg2}, it is not difficult to see, by making
the identifications
	\begin{gather*}
	\texttt{H}_1(\epsilon,f_1,f_2 )=    \breve{\mathscr{S}}_1(\epsilon, \smfu^{\mu\nu})+ \breve{\mathscr{S}}^j_2(\epsilon, \smfu^{\mu\nu}),  \quad
	f_3=   E^2(1) \frac{\Lambda}{3} \smfu^{ i j}+\epsilon E^2(1)  \smfu^{i0}\smfu^{0j}, \quad
	f_4=  \frac{2\Lambda}{3} E^2(1) \delta\breve{\rho},  \\
	f_{11}  = -  E^2(1) \frac{\Lambda}{3} \smfu^{j i}_0,   \quad f_{12}=2 E^2(1)   \sqrt{\frac{\Lambda}{3}}   \breve{\rho} \breve{z}^ j, \quad
	\texttt{H}_7(\epsilon, f_7,\del{i} f_8) =  E^2(1)\smfu^{00} \del{i}\smfu^{ij}_0, \quad
	  f_9 =    \alpha \smfu^{kl},   \\
	\texttt{H}_0(\epsilon, f_0, f_8)  =  \breve{\mathscr{R}}_1(\epsilon, \epsilon\smfu^{\mu\nu}, \epsilon\smfu^{ij}_0, D\smfu^{\mu\nu},\smfu^{ij}_0  ) +  \breve{\mathscr{R}}^j_2(\epsilon, \epsilon \smfu^{\mu\nu}, \epsilon \smfu^{ij}_0, D\smfu^{\mu\nu},\smfu^{ij}_0 ),  \\
	\texttt{H}_5(\epsilon,f_5,\del{i}\del{j}f_6) = E^2(1)( \smfu^{00}\del{i}\del{j}\smfu^{ij}+  \smfu^{ij}\del{i}\del{j}\smfu^{00}+   \smfu^{0i}\del{i}\del{k}\smfu^{kj}+   \smfu^{kl}\del{k}\del{l}\smfu^{0j} )
\intertext{and}
	f_{10} = \beta\smfu^{kl}_0+ \sigma D\smfu^{kl} +  \breve{\mathscr{Q}}_1+\breve{\mathscr{Q}}^j_2 +\epsilon (\breve{\mathscr{B}}_1 +\breve{\mathscr{B}}^j_2)
	+( \breve{\mathscr{F}}_1+\epsilon^2\breve{\mathscr{F}}^j_2) +( \epsilon^2 \breve{\mathscr{G}}_1+\epsilon  \breve{\mathscr{G}}^j_2)
	\end{gather*}
for appropriate constants $\alpha$, $\beta$ and $\sigma$, that we can use Lemma \ref{T:Yuestdec2}
to estimate the terms
$\|\bigl(\Delta -\epsilon^2 (\texttt{a}+\texttt{bd})\bigr)^{-1}f \bigl(\epsilon, \acute{\phi},\acute{\vartheta}^k+\epsilon \texttt{d} \del{}^k (\Delta-\epsilon^2 \texttt{c})^{-1} \acute{\phi},\breve{\xi}\bigr)\|_{R^{s+1}}$ and
$\|(\Delta-\epsilon^2 \texttt{c})^{-1}  \tilde{g}^j(\epsilon, \acute{\phi},\acute{\vartheta}^j,\breve{\xi})\|_{R^{s+1}}$
that appear on the right hand side of \eqref{e:est1R} and \eqref{e:est2R}. Doing so, we find,
with the help of Theorems \ref{Sobolev}.\eqref{sob3} and \ref{moserlemB}, that
	\begin{align*}
		\|\grave{\phi}\|_{R^{s+1}} \leq & \Bigl( \frac{-\texttt{bd}}{\texttt{a}+\texttt{bd}}+C(l,r,\epsilon)\epsilon\Bigr) \|\acute{\phi}\|_{R^{s+1}}+C(l,r,\epsilon)\epsilon(\|\acute{\vartheta}^j\|_{R^{s+1}}+\|\breve{\xi}\|_{s})+\frac{2\Lambda}{3}E^2(1)\|\delta\breve{\rho}\|_{L^{6/5}\cap K^s}  \\
		\leq &\Bigl( \frac{-\texttt{bd}}{\texttt{a}+\texttt{bd}}+C(l,r,\epsilon)\epsilon\Bigr)l +\Bigl(\frac{2\Lambda}{3}E^2(1)+
C(l,r,\epsilon)\epsilon \Bigr)r \\
		\intertext{and}
		\|\grave{\vartheta}^k\|_{R^{s+1}} \leq & C(l,r,\epsilon) \epsilon(\|\acute{\vartheta}^j\|_{R^{s+1}}+\|\acute{\phi}\|_{R^{s+1}}+\|\breve{\xi}\|_{s}) \leq C(l,r,\epsilon)\epsilon l.
	\end{align*}
From this we see that
\begin{equation*}
\|(\grave{\phi},\grave{\vartheta}^k)\|_{R^{s+1}}\bigl|_{\epsilon=0} =\bigl(\|\grave{\phi}\|_{R^{s+1}}
+\|\grave{\vartheta}^k\|_{R^{s+1}}\bigr)\bigl|_{\epsilon=0}
\leq  \Bigl(\frac{-\texttt{bd}}{\texttt{a}+\texttt{bd}}\Bigr)\biggl|_{\epsilon=0}l +\frac{2\Lambda}{3}E^2(1)\bigl|_{\epsilon=0} r,
\end{equation*}
which, recalling that $-\texttt{bd}/(\texttt{a}+\texttt{bd})<1$ by \eqref{e:addcon},  we can satisfy
\begin{equation*}
\|(\grave{\phi},\grave{\vartheta}^k)\|_{R^{s+1}}\bigl|_{\epsilon=0} < l
\end{equation*}
by choosing  $l$ so that
	\begin{align*}
		l>
\biggl(\frac{2\Lambda E^2(1)(\texttt{a}+\texttt{bd})}{3(\texttt{a}+2\texttt{bd})}\biggr)\biggl|_{\epsilon=0} r =\frac{2\Lambda}{3}\frac{7-4\mathring{\Omega}(1)-2\mathring{\Omega}^2(1)}{7-2\mathring{\Omega}(1)-4\mathring{\Omega}^2(1)}\mathring{E}^2(1) r.
	\end{align*}
It then follows from the continuous dependence of the constants on $\epsilon$ that there exists an
$\epsilon_0=\epsilon_0(l,r)$ such that $\|(\grave{\phi},\grave{\vartheta}^k)\|_{R^{s+1}} < l$ for all
$\epsilon \in (0,\epsilon_0)$, or in other words, $\mathbf{G}$ maps the closed ball  $\overline{B_l\bigl(R^{s+1}(\Rbb^3,\Rbb^4)\bigr)}$
to itself for all $\epsilon \in (0,\epsilon_0)$.

Due to the linearity of Yukawa potential operator, derivatives and the Riesz transform, calculation similar to those
used to derive \eqref{e:est1R}-\eqref{e:est2R} show that
	\begin{align}
	\|\grave{\phi}_1-\grave{\phi}_2\|_{R^{s+1}} 
	\leq	&  \|\bigl(\Delta -\epsilon^2 (\texttt{a}+\texttt{bd}) \bigr)^{-1} \bigl[ \epsilon^2 \bigl( \texttt{A}^{0}(\epsilon, \acute{\phi}_1,\acute{\psi}^j_1,\breve{\xi})-\texttt{A}^{0}(\epsilon, \acute{\phi}_2,\acute{\psi}^j_2,\breve{\xi})\bigr) \bigr] \|_{R^{s+1}} \nnb \\
	& +\frac{-\texttt{bd}}{\texttt{a}+\texttt{bd}}\|\acute{\phi}_1-\acute{\phi}_2\|_{R^{s+1}}+ \|(\Delta-\epsilon^2 c)^{-1} \epsilon^2 \bigl( \texttt{A}^{j}(\epsilon, \acute{\phi}_1,\acute{\psi}^j_1,\breve{\xi})-\texttt{A}^{j}(\epsilon, \acute{\phi}_2,\acute{\psi}^j_2,\breve{\xi})\bigr)\|_{R^{s+1}},  \label{e:diff1} \\
	\|\grave{\vartheta}^j_1-\grave{\vartheta}^j_2\|_{R^{s+1}} 
	\lesssim &    \|(\Delta-\epsilon^2 \texttt{c})^{-1} \epsilon^2 \bigl( \texttt{A}^{j}(\epsilon, \acute{\phi}_1,\acute{\psi}^j_1,\breve{\xi})-\texttt{A}^{j}(\epsilon, \acute{\phi}_2,\acute{\psi}^j_2,\breve{\xi})\bigr)\|_{R^{s+1}},
	 \label{e:diff2}
	\end{align}
where  $-\texttt{bd}/(\texttt{a}+\texttt{bd})<1$ by \eqref{e:addcon}.
Defining maps $\texttt{B}^{\mu}(\epsilon, \acute{\phi}_2,\acute{\psi}^j_2,\acute{\phi}_1-\acute{\phi}_2,\acute{\psi}^j_1-\acute{\psi}^j_2 ,\breve{\xi} )$ by
\begin{equation*}
\texttt{B}^{\mu}(\epsilon, \acute{\phi}_2,\acute{\psi}^j_2,\acute{\phi}_1-\acute{\phi}_2,\acute{\psi}^j_1-\acute{\psi}^j_2 ,\breve{\xi} ) = \epsilon^2\bigl( \texttt{A}^{\mu}(\epsilon, \acute{\phi}_1,\acute{\psi}^j_1,\breve{\xi})-\texttt{A}^{\mu}(\epsilon, \acute{\phi}_2,\acute{\psi}^j_2,\breve{\xi})\bigr),
\end{equation*}
which we note are analytic in all variables and vanish to first order in $(\acute{\phi}_1-\acute{\phi}_2,\acute{\psi}^j_1-\acute{\psi}^j_2)$,
we can use Lemma \ref{T:Yuestdec2} in a similar fashion as above, although this time with $f_4=f_{11}=f_{12}=f_9= 0$,
to obtain from \eqref{e:diff1}-\eqref{e:diff2} the estimate
\begin{equation*}
\|(\grave{\phi}_1-\grave{\phi}_2,\grave{\psi}^j_1-\grave{\psi}^j_2)\|_{R^{s+1}}
\leq k\|(\acute{\phi}_1-\acute{\phi}_2,\acute{\psi}^j_1-\acute{\psi}^j_2)\|_{R^{s+1}}
\end{equation*}
for all $(\acute{\phi}_a,\acute{\psi}^j_a)\in \overline{B_l\bigl(R^{s+1}(\Rbb^3,\Rbb^4)\bigr)}$, $a=1,2$, where  $k=\max\{C(l,r,\epsilon)\epsilon, -\texttt{bd}/(\texttt{a}+\texttt{bd})+C(l,r,\epsilon)\epsilon\}$. Since $0<-\texttt{bd}/(\texttt{a}+\texttt{bd})<1$, it is
clear that by shrinking $\epsilon_0$, if necessary, that we can arrange $k \in (0,1)$ for all $\epsilon \in (0,\epsilon_0)$.
\end{proof}

\subsubsection{Existence}
 We now use the contraction map $\mathbf{G}$ to establish the existence of $1$-parameter families of initial data that solve the
constraint equations.

\begin{remark}
	 All solutions in this article, whether they are solutions of the constraint equations or the evolution equations, depend on the singular parameter $\epsilon$ and the free data. Depending on context and what we want to emphasize, we will either make the dependence on $\epsilon$ explicit by including an $\epsilon$
subscript, e.g. $u_\epsilon^{\mu\nu}$, or treat the $\epsilon$ dependence as implicit, e.g.  $u^{\mu\nu}$. We will also use the subscript notation to make explicit 
the dependence of the solution on other initial data
parameters, e.g. $u_{\epsilon,\yve}^{\mu\nu}$.
\end{remark}

\begin{theorem}\label{t:inithm}
	Suppose $s\in \Zbb_{\geq 3}$, $r>0$, $\mu(1)$ satisfies \eqref{inisetup2}, and the free initial data
 $\breve{\xi}=(\smfu^{ij},\smfu^{ij}_0, \breve{z}^k, \delta\breve{\rho})$ is bounded by
	\begin{align*}
	\|\breve{\xi}\|_{s}=\|\smfu^{ij} \|_{R^{s+1}}+\|\smfu^{ij}_{0 }\|_{H^s} +\|\delta\breve{\rho} \|_{L^{\frac{6}{5}}\cap K^s} +\|\breve{z}^j \|_{L^{\frac{6}{5}}\cap K^s}\leq r.
	\end{align*}	
	Then there exists an  $\epsilon_0>0$ and a family of one parameter maps
$(\smfu^{0\mu}_\epsilon, \smfu^{0\mu}_{0,\epsilon}) \in  R^{s+1}(\Rbb^3,\Rbb^4)\times R^s(\Rbb^3,\Rbb^4)$, $0<\epsilon <
\epsilon_0$,
such that
	\begin{align*}
	\underline{\hmfu^{\mu\nu}_\epsilon}\bigl|_{\Sigma}=\p{\smfu^{00}_\epsilon & \smfu^{0j}_\epsilon\\
		\smfu^{i0}_\epsilon & \epsilon \smfu^{ij}} \AND
	\underline{\hmfu^{\mu\nu}_{0,\epsilon}}\bigl|_{\Sigma}=\p{\smfu^{00}_{0,\epsilon} & \smfu^{0j}_{0,\epsilon}\\
		\smfu^{i0}_{0,\epsilon} & \smfu^{ij}_{0}},
	\end{align*}
	where the $\underline{\smfu^{0\mu}_{0,\epsilon}}$ are determined by \eqref{e:D0u00}-\eqref{e:D0uk0},
	solve the constraint equations \eqref{E:CONSTRAINT}-\eqref{E:NORMALIZATION} for $0<\epsilon<\epsilon_0$. Moreover, $\{
\smfu^{00}_\epsilon, \smfu^{0j}_{\epsilon}\}$ and  $\{\del{m} \smfu^{00}_\epsilon, \del{m}\smfu^{0j}_\epsilon \}$ can be expanded as 
	\begin{align}
	\smfu^{00}_\epsilon=
&\frac{2\Lambda}{3 }E^2(1) \bigl(\Delta-\epsilon^2 (\texttt{a}+\texttt{bd})\bigr)^{-1}\delta\breve{\rho} +  \mathcal{R}	(\epsilon,\smfu^{kl},\smfu^{kl}_0, \delta\breve{\rho}, \breve{z}^l), \label{e:iniuexp1} \\
	\smfu^{0j}_\epsilon =& \frac{2\Lambda}{3 }\epsilon \texttt{d} E^2(1) \del{}^j (\Delta-\epsilon^2 \texttt{c})^{-1}\bigl(\Delta-\epsilon^2(\texttt{a}+\texttt{bd})\bigr)^{-1} \delta\breve{\rho}
	+ \mathcal{R}^j	(\epsilon,\smfu^{kl},\smfu^{kl}_0, \delta\breve{\rho}, \breve{z}^l) \label{e:iniuexp2}
	\end{align}
	and
	\begin{align}
	\del{m} \smfu^{00}_\epsilon =  \frac{2\Lambda}{3}E^2(1)\del{m} \Delta^{-1}\delta\breve{\rho}+\epsilon \mathcal{S}_m (\epsilon,\smfu^{kl},\smfu^{kl}_0, \delta\breve{\rho}, \breve{z}^l)  \AND
	\del{m} \smfu^{0j}_\epsilon = \epsilon \mathcal{S}^j_m	(\epsilon,\smfu^{kl},\smfu^{kl}_0, \delta\breve{\rho}, \breve{z}^l), \label{e:iniuexpsp}
	\end{align}
respectively,
	where the remainder terms 
are bounded by
	\begin{equation}
	\|\mathcal{R}(\epsilon,\smfu^{kl},\smfu^{kl}_0, \delta\breve{\rho}, \breve{z}^l)\|_{R^{s+1}}+\|\mathcal{R}^j(\epsilon,\smfu^{kl},\smfu^{kl}_0, \delta\breve{\rho}, \breve{z}^l)\|_{R^{s+1}}+\|\mathcal{S}_m(\epsilon,\smfu^{kl},\smfu^{kl}_0, \delta\breve{\rho}, \breve{z}^l)\|_{R^s}+\|\mathcal{S}^j_m(\epsilon,\smfu^{kl},\smfu^{kl}_0, \delta\breve{\rho}, \breve{z}^l)\|_{R^s}\lesssim
\|\breve{\xi}\|_s,
	\label{e:RSEst}
	\end{equation}
and $\{\smfu^{0\mu}_\epsilon, \smfu^{0\mu}_{0,\epsilon}\}$ satisfy the uniform estimates
	\begin{equation}
	\|\smfu^{0\mu}_\epsilon\|_{R^{s+1}}+\|\smfu^{0\mu}_{0,\epsilon}\|_{R^s}\lesssim
\|\breve{\xi}\|_s
	 \label{e:uconstr}
	\end{equation}
for $\epsilon\in (0,\epsilon_0)$.
\end{theorem}
\begin{proof}
Given $r>0$, choose $l>\frac{2\Lambda}{3}\frac{7-4\mathring{\Omega}(1)-2\mathring{\Omega}^2(1)}{7-2\mathring{\Omega}(1)-4\mathring{\Omega}^2(1)}\mathring{E}^2(1)r>0$ and let $\epsilon_0>0$ and $k\in (0,1)$ be as in Proposition \ref{t:contraction}. Since we know by Proposition  \ref{t:contraction}
that  $\mathbf{G}$ is a contraction mapping on
$\overline{B_l(R^{s+1}(\Rbb^3,\Rbb^4))}$, it follows from Banach's fixed point theorem
that $\mathbf{G}$ has a unique fixed point $\acute{\iota}_*=(\phi,\vartheta^k) =(\smfu^{00}_\epsilon, \smfu^{0k}_\epsilon-\epsilon d \del{ }^j(\Delta-\epsilon^2\texttt{c})^{-1}\smfu^{00}_\epsilon) \in \overline{B_l(R^{s+1}(\Rbb^3,\Rbb^4))}$, that is, $\mathbf{G}(\acute{\iota}_*)=\acute{\iota}_*$.
Furthermore, we know that the successive approximations $\acute{\iota}_{m}=\mathbf{G}^m (\acute{\iota}_0)$
starting from any seed $\acute{\iota}_0 \in \overline{B_l(R^{s+1}(\Rbb^3,\Rbb^4))}$ converge to $\acute{\iota}_*$ and satisfy \begin{align}\label{e:diffNew}
	\|\acute{\iota}_0-\acute{\iota}_*\|_{R^{s+1}}\leq \frac{1}{1-k}\|\mathbf{G} (\acute{\iota}_0)-\acute{\iota}_0\|_{R^{s+1}}.
	\end{align}

In the following, we consider the seed
$\acute{\iota}_0=( \phi_{\text{seed}}, \vartheta_{\text{seed}}^j)$ defined by
	\begin{align*}
	\phi_\text{seed}=
	\frac{2\Lambda}{3 }E^2(1) \bigl(\Delta-\epsilon^2 (\texttt{a}+\texttt{bd})\bigr)^{-1}\delta\breve{\rho}  
	\AND		
	\vartheta_\text{seed}^j=   0.  
	\end{align*}
Since $\delta\breve{\rho} \in L^{\frac{6}{5}}\cap K^s$, it follows from
	 Proposition \ref{t:ddel}.\eqref{ddel2}, \eqref{e:sYupq2}-\eqref{e:sYupq3} and \eqref{e:addcon} that
\begin{equation*}
\|\bigl(\Delta-\epsilon^2 (\texttt{a}+\texttt{bd})\bigr)^{-1}\delta\breve{\rho}\|_{L^6} \lesssim \|\delta\breve{\rho}\|_{L^{\frac{6}{5}}\cap K^s}
\end{equation*}
and
	 \begin{align*}
	 	\|\del{m}\phi_{\text{seed}}\|_{H^s} \lesssim \|\del{m}\bigl(\Delta-\epsilon^2 (\texttt{a}+\texttt{bd})\bigr)^{-1}\delta\breve{\rho}\|_{H^s}&\lesssim \|D\bigl(\Delta-\epsilon^2 (\texttt{a}+\texttt{bd})\bigr)^{-1}\delta\breve{\rho}\|_{L^2}+\|D^2\bigl(\Delta-\epsilon^2 (\texttt{a}+\texttt{bd})\bigr)^{-1}  \delta\breve{\rho}\|_{H^{s-1}}  \\
	 	&\lesssim \|\delta\breve{\rho}\|_{L^{\frac{6}{5}}}+\|\delta\breve{\rho}\|_{H^{s-1}}\lesssim \|\delta\breve{\rho}\|_{L^{\frac{6}{5}}\cap K^s}
	 \end{align*}
from which we deduce that $\acute{\iota}_{0} \in  \overline{B_l(R^{s+1}(\Rbb^3,\Rbb^4))}$ provided $l$ is chosen large enough.

Next, we estimate $\|\mathbf{G} (\acute{\iota}_0)-\acute{\iota}_0\|_{R^{s+1}}$. Before doing so, we let
$\mathbf{G}(\acute{\iota}) =(\mathbf{G}^0(\acute{\iota}),\mathbf{G}^j(\acute{\iota}))$ denote the decomposition
of $\mathbf{G}$ into components, where $\acute{\iota}=(\acute{\phi},\acute{\vartheta}^j)$, and the components
$\mathbf{G}^0(\acute{\iota})$ and
 $\mathbf{G}^j(\acute{\iota})$ are given by the formulas \eqref{e:phi.1} and \eqref{e:psi.1}, respectively. We then find via a direct calculation
involving \eqref{e:f} and \eqref{e:phi.1}-\eqref{e:psi.1} that difference $\mathbf{G}(\acute{\iota}_0)-\acute{\iota}_0$ is given by
	\begin{align*}
	\mathbf{G}^0(\phi_\text{seed}, \vartheta^j_\text{seed})-\phi_\text{seed}=& \bigl(\Delta-\epsilon^2(\texttt{a}+\texttt{bd})\bigr)^{-1}\bigl[\epsilon E^2(1)\frac{\Lambda}{3} \del{i}\del{j}\smfu^{ij}+\epsilon^2\texttt{A}^{0}(\epsilon,\phi_\text{seed},\psi^j_\text{seed},\breve{\xi})+\epsilon^4 \texttt{bcd} (\Delta-\epsilon^2 \texttt{c})^{-1} \phi_{\text{seed}}\bigr] \\
	& + \epsilon \texttt{b} \del{j}(\Delta-\epsilon^2 \texttt{c})^{-1}\bigl(\Delta-\epsilon^2(\texttt{a}+\texttt{bd})\bigr)^{-1} \tilde{g}^j(\epsilon, \phi_{\text{seed}}, \vartheta_{\text{seed}}^k)
	\intertext{and}
	\mathbf{G}^j(\phi_\text{seed}, \vartheta^k_\text{seed})-\vartheta_\text{seed}^j= &
	(\Delta-\epsilon^2 \texttt{c})^{-1}\tilde{g}^k(\epsilon, \phi_{\text{seed}}, \vartheta^j_{\text{seed}})
	\end{align*}
	where
	\begin{align*}
	\psi^j_\text{seed}:= \vartheta^j_\text{seed}+\epsilon \texttt{d} \del{}^j(\Delta-\epsilon^2 \texttt{c})^{-1}  \phi_\text{seed}
	= \epsilon  \frac{2\Lambda}{3 }E^2(1) \texttt{d} \del{}^j(\Delta-\epsilon^2 \texttt{c})^{-1} \bigl(\Delta-\epsilon^2 (\texttt{a}+\texttt{bd})\bigr)^{-1}\delta\breve{\rho}.
	\end{align*}
By similar arguments used to derive \eqref{e:est1R}-\eqref{e:est2R}, we can estimate
$\|\mathbf{G} (\acute{\iota}_0)-\acute{\iota}_0\|_{R^{s+1}}$ using Lemma \ref{T:Yuestdec2}, with $f_4$ set to zero, to get
	\begin{align*}
	\|\mathbf{G} (\acute{\iota}_0)-\acute{\iota}_0\|_{R^{s+1}} 
	\lesssim & 
		\|\smfu^{ij} \|_{R^{s+1}}+\|\smfu^{ij}_{0 }\|_{H^s} +\|\delta\breve{\rho} \|_{L^{\frac{6}{5}}\cap K^s} +\|\breve{z}^j \|_{L^{\frac{6}{5}}\cap K^s} = \norm{\breve{\xi}}_{s},
	\end{align*}
from which it follows that
	\begin{align*}
	\|\acute{\iota}_0-\acute{\iota}_*\|_{R^{s+1}}
	\lesssim \norm{\breve{\xi}}_{s}
	\end{align*}
by \eqref{e:diffNew}. Since $ \acute{\iota}_* =(\smfu^{00}_\epsilon, \smfu^{0k}_\epsilon-\epsilon d \del{ }^j(\Delta-\epsilon^2\texttt{c})^{-1}\smfu^{00}_\epsilon)$ and $\norm{\acute{\iota}_0}_{R^{s+1}} \lesssim  \norm{\breve{\xi}}_{s}$, it is clear
that \eqref{e:iniuexp1}--\eqref{e:iniuexp2} follows from the above estimate.

	In order to bound $\del{m}\smfu^{0j}_\epsilon$, we first note that the estimate
	\begin{align*}
	\|(\Delta-\epsilon^2\lambda)^{-1}\delta\breve{\rho}\|_{R^s}\lesssim & \|(\Delta-\epsilon^2\lambda)^{-1}\delta\breve{\rho}\|_{L^6}+\sum_l\|(\Delta-\epsilon^2\lambda)^{-1}\del{l}\delta\breve{\rho}\|_{L^2}+\sum_{k,l}\|(\Delta-\epsilon^2\lambda)^{-1}\del{l}\del{k}\delta\breve{\rho}\|_{H^{s-2}} \nnb  \\
	\lesssim & \|\delta\breve{\rho}\|_{L^{\frac{6}{5}}}+\|\delta\breve{\rho}\|_{H^{s-2}}, 
	\end{align*}
	which holds for any constant $\lambda\geq 0$, follows from \eqref{e:sYupq2}, \eqref{e:sYupq3}, and Propositions \ref{t:ddel}.\eqref{ddel1} and \ref{t:ddel}.\eqref{ddel2}.
From this,
\eqref{e:varphi} and  \eqref{e:phieq6}--\eqref{e:psieq6}, we then get from an application of Propositions \ref{T:genYu}, \ref{t:ddel}, \ref{T:Yu3est} and  \ref{E:ddelin1}, and Theorem \ref{T:riepot}
the estimates
	\begin{align*}
	& \Bigl\|\del{m} \smfu^{00}_\epsilon  -\frac{2\Lambda}{3}E^2(1)\del{m}\Delta^{-1}\delta\breve{\rho}\Bigr\|_{R^{s}}=\Bigl\|\del{m} \phi -\frac{2\Lambda}{3}E^2(1)\del{m}\Delta^{-1}\delta\breve{\rho}\Bigr\|_{R^{s}} \nnb  \\
	\lesssim & \|\del{m}\bigl(\Delta-\epsilon^2 (\texttt{a}+\texttt{bd})\bigr)^{-1} f(\epsilon, \phi,\psi^k,\breve{\xi}) -\frac{2\Lambda}{3}E^2(1)\del{m}\Delta^{-1}\delta\breve{\rho} \|_{R^s} + \|\del{m}\bigl(\Delta-\epsilon^2 (\texttt{a}+\texttt{bd})\bigr)^{-1} \epsilon^4\texttt{bcd} (\Delta-\epsilon^2\texttt{c})^{-1}\phi\|_{R^s} \nnb  \\
	& +\|\epsilon \texttt{b} \del{m} \del{j} (\Delta-\epsilon^2 \texttt{c})^{-1}\bigl(\Delta-\epsilon^2 (\texttt{a}+\texttt{bd})\bigr)^{-1} \tilde{g}^j(\epsilon, \phi, \vartheta^k,\breve{\xi})\|_{R^s}\nnb  \\
	\lesssim & \Bigl\|\del{m}\bigl(\Delta-\epsilon^2 (\texttt{a}+\texttt{bd})\bigr)^{-1} \Bigl(f(\epsilon, \phi,\psi^k,\breve{\xi})-\frac{2\Lambda}{3}E^2(1)\delta\breve{\rho}\Bigr)+\epsilon^2(\texttt{a}+\texttt{bd}) \frac{2\Lambda}{3}E^2(1)\del{m}\bigl(\Delta-\epsilon^2(\texttt{a}+\texttt{bd})\bigr)^{-1}\Delta^{-1}\delta\breve{\rho} \Bigr\|_{R^s}\nnb  \\
	& + \epsilon \|\phi\|_{R^s}  + \| (\Delta-\epsilon^2 \texttt{c})^{-1}  \tilde{g}^j(\epsilon,\phi, \vartheta^k, \breve{\xi})\|_{R^{s+1}} \nnb  \\
	\lesssim & \Bigl\| \bigl(\Delta-\epsilon^2 (\texttt{a}+\texttt{bd} )\bigr)^{-1} \Bigl(f(\epsilon, \phi,\psi^k,\breve{\xi})-\frac{2\Lambda}{3}E^2(1)\delta\breve{\rho}\Bigr)\Bigr\|_{R^{s+1}}+\epsilon \|\delta\breve{\rho}\|_{L^{\frac{6}{5}}}+\epsilon \|\delta\breve{\rho}\|_{H^{s-2}}+ \epsilon \|\phi\|_{R^s}  \nnb  \\
	& +  \| (\Delta-\epsilon^2 \texttt{c})^{-1}  \tilde{g}^j(\epsilon, \phi , \vartheta^k,\breve{\xi})\|_{R^{s+1}}
	\end{align*}
	and
	\begin{align*}
		\|\del{m}\smfu^{0j}_\epsilon \|_{R^s}=\|\del{m}\psi^j\|_{R^s}\lesssim & \|\del{m}\vartheta^j\|_{R^s}+\|\epsilon \texttt{d} \del{m} \del{}^j(\Delta-\epsilon^2 \texttt{c})^{-1} \phi\|_{R^s}\lesssim \|(\Delta-\epsilon^2 \texttt{c})^{-1}\tilde{g}^j(\epsilon, \phi , \vartheta^k,\breve{\xi})\|_{R^{s+1}}+\epsilon\|  \phi\|_{R^s}
	\end{align*}
Finally, \eqref{e:iniuexpsp}, \eqref{e:RSEst} and \eqref{e:uconstr} follow directly from
Lemma \ref{T:Yuestdec2} and \eqref{e:iniuexp1}--\eqref{e:iniuexp2}.  This completes the proof of the theorem.
\end{proof}

\subsection{Bounding initial evolution variables}   \label{E:Bound}
For the evolution problem, we will need to translate the $\epsilon$-independent bound on the initial data from
Theorem \ref{t:inithm} to an $\epsilon$-independent bound on the initial data $\mathbf{\hat{U}}|_{\Sigma}$ for
the formulation \eqref{E:LOCEQ} of the reduced conformal Einstein-Euler equations. The following
proposition serves this purpose.

\begin{proposition}\label{L:INITIALTRANSFER}
	Suppose that the hypotheses of Theorem \ref{t:inithm} hold and that $\breve{\xi}$, $\|\breve{\xi}\|_s$,
and the maps
	$\{\smfu^{0\mu},\smfu^{0\mu}_0\}$ are as given in Theorem \ref{t:inithm}.
	Then on the initial hypersurface $\Sigma$, the collection
	\begin{equation*}
	\{u^{\mu\nu}_{\epsilon},u^{ij}_{\gamma,\epsilon},
	u^{0\mu}_{i,\epsilon},u^{0\mu}_{0,\epsilon},u_\epsilon,u_{\gamma,\epsilon},z_{j,\epsilon},\delta \zeta_\epsilon\}
	\end{equation*}
of  gravitational and matter fields can be written as
	\begin{align}
	u^{0\mu}_\epsilon|_{\Sigma}&=  
    \epsilon 
	\mathcal{S}^\mu
	(\epsilon,\smfu^{kl},\smfu^{kl}_0, \delta\breve{\rho}, \breve{z}^l), \label{E:U0MUANDINI2}
	\\
	u_\epsilon|_{\Sigma}&=\epsilon^2\frac{2\Lambda}{9}E^2(1)\smfu^{ij}\delta_{ij}+\epsilon^3 \mathcal{S}(\epsilon,\smfu^{kl},\smfu^{kl}_0, \delta\breve{\rho}, \breve{z}^l),  \label{E:UANDINI2}
	\\
	u^{ij}_\epsilon|_{\Sigma}&=\epsilon^2 \left(\smfu^{ij}
	-\frac{1}{3}\smfu^{kl}\delta_{kl}\delta^{ij}\right)+\epsilon^3\mathcal{S}^{ij}
	(\epsilon,\smfu^{kl},\smfu^{kl}_0, \delta\breve{\rho}, \breve{z}^l), \label{E:UIJANDINI2}
	\\
	z_{j,\epsilon}|_{\Sigma}&= E^2(1) \delta_{kl}\breve{z}^k
	+\epsilon\mathcal{R}_j
	(\epsilon,\smfu^{kl},\smfu^{kl}_0, \delta\breve{\rho}, \breve{z}^l),\label{E:ZJINI}
	\\
	\delta\zeta_\epsilon|_{\Sigma}&=\frac{1}{1+\epsilon^2 K}\ln{\left(1+\frac{\delta\breve{\rho}}{\mu(1)} \right)}, \label{E:DELTAZETA}
	\\
	u^{0\mu}_{i,\epsilon}|_{\Sigma} &= \frac{\Lambda}{3}E^2(1)\delta^\mu_0\del{i} \Delta^{-1}\delta\breve{\rho} +\epsilon \mathcal{S}^\mu_i
	(\epsilon,\smfu^{kl},\smfu^{kl}_0, \delta\breve{\rho}, \breve{z}^l), \label{wexp}
	\\
	u^{0\mu}_{0,\epsilon}|_{\Sigma} &= \epsilon \mathcal{S}^\mu_0
	(\epsilon,\smfu^{kl},\smfu^{kl}_0, \delta\breve{\rho}, \breve{z}^l), \label{u0muexp}
	\\
	u_{\gamma,\epsilon}|_{\Sigma} &= \epsilon \mathcal{S}_\gamma
	(\epsilon,\smfu^{kl},\smfu^{kl}_0, \delta\breve{\rho}, \breve{z}^l)  \label{ugammaexp}
	\intertext{and}
	u^{ij}_{\gamma,\epsilon}|_{\Sigma} &= \epsilon \mathcal{S}^{ij}_\gamma
	(\epsilon,\smfu^{kl},\smfu^{kl}_0, \delta\breve{\rho}, \breve{z}^l) \label{uijgammaexp}
	\end{align}
where the remainder terms  satisfy bounds of the form
	\begin{align*}
		\|\mathcal{S}^\mu(\epsilon,\smfu^{kl},\smfu^{kl}_0, \delta\breve{\rho}, \breve{z}^l)\|_{R^{s+1}}+\|\mathcal{S}(\epsilon,\smfu^{kl},\smfu^{kl}_0, \delta\breve{\rho}, \breve{z}^l)\|_{R^{s+1}}+\|\mathcal{S}^{ij} (\epsilon,\smfu^{kl},\smfu^{kl}_0, \delta\breve{\rho}, \breve{z}^l)\|_{R^{s+1}}& \nnb\\
+\|\mathcal{R}_j(\epsilon,\smfu^{kl},\smfu^{kl}_0, \delta\breve{\rho}, \breve{z}^l)\|_{R^{s+1}}
		+\|\mathcal{S}^\mu_i(\epsilon,\smfu^{kl},\smfu^{kl}_0, \delta\breve{\rho}, \breve{z}^l)\|_{R^{s+1}} +\|\mathcal{S}^\mu_0(\epsilon,\smfu^{kl},\smfu^{kl}_0, \delta\breve{\rho}, \breve{z}^l)\|_{R^{s+1}} \nnb \\
+\|\mathcal{S}_\gamma (\epsilon,\smfu^{kl},\smfu^{kl}_0, \delta\breve{\rho}, \breve{z}^l)\|_{R^{s+1}}
		 +\|\mathcal{S}^{ij}_\gamma(\epsilon,\smfu^{kl},\smfu^{kl}_0, \delta\breve{\rho}, \breve{z}^l)\|_{R^{s+1}}&\lesssim \|\breve{\xi}\|_s.
	\end{align*}
	Moreover, the estimates
	\begin{align}\label{E:INITIALDATAESTIMATE2}
		\|u^{\mu\nu}_\epsilon|_{\Sigma}\|_{R^{s+1}}+\|u_\epsilon|_{\Sigma}&\|_{R^{s+1}}+\|u^{0k}_{i,\epsilon}|_{\Sigma}
	\|_{R^s}+\|u^{0\mu}_{0,\epsilon}|_{\Sigma}\|_{R^s}+
	\|u_{\mu,\epsilon}|_{\Sigma}\|_{R^s} +\|u^{ij}_{\mu,\epsilon}|_{\Sigma}\|_{R^s}
	\lesssim \epsilon  \|\breve{\xi}\|_s
	\end{align}
	and
	\begin{align}\label{E:INITIALDATAESTIMATE3}
\|u^{00}_{i,\epsilon}|_{\Sigma}\|_{R^s}+\|z_{j,\epsilon}|_{\Sigma}\|_{R^s}+\|\delta\zeta_\epsilon|_{\Sigma}\|_{L^{\frac{6}{5}}\cap K^s}\lesssim
	 \|\breve{\xi}\|_s
	\end{align}
	hold uniformly for $\epsilon \in (0,\epsilon_0)$.
\end{proposition}
\begin{proof}	
To start, we observe that \eqref{E:U0MUANDINI2}-\eqref{E:UIJANDINI2} are a direct consequence of Lemma \ref{L:RELATION1} and the expansions \eqref{e:iniuexp1}-\eqref{e:iniuexp2}, while
 \eqref{E:ZJINI} follows from \eqref{E:CHECKGIJ}, \eqref{E:V_0}, \eqref{E:VELOCITY} and  \eqref{E:U0MUANDINI2}-\eqref{E:UIJANDINI2}.
Next,  we deduce \eqref{E:DELTAZETA} directly from \eqref{E:ZETA}, \eqref{E:DELZETA} and  \eqref{E:DELRHO},
and we observe by \eqref{E:PIG}, \eqref{E:DEFHATG}, \eqref{e:gtherel}, \eqref{e:iniuexpsp} and Lemma \ref{L:IDENTITY} that
	\begin{align*}
	u^{0\mu}_{i,\epsilon}|_{\Sigma}= \frac{1}{2}\delta^\mu_0\partial_i\smfu^{00}_\epsilon+\delta^\mu_k \partial_i \smfu^{0k}_\epsilon+\epsilon
	\mathcal{T}^\mu_i(\epsilon,\smfu^{kl},\smfu^{kl}_0, \breve{\rho}_0, \breve{\nu}^l)= \frac{\Lambda}{3 }E^2(1)\delta^\mu_0\del{i} \Delta^{-1}\delta\breve{\rho} +\epsilon
	\mathcal{S}^\mu_i(\epsilon,\smfu^{kl},\smfu^{kl}_0, \delta\breve{\rho}, \breve{z}^l),
	\end{align*}
	where $\mathcal{S}_i^\mu(\epsilon,0,0,0,0)=0$,
which gives \eqref{wexp}. Furthermore, similar calculations using \eqref{E:PIG}, \eqref{E:DEFHATG}, \eqref{E:U0MUANDINI}, Lemma \ref{L:IDENTITY} and Theorem \ref{t:inithm} give \eqref{u0muexp}.
From \eqref{idalpha}, \eqref{E:U0MUANDINI2}, \eqref{wexp}, \eqref{u0muexp} and Theorem \ref{t:inithm},
we find, with the help of \eqref{e:alpha}, that
    \begin{align} \label{e:ub}
    	u_{\beta,\epsilon}=3u^{00}_\epsilon\delta^0_\beta+u^{00}_{\beta,\epsilon}-\frac{1}{\epsilon}\frac{\Lambda}{3}\frac{1}{\underline{\alpha}}\underline{\bar{ \partial}_\beta \alpha},
    \end{align}
from which \eqref{ugammaexp} follows via a straightforward calculation.
We also find, with the help of \eqref{e:ub},  Lemma \ref{L:IDENTITY} and Theorem \ref{t:inithm}, that
	\begin{equation*}\label{E:UIJ0ANDINI}
	u^{ij}_{\gamma,\epsilon}|_{\Sigma}=  \frac{1}{\epsilon}\underline{\udn{\gamma}(\alpha^{-1}\theta^{-1}\hat{g}^{ij})}
	=  \epsilon\mathcal{S}^{ij}_\gamma (\epsilon,\smfu^{kl},\smfu^{kl}_0, \breve{\rho}_0, \breve{\nu}^l),
	\end{equation*}
	where $\mathcal{S}^{ij}(\epsilon,0,0,0,0)=0$, follows from  \eqref{E:u.e} and \eqref{E:GAMMA}, which establishes \eqref{uijgammaexp}.
Finally, it is not difficult to verify that the estimates \eqref{E:INITIALDATAESTIMATE2} and \eqref{E:INITIALDATAESTIMATE3}
are a direct consequence of the expansions \eqref{E:U0MUANDINI2}-\eqref{uijgammaexp} and  Theorem \ref{t:inithm}.
\end{proof}

\subsection{Matter fluctuations away from homogeneity}
As discussed in the introduction, we are interested in initial data where the density and velocity fluctuations away from homogeneity
are of the form
\begin{align}\label{e:epini}
\delta\breve{\rho}_{\epsilon,\yve}(\xv)=\sum_{\lambda=1}^{N}\delta\breve{\rho}_\lambda\biggl(\xv-\frac{\yv_\lambda}{\epsilon}\biggr) \AND
\breve{z}^j_{\epsilon,\yve}(\xv)=\sum_{\lambda=1}^{N}\breve{z}_\lambda^j\biggl(\xv-\frac{\yv_\lambda}{\epsilon}\biggr),
\end{align}
where $\yve=(\yv_1,\cdots,\yv_N)\in \Rbb^{3N}$. This fixes part of the free initial data. 
We will assume that the remainder of the free initial data $\{\smfu^{ij}_{\epsilon},\smfu^{ij}_{0,\epsilon}\}$ is bounded as $\epsilon\searrow 0$ with the simplest choice being $\smfu^{ij}_{\epsilon}=\smfu^{ij}_{0,\epsilon}=0$.
Noting that the bounds
\begin{align} \label{e:unifini}
\|\delta\breve{\rho}_{\epsilon, \yve}\|_{L^{\frac{6}{5}}\cap K^s}  \leq \sum_{\lambda=1}^{N} \|
\delta\breve{\rho}_\lambda\|_{L^{\frac{6}{5}}\cap K^s}  \AND \|\breve{z}^j_{\epsilon, \yve}\|_{L^{\frac{6}{5}}\cap K^s}  \leq \sum_{\lambda=1}^{N} \|
\breve{z}^j_\lambda\|_{L^{\frac{6}{5}}\cap K^s}
\end{align}
follow immediately from the triangle inequality and the translation invariance of the norms $L^{\frac{6}{5}}\cap K^s$,
it is clear that we can apply Theorem \ref{t:inithm} and Proposition \ref{L:INITIALTRANSFER} to this class of free initial data
to obtain the following result.

\begin{theorem}\label{e:inithm}
	Suppose $s\in \Zbb_{\geq 3}$, $r>0$, $\epsilon_1>0$,  $\yve=(\yv_1,\cdots,\yv_N)\in \Rbb^{3N}$, $\smfu^{ij}_{\epsilon}\in R^{s+1}(\Rbb^3,\mathbb{S}_3)$ and $\smfu^{ij}_{0,\epsilon}\in R^s(\Rbb^3,\mathbb{S}_3)\cap L^2(\Rbb^3,\mathbb{S}_3)$ for $\epsilon\in (0,\epsilon_1)$, $\delta\breve{\rho}_{\lambda}\in L^{\frac{6}{5}}\cap K^s(\Rbb^3,\Rbb)$ and $\breve{z}^j_{\lambda} \in L^{\frac{6}{5}}\cap K^s(\Rbb^3,\Rbb^3)$ for $\lambda=1,\cdots, N$, $\delta\breve{\rho}_{\epsilon,\yve}$ and $\breve{z}^j_{\epsilon,\yve}$ are defined by \eqref{e:epini} and $\mu(1)$ satisfies \eqref{inisetup2}. Then there exists a constant $\epsilon_0\in (0,\epsilon_1)$ such that if the free initial data satisfies
	\begin{align*}
	\|\breve{\xi}_\epsilon\|_s:=\|\smfu^{ij}_\epsilon\|_{R^{s+1}}+\|\smfu^{ij}_{0, \epsilon}\|_{H^s} +\sum_{\lambda=1}^N \|\delta\breve{\rho}_\lambda\|_{L^{\frac{6}{5}}\cap K^s} +\sum_{\lambda=1}^N \|\breve{z}^j_\lambda\|_{L^{\frac{6}{5}}\cap K^s} \leq r, \quad 0<\epsilon < \epsilon_0,
	\end{align*}
then there exists a family $(\epsilon,\yve)$-dependent maps
	\begin{align*}
	\mathbf{\hat{U}}_{\epsilon, \yve}|_{\Sigma}=\{u^{\mu\nu}_{\epsilon, \yve},u_{\epsilon, \yve},u^{ij}_{\gamma,\epsilon, \yve},
	u^{0\mu}_{i,\epsilon, \yve},u^{0\mu}_{0,\epsilon, \yve},u_{\gamma,\epsilon, \yve},z_{j,\epsilon, \yve},\delta \zeta_{\epsilon, \yve}\}|_{\Sigma}, \quad (\epsilon,\yve)\in (0,\epsilon_0)\times \Rbb^{3N},
	\end{align*}
such that $\mathbf{\hat{U}}_{\epsilon, \yve}|_{\Sigma}\in X^s(\Rbb^3)$, $\mathbf{\hat{U}}_{\epsilon, \yve}|_{\Sigma}$
determines a solution of the constraint equations
\eqref{E:CONSTRAINT}-\eqref{E:NORMALIZATION}, and the components of
$\mathbf{\hat{U}}_{\epsilon, \yve}|_{\Sigma}$ can be expressed as
	\begin{align}
	u^{0\mu}_{\epsilon, \yve}|_{\Sigma}=&  
    \epsilon   
	\mathcal{S}^\mu
	(\epsilon,\smfu^{kl}_\epsilon,\smfu^{kl}_{0,\epsilon}, \delta\breve{\rho}_{\epsilon, \yve}, \breve{z}^l_{\epsilon, \yve}), \label{E:U0MUANDINI3}
	\\
	u_{\epsilon, \yve}|_{\Sigma}=&\epsilon^2\frac{2\Lambda}{9}E^2(1)\smfu^{ij}_\epsilon \delta_{ij}+\epsilon^3 \mathcal{S}(\epsilon,\smfu^{kl}_\epsilon,\smfu^{kl}_{0,\epsilon}, \delta\breve{\rho}_{\epsilon, \yve}, \breve{z}^l_{\epsilon, \yve}),  \label{E:UANDINI3}
	\\
	u^{ij}_{\epsilon, \yve}|_{\Sigma}=&\epsilon^2 \left(\smfu^{ij}_\epsilon
	-\frac{1}{3}\smfu^{kl}_\epsilon\delta_{kl}\delta^{ij}\right)+\epsilon^3\mathcal{S}^{ij}
	(\epsilon,\smfu^{kl}_\epsilon,\smfu^{kl}_{0,\epsilon}, \delta\breve{\rho}_{\epsilon, \yve}, \breve{z}^l_{\epsilon, \yve}), \label{E:UIJANDINI3}
	\\
	z_{j,\epsilon, \yve}|_{\Sigma}=& E^2(1) \delta_{kl}\breve{z}^k_{\epsilon, \yve}
	+\epsilon\mathcal{R}_j
	(\epsilon,\smfu^{kl}_\epsilon,\smfu^{kl}_{0,\epsilon}, \delta\breve{\rho}_{\epsilon, \yve}, \breve{z}^l_{\epsilon, \yve}),\label{E:ZJINI1}
	\\
	\delta\zeta_{\epsilon, \yve}|_{\Sigma}=&\frac{1}{1+\epsilon^2 K}\ln{\left(1+\frac{\delta\breve{\rho}_{\epsilon, \yve}}{\mu(1)} \right)}, \label{E:DELTAZETA1}
	\\
	u^{0\mu}_{i,\epsilon, \yve}|_{\Sigma} =& \frac{\Lambda}{3 }E^2(1)\delta^\mu_0\del{i} \Delta^{-1}\delta\breve{\rho}_{\epsilon, \yve} +\epsilon \mathcal{S}^\mu_i
	(\epsilon,\smfu^{kl}_\epsilon,\smfu^{kl}_{0,\epsilon}, \delta\breve{\rho}_{\epsilon, \yve}, \breve{z}^l_{\epsilon, \yve}), \label{wexp1}
	\\
	u^{0\mu}_{0,\epsilon,\yve}|_{\Sigma} = &\epsilon \mathcal{S}^\mu_0
	(\epsilon,\smfu^{kl}_\epsilon,\smfu^{kl}_{0,\epsilon}, \delta\breve{\rho}_{\epsilon, \yve}, \breve{z}^l_{\epsilon, \yve}), \label{u0muexp1}
	\\
	u_{\gamma,\epsilon,\yve}|_{\Sigma} =& \epsilon \mathcal{S}_\gamma
	(\epsilon,\smfu^{kl}_\epsilon,\smfu^{kl}_{0,\epsilon}, \delta\breve{\rho}_{\epsilon, \yve}, \breve{z}^l_{\epsilon, \yve}),  \label{ugammaexp1}
	\intertext{and}
	u^{ij}_{\gamma,\epsilon,\yve}|_{\Sigma} =& \epsilon \mathcal{S}^{ij}_\gamma
	(\epsilon,\smfu^{kl}_\epsilon,\smfu^{kl}_{0,\epsilon}, \delta\breve{\rho}_{\epsilon, \yve}, \breve{z}^l_{\epsilon, \yve}), \label{uijgammaexp1}
	\end{align}
where the remainders are bounded by
	\begin{align*}
&\|\mathcal{S}^\mu	(\epsilon,\smfu^{kl}_\epsilon,\smfu^{kl}_{0,\epsilon}, \delta\breve{\rho}_{\epsilon, \yve}, \breve{z}^l_{\epsilon, \yve})\|_{R^{s+1}}+\|\mathcal{S}	(\epsilon,\smfu^{kl}_\epsilon,\smfu^{kl}_{0,\epsilon}, \delta\breve{\rho}_{\epsilon, \yve}, \breve{z}^l_{\epsilon, \yve})\|_{R^{s+1}}+\|\mathcal{S}^{ij}	(\epsilon,\smfu^{kl}_\epsilon,\smfu^{kl}_{0,\epsilon}, \delta\breve{\rho}_{\epsilon, \yve}, \breve{z}^l_{\epsilon, \yve})\|_{R^{s+1}} \nnb  \\
&\hspace{1cm} +\|\mathcal{R}_j	(\epsilon,\smfu^{kl}_\epsilon,\smfu^{kl}_{0,\epsilon}, \delta\breve{\rho}_{\epsilon, \yve}, \breve{z}^l_{\epsilon, \yve})\|_{R^{s+1}}  +\|\mathcal{S}^\mu_i	(\epsilon,\smfu^{kl}_\epsilon,\smfu^{kl}_{0,\epsilon}, \delta\breve{\rho}_{\epsilon, \yve}, \breve{z}^l_{\epsilon, \yve})\|_{R^{s+1}} +\|\mathcal{S}^\mu_0	(\epsilon,\smfu^{kl}_\epsilon,\smfu^{kl}_{0,\epsilon}, \delta\breve{\rho}_{\epsilon, \yve}, \breve{z}^l_{\epsilon, \yve})\|_{R^{s+1}} \nnb  \\ &\hspace{5cm}+\|\mathcal{S}_\gamma 	(\epsilon,\smfu^{kl}_\epsilon,\smfu^{kl}_{0,\epsilon}, \delta\breve{\rho}_{\epsilon, \yve}, \breve{z}^l_{\epsilon, \yve})\|_{R^{s+1}}    +\|\mathcal{S}^{ij}_\gamma	(\epsilon,\smfu^{kl}_\epsilon,\smfu^{kl}_{0,\epsilon}, \delta\breve{\rho}_{\epsilon, \yve}, \breve{z}^l_{\epsilon, \yve})\|_{R^{s+1}}\lesssim \|\breve{\xi}\|_s
\end{align*}
for all  $ (\epsilon,\yve)\in (0,\epsilon_0)\times \Rbb^{3N}$.
Moreover, the components of $\mathbf{\hat{U}}_{\epsilon, \yve}|_{\Sigma}$ satisfy the uniform bounds
	\begin{align*}
	\|u^{\mu\nu}_{\epsilon,\yve}|_{\Sigma}\|_{R^{s+1}}&+\|u_{\epsilon,\yve}|_{\Sigma}\|_{R^{s+1}}+\|u^{0k}_{i,\epsilon,\yve}|_{\Sigma}
	\|_{R^s}+\|u^{0\mu}_{0,\epsilon,\yve}|_{\Sigma}\|_{R^s}+
	\|u_{\mu,\epsilon,\yve}|_{\Sigma}\|_{R^s} +\|u^{ij}_{\mu,\epsilon,\yve}|_{\Sigma}\|_{R^s}  \lesssim \epsilon \|\breve{\xi}_\epsilon\|_s 
	\end{align*}
	and
	\begin{align*}
	\|u^{00}_{i,\epsilon, \yve}|_{\Sigma}\|_{R^s}&+\|\breve{z}_{j,\epsilon, \yve}|_{\Sigma}\|_{R^s}+\|\delta\zeta_{\epsilon, \yve}|_{\Sigma}\|_{ L^{\frac{6}{5}}\cap K^s
	} \lesssim \|\breve{\xi}_\epsilon\|_s
	\end{align*}
for all  $ (\epsilon,\yve)\in (0,\epsilon_0)\times \Rbb^{3N}$.	
\end{theorem}

\section{Local existence and continuation}\label{EE:ules}

\subsection{Reduced conformal Einstein-Euler equations}\label{E:ulex}
The formulation \eqref{E:LOCEQ} of the reduced conformal Einstein-Euler equations is symmetric hyperbolic. Consequently,
we can apply standard results, e.g. \cite[\S2.3]{Majda2012}, to obtain the local-in-time existence and uniqueness of solutions
in uniformly local Sobolev spaces $H^s_{\text{ul}}(\Rbb^3)$, $s\in \Zbb_{\geq 3}$,
along with a continuation principle; see Proposition \ref{t:ulst} below for the precise statement. However, in order to obtain the global existence of solutions to the future that exist for all parameter values $\epsilon \in (0,\epsilon_0)$ and all $t\in (0,1]$, we cannot use the formulation \eqref{E:LOCEQ}
of the conformal Einstein-Euler equations. Instead, we rely on a non-local formulation, which is defined below by \eqref{E:REALEQ}.
Due to the non-locality, it is not enough to have local existence and continuation in the uniformly local Sobolev spaces. Instead,
we need to establish the local-in-time existence of solutions and a continuation principle in the spaces $R^s(\Rbb^3)$,
$s\in \Zbb_{\geq 3}$, which we do below in Corollary \ref{T:poes}.

\begin{proposition} \label{t:ulst}
	Suppose $s\in \Zbb_{\geq 3}$ and
	\begin{align*}
	\mathbf{\hat{U}}_{\epsilon, \yve}|_{\Sigma}=\{u^{\mu\nu}_{\epsilon, \yve},u_{\epsilon, \yve},u^{ij}_{\gamma,\epsilon, \yve},
	u^{0\mu}_{i,\epsilon, \yve},u^{0\mu}_{0,\epsilon, \yve},u_{\gamma,\epsilon, \yve},z_{j,\epsilon, \yve},\delta \zeta_{\epsilon, \yve}\}|_{\Sigma}\in X^s(\Rbb^3), \quad (\epsilon,\yve) \in (0,\epsilon_0)\times \Rbb^{3N},
	\end{align*}
is the initial data from Theorem \ref{e:inithm}. Then
there exists a constant  $T>0$ and a unique classical solution
	\begin{equation*}
		\bhU_{\epsilon, \yve} \in C((T, 1], H^s_{\emph{loc}}(\Rbb^3,\mathbb{K}))\bigcap C^1 ((T, 1], \Hs_{\emph{loc}}(\Rbb^3,\mathbb{K})) \bigcap \Li ((T, 1], H^s_{\emph{ul}}(\Rbb^3,\mathbb{K})),
	\end{equation*}
where $\mathbb{K}=\mathbb{S}_4\times \mathbb{R}\times \mathbb{S}_3\times (\mathbb{R}^3)^2\times\mathbb{R}
\times \mathbb{R}^3\times \mathbb{R}$,
	to \eqref{E:LOCEQ} on the spacetime region $(T, 1] \times \Rbb^3$ that agrees with the initial data
from Theorem \ref{e:inithm} on the initial hypersurface $\Sigma$. Moreover, there exists a constant $\sigma > 0$, independent
of the initial data, such that if
$\mathbf{\hat{U}}_{\epsilon, \yve}$ exists for $t\in (T_1,1]$ with the same regularity as above and satisfies
$\norm{\mathbf{\hat{U}}_{\epsilon, \yve}}_{L^\infty((T_1,1],W^{1,\infty})} < \sigma$,
	then the solution $\mathbf{\hat{U}}_{\epsilon, \yve}$ can be uniquely continued as a classical solution
	with the same regularity
	to the larger spacetime region $(T^*
	,1]\times \mathbb{R}^3
	$ for some $T^*
	 \in (0,T_1)$.
\end{proposition}
\begin{proof}
First, we observe by Theorem \ref{Sobolev},  \eqref{E:GIJ}-\eqref{E:G0MU}, \eqref{E:ZETA2},  \eqref{e:vup0}--\eqref{E:Z_IANDZ^I} and \eqref{E:NORMALIZATION} that there exists a constant $\sigma > 0$ such that if $\hat{\mathbf{U}}(t,x)$ satisfies
$\|\hat{\mathbf{U}}(t)\|_{W^{1,\infty}} < \sigma$,
then the metric $\bar{g}^{\mu\nu}$, conformal four-velocity $\bar{v}^\mu$ and proper energy density $\bar{\rho}$  will remain non-degenerate, future directed, and strictly positive, respectively. We also observe by Theorem \ref{e:inithm} that there exists
initial  data $\hat{\mathbf{U}}_{\epsilon,\yve}|_{\Sigma}$ for the
evolution equation \eqref{E:LOCEQ} that satisfies\footnote{We can always arrange this by shrinking $\epsilon_0$, if necessary, to
guarantee via Sobolev's inequality that the bound  $\|\hat{\mathbf{U}}_{\epsilon,\yve}|_{\Sigma}\|_{W^{1,\infty}} <  \sigma$
is satisfied for any particular choice of $\sigma$.}
$\|\hat{\mathbf{U}}_{\epsilon,\yve}|_{\Sigma}\|_{W^{1,\infty}} <  \sigma$.  It then follows from\footnote{Here, we are using the inclusion $X^s(\Rbb^3) \subset H_{\text{ul}}^s(\Rbb^3,\mathbb{K})$ which is a direct consequence
of Sobolev's inequality and the definition of the space $X^s(\Rbb^3)$.}
Theorem 2.1 from \cite[\S2.3]{Majda2012} that there exists a $T\in (0,1)$
and a unique classical solution
	\begin{equation*}
		\bhU_{\epsilon, \yve} \in C((T, 1], H^s_{\textrm{loc}}(\Rbb^3,\mathbb{K}))\bigcap C^1 ((T, 1], \Hs_{\textrm{loc}}(\Rbb^3,\mathbb{K})) \bigcap \Li ((T, 1], H^s_{\textrm{ul}}(\Rbb^3,\mathbb{K}))
	\end{equation*}
to \eqref{E:LOCEQ} that satisfies $\|\hat{\mathbf{U}}_{\epsilon,\yve}\|_{L^\infty((T,1],W^{1,\infty})} < \sigma$
and agrees with the initial data $\hat{\mathbf{U}}_{\epsilon,\yve}|_{\Sigma}$ at $t=1$. This proves the existence and uniqueness part of
the statement. The continuation part of the statement is a direct consequence of Theorem 2.2 from \cite[\S2.3]{Majda2012} since
the bound $\|\hat{\mathbf{U}}\|_{L^\infty((T_1,1],W^{1,\infty}(\Rbb^3))} < \sigma$ together with the equations of motion
\eqref{E:LOCEQ} imply a bound of the form  $\|\hat{\mathbf{U}}\|_{W^{1,\infty}((T_1,1]\times \Rbb^3))} < C(\sigma)$.
\end{proof}

\begin{corollary} \label{T:poes}
	Suppose $s\in \Zbb_{\geq 3}$,
	\begin{align*}
	\mathbf{\hat{U}}_{\epsilon, \yve}|_{\Sigma}=\{u^{\mu\nu}_{\epsilon, \yve},u_{\epsilon, \yve},u^{ij}_{\gamma,\epsilon, \yve},
	u^{0\mu}_{i,\epsilon, \yve},u^{0\mu}_{0,\epsilon, \yve},u_{\gamma,\epsilon, \yve},z_{j,\epsilon, \yve},\delta \zeta_{\epsilon, \yve}\}|_{\Sigma}\in X^s(\Rbb^3), \quad (\epsilon,\yve)\in (0,\epsilon_0)\times \Rbb^{3N},
	\end{align*}
is the initial data from Theorem \ref{e:inithm}. Then
	there exists a constant  $T>0$  and a unique classical solution
	\begin{equation*}
		\bhU_{\epsilon, \yve} \in \bigcap_{\ell=0}^1C^\ell((T,1], R^{s-\ell}(\mathbb{R}^3,\mathbb{K}) )
	\end{equation*}
	to \eqref{E:LOCEQ} on the spacetime region $(T, 1] \times \Rbb^3$
that agrees with the initial data
from Theorem \ref{e:inithm} on the initial hypersurface $\Sigma$. Moreover, there exists a constant $\sigma > 0$, independent
of the initial data, such that if
$\mathbf{\hat{U}}_{\epsilon, \yve}$ exists for $t\in (T_1,1]$ with the same regularity as above and satisfies
$\norm{\mathbf{\hat{U}}_{\epsilon, \yve}}_{L^\infty((T_1,1],R^s))} < \sigma$,
	then the solution $\mathbf{\hat{U}}_{\epsilon, \yve}$ can be uniquely continued as a classical solution
	with the same regularity
	to the larger spacetime region $(T^*
	,1]\times \mathbb{R}^3
	$ for some $T^*
	 \in (0,T_1)$.
\end{corollary}
\begin{proof}
While the solution  $\bhU_{\epsilon,\yve}$ to \eqref{E:LOCEQ} from Proposition \ref{t:ulst} satisfies $\bhU_{\epsilon,\yve}(1) \in
X^s(\Rbb^3) \subset R^s(\Rbb^3,\mathbb{K})$,  we only
know that  $\bhU_{\epsilon,\yve}(t)\in H^s_{\text{ul}}(\Rbb^3,\mathbb{K})$ for $t<1$. Thus, the first step involves showing
that $\bhU_{\epsilon,\yve}(t)$ remains in $R^s(\Rbb^3,\mathbb{K})$ for $t<1$. To accomplish this, we use energy estimates.
In fact, it will be enough to show that solutions $\bhU_{\epsilon,\yve}(t)$ that stay in $R^s(\Rbb^3,\mathbb{K})$ satisfy an
energy estimate that yield a bound on the norm  $\|\bhU_{\epsilon,\yve}(t)\|_{R^s}$. While this may seem
circular, it is easy to justify using the finite
speed of propagation to first prove energy estimates on truncated spacetime cones, which are well defined for solutions
$\bhU_{\epsilon,\yve}(t)$ that lie in $H^s_{\text{ul}}(\Rbb^3,\mathbb{K})$,  followed by letting the
width of the cone go to infinity to obtain estimates on the spacetime slab of the form $(T_1,1]\times \Rbb^3$. For
reasons of economy, we omit these easily reproducible details.

	In the following, we will suppress the subscripts and just write $\bhU$ for the solution.  We then derive energy estimates for
$D\bhU$ by differentiating the evolution equations  \eqref{E:LOCEQ} to get
\begin{align}
        \Bb^0 D^\alpha \del{t}\bhU=&-\Bb^i D^\alpha \del{i} \bhU-\frac{1}{\epsilon}\mathbf{C}^iD^\alpha\del{i}\bhU  + \frac{1}{t}\Bb^0 D^\alpha (\Bb^0)^{-1}\Bb \mathbf{\Pbb} \bhU +\Bb^0 [(\Bb^0)^{-1}\Bb^i, D^\alpha] \del{i} \bhU   \nnb \\
        & \hspace{5.5cm} +\frac{1}{\epsilon}\Bb^0[(\Bb^0)^{-1}, D^\alpha]\mathbf{C}^i \del{i} \bhU+\Bb^0 D^\alpha \bigl((\Bb^0)^{-1} \mathbf{\hat{H}}\bigr) \label{dxeveqn} \\
        \intertext{while  differentiating $\mathbf{B}^0$ with respect to $t$ yields}
		\del{t} \Bb^0
		=	&D_t \Bb^0+D_\bhU \Bb^0 \cdot \left( -(\Bb^0)^{-1} \Bb^i \del{i}\bhU-\frac{1}{\epsilon}(\Bb^0)^{-1}C^i\del{i}\bhU+\frac{1}{t}(\Bb^0)^{-1}\Bb\mathbf{\Pbb}\bhU+(\Bb^0)^{-1}\mathbf{\hat{H}}\right).  \label{dteveqn} 
\end{align}
Multiplying \eqref{dxeveqn} by the transpose of $D^\alpha \bhU$, $|\alpha|\geq 1$, followed by integrating over $\Rbb^3$ and summing over
$\alpha$ for $1\leq |\alpha|\leq s+1$, we obtain, with the
help of \eqref{dteveqn}, integration by parts and the calculus inequalities from Appendix \ref{A:INEQUALITIES}, the following variation on the standard energy estimate for symmetric hyperbolic systems:
	\als{
		-\del{t}\|D\bhU\|_{\Hs}^2= & -\sum_{1\leq |\alpha|\leq s}\la D^\alpha \bhU, ( \del{t} \Bb^0 ) D^\alpha \bhU\ra - 2\sum_{1\leq |\alpha|\leq s}\la D^\alpha \bhU,  \Bb^0 D^\alpha \del{t} \bhU\ra  \nonumber \\
		\leq & C(\|\bhU\|_{W^{1, \infty} },\epsilon^{-1})\|D \bhU\|_{\Hs}(\|D \bhU\|_{\Hs}+\|\bhU\|_{\Li}),
	}
which we note is equivalent to
\begin{equation}
-\del{t}\| D \bhU\|_{\Hs} \leq   C(\|\bhU\|_{W^{1, \infty}})(\|D \bhU\|_{\Hs}+\| \bhU\|_{\Li}). \label{dUengest}
\end{equation}
Assuming that $T_0 \in (0,1]$, we obtain from integrating \eqref{dUengest} the estimate
	\al{bhUDH}{
		\|D\bhU(t)\|_{\Hs} \leq \|D\bhU(T_0)\|_{\Hs}+\int_t^{T_0}
C(\|\bhU(\tau)\|_{W^{1, \infty}})(\|D \bhU(\tau)\|_{\Hs}+\|\bhU(\tau)\|_{\Li}) d\tau
	}
for $0<T_1< t \leq T_0\leq 1$.
On the other hand, multiplying the evolution equation \eqref{E:LOCEQ} on the left by $(\mathbf{B}^0)^{-1}$ followed by integrating in time, we find,
after taking the $L^6$ norm, that
	\al{bhUL6}{
		\|\bhU(t)\|_{L^6} \leq & \|\bhU(T_0)\|_{L^6}+\int_{t}^{T_0} \left( \|(\Bb^0)^{-1} \Bb^i \del{i}\bhU\|_{L^6}+\frac{1}{\epsilon}\|(\Bb^0)^{-1}C^i\del{i}\bhU\|_{L^6} +\frac{1}{t}\|(\Bb^0)^{-1}\Bb\mathbf{\Pbb}\bhU\|_{L^6}+\|(\Bb^0)^{-1}\mathbf{\hat{H}}\|_{L^6}\right) d\tau  \nonumber\\
		\leq & \|\bhU(T_0)\|_{L^6}+\int_{t}^{T_0}  C(\|\bhU(\tau)\|_{\Li}) \| \bhU(\tau)\|_{W^{1,6}} d\tau
	}
for $0<T_1< t \leq T_0\leq 1$.
	Adding the two inequalities \eqref{E:bhUDH} and \eqref{E:bhUL6}, we get, with the help of  \eqref{E:NORMEQ1}, that
\begin{align*}
	\|\bhU(t)\|_{R^s} \leq \|\bhU(T_0)\|_{R^s}+	 \int_t^{T_0} C(\|\bhU(\tau)\|_{R^s}) \| \bhU(\tau)\|_{R^s }  d\tau, \quad  0<T_1<t \leq T_0\leq 1.
\end{align*}
Then by the Gr\"onwall's inequality, there exists a $T_*\in [T_1,T_0)$ such that
	\begin{align*}
		\|\bhU(t)\|_{R^s}\leq C(\|\bhU(T_0)\|_{R^s}),  \quad  0<T_*<t \leq T_0\leq 1.
	\end{align*}
From this inequality and Proposition \ref{t:ulst}, we deduce that the space $R^s(\Rbb^3,\mathbb{K})$ is preserved under evolution, and hence
there exists a $T\in (0,1]$ such that
\begin{equation} \label{Rsregularity}
		\bhU_{\epsilon, \yve} \in \bigcap_{\ell=0}^1C^\ell((T,1], R^{s-\ell}(\mathbb{R}^3,\mathbb{K}) ).
	\end{equation}
Moreover, since $\norm{\mathbf{\hat{U}}_{\epsilon, \yve}}_{L^\infty((T_1,1],W^{1,\infty})} \lesssim \norm{\mathbf{\hat{U}}_{\epsilon, \yve}}_{L^\infty((T_1,1],R^s)}$ by \eqref{E:NORMEQ1}, it follows from the continuation principle from Proposition \ref{t:ulst}
that there exists a constant $\sigma > 0$, independent
of the initial data, such that if
$\mathbf{\hat{U}}_{\epsilon, \yve}$ exists for $t\in (T_1,1]$ with the same regularity as \eqref{Rsregularity} and
satisfies $\norm{\mathbf{\hat{U}}_{\epsilon, \yve}}_{L^\infty((T_1,1],R^s)} < \sigma$,
	then the solution $\mathbf{\hat{U}}_{\epsilon, \yve}$ can be uniquely continued as a classical solution
	with the same regularity
	to the larger spacetime region $(T^*,1] \times \mathbb{R}^3$ for some $T^*\in (0,T_1)$.
\end{proof} 

\begin{remark}
While it is clear from the energy estimates that the time of existence $T$ from the above corollary does not depend on the
parameter $\yve$ since the norm of the initial data $\|\mathbf{\hat{U}}_{\epsilon, \yve}|_{\Sigma}\|_{R^s}$ is independent of $\yve$ due
to the translational invariance of the norm $\|\cdot \|_{R^s}$, the time of existence does appear to depend on $\epsilon$
due to the appearance of $\epsilon^{-1}$ in the energy estimates. To show that the time of existence does not, in fact, depend
on $\epsilon$ and that for small enough data the solution exists on the whole time interval $(0,1]$ relies on
the  non-local version of the reduced conformal Einstein equations defined by \eqref{E:REALEQ}.
\end{remark}

\subsection{Conformal Poisson-Euler equations} \label{S:locPE}
In this section, we consider the local-in-time existence and uniqueness of solutions to the conformal cosmological Poisson-Euler
equations, and we
establish a continuation principle  based on bounding the
$R^s$ norm of $(\delta \mathring{\zeta},\mathring{z}^j)$. For convenience, we define
\begin{align}\label{e:varpi}
\varpi^j:=\frac{\mathring{E}^3}{t^3}\mathring{\rho}\mathring{z}^j,
\end{align}
and let
\begin{equation} \label{deltazetaringdef}
\delta \mathring{\zeta} = \mathring{\zeta} - \mathring{\zeta}_H \AND \mathring{z}_j=\mathring{E}^2 \delta_{ij}\mathring{z}^i
\end{equation}
where, see \eqref{arhoringdef}, \eqref{zetaHringform} and \eqref{E:OMEGAREP},
\begin{equation} \label{deltazetaringH}
\mathring{\zeta}_H = \ln(t^{-3}\mathring{\mu}).
\end{equation}
For each positive constant $\beta > 0$, we further define the quantity
\begin{align} \label{Upsi}
\mathring{\Upsilon}= \epsilon \beta \frac{\Lambda}{3t^3}\mathring{E}^3 (\Delta-\epsilon^2 \beta)^{-1} \delta\mathring{\rho},
\end{align}
which will be used to simplify $\mathring{\mathbf{F}}$ later in \S\ref{S:MAINPROOF}.

\begin{proposition} \label{PEexist}
	Suppose $s\in \Zbb_{\geq 3}$, $\epsilon_0>0$, $\epsilon \in (0,\epsilon_0)$, 
	$\vec{\yv}\in \Rbb^{3N}$,  $\delta\breve{\rho}_{\lambda}\in L^{\frac{6}{5}}\cap K^s(\Rbb^3,\Rbb)$ and $\breve{z}^j_{\lambda} \in L^{\frac{6}{5}}\cap K^s(\Rbb^3,\Rbb^3)$ for $\lambda=1,\cdots, N$, $\delta\breve{\rho}_{\epsilon,\yve}$ and $\breve{z}^j_{\epsilon,\yve}$ are defined by \eqref{e:epini},
and $\mathring{\mu}(1)>0$. 
	Then there exists a  $T
	\in (0,1]$ and a unique classical  solution $( \mathring{\zeta}_{\epsilon,\yve},\mathring{z}^i_{\epsilon,\yve},\mathring{\Phi}_{\epsilon,\yve})$ 
	of the conformal cosmological Poisson-Euler equations, given by \eqref{CPeqn1}-\eqref{CPeqn3}, such that
	\begin{align*}
	(\delta \mathring{\zeta}_{\epsilon,\yve},\mathring{z}^i_{\epsilon,\yve},\mathring{\Phi}_{\epsilon,\yve})
	\in \bigcap_{\ell=0}^1 C^{\ell}((T, 1],H^{s-\ell}(\mathbb{R}^3))
	\times \bigcap_{\ell=0}^1 C^{\ell}((T, 1],H^{s-\ell}(\mathbb{R}^3, \Rbb^3)) \times
	\bigcap_{\ell=0}^1 C^{\ell}((T, 1],R^{s+2-\ell}(\mathbb{R}^3))
	\end{align*}
	on the spacetime region
	$(T,1]\times \mathbb{R}^3$ that satisfies
\begin{equation}\label{e:Nini}
	(\delta\mathring{\zeta}_{\epsilon,\yve},\mathring{z}^i_{\epsilon, \yve})|_{\Sigma} =\biggl( \ln\biggl(1+\frac{\delta\breve{\rho}_{\epsilon,\yve}}{\mathring{\mu}(1)}\biggr),\breve{z}^i_{\epsilon,\yve}\biggr) \in L^{\frac{6}{5}}\cap
K^s(\Rbb^3,\Rbb^4) \subset H^s(\Rbb^3,\Rbb^4)
	\end{equation}
on the initial hypersurface $\Sigma $,
and the estimates
	\begin{align}
	\|\varpi^j_{\epsilon,\yve}\|_{H^s} +\|\del{t} \mathring{\Phi}_{\epsilon,\yve}\|_{R^{s+1}} +\|(-\Delta)^{-\frac{1}{2}}\mathfrak{R}_k\varpi^j_{\epsilon,\yve}\|_{R^{s+1}}+	\|\del{t} \mathring{\Phi}_{i,\epsilon,\yve}\|_{H^s} \leq & C\bigl(\|\delta\mathring{\zeta}_{\epsilon,\yve}\|_{L^\infty((t,1],H^s)}\bigr)\|\mathring{z}^j_{\epsilon,\yve}\|_{H^s},  \label{e:vpi}
	\end{align}
	\begin{align}
	\|\mathring{\Phi}_{\epsilon,\yve}\|_{R^{s+1}} +	\|\mathring{\Phi}_{i,\epsilon,\yve}\|_{H^s} \leq &   C\|\delta\breve{\rho}_{\epsilon,\yve}\|_{L^{\frac{6}{5}}\cap H^s} 
	+C\bigl(\|\delta\mathring{\zeta}_{\epsilon,\yve}\|_{L^\infty((t
		,1],H^s)}\bigr)\int_t^1 \| \mathring{z}^k_{\epsilon,\yve}(\tau)\|_{H^s} d\tau,  \label{e:phine1}
	\end{align}
	\begin{align}
	&\|t\del{t}\mathfrak{R}_j(-\Delta)^{-\frac{1}{2}}\varpi^j_{\epsilon,\yve}\|_{R^s}+ \|t\partial_t^2\mathring{\Phi}_{\epsilon,\yve}\|_{R^s} +	\|t\del{t}\varpi^j_{\epsilon,\yve}\|_{R^{s-1}} \nnb  \\
	& \hspace{0.5cm} \leq  C\bigl(\|\delta\mathring{\zeta}_{\epsilon,\yve}\|_{L^\infty((t
		,1],H^s)},\|\mathring{z}^j_{\epsilon,\yve}\|_{L^\infty((t
		,1],H^s)}\bigr)\Bigl( \|\delta\mathring{\zeta}_{\epsilon,\yve}\|_{R^s}+\|\mathring{z}^j_{\epsilon,\yve}\|_{H^s}+\|\delta\breve{\rho}_{\epsilon,\yve}\|_{L^{\frac{6}{5}}\cap H^s}  
	+
	\int_t^1\| \mathring{z}^k_{\epsilon,\yve}(\tau)\|_{H^s} d\tau
	\Bigr),  \label{e:tdt2ph}
	\end{align}
	\begin{align}
	\|\mathring{\Upsilon}_{\epsilon,\yve}\|_{H^s} \leq  C\|\delta\breve{\rho}_{\epsilon,\yve}\|_{L^{\frac{6}{5}}\cap H^s}+
 C\bigl(\|\delta\mathring{\zeta}_{\epsilon,\yve}\|_{L^\infty((t
		,1],H^s)}\bigr)\int_t^1 \|\mathring{z}^j_{\epsilon,\yve}(\tau)\|_{H^s}d\tau, \label{e:Ups1}
\end{align}
and
\begin{align}
	\|\del{t}\mathring{\Upsilon}_{\epsilon,\yve}\|_{R^s}   \leq  C\bigl(\|\delta\mathring{\zeta}_{\epsilon,\yve}\|_{L^\infty((t
		,1],H^s)}\bigr)\|\mathring{z}^j_{\epsilon,\yve}\|_{H^s}\label{e:Ups2}
	\end{align}
for all $t\in(T,1]$.
	Furthermore,  there exists a constant $\sigma > 0$, independent of the initial data and $T_1\in (0,1)$, such that if $(\delta\mathring{\zeta}_{\epsilon,\yve},\mathring{z}^i_{\epsilon,\yve},\mathring{\Phi}_{\epsilon,\yve})$ exists for $t\in (T_1,1]$ with the same regularity as above and satisfies
$\norm{(\delta\mathring{\zeta}_{\epsilon,\yve},\mathring{z}^i_{\epsilon,\yve})}_{L^\infty((T_1,1],R^s)} < \sigma$,
	then the solution $(\delta\mathring{\zeta}_{\epsilon,\yve},\mathring{z}^i_{\epsilon,\yve},\mathring{\Phi}_{\epsilon,\yve})$ can be uniquely continued as a classical solution with the same regularity to the larger spacetime region $(T^*,1]\times \mathbb{R}^3$ for some $T^*\in (0,T_1)$.
\end{proposition}
\begin{proof}
	With the help of \eqref{E:PTZETAH2} and \eqref{deltazetaringdef}-\eqref{deltazetaringH}, we can write the first two equations from the conformal cosmological Poisson-Euler system, see  \eqref{CPeqn1}-\eqref{CPeqn2}, as
	\begin{align}
	\del{t}\delta\mathring{\zeta}+\sqrt{\frac{3}{\Lambda}}(\mathring{z}^j\del{j}\delta\mathring{\zeta}+\del{j}\mathring{z}^j) & = 0, \label{PEexist4}\\
	\sqrt{\frac{\Lambda}{3}}\partial_t\mathring{z}^j+ \mathring{z}^i\partial_i\mathring{z}^j+ K
	\frac{\delta^{ji}}{\mathring{E}^2} \partial_i \mathring{\zeta}
	&=\sqrt{\frac{\Lambda}{3}}\frac{1}{t}\mathring{z}^j-\frac{1}{2}
	\frac{3}{\Lambda}t \mathring{E}^{-3} \delta^{ij} \mathring{\Phi}_i,  \label{PEexist5}
	\end{align}
	where we have set
	\begin{align}\label{e:phij}
	\mathring{\Phi}_i= \del{i}\mathring{\Phi}.
	\end{align}
Rather than considering the Poisson equation \eqref{E:COSEULERPOISSONEQ.c} (see also \eqref{CPeqn3}) directly, we find it convenient instead
to consider the equation satisfied by $\mathring{\Phi}_i$. To derive this equation, we write \eqref{E:COSEULERPOISSONEQ.c} as
\begin{equation} \label{PEexist1}
	\mathring{\Phi} = \frac{\Lambda}{3t^3}  \mathring{E}^3 \Delta^{-1} \delta \mathring{\rho} .
	\end{equation}
Applying $\frac{\Lambda}{3t^3}  \mathring{E}^3 \Delta^{-1}$
to \eqref{E:COSEULERPOISSONEQ.a} we derive, with the help  of \eqref{e:delrr} and \eqref{PEexist1}, the equation
\begin{equation}\label{E:PTPHI1}
	\partial_t \mathring{\Phi} =-\sqrt{\frac{\Lambda}{3}} \frac{\mathring{E}^3}{t^3}  \partial_k\Delta^{-1}
	\left( \mathring{\rho} \mathring{z}^k\right)=-\sqrt{\frac{\Lambda}{3}}   \mathfrak{R}_j(-\Delta)^{-\frac{1}{2}} \varpi^j.
	\end{equation}
In a similar fashion, we derive
\begin{equation}\label{E:PTPHI1b}
	\partial_t \mathring{\Upsilon} =-\epsilon \beta \sqrt{\frac{\Lambda}{3}}   \del{j} ( \Delta-\epsilon^2\beta)^{-1} \varpi^j
	\end{equation}
by applying  $\epsilon\beta \frac{\Lambda}{3t^3}  \mathring{E}^3 (\Delta-\epsilon^2 \beta)^{-1}$ to \eqref{E:COSEULERPOISSONEQ.a}.
Applying the spatial partial derivative $\partial_j$ to \eqref{E:PTPHI1} then yields the desired evolution equation
\begin{align}\label{e:Npoev}
	\del{t} \mathring{\Phi}_i= \sqrt{\frac{\Lambda}{3}} \frac{\mathring{E}^3}{t^3} \mathfrak{R}_i\mathfrak{R}_j
	\left(\mathring{\rho} \mathring{z}^j\right)= \sqrt{\frac{\Lambda}{3}} \mathfrak{R}_i\mathfrak{R}_j \varpi^j
	\end{align}
for  $\mathring{\Phi}_i$.

Together, the equations \eqref{PEexist4}, \eqref{PEexist5} and \eqref{e:Npoev} can be cast into a non-local symmetric hyperbolic form in
the unknowns $\{\delta\mathring{\zeta},\mathring{z}{}^i,\mathring{\Phi}_i\}$ by multiplying \eqref{PEexist5} by
	$\mathring{E}^2 K^{-1}\sqrt{\frac{3}{\Lambda}} $. Moreover, we observe that the initial data is bounded by
\begin{align*}
	\|\mathring{\Phi}_{i,\epsilon,\yve}(1)\|_{H^{s+1}}=\Bigl\|\frac{\Lambda}{3}\mathring{E}^3(1) \del{i}\Delta^{-1} \delta \mathring{\rho}_{\epsilon,\yve}(1)\Bigr \|_{H^{s+1}} \lesssim \|  (-\Delta)^{-\frac{1}{2}} \delta \mathring{\rho}_{\epsilon,\yve}(1) \|_{H^{s+1}} \lesssim \|\delta\breve{\rho}_{\epsilon,\yve}\|_{L^{\frac{6}{5}}\cap H^s}
	\end{align*}
and $\|\delta\mathring{\zeta}_{\epsilon,\yve}\|_{H^s} + \|\mathring{z}{}^i_{\epsilon,\yve}\|_{H^s} \leq \|\delta\mathring{\zeta}_{\epsilon,\yve}\|_{L^{\frac{6}{5}}\cap K^s}+
\|\mathring{z}{}^i_{\epsilon,\yve}\|_{L^{\frac{6}{5}}\cap K^s}$. We can therefore conclude from standard local-in-time existence and uniqueness results and
continuation principles for symmetric hyperbolic systems,  e.g. Theorems 2.1 and 2.2 of \cite[\S$2.1$]{Majda2012}, which continue to apply
for non-local systems, that there exists,  for some time $T
	\in (0,1)$, a
	unique local-in-time classical solution
	\begin{equation}\label{PElocreg}
	(\delta\mathring{\zeta}_{\epsilon,\yve},\mathring{z}^i_{\epsilon,\yve}, \mathring{\Phi}_{i,\epsilon,\yve})\in \bigcap_{\ell=0}^1 C^{\ell}((T
	,1],H^{s-\ell}(\mathbb{R}^3))
	\times \bigcap_{\ell=0}^1  C^{\ell}((T
	,1],H^{s-\ell}(\mathbb{R}^3,\mathbb{R}^3))\times \bigcap_{\ell=0}^1  C^{\ell}((T
	,1],H^{s+1-\ell}(\mathbb{R}^3,\mathbb{R}^3))
	\end{equation}
	of \eqref{PEexist4}, \eqref{PEexist5} and \eqref{e:Npoev} that agrees with the initial data \eqref{e:Nini} on $\Sigma$. Moreover,
	if the solution satisfies $\norm{\delta\mathring{\zeta}_{\epsilon,\yve}}_{L^\infty((T
		,1],W^{1,\infty})}+
	\norm{\mathring{z}^i_{\epsilon,\yve}}_{L^\infty((T
		,1],W^{1,\infty})}
	< \sigma$,
	then there exists a time $T^*
	\in (0,T
	)$ such that the solution \eqref{PElocreg} uniquely extends
	to the spacetime region $(T^*
	,1]\times \mathbb{R}^3$ with the same regularity. By \eqref{E:NORMEQ1}, 
	this is clearly implied by the stronger condition $\norm{(\delta\mathring{\zeta}_{\epsilon,\yve} ,\mathring{z}^i_{\epsilon,\yve})}_{L^\infty((T
		,1],R^s)} < \sigma$.

	From the definition \eqref{e:varpi} of $\varpi^j$, it is clear that the bound
	\begin{align}
	\|\varpi^j_{\epsilon,\yve}(t)\|_{H^s} \leq C\bigl(\|\delta\mathring{\zeta}_{\epsilon,\yve}\|_{L^\infty((t
		,1],H^s)}\bigr)\|\mathring{z}^j_{\epsilon,\yve}(t)\|_{H^s}, \quad T < t \leq 1, \label{varpiest}
	\end{align}
	follows from the calculus inequalities, see
Appendix \ref{A:INEQUALITIES}. Next,
	integrating \eqref{E:PTPHI1} in time and then taking $\|\cdot\|_{L^6}$ norm, we obtain, with the help of the calculus inequalities
and the potential theory from Appendix \ref{S:Potop}, the estimate
	\begin{align}
	\|\mathring{\Phi}_{\epsilon,\yve}(t)\|_{L^6}
	\leq  C \|\delta\mathring{\rho}_{\epsilon,\yve}(1)\|_{L^{\frac{6}{5}}}
	+C(\|\delta\mathring{\zeta}_{\epsilon,\yve}\|_{\Li((t,1],H^s)})\int_{t}^1\| \mathring{z}^k_{\epsilon,\yve}(\tau)\|_{L^2} d\tau
	, \quad T <t\leq1, \label{e:ph6}
	\end{align}
while the estimate
	\begin{align}
	\|\del{t} \mathring{\Phi}_{\epsilon,\yve}(t)\|_{L^6}
\lesssim \|(-\Delta)^{-\frac{1}{2}}(\mathring{\rho}_{\epsilon,\yve}(t) \mathring{z}^k_{\epsilon,\yve}(t))\|_{L^6}
\lesssim \|\mathring{\rho}_{\epsilon,\yve}(t) \mathring{z}^k_{\epsilon,\yve}(t)\|_{L^2}\leq C\bigl(\|\delta\mathring{\zeta}_{\epsilon,\yve}\|_{L^\infty((t
		,1],H^s)}\bigr)\|\mathring{z}^k_{\epsilon,\yve}(t)\|_{L^2}, \label{e:dph6}
	\end{align}
for  $T <t\leq1$, follows directly from applying the  $\|\cdot\|_{L^6}$ norm to \eqref{E:PTPHI1}.
Integrating \eqref{e:Npoev}, we obtain, after taking $\|\cdot\|_{H^s}$-norm, the estimate
	\begin{align}\label{e:phd}
	\|\mathring{\Phi}_{i,\epsilon,\yve}(t)\|_{H^s}\leq & \|\mathring{\Phi}_{i,\epsilon,\yve}(1)\|_{H^s}+ \int_{t}^1 \|\varpi^j_{\epsilon,\yve}\|_{H^s}d\tau
	\leq  C\|\delta\breve{\rho}_{\epsilon,\yve}\|_{L^{\frac{6}{5}}\cap H^s} 
	+C\bigl(\|\delta\mathring{\zeta}_{\epsilon,\yve}\|_{L^\infty((t
		,1],H^s)}\bigr)\int_{t}^1\| \mathring{z}^k_{\epsilon,\yve}(\tau)\|_{H^s} d\tau  
	\end{align}
for  $T <t\leq1$.
From this estimate and \eqref{e:ph6}, we then deduce that
	\begin{equation*}\label{e:phine}
	\|\mathring{\Phi}_{\epsilon,\yve}(t)\|_{R^{s+1}} \leq 
	C\|\delta\breve{\rho}_{\epsilon,\yve}\|_{L^{\frac{6}{5}}\cap H^s} 
	+C\bigl(\|\delta\mathring{\zeta}_{\epsilon,\yve}\|_{L^\infty((t
		,1],H^s)}\bigr)\int_{t}^1\| \mathring{z}^k_{\epsilon,\yve}(\tau)\|_{H^s} d\tau, \quad T <t\leq1.
	\end{equation*}
Applying the norm $\|\cdot\|_{H^s}$-norm to \eqref{e:Npoev} yields the estimate
	\begin{equation*}\label{e:Npoev2}
	\|\del{t} \mathring{\Phi}_{i,\epsilon,\yve}\|_{H^s} \lesssim \|  \mathfrak{R}_i\mathfrak{R}_j\varpi^j_{\epsilon,\yve}\|_{H^s}
\lesssim C\bigl(\|\delta\mathring{\zeta}_{\epsilon,\yve}\|_{L^\infty((t
		,1],H^s)}\bigr)\|\mathring{z}^j_{\epsilon,\yve}\|_{H^s}, \quad T <t\leq1,
	\end{equation*}
which when combined with \eqref{e:dph6} gives
	\begin{align*}
	\|\del{t} \mathring{\Phi}_{\epsilon,\yve}\|_{R^{s+1}} \leq  C\bigl(\|\delta\mathring{\zeta}_{\epsilon,\yve}\|_{L^\infty((t
		,1],H^s)}\bigr)\|\mathring{z}^j_{\epsilon,\yve}\|_{H^s}, \quad T <t\leq1.
	\end{align*}

	By adding the conformal cosmological Poisson-Euler equations \eqref{E:COSEULERPOISSONEQ.a}-\eqref{E:COSEULERPOISSONEQ.b} together, we obtain the following
	equation for $\mathring{\rho}\mathring{z}^j$:
	\begin{align*}
	\partial_t\left(\mathring{\rho}\mathring{z}^j\right)+\sqrt{\frac{3}{\Lambda}}K \frac{\delta^{ji}}{\mathring{E}^2} \partial_i \mathring{\rho}+\sqrt{\frac{3}{\Lambda}}\partial_i\left( \mathring{\rho}\mathring{z}^i\mathring{z}^j\right)=\frac{4-3\mathring{\Omega}}{t}\mathring{\rho} \mathring{z}^j -\frac{1}{2}\left(\frac{3}{\Lambda}\right)^{\frac{3}{2}} \frac{t}{\mathring{E}}\delta^{ij} \mathring{\rho}\mathring{\Phi}_i.
	\end{align*}
With the help of \eqref{e:varpi}, it is not difficult to verify that this equation is equivalent to
	\begin{equation}\label{E:PTRHOZ0}
	t\del{t} \varpi^j +\frac{\mathring{E} }{t^2}\sqrt{\frac{3}{\Lambda}}K \delta^{ij}\del{i}\delta \mathring{\rho}+\frac{\mathring{E}^3}{t^2}\sqrt{\frac{3}{\Lambda}}\del{i}\left( \mathring{\rho}\mathring{z}^i\mathring{z}^j\right) = \varpi^j -\frac{1}{2}\left(\frac{3}{\Lambda}\right)^{\frac{3}{2}}\frac{\mathring{E}^2}{t } \delta^{ij} ( \mathring{\rho}\mathring{\Phi}_i).
	\end{equation}
From this equation, \eqref{varpiest} and \eqref{e:phd} , we obtain, with the help of calculus inequalities, the estimate
	\begin{align*}
	\|t\del{t}\varpi^j_{\epsilon,\yve}\|_{R^{s-1}}\leq & C\bigl(\|\delta\mathring{\zeta}_{\epsilon,\yve}\|_{L^\infty((t
		,1],H^s)},\|\mathring{z}^j_{\epsilon,\yve}\|_{L^\infty((t
		,1],H^s)}\bigr)\Bigl( \|\delta\mathring{\zeta}_{\epsilon,\yve}\|_{R^s}+\|\mathring{z}^j_{\epsilon,\yve}\|_{H^s}+\|\delta\breve{\rho}_{\epsilon,\yve}\|_{L^{\frac{6}{5}}\cap H^s} \nnb  \\
	& 
	+
	\int_{t}^1\| \mathring{z}^k_{\epsilon,\yve}(\tau)\|_{H^s} d\tau 
	\Bigr).
	\end{align*}

Next, applying the operator $\mathfrak{R}_j(-\Delta)^{-\frac{1}{2}}$ to \eqref{E:PTRHOZ0} gives
	\begin{align}\label{E:PTRHOZ}
	&t\del{t}\mathfrak{R}_j(-\Delta)^{-\frac{1}{2}}\varpi^j -\frac{\mathring{E} }{t^2}\sqrt{\frac{3}{\Lambda}}K \delta^{ij} \mathfrak{R}_j \mathfrak{R}_i\delta \mathring{\rho}-\frac{\mathring{E}^3}{t^2}\sqrt{\frac{3}{\Lambda}}\mathfrak{R}_j\mathfrak{R}_i\left( \mathring{\rho}\mathring{z}^i\mathring{z}^j\right) \nnb  \\
	&\hspace{5cm}= \mathfrak{R}_j(-\Delta)^{-\frac{1}{2}}\varpi^j -\frac{1}{2}\left(\frac{3}{\Lambda}\right)^{\frac{3}{2}}\frac{\mathring{E}^2}{t } \delta^{ij}\mathfrak{R}_j(-\Delta)^{-\frac{1}{2}}( \mathring{\rho}\mathring{\Phi}_i).
	\end{align}	
 Using the potential theory estimates from Appendix \ref{S:Potop} and \eqref{varpiest}, we observe that the bound
	\begin{align*}
	\|(-\Delta)^{-\frac{1}{2}}\mathfrak{R}_k\varpi^j\|_{R^{s+1}}\lesssim & \|(-\Delta)^{\frac{1}{2}}\varpi^j\|_{L^6}+\sum_{l=1}^3\|\mathfrak{R}_l\mathfrak{R}_k\varpi^j\|_{H^{s}} \lesssim \|\varpi^j\|_{H^{s}}\leq C\bigl(\|\delta\mathring{\zeta}\|_{L^\infty((t
		,1],H^s)}\bigr)\|\mathring{z}^j\|_{H^s}
	\end{align*}
holds for $T <t\leq1$. It is then not difficult to verify that the estimate
	\begin{align}\label{e:tdtdpi}
	\|t\del{t}\mathfrak{R}_j(-\Delta)^{-\frac{1}{2}}\varpi^j_{\epsilon,\yve}\|_{R^s}
	\leq & C\bigl(\|\delta\mathring{\zeta}_{\epsilon,\yve}\|_{L^\infty((t
		,1],H^s)},\|\mathring{z}^j_{\epsilon,\yve}\|_{L^\infty((t
		,1],H^s)}\bigr)\bigl( \|\delta\mathring{\zeta}_{\epsilon,\yve}\|_{R^s}+\|\mathring{z}^j_{\epsilon,\yve}\|_{H^s} + \|\delta\breve{\rho}_{\epsilon,\yve}\|_{L^{\frac{6}{5}}\cap H^s}  \nnb  \\ & 
	+ 
	\int_{t}^1\| \mathring{z}^k_{\epsilon,\yve}(\tau)\|_{H^s} d\tau 
	\bigr)
	\end{align}
follows from \eqref{e:phd}, \eqref{E:PTRHOZ}, the potential theory estimates, and the calculus inequalities.

	Differentiating \eqref{E:PTPHI1} with respect to $t$ shows that
$\partial_t^2 \mathring{\Phi}_{\epsilon,\yve}  =-\sqrt{\frac{\Lambda}{3}} \mathfrak{R}_j(-\Delta)^{-\frac{1}{2}}\del{t}\varpi^j_{\epsilon,\yve}$. Using this,
we obtain, with the help of \eqref{e:tdtdpi}, the estimate
	\begin{align*}
	\|t\partial_t^2\mathring{\Phi}_{\epsilon,\yve}\|_{R^s} \lesssim  \| \mathfrak{R}_j(-\Delta)^{-\frac{1}{2}}t\del{t}\varpi^j_{\epsilon,\yve}\|_{R^s}
	\leq & C\bigl(\|\delta\mathring{\zeta}_{\epsilon,\yve}\|_{L^\infty((t
		,1],H^s)},\|\mathring{z}^j_{\epsilon,\yve}\|_{L^\infty((t
		,1],H^s)}\bigr)\Bigl( \|\delta\mathring{\zeta}_{\epsilon,\yve}\|_{R^s}+\|\mathring{z}^j_{\epsilon,\yve}\|_{H^s}
	\nnb  \\
	& +\|\delta\breve{\rho}_{\epsilon,\yve}\|_{L^{\frac{6}{5}} \cap H^s}+
	\int_{t}^1\| \mathring{z}^k_{\epsilon,\yve}(\tau)\|_{H^s} d\tau 
	\Bigr)
	\end{align*}
for $T <t\leq1$.

Applying the $\|\cdot\|_{R^s}$ norm to \eqref{E:PTPHI1b}, it is clear that
the estimate
	\begin{align*}
	\|\del{t}\mathring{\Upsilon}\|_{R^s} \lesssim \|\varpi^j\|_{R^s}\lesssim \|\varpi^j\|_{H^s} \leq C\bigl(\|\delta\mathring{\zeta}\|_{L^\infty((t
		,1],H^s)}\bigr)\|\mathring{z}^j\|_{H^s}, \quad t\in (T,1],
	\end{align*}
follows directly from  \eqref{varpiest} and the Yukawa operator estimate form Proposition \ref{T:Yu3est}.
Finally, from the bound
	\begin{align*}
	\|\mathring{\Upsilon}_{i,\epsilon,\yve}(1)\|_{H^{s+1}}=\Bigl\|\epsilon\beta \frac{\Lambda}{3}\mathring{E}^3(1) (\Delta-\epsilon^2\beta)^{-1} \delta \mathring{\rho}_{\epsilon,\yve}(1)\Bigr \|_{H^{s+1}} \lesssim \|  (\epsilon^2\beta-\Delta)^{-\frac{1}{2}} \delta \mathring{\rho}_{\epsilon,\yve}(1) \|_{H^{s+1}} \lesssim \|\delta\breve{\rho}_{\epsilon,\yve}\|_{L^{\frac{6}{5}}\cap H^s},		
	\end{align*}
which follows from the Yukawa operator estimates from Proposition \ref{T:Yu3est} along with \eqref{e:halfYu} and \eqref{e:sYupq3},
we see, after integrating \eqref{E:PTPHI1b} in time and applying the $\|\cdot\|_{H^s}$ norm, that
\begin{align*}
	\|\mathring{\Upsilon}_{\epsilon,\yve}(t)\|_{H^s}\lesssim \|\mathring{\Upsilon}_{\epsilon,\yve}(1)\|_{H^s} +& \int^1_t \|\epsilon \del{j}(\Delta-\epsilon^2 \beta)^{-1}\varpi^j_{\epsilon,\yve}\|_{H^s} d\tau  \leq C\|\delta\breve{\rho}_{\epsilon,\yve}\|_{L^{\frac{6}{5}}\cap H^s}+\int^1_t \| \varpi^j_{\epsilon,\yve}(\tau)\|_{H^s} d\tau  \nnb  \\
	& \leq C\|\delta\breve{\rho}_{\epsilon,\yve}\|_{L^{\frac{6}{5}}\cap H^s}+ C\bigl(\|\delta\mathring{\zeta}_{\epsilon,\yve}\|_{L^\infty((t
		,1],H^s)}\bigr)\int_t^1 \|\mathring{z}^j_{\epsilon,\yve}(\tau)\|_{H^s}d\tau
	\end{align*}
for $T <t\leq1$, which complete the proof.
\end{proof}

\begin{remark}
It follows from \eqref{e:unifini} and \eqref{e:Nini} that the size of the initial $(\delta\mathring{\zeta}_{\epsilon,\yve}|_{\Sigma},\mathring{z}^i_{\epsilon,\yve}|_\Sigma)$, as measured with respect to the $H^s$ norm, is independent of the parameters $(\epsilon,\yve)$. An immediate
consequence is that the time of existence $T$ from Proposition \ref{PEexist} is independent of $(\epsilon, \yve)$.
\end{remark}


\section{A non-local formulation of the reduced conformal Einstein-Euler system}\label{S:NloEE}

\subsection{Poisson potential estimates}
In \S \ref{rcEEeqns}, we brought the reduced conformal Einstein-Euler equations into a form, see \eqref{E:LOCEQ},  that is suitable for obtaining the global existence
of solutions to the future at fixed $\epsilon > 0$ using the theory developed in \cite{Oliynyk2016a}. However, due to the singular dependence of the
source term $\mathbf{\hat{H}}$ in the evolution equations \eqref{E:LOCEQ} on the parameter $\epsilon$, these equations, in their current form, are not useful for
analyzing the global existence problem in the limit $\epsilon \searrow 0$. To remedy this situation, we perform a non-local change of
variables designed to eliminate the singular term from $\mathbf{\hat{H}}$.  We note that a similar transformation was used
previously in the articles \cite{Liu2017,Oliynyk2008a,Oliynyk2008,Oliynyk2009a,Oliynyk2009b,Oliynyk2009,Oliynyk2014,Oliynyk2015}.

The transformation is based on shifting the metric variable $u^{0\mu}_i$, see \eqref{E:WPHI}, by the following
non-local term
\al{POTENTIAL}{
	\Phi_k^\mu= \frac{\Lambda}{3}\frac{E^3}{t^3}\delta^\mu_0 \del{k} (\Delta-\epsilon^2 \beta)^{-1}    \bigl(E^{-3}\sqrt{|\underline{\bar{g}}|} \underline{\bar{v}^0}\varrho- \mu^{\frac{1}{1+\epsilon^2K}}\bigr) 	
}
where $\varrho=\rho^{\frac{1}{1+\epsilon^2K}}$ and $|\underline{\bar{g}}|:=-\det\bigl (\underline{\bar{g}_{\mu\nu}}\bigr)$. We note that this term is closely related to the  spatial derivative of the Newtonian potential. We also observe that two
equations
	\begin{align}
\partial_l
	\Phi_k^\mu=& \frac{\Lambda}{3t^3} \delta^\mu_0 (\Delta-\epsilon^2\beta)^{-1}\del{k}\del{l}  \bigl(\sqrt{|\underline{\bar{g}}|}   \underline{\bar{v}^0}\varrho- E^3\mu^{\frac{1}{1+\epsilon^2K}}\bigr) \label{e:dsphi2}	
	\intertext{and}
\delta^{ik} \del{i}\Phi_k^\mu=&\frac{\Lambda}{3t^3} \delta^\mu_0   \bigl(\sqrt{|\underline{\bar{g}}|}   \underline{\bar{v}^0}\varrho- E^3\mu^{\frac{1}{1+\epsilon^2K}}\bigr)+\epsilon^2\beta \frac{\Lambda}{3t^3} \delta^\mu_0 (\Delta-\epsilon^2\beta)^{-1}  \bigl(\sqrt{|\underline{\bar{g}}|}   \underline{\bar{v}^0}\varrho- E^3\mu^{\frac{1}{1+\epsilon^2K}}\bigr) \label{e:dsphi}
	\end{align}
follow directly from differentiating \eqref{E:POTENTIAL}.

For use below in  simplifying the expression $\mathbf{F}$ from \S\ref{S:MAINPROOF}, we define
\begin{align}\label{e:Upsi2}
\Upsilon= \frac{\Lambda}{3} \frac{E^3}{t^3}\epsilon \beta (\Delta-\epsilon^2 \beta)^{-1}    \bigl(E^{-3}\sqrt{|\underline{\bar{g}}|}   \underline{\bar{v}^0}\varrho- \mu^{\frac{1}{1+\epsilon^2K}}\bigr),
\end{align}
where $\beta>0$ is an arbitrary constant.
We further note the expansions 
\begin{gather}
\bigl(\rho^{\frac{1}{1+\epsilon^2K}}-\mu^{\frac{1}{1+\epsilon^2K}}\bigr) -(\rho-\mu)=\epsilon^2\mathscr{S}(\epsilon, t, \delta\zeta)  \label{e:rhodiff}
\intertext{and} 
E^{-3}\sqrt{|\underline{\bar{g}}|}   \underline{\bar{v}^0}\varrho- \mu^{\frac{1}{1+\epsilon^2K}}=t^3 e^{\zeta_H} (e^{\delta\zeta}-1)+ \epsilon \mathscr{T}_1(\epsilon, t,u^{\mu\nu},u,\delta\zeta)+ \epsilon^2 \mathscr{T}_2(\epsilon, t,u^{\mu\nu},u,\delta\zeta,z_j), \label{e:de}
\end{gather}
where 
$\mathscr{S} $, $\mathscr{T}_1 $, and $\mathscr{T}_2$ vanish to first order in
 $\delta\zeta$,  $(u^{\mu\nu},u)$, and
 $(u^{\mu\nu},u,\delta\zeta,z_j)$, respectively.

\begin{proposition}\label{T:phiex}
	Suppose $s\in \Zbb_{\geq 3}$
     and $\bhU_{\epsilon,\yve}  \in \bigcap_{\ell=0}^1C^\ell((T,1], R^{s-\ell}(\mathbb{R}^3,\mathbb{K}) )$ is the solution to \eqref{E:LOCEQ} from Corollary \ref{T:poes}.
    Then $\Phi^\mu_k$ and $\Upsilon$ 
	satisfy the estimates
	\begin{align}
	\|\Phi^\mu_{k, \epsilon,\yve}  \|_{R^{s }} \leq & C_0
	\Bigl(\|\breve{\xi}_\epsilon \|_{s}  +
	\int_t^1 (\| u^{0i}_{\epsilon,\yve}(\tau)\|_{ R^s }+\| z_{l,\epsilon,\yve} (\tau)\|_{ R^s } ) d\tau \Bigr), \label{e:phi3} \\
	\|\del{l}\Phi^\mu_{k,\epsilon,\yve}\|_{R^s}\leq &C_0(\| \delta\zeta_{\epsilon,\yve}\|_{R^s}+\|  u^{\mu\nu}_{\epsilon,\yve}\|_{R^s}+\| u_{\epsilon,\yve}\|_{R^s} +\| z_{j,\epsilon,\yve}\|_{R^s}),   \label{e:ddphi3}\\
	\|\del{t} \Phi ^\mu_{k,\epsilon,\yve}\|_{R^s} \leq &   C_0(\|u^{0j}_{\epsilon,\yve}\|_{R^s}+\|z_{i,\epsilon,\yve} \|_{R^s}) \label{e:dtphi3}
\intertext{and}
	\|\Upsilon_{ \epsilon,\yve}  \|_{R^{s }} \leq & C_0 
	\Bigl(\|\breve{\xi}_\epsilon \|_{s}  +
	\int_t^1 (\| u^{0i}_{\epsilon,\yve}(\tau)\|_{ R^s }+\| z_{l,\epsilon,\yve} (\tau)\|_{ R^s } ) d\tau \Bigr)\label{e:upsilon}
	\end{align}
for $T < t \leq 1$, where $C_0=C_0\bigl(\|(u^{\mu\nu}_{\epsilon,\yve},u_{\epsilon,\yve}, \delta\zeta_{\epsilon,\yve}, z_{j,\epsilon,\yve})\|_{L^\infty((t,1],R^s)}\bigr)$.
\end{proposition}
\begin{proof}
To simplify notation, we drop the subscripts involving the parameters $(\epsilon,\yve)$ from all quantities for the remainder of this proof.
Fixing a solution $\bhU  \in \bigcap_{\ell=0}^1C^\ell((T,1], R^{s-\ell}(\mathbb{R}^3,\mathbb{K}) )$ to the reduced Einstein-Euler system  \eqref{E:LOCEQ} from Corollary \ref{T:poes}, we let $\bar{v}^\mu$, $\bar{\rho}$ and $\bar{g}_{\mu\nu}$ denote the fluid variables and spacetime metric
in relativistic  coordinates determined by this solution. Then,
 contracting the conformal Euler equations \eqref{Confeul} with $\bar{v}_\nu$ yields the conformal continuity equation 
$\bar{v}^\mu\bar{\nabla}_\mu\bar{\rho}+(1+\epsilon^2 K)\bar{\rho} \bar{\nabla}_\mu \bar{v}^\mu= -3(1+\epsilon^2 K)\bar{\rho} \bar{v}^\mu \bar{\nabla}_\mu \Psi$, which in turn implies that
$\bar{\nabla}_\mu(\bar{\varrho}\bar{v}^\mu)=\frac{1}{\sqrt{|\bar{g}|}}\bar{\partial}_\mu(\sqrt{|\bar{g}|}\bar{\varrho} \bar{v}^\mu) =\frac{3}{t} \bar{\varrho} \bar{v}^0$, where $\bar{\varrho} =  \bar{\rho}^{\frac{1}{1+\epsilon^2K}}$. From this equation, we then find that
	\al{POTEVOL}{
		\del{t}(\sqrt{|\underline{\bar{g}}|}   \underline{\bar{v}^0}\varrho)+   \del{i}(\sqrt{|\underline{\bar{g}|}}\varrho z^i)=\frac{3}{t} \sqrt{|\underline{\bar{g}}|}  \underline{\bar{v}^0} \varrho.
	}

Next, we see that the equations
	\begin{align}
	\del{t} \Phi_k^\mu=&- \frac{\Lambda}{3} (\Delta-\epsilon^2 \beta)^{-1}\del{k}\del{l}  \bigl(  \delta^\mu_0\sqrt{|\underline{\bar{g}}|} e^\zeta z^l\bigr) \label{E:POTEQ}\\
	\intertext{and}
	\del{t} \Upsilon=&- \epsilon \beta \frac{\Lambda}{3} (\Delta-\epsilon^2 \beta)^{-1} \del{l}  \bigl( \sqrt{|\underline{\bar{g}}|} e^\zeta z^l\bigr). \label{E:UPses1}
	\end{align}
follow from acting on \eqref{E:POTEVOL} with the operators $\frac{\Lambda}{3}\frac{1}{t^3} \delta^\mu_0 (\Delta-\epsilon^2 \beta)^{-1} \del{k}$ and $\epsilon \beta \frac{\Lambda}{3}\frac{1}{t^3}   (\Delta-\epsilon^2 \beta)^{-1} $ along with the help of \eqref{E:ZETA} and \eqref{E:POTENTIAL}. Applying the $\|\cdot\|_{R^s}$ norm to \eqref{E:POTEQ}, we obtain, with the help of \eqref{E:VELOCITY}, Proposition \ref{t:ddel}.\eqref{ddel2} and the calculus inequalities from Appendix \ref{A:INEQUALITIES} , the estimate
	\als{
		\|\del{t} \Phi_k^\mu(t)\|_{R^s}\leq  C_0(\|u^{0j}(t)\|_{R^s}+\|z_i(t) \|_{R^s}),\quad T < t \leq 1,  
	}
where, here and for the remainder of the proof, we let $C_0$ denote a constant of the form
\begin{equation*}
C_0=C_0\bigl(\|(u^{\mu\nu}_{\epsilon,\yve},u_{\epsilon,\yve}, \delta\zeta_{\epsilon,\yve}, z_{j,\epsilon,\yve})\|_{L^\infty((t,1],R^s)}\bigr).
\end{equation*}
Integrating \eqref{E:POTEQ} and \eqref{E:UPses1} in time, we obtain, after taking $\|\cdot\|_{R^s}$ norm and using \eqref{E:VELOCITY},  \eqref{e:sYupq}, \eqref{E:DELTAZETA1},  \eqref{E:POTENTIAL}, \eqref{e:de},  Propositions \ref{T:genYu} and \eqref{T:Yu3est}, and Corollary \ref{T:poes},
the estimates
	\als{
		\|\Phi_k^\mu (t) \|_{R^{s }} \lesssim & \|\Phi_k^\mu(1)\|_{R^{s }}+\int_t^1\left\|(\Delta-\epsilon^2 \beta)^{-1}\del{k}\del{l} \bigl(  \sqrt{|\underline{\bar{g}}|} e^\zeta z^i\delta^\mu_0\bigr)\right\|_{R^{s }}d\tau  \nnb  \\
		\leq & C\| (\epsilon^2\beta-\Delta)^{-\frac{1}{2}}   \bigl(t^3 e^{\zeta_H} (e^{\delta\zeta}-1)+ \epsilon \mathscr{T}_1(\epsilon, t,u^{\mu\nu},u,\delta\zeta)+ \epsilon^2 \mathscr{T}_2(\epsilon, t,u^{\mu\nu},u,\delta\zeta,z_j)\bigr)(1)\|_{R^s} \nnb  \\
		& +C_0 \int_t^1 (\| u^{0i}(\tau)\|_{ R^s }+\| z _l (\tau)\|_{ R^s } ) d\tau \nnb  \\
		\leq & C_0\biggl( \|u^{\mu\nu}(1)\|_{R^s}+ \|u(1)\|_{R^s}+ \|\delta\zeta(1)\|_{R^s}+ \|z_j(1)\|_{R^s}
		+\int_t^1 (\| u^{0i}(\tau)\|_{ R^s }+\| z _l (\tau)\|_{ R^s } ) d\tau\biggr)
	}
	and, similarly,
	\als{
		\|\Upsilon (t) \|_{R^{s }} 
		\leq & \biggl(\|u^{\mu\nu}(1)\|_{R^s}+ \|u(1)\|_{R^s}+ \|\delta\zeta(1)\|_{R^s}+ \|z_j(1)\|_{R^s}  
		+\int_t^1 (\| u^{0i}(\tau)\|_{ R^s }+\| z _l (\tau)\|_{ R^s } ) d\tau\biggl)
	}
for $T < t \leq 1$,
where in the above derivations, we  have used \eqref{e:halfYu} to conclude that
$\| (\epsilon^2\beta-\Delta)^{-\frac{1}{2}}   \epsilon \mathscr{T}_1\|_{R^s}$ $\lesssim$ $\|\mathscr{T}_1\|_{R^s}$ and
 $\| (\epsilon^2\beta-\Delta)^{-\frac{1}{2}}   \epsilon^2 \mathscr{T}_2\|_{R^s}$ $\lesssim$ $\epsilon \|\mathscr{T}_2\|_{R^s}$.
In addition, we deduce from \eqref{e:dsphi2} and \eqref{e:de} the estimates
	\begin{align}\label{e:ddph0}
	\|\partial_l
	\Phi_k^\mu(t)\|_{L^6} \lesssim & \|   (\Delta-\epsilon^2\beta)^{-1}\del{k}\del{l}  \bigl(\sqrt{|\underline{\bar{g}}|}  \underline{\bar{v}^0} \varrho - E^3\mu^{\frac{1}{1+\epsilon^2K}}\bigr)(t)\|_{L^6}  \leq C_0(\| \delta\zeta(t)\|_{L^6}+\|  u^{\mu\nu}(t)\|_{L^6}+\| u(t)\|_{L^6} +\| z_j(t)\|_{L^6})
	\end{align}
	and
	\begin{align}
	\|\del{j}\partial_l
	\Phi_k^\mu(t)\|_{H^{s-1}}\lesssim & \|\del{j} (\Delta-\epsilon^2\beta)^{-1}\del{k}\del{l}  \bigl(\sqrt{|\underline{\bar{g}}|} \underline{\bar{v}^0} \varrho- E^3\mu^{\frac{1}{1+\epsilon^2K}}\bigr)(t)\|_{H^{s-1}} \nnb\\
	\leq & C_0 \bigl( \|\del{j}\delta\zeta(t)\|_{H^{s-1}}+\|\del{j} u^{\mu\nu}(t)\|_{H^{s-1}}+\|\del{j}u(t)\|_{H^{s-1}} +\|\del{j}z_j(t)\|_{H^{s-1}} \bigr),  
	\label{e:ddph}
	\end{align}
which hold for $T < t \leq 1$. Together, \eqref{e:ddph0} and \eqref{e:ddph} imply that
	\begin{align*}
	\|\del{l}\Phi^\mu_k(t)\|_{R^s}\leq C_0(\| \delta\zeta(t)\|_{R^s}+\|  u^{\mu\nu}(t)\|_{R^s}+\| u(t)\|_{R^s} +\| z_j(t)\|_{R^s}),
\quad T < t \leq 1,
	\end{align*}
which completes the proof.
\end{proof}

\subsection{The non-local transformation} \label{nonlocaleq} 
From the definition of source term $\hat{H}$, see \eqref{E:LOCEQ}, it is clear, see \eqref{E:EINS1}-\eqref{E:EINS23} and \eqref{E:SREMAINDER},
that the only $\epsilon$-singular terms appear in $\hat{S}_1$ and are of the form   $\frac{1}{\epsilon}\left(\underline{\bar{v}^0}-\sqrt{\frac{\Lambda}{3}}\right)$ and
$-\frac{1}{\epsilon}\frac{\Lambda}{3}\frac{1}{t^2}E^2\delta\rho\delta^\mu_0$. Noting that $\frac{1}{\epsilon}\left(\underline{\bar{v}^0}-\sqrt{\frac{\Lambda}{3}}\right)$ is actually regular in $\epsilon$ as can be seen directly from the expansion \eqref{E:V^0},  the only $\epsilon$-singular term left to deal with
is $-\frac{1}{\epsilon}\frac{\Lambda}{3}\frac{1}{t^2}E^2\delta\rho\delta^\mu_0$.
Following the method introduced in \cite{Oliynyk2008a} and then adapted to the cosmological
setting in \cite{Oliynyk2009a}, we can remove this singular term from \eqref{E:LOCEQ} while preserving
its desirable structure via the introduction of the shifted variable
\begin{align}\label{E:WPHI}
w^{0\mu}_k=u^{0\mu}_k- t E^{-1} \Phi_k^\mu,
\end{align}
where $\Phi_k^\mu$ is as defined above by \eqref{E:POTENTIAL}.

Under the change of variables \eqref{E:WPHI}, a short calculation using \eqref{e:dsphi}, where we note that
\begin{align*}
\frac{1}{\epsilon} \frac{\Lambda}{3t^2}   \delta^\mu_0  \bigl(E^{-3}\sqrt{|\underline{\bar{g}}|} \underline{\bar{v}^0} \varrho- \mu^{\frac{1}{1+\epsilon^2K}} \bigr) -  \frac{1}{\epsilon}\frac{\Lambda}{3t^2} \delta^\mu_0 \delta\rho
=&
\mathscr{X}^\mu_1(\epsilon, t,  u^{\alpha\beta}, u, \delta\zeta) +\epsilon \mathscr{X}^\mu_2(\epsilon, t,  u^{\alpha\beta}, u, \delta\zeta, z_j )
\end{align*}
with $\mathscr{X}^\mu_1$ and $\mathscr{X}^\mu_2$ vanishing to first order in $(u^{\alpha\beta}, u)$ and $(u^{\alpha\beta}, u, \delta\zeta, z_j )$,
respectively, shows that equation \eqref{E:EIN1} transforms into
\begin{align}
\tilde{B}^0  \partial_0\begin{pmatrix}
u^{0\mu}_0\\w^{0\mu}_k\\u^{0\mu}
\end{pmatrix}+ \tilde{B}^k \partial_k\begin{pmatrix}
u^{0\mu}_0\\w^{0\mu}_l\\u^{0\mu}
\end{pmatrix}+\frac{1}{\epsilon}\tilde{C}^k\partial_k\begin{pmatrix}
u^{0\mu}_0\\w^{0\mu}_l\\u^{0\mu}
\end{pmatrix}&=\frac{1}{t}\mathfrak{\tilde{B}}\mathbb{P}_2\begin{pmatrix}
u^{0\mu}_0\\w^{0\mu}_l\\u^{0\mu}
\end{pmatrix} +\tilde{G}_1+\tilde{S}_1, \label{E:EIN1f}
\end{align}
where
\begin{align} \label{E:EING2}
\tilde{G}_1=E^2\begin{pmatrix}
2  E^{-2}\frac{\Omega}{t}\delta^{kj}  \delta^\mu_j w^{00}_k -\frac{2(1-\epsilon^2K)}{t}  \delta \rho u^{0\mu} + \tilde{f}^{0\mu}+  \mathscr{L}^{0\mu}   \\ \mathscr{L}^{0\mu l}    \\
\mathscr{L}^{00\mu}
\end{pmatrix},
\end{align}
\begin{align*} 
\tilde{S}_1=E^2\begin{pmatrix}
2  E^{-3}\Omega\delta^{kj}  \delta^\mu_j \Phi^0_k  + \Theta^{kl} t E^{-1} \del{k} \Phi^\mu_l - \frac{2}{t^2}  \rho\delta^\mu_i z^i \sqrt{\frac{\Lambda}{3}}+E^{-3} t \delta^\mu_0 \Upsilon  +\epsilon \mathscr{J}^{0\mu} \\ \bigl(\frac{1}{2}+\Omega\bigr)E^{-1}\underline{\bar{g}^{kl}}\Phi^\mu_k
-\underline{\bar{g}^{kl}}tE^{-1}\del{0}\Phi^\mu_k  +\epsilon \mathscr{J}^{0\mu l}    \\
\epsilon \mathscr{J}^{00\mu},
\end{pmatrix},
\end{align*}
and we have set
\begin{align*}
\tilde{f}^{0\mu}
= &  \mathscr{X}^\mu_1+\epsilon \mathscr{X}^\mu_2 -\frac{\Lambda}{3}\frac{\epsilon K}{ t^2}  \delta\rho \delta^\mu_0 - \frac{4}{t^2}\frac{1}{\epsilon}\sqrt{\frac{\Lambda}{3}}\delta\rho\delta^\mu_0\biggl(\underline{\bar{v}^0}-\sqrt{\frac{\Lambda}{3}}\biggr)- \frac{2}{t^2}\frac{1}{\epsilon}\delta\rho\delta^\mu_0\biggl(\underline{\bar{v}^0}-\sqrt{\frac{\Lambda}{3}}\biggr)^2 - \frac{2}{t^2}  \rho\delta^\mu_i z^i \biggl(\underline{\bar{v}^0}-\sqrt{\frac{\Lambda}{3}}\biggr)
\nnb  \\
& - \frac{2}{t^2}\frac{1}{\epsilon} \mu  \delta^\mu_0 \biggl( \underline{\bar{v}^0}-\sqrt{\frac{\Lambda}{3}} \biggr)\biggl( \underline{\bar{v}^0}+ \sqrt{\frac{\Lambda}{3}} \biggr)   - \frac{2}{t^2}  \epsilon K \left(\delta\rho \underline{\bar{v}^\mu} \underline{\bar{v}^0} + \mu  \Bigl( \underline{\bar{v}^\mu} \underline{\bar{v}^0}-  \frac{\Lambda}{3}\delta^\mu_0 \Bigr) \right).
\end{align*}
Observe now that the right hand side of this equation is regular in $\epsilon$.
For later use in \S \ref{S:MAINPROOF}, we decompose the remainder term $\hat{S}$, see \eqref{E:SREMAINDER}, from the Euler equation \eqref{E:FINALEULEREQUATIONS}
as
\als{\hat{S}=G+S,}
where
\al{G}{G=\p{0 \\
		-K^{-1}\left[\sqrt{\frac{3}{\Lambda}}\bigl(-u^{0l}_0+(-3+4\Omega)u^{0l}\bigr) +\frac{1}{2}\left(\frac{3}{\Lambda}\right)^{\frac{3}{2}}E^{-2}\delta^{lk}w^{00}_k\right]
	}	}
and
\al{S2a}{
	S= \p{
		0\\ -K^{-1}\frac{1}{2}\left(\frac{3}{\Lambda}\right)^{\frac{3}{2}}E^{-3}\delta^{lk} t \Phi^0_k
	} + \epsilon \mathscr{S}(\epsilon,t,u,u^{\alpha\beta},u_\gamma,u^{\alpha\beta}_\gamma,z_j).
}

\subsection{The complete evolution system} \label{completeevolution}
We incorporate the shifted variable \eqref{E:WPHI}
into our set of gravitational variables by defining the vector quantity
\begin{align} \label{E:U1}
\mathbf{U}_1=(u^{0\mu}_0, w^{0\mu}_k, u^{0\mu}, u^{ij}_0, u^{ij}_k, u^{ij}, u_0, u_k, u)^T,
\end{align}
and then combine this with the fluid variables by defining
\begin{align} \label{E:REALVAR}
\mathbf{U}=(\mathbf{U}_1, \mathbf{U}_2 )^T,
\end{align}
where  $\mathbf{U}_2=(\delta\zeta, z _i)^T$ is as previously defined by \eqref{E:REALVAR1}.
Gathering \eqref{E:EIN2}, \eqref{E:EIN3}, \eqref{E:EIN1f} and  \eqref{E:FINALEULEREQUATIONS} together, we arrive at the following complete evolution equation for $\mathbf{U}$:
\begin{equation}\label{E:REALEQ}
\begin{aligned}
\mathbf{B}^0\partial_t \mathbf{U}+\mathbf{B}^i\partial_i \mathbf{U}+\frac{1}{\epsilon}\mathbf{C}^i\partial_i\mathbf{U}=\frac{1}{t}\mathbf{B}\mathbf{P}
\mathbf{U}+\mathbf{H}+\mathbf{F},
\end{aligned}
\end{equation}
where we recall that $\mathbf{B}^0$, $\mathbf{B}^i$, $\mathbf{C}^i$, $\mathbf{B}$ and $\mathbf{P}$ are defined by \eqref{E:REALEQa}-\eqref{E:REALEQb} and 
\begin{align}
\mathbf{H}
=(\tilde{G}_1,
\tilde{G}_2,
\tilde{G}_3,
G ) ^T  \AND
\mathbf{F}
=(\tilde{S}_1,
\tilde{S}_2,
\tilde{S}_3,
S ) ^T . \label{E:REALEQc}
\end{align}
The importance of equation \eqref{E:REALEQ} is twofold. First, it is completely equivalent to the formulation
\eqref{E:LOCEQ} of the reduced conformal Einstein-Euler equations.
Second, it is of the required form so that the a priori estimates established below in \S\ref{S:MODEL} apply to its solutions.
These two properties will be crucial for the proof of Theorem \ref{T:MAINTHEOREM}; see \S\ref{S:MAINPROOF} for details.

Before completing this section, we state the following proposition, which is a direct consequence of Corollary \ref{T:poes}, Proposition \ref{T:phiex} and
the change of variables  \eqref{E:WPHI}.
\begin{proposition}\label{T:locfull}
	Suppose $s\in \Zbb_{\geq 3}$, $\epsilon_0>0$, $\epsilon \in (0,\epsilon_0)$, $\vec{\yv}\in \Rbb^{3N}$ and
	\begin{align*}
	\mathbf{\hat{U}}_{\epsilon, \yve}|_{\Sigma}=\{u^{\mu\nu}_{\epsilon, \yve},u_{\epsilon, \yve},u^{ij}_{\gamma,\epsilon, \yve},
	u^{0\mu}_{i,\epsilon, \yve},u^{0\mu}_{0,\epsilon, \yve},u_{\gamma,\epsilon, \yve},z_{j,\epsilon, \yve},\delta \zeta_{\epsilon, \yve}\}|_{\Sigma}\in X^s(\Rbb^3) 
	\end{align*}
	is the initial data from Theorem \ref{e:inithm}. Then
	\begin{enumerate}
		\item there exists a constant  $T >0$  and a unique classical solution
		\begin{equation*}
		\mathbf{U}_{\epsilon, \yve} \in \bigcap_{\ell=0}^1C^\ell((T ,1], R^{s-\ell}(\mathbb{R}^3,\mathbb{K}) )
		\end{equation*}
		to \eqref{E:REALEQ} on the spacetime region $(T
		, 1] \times \Rbb^3$ that agrees, after applying the transformation \eqref{E:WPHI} to the $w^{0\mu}_{k,\epsilon,\yve}$ component of $\mathbf{U}_{\epsilon, \yve}$ , with the initial data $\mathbf{\hat{U}}_{\epsilon, \yve}|_{\Sigma}$ on the initial hypersurface $\Sigma$,
\item \label{T:locfull2}  the $w^{0\mu}_{k,\epsilon,\yve}$ component of $\mathbf{U}_{\epsilon, \yve}$  can be expanded as
 \begin{align*}\label{e:wini}
		w^{0\mu}_{k,\epsilon,\yve}|_{\Sigma}=\epsilon \mathcal{S}^\mu_k(\epsilon,\smfu^{kl}_{\epsilon,\yve} ,
\smfu^{kl}_{0,\epsilon,\yve} , \delta\breve{\rho}_{\epsilon, \yve} , \breve{z}^l_{\epsilon, \yve} )
\end{align*}
on the initial hypersurface, where $\mathcal{S}^\mu_k$ is defined in Theorem \ref{e:inithm}, 
		\item and there exists a constant $\sigma > 0$, independent
		of the initial data and $T_1 \in (0,1]$, such that if
		$\mathbf{U}_{\epsilon, \yve}$ exists for $t\in (T_1,1]$ with the same regularity as above and satisfies
		$\norm{\mathbf{U}_{\epsilon, \yve}}_{L^\infty((T_1,1],R^s))} < \sigma$,
		then the solution $\mathbf{U}_{\epsilon, \yve}$ can be uniquely continued as a classical solution
		with the same regularity
		to the larger spacetime region $(T^*
		,1]\times \mathbb{R}^3
		$ for some $T^*
		\in (0,T_1)$.
	\end{enumerate}
\end{proposition}
\begin{remark}
It is worthwhile noting that the time of existence $T$ from the above proposition can be chosen to be 
independent of $\epsilon$. This follows from the form of the evolution equations \eqref{E:REALEQ}, which would allow us to use the method
from \cite{Browning1982,Klainerman1981,Klainerman1982,Kreiss1980}
for deriving 
$\epsilon$-independent energy estimates. We omit the details since we will establish this and more in the following section; see Theorem \ref{T:MAINMODELTHEOREM} for details.
\end{remark}

\section{Singular Symmetric Hyperbolic Systems} \label{S:MODEL}
In this section, we establish uniform a priori estimates for solutions to a class of symmetric hyperbolic systems that are
jointly singular in $\epsilon$ and $t$, and include both the formulation of the reduced conformal Einstein-Euler equations given by \eqref{E:REALEQ} and the $\epsilon \searrow 0$
limit of these equations. We also establish \textit{error estimates}, that is, a priori estimates for the difference between solutions of the
$\epsilon$-dependent singular symmetric hyperbolic systems and their corresponding $\epsilon \searrow 0$ limit equations.

The $\epsilon$-dependent singular terms that appear in the symmetric hyperbolic systems we consider
are of a type that have been well studied, see \cite{Browning1982,Klainerman1981,Klainerman1982,Kreiss1980}, while the
$t$-dependent singular terms are of the type analyzed in \cite{Oliynyk2016a}. Previously, we analyzed such systems on the torus $\Tbb^n$ \cite{Liu2017}.  Here, we will generalize the results of \cite{Liu2017} in three spatial dimensions from $\Tbb^3$ to $\Rbb^3$.

\begin{remark}
In this section, we switch to the standard time orientation, where the future is located in the direction of increasing time, while keeping the singularity located at $t=0$. We
do this in order to make the derivation of the energy estimates in this section as similar as possible to those for non-singular symmetric hyperbolic systems, which we expect will
make it easier for readers familiar with such estimates to follow the arguments below. To get back to the time orientation used to formulate the conformal Einstein-Euler equations, we need only apply the trivial time transformation $t \mapsto -t$.
\end{remark}
\subsection{Uniform estimates\label{S:MODELuni}}
The class of singular hyperbolic systems that we will consider are of the following form:
\begin{equation}\label{E:MODELEQ2a}
A^0  \partial_t U+A^i  \partial_i U+\frac{1}{\epsilon}C^i\partial_i U=\frac{1}{t}\mathfrak{A} \mathbb{P}  U +H \quad \mbox{in} \quad[T_0, T_1)\times\mathbb{R}^3,
\end{equation}
where
\begin{align*}
U &= (w, u)^T, \\
A^0&=\p{A^0_1(\epsilon,t,x,w) & 0\\
	0 & A^0_2(\epsilon,t,x,w)}, \\
A^i&= \p{A^i_1 (\epsilon,t,x,w)& 0\\
	0 & A^i_2(\epsilon,t,x,w)}, \\
C^i&=\p{C^i_1 & 0\\
	0 & C^i_2 },  \quad  \Pbb=\p{\Pbb_1 & 0\\
	0 & \Pbb_2 }, \\
\mathfrak{A}&= \p{\mathfrak{A}_1 (\epsilon,t,x,w)& 0\\
	0 & \mathfrak{A}_2(\epsilon,t,x,w)},\\
H&= \p{H_1(\epsilon,t,x,w)\\ H_2(\epsilon,t,x,w,u)+R_2 }+\p{F_1(\epsilon,t,x)\\ F_2(\epsilon,t,x)},\\
R_2&=\frac{1}{t}
M_2(\epsilon, t,x,w,u) \Pbb_3 U,
\end{align*}
and the following assumptions hold for fixed constants $\epsilon_0,R >0$, $T_0 < T_1 < 0$ and $s\in \Zbb_{\geq 3}$:

\begin{ass}\label{ASS1}$\;$
	\begin{enumerate}
		\item  \label{A:CONSC} The $C^i_a$, $i=1,\ldots,n$ and $a=1,2$, are constant, symmetric $N_a\times N_a$ matrices.
		\item  The $\mathbb{P}_a$, $a=1,2$, are constant, symmetric $N_a\times N_a$ projection matrices, i.e. $\mathbb{P}_a^2= \mathbb{P}_a$. We use $\mathbb{P}_a^{\perp}=\mathds{1}-\mathbb{P}_a$  to denote the complementary projection matrix.
		\item \label{A:P4} $\mathbb{P}_4$ is a constant, symmetric $N_1\times N_1$ projection matrix that commutes with $A^\mu_1$, $\mathfrak{A}_1$, $\Pbb_1$ and $C^i_1$, that is,
		\begin{equation*}
		[\Pbb_4, A^\mu_1]=[\Pbb_4,\mathfrak{A}_1]=[\Pbb_4,\Pbb_1]=[\Pbb_4,C^i_1]=0. \label{e:p4com}
		\end{equation*}
		\item  \label{A:GH} The source terms $H_a(\epsilon,t,x,w)$, $a=1,2$, $F_a(\epsilon, t, x)$, $a=1,2$, and
		$M_2(\epsilon,t,x,w,u)$ satisfy
		$H_1 \in E^0\bigl((0,\epsilon_0)\times (2 T_0,0)\times \Rbb^3 \times B_R(\Rbb^{N_1}),\Rbb^{N_1}\bigr)$,
		$H_2 \in E^0\bigl((0,\epsilon_0)\times (2 T_0,0)\times \Rbb^3 \times B_R(\Rbb^{N_1})
		\times B_R(\Rbb^{N_2}) ,\Rbb^{N_2}\bigr)$,
		$F_a \in C^0\bigl((0,\epsilon_0)\times [T_0,T_1), H^s(\Rbb^3,\Rbb^{N_a})\bigr)$,
		$M_2  \in E^0\bigl((0,\epsilon_0)\times (2 T_0,0)\times \Rbb^3 \times B_R(\Rbb^{N_1})
		\times B_R(\Rbb^{N_2}),\mathbb{M}_{N_2\times N_2}\bigr)$, and
		\begin{equation*}
		\Pbb_4 H_1(\epsilon,t,x,\Pbb_4^\perp w)=0 , \quad H_1(\epsilon,t,x,0) = 0, \quad H_2(\epsilon,t,x,0,0) = 0 \AND M_2(\epsilon,t,x,0,0) = 0
		\end{equation*}
		for all $(\epsilon,t,x)\in (0,\epsilon_0)\times (2 T_0,0)\times \Rbb^3 $.
		
		\item  \label{A:Bi} The matrix valued maps $A_a^\mu(\epsilon,t,x,w)$, $\mu=0,\ldots,3$ and $a=1,2$, satisfy $A^\mu_a \in E^0\bigl((0,\epsilon_0)\times (2 T_0,0)\times \Rbb^3 \times B_R(\Rbb^{N_a}),\mathbb{S}_{N_a}\bigr)$.
		
		\item \label{A:B0} The matrix valued maps
		$A_a ^0(\epsilon,t,x, w)$, $a=1,2$, and $\mathfrak{A}_a(\epsilon,t,x, w)$, $a=1,2$, can be decomposed as
		\begin{gather}
		A_a^0(\epsilon,t,x, w)=\mathring{A}_a^0(t)+\epsilon \tilde{A}_a^0(\epsilon,t,x, w),\label{E:DECOMPOSITIONOFA01}\\
		\mathfrak{A}_a(\epsilon,t,x, w)=\mathring{\mathfrak{A}}_a(t)+\epsilon \tilde{\mathfrak{A}}_a(\epsilon,t,x, w),\label{E:DECOMPOSITIONOFCALB}
		\end{gather}
		where $\mathring{A}_a^0 \in E^1\bigl((2 T_0,0),\mathbb{S}_{N_a}\bigr)$,
		$\mathring{\mathfrak{A}}_a \in E^1\bigl((2 T_0,0),\mathbb{M}_{N_a\times  N_a}\bigr)$,
		$\tilde{A}_a^0 \in E^1\bigl((0,\epsilon_0)\times (2 T_0,0)\times \Rbb^3 \times B_R(\Rbb^{N_1}),\mathbb{S}_{N_a}\bigr)$,
		$\tilde{\mathfrak{A}}_a \in E^0\bigl((0,\epsilon_0)\times (2 T_0,0)\times \Rbb^3 \times B_R(\Rbb^{N_1}),\mathbb{M}_{N_a\times N_a}\bigr)$, and\footnote{Or in other words,  the matrices $\tilde{\mathfrak{A}}_a|_{w=0}$ and
			$\tilde{A}_a^0|_{w=0}$ depend only on $(\epsilon,t)$.}
		\begin{equation} \label{DA0}
		D_x\tilde{\mathfrak{A}}_a(\epsilon,t,x,0)=D_x\tilde{A}_a^0(\epsilon,t,x,0)=0
		\end{equation}
		for all $(\epsilon,t,x)\in (0,\epsilon_0)\times (2 T_0,0)\times \mathbb{R}^3$.
		
		\item   \label{A:B}  For $a=1,2$, the matrix $\mathfrak{A}_a$ commutes with $\mathbb{P}_a$, i.e.
		\begin{align}\label{E:COMMUTEPANFB}
		[\mathbb{P}_a, \mathfrak{A}_a(\epsilon,t,x,w)]=0
		\end{align}
		for all $(\epsilon,t,x,w)\in(0, \epsilon_0)\times (2T_0, 0) \times\mathbb{R}^3 \times B(\mathbb{R}^{N_1}) $.
		
		\item  $\Pbb_3$ is a symmetric  $(N_1+N_2)\times (N_1+N_2)$ projection matrix that satisfies
		\begin{gather}
		\Pbb\Pbb_3 =\Pbb_3\Pbb=\Pbb_3,  \label{E:P32a} \\
		\Pbb_3 A^i(\epsilon,t,x,w) \Pbb_3^\perp =
		\Pbb_3 C^i \Pbb_3^\perp= \Pbb_3 \mathfrak{A}(\epsilon,t,x,w) \Pbb_3^\perp = 0  \label{E:P32b}
		\intertext{and}
		[\Pbb_3,A^0(\epsilon,t,x,w)] = 0  \label{E:P32c}
		\end{gather}
		for all $(\epsilon,t,x,w)\in(0, \epsilon_0)\times (2T_0, 0) \times\mathbb{R}^3 \times B_R(\mathbb{R}^{N_1})$,
		where $\mathbb{P}_3^\perp = \mathds{1}-\mathbb{P}_3$ defines the complementary projection matrix.

		\item   \label{E:B0ANDCALB}  There exists constants $\kappa$, $\gamma_1$, $\gamma_2>0$, such that
		\begin{align}
		\frac{1}{\gamma_1}\mathds{1}\leq A_a^0(\epsilon,t,x,w)\leq\frac{1}{\kappa}\mathfrak{A}_a(\epsilon,t,x,w)\leq \gamma_2\mathds{1} \label{E:KAPPAB0CALB}
		\end{align}
		for all $(\epsilon,t,x,w)\in(0, \epsilon_0)\times (2T_0, 0) \times\mathbb{R}^3 \times B(\mathbb{R}^{N_1})$ and $a=1,2$.
		
		\item  \label{A:PBP}  For $a=1,2$, the matrix $A_a^0$ satisfies
		\begin{align}\label{E:PAP}
		\mathbb{P}_a ^{\perp}A_a^0(\epsilon,t,x, \mathbb{P}_1 ^{\perp}w)\mathbb{P}_a =\mathbb{P}_a A_a^0(\epsilon,t,x, \mathbb{P}_1 ^{\perp}w)\mathbb{P}_a ^{\perp}=0
		\end{align}
		for all $(\epsilon,t,x,w)\in(0, \epsilon_0)\times (2T_0, 0)\times\mathbb{R}^3 \times B(\mathbb{R}^{N_1}) $.
		
		\item   \label{A:PDECOMPOSABLE} For $a=1,2$, the matrix $\Pbb_a^\perp[D_w A_a^0\cdot(A_1^0)^{-1}\mathfrak{A}_1\mathbb{P}_1w]
		\Pbb_a^\perp$ can be decomposed as
		\al{ADEC}{
			\Pbb_a^\perp \bigl[D_w A_a^0(\epsilon,t,x,w)\cdot \bigl(A^0_1(\epsilon,t,x,w)\bigr)^{-1}\mathfrak{A}_1(\epsilon,t,x,w)\mathbb{P}_1w\bigr] \Pbb_a^\perp =
			t\mathscr{S}_a(\epsilon, t, x, w)+\mathscr{T}_a(\epsilon, t, x, w, \Pbb_1 w)
		}
		for some
		$\mathscr{S}_a \in E^0\bigl((0,\epsilon_0)\times (2T_0,0)\times \Rbb^3\times B_R(\Rbb^{N_1}),
		\mathbb{M}_{N_a\times N_a}\bigr)$, $a=1,2$, and
		$\mathscr{T}_a \in E^0\bigl((0,\epsilon_0)\times (2T_0,0)\times \Rbb^3\times B_R(\Rbb^{N_1})\times \Rbb^{N_1},
		\mathbb{M}_{N_a\times N_a}\bigr)$, $a=1,2$, where the  $\mathscr{T}_a(\epsilon, t, x, w, \xi)$ vanish to second order in $\xi$.
	\end{enumerate}
\end{ass}

Before proceeding with the analysis, we take a moment to make a few
observations about the structure of the singular system \eqref{E:MODELEQ2a}. First, if $\mathfrak{A}=0$, then the singular term
$\frac{1}{t}\mathfrak{A}\mathbb{P}U$  disappears from \eqref{E:MODELEQ2a} and it becomes a regular symmetric hyperbolic system.
Uniform $\epsilon$-independent a priori estimates that are valid for $t\in [T_0,0)$ would then follow,  under a suitable small initial data assumption, as a direct consequence
of the energy estimates from \cite{Browning1982,Klainerman1981,Klainerman1982,Kreiss1980}. When $\mathfrak{A}\neq 0$, the  positivity assumption
\eqref{E:KAPPAB0CALB} guarantees that the singular term $\frac{1}{t}\mathfrak{A}\mathbb{P}U$ acts like
a friction term. This allows us to generalize the energy estimates from \cite{Browning1982,Klainerman1981,Klainerman1982,Kreiss1980} in such a way as to obtain, under a suitable small initial data assumption, uniform $\epsilon$-independent a priori estimates that are valid on the time interval $[T_0,0)$; see
\eqref{E:ENERGEST1}, \eqref{E:ENERGEST2}, \eqref{E:ENERGEST4}  and \eqref{E:ENERGEST3}
 for the key differential inequalities used to derive these a priori estimates.

\begin{remark} \label{decouple}
	The equation for $w$ decouples from the system
	\eqref{E:MODELEQ2a} and is given by
	\begin{align}
	A_1^0  \partial_t w+A_1^i  \partial_i w+\frac{1}{\epsilon}C_1^i\partial_i w =\frac{1}{t}\mathfrak{A}_1 \mathbb{P}_1  w +H_1+F_1 \quad &\text{in} \quad[T_0, T_1)\times\mathbb{R}^3.  \label{E:MODELEQ1a}
	\end{align}
\end{remark}

\begin{remark} $\;$
	
	\begin{enumerate}
		\item
		By Taylor expanding $A^0_a(\epsilon,t,x,\Pbb^\perp_1 w+ \Pbb_1 w)$ in the variable $\Pbb_1 w$, it follows
		from  \eqref{E:PAP} that there exist matrix valued maps $\hat{A}^0_a, \breve{A}^0_a \in E^1\bigl((0,\epsilon_0)\times (2 T_0,0)\times \Rbb^3 \times B_R\bigl(\Rbb^{N_1}\bigr),\mathbb{M}_{N_a \times N_a}\bigr)$, $a=1,2$, such that
		\begin{align}
		\mathbb{P} ^\perp_a  A^0_a(\epsilon,t,x, w) \mathbb{P}_a =\mathbb{P}_a ^\perp [\hat{A}^0_a(\epsilon,t,x, w)\cdot\mathbb{P}_1 w]\mathbb{P}_a \label{E:PPERPB0P}
		\intertext{and}
		\mathbb{P}_a A^0_a(\epsilon,t,x, w) \mathbb{P}_a^\perp=\mathbb{P}_a  [\breve{A}^0_a(\epsilon,t,x, w)\cdot\mathbb{P}_1 w]\mathbb{P}_a ^\perp\label{E:PB0PPERP}
		\end{align}
		for all $(\epsilon,t,x,w)\in(0, \epsilon_0)\times (2T_0, 0)\times\mathbb{R}^3 \times B(\mathbb{R}^{N_1}) $.
		\item
		It is not difficult to see that the assumptions \eqref{E:KAPPAB0CALB} and  \eqref{E:PAP} imply that
		\begin{align*}
		\mathbb{P}_a ^{\perp}\bigl(A_a^0(\epsilon,t,x, \mathbb{P}_1 ^{\perp}w)\bigr)^{-1}\mathbb{P}_a =
		\mathbb{P}_a \bigl(A_a^0(\epsilon,t,x, \mathbb{P}_1 ^{\perp}w)\bigr)^{-1}\mathbb{P}_a ^{\perp}=0
		\end{align*}
		for all $(\epsilon,t,x,w)\in(0, \epsilon_0)\times (2T_0, 0)\times\mathbb{R}^3 \times B(\mathbb{R}^{N_1}) $.
		By Taylor expanding $(A^0_a(\epsilon,t,x,\Pbb^\perp_1 w+ \Pbb_1 w))^{-1}$ in the variable $\Pbb_1 w$, it follows that
		there exist matrix valued maps $\hat{B}^0_a, \breve{B}^0_a \in E^1\bigl((0,\epsilon_0)\times (2 T_0,0)\times \Rbb^3 \times B_R\bigl(\Rbb^{N_1}\bigr),\mathbb{M}_{N_a \times N_a}\bigr)$, $a=1,2$, such that
		\begin{align}
		\mathbb{P} ^\perp_a\bigl(A^0_a(\epsilon,t,x, w)\bigr)^{-1}\mathbb{P}_a =\mathbb{P}_a ^\perp [\hat{B}^0_a(\epsilon,t,x, w)\cdot\mathbb{P}_1 w]\mathbb{P}_a \label{E:PPERPB0Pa}
		\intertext{and}
		\mathbb{P}_a\bigl(A^0_a(\epsilon,t,x, w)\bigr)^{-1}\mathbb{P}_a^\perp=\mathbb{P}_a  [\breve{B}^0_a(\epsilon,t,x, w)\cdot\mathbb{P}_1 w]\mathbb{P}_a ^\perp\label{E:PB0PPERPa}
		\end{align}
		for all $(\epsilon,t,x,w)\in(0, \epsilon_0)\times (2T_0, 0)\times\mathbb{R}^3 \times B(\mathbb{R}^{N_1}) $.
	\end{enumerate}
\end{remark}
To facilitate the statement and proof of our a priori estimates for solutions of the system \eqref{E:MODELEQ2a},
we introduce the following energy norms:
\begin{definition}\label{energynorms}
	Suppose  $w \in L^\infty([T_0,T_1)\times \mathbb{R}^3,\Rbb^{N_1})$, $k \in \mathbb{Z}_{\geq 0}$, and
	$\{\mathbb{P}_a,A_a^0\}$, $a=1,2$, are as defined above. Then for maps $f_a$, $a=1,2$, and $U$ from $\mathbb{R}^3$ into
	$R^{N_a}$ and $R^{N_1}\times R^{N_2}$, respectively, the
	\emph{energy norms}
of $f_a$ and $U$ are defined by
	\begin{align*}
	\vertiii{f_a}^2_{a,H^k}=&\sum_{0\leq |\alpha|\leq k}\langle D^\alpha f_a, A_a^0\bigl(\epsilon,t,\cdot,w(t,\cdot)\bigr) D^\alpha f_a\rangle, \\
	\vertiii{f_a}^2_{a,R^k}=&\vertiii{Df_a}^2_{a,H^{k-1}}+\|f_a\|^2_{L^6},\\
	\vertiii{  U}_{H^k}^2=&\sum_{0\leq |\alpha|\leq k}\langle D^\alpha U, A^0\bigl(\epsilon,t,\cdot,w(t,\cdot)\bigr) D^\alpha U\rangle,
\intertext{and}
	\vertiii{U}^2_{R^k}= & \vertiii{DU}^2_{H^{k-1}}
	+ \|U\|^2_{L^6}.
	\end{align*}
	In addition to the energy norms, we also define, for $T_0 < T \leq T_1$, the spacetime norm
	of maps $f_a$, $a=1,2$, from $[T_0,T)\times \mathbb{R}^3$ to $R^{N_a}$  defined by
	\begin{equation*}
	\|f_a\|_{M_{\mathbb{P}_a ,k}^\infty([T_0, T)\times\mathbb{R}^3)}=\|f_a\|_{L^{\infty}([T_0, T),R^k(\Rbb^3))}+ \left(-\int_{T_0}^T\frac{1}{t}\|\mathbb{P}_a f_a (t)\|^2_{R^k(\Rbb^3)}dt\right)^\frac{1}{2}.
	\end{equation*}
\end{definition}
\begin {remark}
For  $w \in L^\infty([T_0,T_1)\times \Rbb^3,\Rbb^{N_1})$ satisfying $\norm{w}_{L^\infty([T_0,T_1)\times \mathbb{R}^3)} < R$,
we observe, by \eqref{E:KAPPAB0CALB}, that the standard Sobolev norm $\norm{\cdot}_{H^k}$ and
the energy norms  $\vertiii{\,\cdot\,}_{a,H^k}$, $a=1,2$, are equivalent since they satisfy
\begin{align}\label{e:eqnorm0}
\frac{1}{\sqrt{\gamma_1}}\|\cdot\|_{H^k}\leq \vertiii{\,\cdot\,}_{a,{H^k}} \leq \sqrt{\gamma_2}\|\cdot\|_{H^k}.
\end{align}
Furthermore, if $k\geq 2$, we have that
\begin{align}
\|f_a\|_{R^k}\lesssim \vertiii{f_a}_{a,R^k}\lesssim \|f_a\|_{R^k}
\AND
\|U\|_{R^k}\lesssim \vertiii{U}_{R^k} \lesssim \|U\|_{R^k}.
\label{e:eqnorm}
\end{align}
These norm equivalences will be used without comment throughout this section.
\end{remark}

With the preliminaries out of the way, we are now ready to state and prove a priori estimates for solutions of the system
\eqref{E:MODELEQ2a} that are uniform in $\epsilon$.
\begin{theorem}\label{L:BASICMODEL}
Suppose $R>0$, $s\in \mathbb{Z}_{\geq 3}$, $T_0 < T_1 < 0$, $\epsilon_0 > 0$, $\epsilon\in(0, \epsilon_0)$, Assumption \ref{ASS1} holds, the map
\als
{
	U=(w,u) \in \bigcap_{\ell=0}^1 C^\ell([T_0,T_1), R^{s-\ell}(\mathbb{R}^3,\Rbb^{N_1})) \times
	\bigcap_{\ell=0}^1 C^\ell([T_0,T_1), R^{s-1-\ell}(\mathbb{R}^3,\Rbb^{N_2})),
}
defines a solution of the system \eqref{E:MODELEQ2a}, $\Pbb_4 w\in \bigcap_{\ell=0}^1 C^\ell([T_0,T_1), H^{s-\ell}(\mathbb{R}^3, \Rbb^{N_1}))$, and for $t\in [T_0,T_1)$, the source terms $F_a$, $a=1,2$, satisfy
the estimates
\al{F_I}{
	\|\Pbb_4 F_1(\epsilon,t)\|_{L^2}+\|F_1(\epsilon,t)\|_{R^s}
	\leq & C(\|w\|_{\Li([T_0,t),R^s)},\|\Pbb_4 w\|_{\Li([T_0,t),H^s)})(
	\mathcal{C}_* (t) +\mathcal{C}^{*}   +\|w(t)\|_{R^s}+\|\Pbb_4w(t)\|_{H^s})
}
and
\al{F_I2}{
	\| F_2(\epsilon,t)\|_{R^{s-1}}	\leq & C\bigl(\|w\|_{\Li([T_0,t),R^s)},\|\Pbb_4 w\|_{\Li([T_0,t),H^s)},\|u\|_{\Li([T_0,t),\Rs)}\bigr)(\mathcal{C}_{*}(t) +\mathcal{C}^{*}   +\|w(t)\|_{R^s}\nnb   \\
	&+\|\Pbb_4 w(t)\|_{H^s}+\|u(t)\|_{\Qs}),
}
where $\mathcal{C}_{*}(t) =\int^t_{T_0}\bigl(\|w(\tau)\|_{ R^s }+\|\Pbb_4 w(\tau)\|_{ H^s } +\|u(\tau)\|_{R^{s-1}}\bigr) d\tau$
and the constants $\mathcal{C}^*$, $C\bigl(\|w\|_{\Li([T_0,t),R^s)},\|\Pbb_4 w\|_{\Li([T_0,t),H^s)}\bigr)$     and
$C\bigl(\|w\|_{\Li([T_0,t),R^s)},\|\Pbb_4 w\|_{\Li([T_0,t),H^s)},\|u\|_{\Li([T_0,t),\Rs)}\bigr)$ are independent of $\epsilon \in (0,\epsilon_0)$
and $T_1 \in (T_0,0]$.
Then there exists a $\sigma>0$ independent of $\epsilon \in (0,\epsilon_0)$ and $T_1 \in (T_0,0)$, such that if initially
\begin{align*}
\|w(T_0)\|_{R^s} + \|\Pbb_4  w(T_0)\|_{H^s}+\|u(T_0)\|_{R^{s-1}} +\mathcal{C}^{*} \leq \sigma,
\end{align*}
then
\begin{equation*} 
\norm{w}_{L^\infty([T_0,T_1)\times \Rbb^3)} \leq \frac{R}{2}
\end{equation*}
and there exists a constant $C>0$, independent of $\epsilon\in (0,\epsilon_0)$ and $T_1 \in (T_0,0)$, such that
\als
{
	\|\Pbb_4w\|_{L^\infty([T_0,t),L^2)} +&
	\left(-\int_{T_0}^{t} \frac{1}{\tau} \|\Pbb_1\Pbb_4 w(\tau)\|^2_{L^2}\, d\tau\right)^{\frac{1}{2}}+\|w\|_{M^\infty_{\Pbb_1, s}([T_0,t)\times \Rbb^3)}\nnb  \\
	&\hspace{1cm}+\|u\|_{M^\infty_{\Pbb_2, s-1}([T_0,t)\times \Rbb^3)}  -
	\int_{T_0}^{t} \frac{1}{\tau} \|\Pbb_3 U(\tau)\|_{\Qs}\, d\tau     \leq C\sigma
}
for $T_0 \leq t < T_1$.
\end{theorem}
\begin{proof}
According to the definition of $R^s$ and \eqref{E:NORMEQ1}, there exists a constant $C_1$, such that
\begin{equation*}
\norm{w(T_0)}_{L^\infty} \leq C_1\bigl(\norm{w(T_0)}_{R^s}+\|\Pbb_4w(T_0)\|_{L^2}\bigr) \leq C_1\sigma.
\end{equation*}
We then choose $\sigma$ to satisfy
\begin{equation} \label{sigmaC1}
\sigma \leq \min\biggl\{1, \frac{\hat{R}}{4}\biggr\},
\end{equation}
where $\hat{R} = \frac{R}{2 C_1}$,
so that
\begin{equation*} \label{winit1}
\norm{w(T_0)}_{L^\infty} \leq \frac{R}{8}.
\end{equation*}
Next, we define
\begin{align*}
K_1(t)=\|w\|_{L^\infty([T_0, t), R^s)},  \quad  K_2(t)=\|u\|_{L^\infty([T_0, t), R^{s-1})} \AND K_3(t)= \|\Pbb_4 w\|_{\Li([T_0,t),H^s)},
\end{align*}
and we observe that $K_1(T_0)+K_2(T_0)+K_3(T_0) \leq \hat{R}/2$, and hence, by continuity,
either $K_1(t)+K_2(t)+K_3(t) < \hat{R}$ for all $t\in [T_0,T_1)$, or
else there exists a first time $T_* \in (T_0,T_1)$ such that $K_1(T_*)+K_2(T_*)+K_3(T_*) = \hat{R}$. Letting $T_* = T_1$ if the first case holds, we then have
that
\begin{equation} \label{K1ineq}
K_1(t)+K_2(t)+K_3(t) < \hat{R}, \quad 0\leq t < T_*,
\end{equation}
where $T_* = T_1$ or else $T_*$ is the first time in $(T_0,T_1)$ for which $K_1(T_*) +K_2(T_*)+K_3(T_*)= \hat{R}$.

Before proceeding the proof, we first establish a number of estimates that will be needed in the proof; we collect them together in the following lemma.
\begin{lemma}\label{L:PREEST}
	There exists constants $C(K_1(t))$ and  $C(K_1(t),K_2(t))$, both independent of $\epsilon\in (0,\epsilon_0)$ and $T_*\in (T_0,T_1]$, such that the following estimates hold
	for $T_0 \leq t < T_*<0$:
	\gat{
		-\frac{2}{t}\sum_{1\leq|\alpha|\leq s }\langle D^\alpha w, A_1^0[(A_1^0)^{-1}\mathfrak{A}_1, D^\alpha]\mathbb{P}_1 w\rangle
		\leq    -\frac{1}{t}C(K_1 ) \|w\|_{R^s} \vertiii{\mathbb{P}_1 w}^2_{1, R^s},  \label{E:INEQ1a}\\
		-\frac{2}{t}\sum_{1\leq|\alpha|\leq s-1 }\langle D^\alpha u, A_2^0[(A_2^0)^{-1}\mathfrak{A}_2, D^\alpha]\mathbb{P}_2 u\rangle
		\leq   -\frac{1}{t}C(K_1 ) (\|u\|_{\Qs}+\|w\|_{R^s})(\vertiii{\mathbb{P}_2u}^2_{2,\Rs}+\vertiii{\mathbb{P}_2w}^2_{1,R^s}),  \label{E:INEQ1b}
	}
	\ali{
		-\sum_{1\leq |\alpha|\leq s}\langle D^\alpha w, A_1^0[ D^\alpha,(A_1 ^0)^{-1}A_1^i]\partial_i w\rangle
		\leq & C(K_1 ) \|w\|_{R^s}^2,  \label{E:INEQ2a}\\
		-\sum_{1\leq |\alpha|\leq s-1}\langle D^\alpha u, A_2 ^0[ D^\alpha,(A_2 ^0)^{-1}A_2^i]\partial_i u\rangle
		\leq  & C(K_1 ) \|u\|_{\Qs}^2,  \label{E:INEQ2b}
	}
	\gat{
		-\sum_{1\leq |\alpha|\leq s}\langle D^\alpha w, [\tilde{A}_1^0, D^\alpha](A_1^0)^{-1}C_1^i\partial_i w\rangle \leq    C(K_1)\|w\|^2_{R^s},  \label{E:INEQ3a}\\
		-\sum_{1\leq |\alpha|\leq s-1}\langle D^\alpha u, [\tilde{A}_2^0, D^\alpha](A_2^0)^{-1}C_2^i\partial_i u\rangle \leq  C(K_1 )\|u\|^2_{\Qs},  \label{E:INEQ3b}\\
		\sum_{1\leq |\alpha|\leq s}\langle D^\alpha w, (\partial_t A_1^0) D^\alpha w\rangle
		\leq    C(K_1) \|w\|^2_{R^s}-\frac{1}{t}C(K_1 )  \|w\|_{R^s} \vertiii{\mathbb{P}_1w}^2_{1,R^s}, \label{E:INEQ5a}\\
		\sum_{1\leq |\alpha|\leq s-1}\langle D^\alpha u, (\partial_tA_2^0) D^\alpha u\rangle
		\leq C(K_1 )  \|u\|^2_{\Qs}-\frac{1}{t}C(K_1, K_2) (\|u\|_{\Qs}+\|w\|_{R^s})
		(\vertiii{\mathbb{P}_2u}^2_{2,\Rs}+\vertiii{\mathbb{P}_1w}^2_{1, R^s})\label{E:INEQ5b}
	}
	and
	\al{INEQ6}{
		\sum_{1\leq|\alpha| \leq s-1} \la D^\alpha \Pbb_3 U, (\del{t} A^0) D^\alpha \Pbb_3 U \ra
		\leq -\frac{1}{t} C(K_1)\|\Pbb_1 w\|_{R^s}\vertiii{\Pbb_3 U}^2_{\Rs}+C(K_1) \|\Pbb_3 U\|^2_{\Qs}.
	}
\end{lemma}
\begin{proof}
	Using the properties $\mathbb{P}_1^2=\mathbb{P}_1$, $\mathbb{P}_1 + \mathbb{P}_1^\perp = \mathds{1}$,
	$\mathbb{P}_1^\textrm{T} = \mathbb{P}_1$, and $D\mathbb{P}_1 = 0$ of the projection matrix $\mathbb{P}_1$ repeatedly, we compute
	\begin{align}
	& -\frac{2}{t}\sum_{1\leq|\alpha|\leq s}\langle D^\alpha w, A_1^0[(A_1 ^0)^{-1}\mathfrak{A}_1, D^\alpha]\mathbb{P}_1 w\rangle \notag \\
	&\hspace{0.5cm} = -\frac{2}{t}\sum_{1\leq|\alpha|\leq s}\langle D^\alpha \mathbb{P}_1 w, A_1^0[(A_1^0)^{-1}\mathfrak{A}_1, D^\alpha]\mathbb{P}_1 w\rangle
	-\frac{2}{t}\sum_{1\leq|\alpha|\leq s}\langle D^\alpha \mathbb{P}_1^\perp w, \mathbb{P}_1^\perp A_1^0 [(A_1^0)^{-1}\mathfrak{A}_1, D^\alpha]\mathbb{P}_1 w \rangle  \notag \\
	&\hspace{0.5cm}= -\frac{2}{t}\sum_{1\leq|\alpha|\leq s}\langle D^\alpha \mathbb{P}_1 w, A_1^0[(A_1^0)^{-1}\mathfrak{A}_1, D^\alpha]\mathbb{P}_1 w\rangle
	-\frac{2}{t}\sum_{1\leq|\alpha|\leq s}\langle D^\alpha \mathbb{P}_1^\perp w, \mathbb{P}_1^\perp A_1^0 [(A_1^0)^{-1}\mathbb{P}_1\mathfrak{A}_1, D^\alpha]\mathbb{P}_1 w \rangle  && \text{(by \eqref{E:COMMUTEPANFB})} \notag \\
	&\hspace{0.5cm}= -\frac{2}{t}\sum_{1\leq|\alpha|\leq s}\langle D^\alpha \mathbb{P}_1 w, A_1^0[(A_1^0)^{-1}\mathfrak{A}_1, D^\alpha]\mathbb{P}_1 w\rangle
	-\frac{2}{t}\sum_{1\leq|\alpha|\leq s}\langle D^\alpha \mathbb{P}_1^\perp w, \mathbb{P}_1^\perp A_1^0\mathbb{P}_1^\perp [\mathbb{P}_1^\perp (A_1^0)^{-1}\mathbb{P}_1\mathfrak{A}_1, D^\alpha]\mathbb{P}_1 w \rangle   \nnb\\
	&\hspace{7.5cm}-\frac{2}{t}\sum_{1\leq |\alpha|\leq s}\langle D^\alpha \mathbb{P}_1^\perp w, \mathbb{P}_1^\perp A_1^0\mathbb{P}_1 [\mathbb{P}_1(A_1^0)^{-1}\mathbb{P}_1 \mathfrak{A}_1, D^\alpha] \mathbb{P}_1 w\rangle. \notag
	\end{align}
	From this expression,  we obtain, with the help of the Cauchy-Schwarz inequality, the calculus inequalities \eqref{E:COMMUTATOR6} and Proposition \ref{T:Moser2},
	the expansions \eqref{E:DECOMPOSITIONOFA01}-\eqref{E:DECOMPOSITIONOFCALB},
	the relations \eqref{DA0}, \eqref{E:PPERPB0P}, and \eqref{E:PPERPB0Pa}, the inequality \eqref{K1ineq} and the equivalence of norms \eqref{e:eqnorm}, the estimate
	\begin{align*}
	&-\frac{1}{t}\sum_{1\leq |\alpha|\leq s}\langle D^\alpha w, A_1^0[(A_1 ^0)^{-1}\mathfrak{A}_1, D^\alpha]\mathbb{P}_1 w\rangle \notag
	\\
	&\hspace{0.5cm}  \lesssim \notag -\frac{1}{t}\bigl[\|A_1^0\|_{L^\infty} \|\mathbb{P}_1 w\|_{R^{s}}\|D \bigl( (A_1^0)^{-1} \mathfrak{A}_1 \bigr)\|_{R^{s-1}\cap L^2}
	+\|A_1^0\|_{L^\infty}
	\|\mathbb{P}_1^\perp w \|_{R^s} \|D\bigl(\mathbb{P}_1^\perp(A_1^0)^{-1}\mathbb{P}_1\mathfrak{A}_1\bigr)\|_{R^{s-1}\cap L^2}
	\notag \\
	& \hspace{0.5cm}  + \|\mathbb{P}_1^\perp A_1^0\mathbb{P}_1\|_{L^\infty}
	\|\mathbb{P}_1^\perp w\|_{R^s}  \|D\bigl(\mathbb{P}_1(A_1^0)^{-1}\mathbb{P}_1\mathfrak{A}_1\bigr)\|_{R^{s-1}\cap L^2}
	\bigr] \|\mathbb{P}_1 w\|_{R^{s-1}}  \nnb  \\
	\leq &  -C(K_1 )\frac{1}{t}\|w\|_{R^s}
	\|\mathbb{P}_1 w\|_{R^s}^2 \leq    -\frac{1}{t}C(K_1 ) \|w\|_{R^s} \vertiii{\mathbb{P}_1 w}^2_{1, R^s} 
	\end{align*}
	for $T_0\leq t < T_*$, where the constant $C(K_1)$ is independent of $\epsilon \in (0,\epsilon_0)$ and $T_* \in (T_0,T_1]$. This establishes
	the first estimate \eqref{E:INEQ1a}. By a similar calculation, we find that
	\begin{align*}
	& -\frac{2}{t}\sum_{1\leq|\alpha|\leq s-1}\langle D^\alpha u, A_2^0[(A_2 ^0)^{-1}\mathfrak{A}_2, D^\alpha]\mathbb{P}_2 u\rangle
	=-\frac{2}{t}\sum_{1\leq|\alpha|\leq s-1}\langle D^\alpha \mathbb{P}_2u, A_2^0[(A_2^0)^{-1}\mathfrak{A}_2, D^\alpha]\mathbb{P}_2u\rangle
	\notag \\
	&\hspace{0.5cm} -\frac{2}{t}\sum_{1\leq|\alpha|\leq s-1}\langle D^\alpha \mathbb{P}_2^\perp u, \mathbb{P}_2^\perp A_2^0\mathbb{P}_2^\perp [\mathbb{P}_2^\perp (A_2^0)^{-1}\mathbb{P}_2  \mathfrak{A}_2, D^\alpha]\mathbb{P}_2 u\rangle
	-\frac{2}{t}\sum_{1\leq |\alpha|\leq s-1}\langle D^\alpha \mathbb{P}_2^\perp u, \mathbb{P}_2^\perp A_2^0\mathbb{P}_2 [\mathbb{P}_2(A_2^0)^{-1}\mathbb{P}_2\mathfrak{A}_2, D^\alpha]\mathbb{P}_2u\rangle   \nnb\\
	&\hspace{0.5cm}\leq  -\frac{1}{t}C(K_1) \|w\|_{R^s}\|\mathbb{P}_2u\|^2_{\Qs}
	-\frac{1}{t}C(K_1)\|u\|_{\Qs}\|\mathbb{P}_1 w\|_{R^s}\|\mathbb{P}_2u\|_{\Qs}
	-\frac{1}{t}C(K_1)\|u\|_{\Qs}\|\mathbb{P}_1 w\|_{R^s}\|\mathbb{P}_2u\|_{\Qs}   \nnb\\
	&\hspace{6.0cm}
	\leq   -\frac{1}{t}C(K_1 ) (\|u\|_{\Qs}+\|w\|_{R^s})(\vertiii{\mathbb{P}_2u}^2_{2,\Rs}+\vertiii{\mathbb{P}_2w}^2_{1,R^s}),  
	\end{align*}
	which establishes the second estimate \eqref{E:INEQ1b}.

	Next, using the calculus inequalities \eqref{E:COMMUTATOR6} and Proposition \ref{T:Moser2}, we observe that
	\begin{align*}
	\sum_{1 \leq |\alpha|\leq s-1}\langle D^\alpha u, -A_2^0[ D^\alpha,(A_2^0)^{-1} A_2^i]\partial_i u\rangle
	\lesssim   \|A_2^0\|_{L^\infty}\|u\|_{\Qs}^2\| D((A_2^0)^{-1}A_2^i)\|_{\Qs\cap L^2}
	\leq    C(K_1 )\|u\|_{\Qs}^2,
	\end{align*}
	which establishes the fourth estimate \eqref{E:INEQ2b}. Since the estimates \eqref{E:INEQ2a}, \eqref{E:INEQ3a} and \eqref{E:INEQ3b}
	can be obtained in a similar fashion, we omit the details.
	
	Finally, we consider the estimates \eqref{E:INEQ5a}-\eqref{E:INEQ5b}. We begin establishing these estimates by
	writing \eqref{E:MODELEQ1a} as
	\begin{align*}
	\epsilon \partial_t w=\epsilon\frac{1}{t}(A_1^0)^{-1} \mathfrak{A}_1\mathbb{P}_1 w -\epsilon(A_1^0)^{-1} A_1^i \partial_i w- (A_1^0)^{-1} C_1^i\partial_i w+\epsilon(A_1^0)^{-1} H_1 +\epsilon(A_1^0)^{-1} F_1.
	\end{align*}
	Using this  and the expansion \eqref{E:DECOMPOSITIONOFA01}, we can express the
	time derivatives $\del{t}A^0_a$, $a=1,2$, as
	\begin{align}
	\partial_t A_a^0=&  D_w A_a^0 \cdot \partial_t w +D_t A_a^0    \nnb\\
	=&-D_w A_a^0\cdot(A_1^0)^{-1} A_1^i \partial_i w -[ D_w \tilde{A}_a^0 \cdot (A_1^0)^{-1} C_1^i\partial_i w] \nnb\\
	&+[ D_w A_a^0 \cdot (A_1^0)^{-1} H_1]+D_t A_a^0   +[ D_w A_a^0\cdot(A_1^0)^{-1} F_1] + \frac{1}{t}[D_w A_a^0\cdot (A_1^0)^{-1} \mathfrak{A}_1\mathbb{P}_1 w]. \label{E:PTB0}
	\end{align}
	Using \eqref{E:PTB0} with  $a=2$, we see, with the help of the calculus inequalities from Appendix  \ref{A:INEQUALITIES},
	the Cauchy-Schwarz inequality,
	the estimate \eqref{E:F_I}, and  the expansion \eqref{E:ADEC} for $a=2$, that
	\begin{align*}
	\sum_{1\leq|\alpha|\leq s-1}\langle D^\alpha u, (\partial_tA_2^0) D^\alpha u\rangle
	\leq &\sum_{1\leq|\alpha|\leq s-1}\left[\langle D^\alpha u, \mathbb{P}_2^\perp(\partial_tA_2^0)\mathbb{P}_2^\perp D^\alpha u\rangle+\langle D^\alpha u, \mathbb{P}_2^\perp(\partial_tA_2^0)\mathbb{P}_2 D^\alpha u\rangle \right.\nnb\\
	&\hspace{2.0cm} \left.+\langle D^\alpha u, \mathbb{P}_2(\partial_tA_2^0)\mathbb{P}_2^\perp D^\alpha u\rangle
	+\langle D^\alpha u, \mathbb{P}_2(\partial_tA_2^0)\mathbb{P}_2 D^\alpha u\rangle\right]\nnb\\
	\leq & C(K_1) \|u\|^2_{\Qs}    -\frac{2}{t}\|u\|_{\Qs}\|(A_1^0)^{-1}\mathfrak{A}_1\|_{L^\infty}\|D_wA_2^0\|_{L^\infty}
	\|\mathbb{P}_2u\|_{\Qs}\|\mathbb{P}_1w\|_{\Qs}\nnb\\
	&-\frac{1}{t}\|\mathbb{P}_1w\|_{R^s}\|(A^0)^{-1}\mathfrak{A}\|_{L^\infty}\|D_w A_2^0\|_{L^\infty}
	\|\mathbb{P}_2u
	\|_{\Hs}^2 - \frac{1}{t}\|u\|_{\Qs}^2C(K_1)\|\mathbb{P}_1w\|^2_{\Qs}\nnb\\
	\leq &  C(K_1 )  \|u\|^2_{\Qs}-\frac{1}{t}C(K_1, K_2) (\|u\|_{\Qs}+\|w\|_{R^s})
	(\|\mathbb{P}_2u\|^2_{\Qs}+\|\mathbb{P}_1w\|^2_{R^s}). \notag
	\end{align*}
	With the help of \eqref{e:eqnorm}, this establishes the estimate \eqref{E:INEQ5b}. Since the estimate \eqref{E:INEQ5a} can be established using similar
	arguments, we omit the details. The last estimate \eqref{E:INEQ6} can also be established using similar arguments with
	the help of the identity $\Pbb_3\Pbb=\Pbb\Pbb_3=\Pbb_3$. We again omit the details.
\end{proof}

Applying $A^0 D^\alpha (A^0)^{-1}$ to both sides of \eqref{E:MODELEQ2a}, we find that
\begin{align}
A^0\partial_t D^\alpha U+A^i\partial_i D^\alpha U+  \frac{1}{\epsilon}C^i\partial_i D^\alpha U
=&-A ^0[ D^\alpha,(A ^0)^{-1}A^i]\partial_i U-[\tilde{A}^0, D^\alpha](A^0)^{-1}C^i\partial_i U \notag \\
&+\frac{1}{t}\mathfrak{A} D^\alpha \mathbb{P} U +\frac{1}{t}A ^0[ D^\alpha, (A ^0)^{-1}\mathfrak{A}]\mathbb{P} U +A ^0 D^\alpha[(A ^0)^{-1} H],
\label{E:ABSOLUTEPTB1}
\end{align}
where in deriving this we have used
\begin{align*}
\frac{1}{\epsilon}[A^0, D^\alpha](A ^0)^{-1}C^i\partial_i U &\overset{\eqref{E:DECOMPOSITIONOFA01}}{=}\frac{1}{\epsilon}[\mathring{A}^0+\epsilon\tilde{A}^0, D^\alpha](A ^0)^{-1}C^i\partial_i U=[\tilde{A}^0, D^\alpha](A ^0)^{-1}C^i\partial_i U
\intertext{and}
A^0[D^\alpha,(A^0)^{-1}]C^i\del{i}U &= A^0 D^\alpha \bigl( (A^0)^{-1}C^i\del{i}U\bigr)- D^\alpha \bigl(C^i\del{i}U\bigr) \\
&=  A^0 D^\alpha \bigl( (A^0)^{-1}C^i\del{i}U\bigr) - D^\alpha( A^{0} (A^0)^{-1} C^i \del{i} U \bigr) = [A^0, D^\alpha] (A^0)^{-1} C^i \del{i} U .
\end{align*}
Writing $A^0_a$, $a=1,2$, as $A ^0_a=(A ^0_a)^{\frac{1}{2}}(A ^0_a)^{\frac{1}{2}}$, which we can do since
$A ^0_a$ is a real symmetric and positive-definite, we see from \eqref{E:KAPPAB0CALB} that
\begin{align} \label{Afrbound}
(A ^0_a)^{-\frac{1}{2}}\mathfrak{A_a}(A_a^0)^{-\frac{1}{2}}\geq \kappa\mathds{1}.
\end{align}
Since, by \eqref{E:COMMUTEPANFB},
\begin{align*}
\frac{2}{t} \langle D^\alpha f,\mathfrak{A}_a D^\alpha \mathbb{P}_a f\rangle=&\frac{2}{t} \langle D^\alpha \mathbb{P}_a f, (A^0)^{\frac{1}{2}}[(A^0_a)^{-\frac{1}{2}}\mathfrak{A}_a(A^0_a)^{-\frac{1}{2}}](A^0_a)^{\frac{1}{2}} D^\alpha \mathbb{P}_a f\rangle,
\quad a=1,2,
\end{align*}
it follows immediately from \eqref{Afrbound} that
\begin{align} \label{E:KAPPACONTR}
\frac{2}{t}\sum_{1\leq |\alpha|\leq s-1}\langle D^\alpha u,\mathfrak{A}_2 D^\alpha \mathbb{P}_2  u\rangle \leq  \frac{2\kappa}{t}\vertiii{D\mathbb{P}_2 u}^2_{2,H^{s-2}}\AND
\frac{2}{t}\sum_{1 \leq |\alpha|\leq s }\langle D^\alpha w,\mathfrak{A}_1 D^\alpha \mathbb{P}_1  w\rangle \leq \frac{2\kappa}{t}\vertiii{D \mathbb{P}_1 w}^2_{1,H^{s-1}}.
\end{align}
Letting $f_1$ denote one of $\Pbb_1w$, $w$ or $\Pbb_4w$, and $f_2$ denote one of $\Pbb_2u$ or $u$ ,
we have, by Theorem \ref{Sobolev}.\eqref{sob1} and \eqref{e:eqnorm0}, that
\begin{align}\label{e:normcp}
\|f_a\|_{L^6}\leq C_{\text{S}}\|Df_a\|_{L^2}\leq C_{\text{S}}\|Df_a\|_{H^{s-a}}\leq C_{\text{S}}\sqrt{\gamma_1}\vertiii{Df_a}_{a, H^{s-a}},
\end{align}
for $a=1,2,$
which yields that
\begin{align}
\vertiii{f_a}^2_{a,R^{s-a+1}}\leq & (C_{\text{S}}^2\gamma_1+1) \vertiii{Df_a}_{a,H^{s-a}}^2  \label{e:normcp0}\\
\intertext{and}
\frac{\kappa}{t}\vertiii{D f_a}^2_{a, H^{s-a}} \leq &  \frac{\kappa}{t}\frac{1}{C_{\text{S}}^2\gamma_1}\|f_a\|^2_{L^6}. \label{e:normcp1}
\end{align}
Adding $\frac{\kappa}{t}\vertiii{Df_a}^2_{a,H^{s-a}}$ on both sides of above \eqref{e:normcp1}, recall $t<0$, yields
\begin{align}\label{e:normcp2}
\frac{2\kappa}{t}\vertiii{D f_a}^2_{a, H^{s-a}} \leq \frac{\kappa}{t}\frac{1}{C_{\text{S}}^2\gamma_1}\|f_a\|^2_{L^6}+\frac{\kappa}{t}\vertiii{Df_a}^2_{a,H^{s-a}}\leq \frac{2\hat{\kappa}}{t}
\vertiii{f_a}^2_{a,R^{s-a+1}}
\end{align}
where we have set
\begin{align*}
\hat{\kappa}=\frac{1}{2}\kappa\min{\biggl(\frac{1}{C_{\text{S}}^2\gamma_1}, 1\biggr)}.
\end{align*}
Using \eqref{e:normcp2}, it is clear that the inequalities \eqref{E:KAPPACONTR} imply that
\begin{equation*}\label{E:KAPPACONTR1}
\frac{2}{t}\sum_{1\leq |\alpha|\leq s-1}\langle D^\alpha u,\mathfrak{A}_2 D^\alpha \mathbb{P}_2  u\rangle \leq  \frac{2\hat{\kappa}}{t}\vertiii{ \mathbb{P}_2 u}^2_{2,R^{s-1}}\AND
\frac{2}{t}\sum_{1 \leq |\alpha|\leq s }\langle D^\alpha w,\mathfrak{A}_1 D^\alpha \mathbb{P}_1  w\rangle \leq \frac{2\hat{\kappa}}{t}\vertiii{  \mathbb{P}_1 w}^2_{1,R^s}.
\end{equation*}
Then differentiating $\langle D^\alpha w, A^0_1 D^\alpha w\rangle$ with respect to $t$, we see, from
the identities  $\langle D^\alpha w, C^i_1\partial_i D^\alpha w\rangle = 0$
and $2 \langle D^\alpha w, A^i_1\partial_i D^\alpha w\rangle =  - \langle D^\alpha w, (\partial_i A^i_1) D^\alpha w\rangle$, the block decomposition of
\eqref{E:ABSOLUTEPTB1}, which
we can use to determine $ D^\alpha \del{t}w$, the estimates \eqref{E:F_I} and \eqref{E:KAPPACONTR} together with
those
from Lemma \ref{L:PREEST}, the relation \eqref{e:normcp0}, Young's inequality (Proposition \ref{T:young}) and the calculus inequalities from Appendix  \ref{A:INEQUALITIES}, that
\begin{align}\label{E:DTW1}
\partial_t\vertiii{Dw}^2_{1,H^{s-1}}=&\sum_{1 \leq |\alpha|\leq s}\langle D^\alpha w, (\partial_tA_1^0) D^\alpha w\rangle+2\sum_{1 \leq |\alpha|\leq s}\langle D^\alpha w, A_1^0 D^\alpha \partial_t w\rangle   \nnb \\
\leq &  C(K_1 ) \|w\|^2_{R^s}-\frac{1}{t}C(K_1) \|w\|_{R^s}
\vertiii{\mathbb{P}_1w}^2_{1,R^s} +\sum_{1 \leq |\alpha|\leq s}\langle D^\alpha w, (\del{i}A_1^i) D^\alpha w\rangle\nnb\\
& -\frac{2}{\epsilon}\sum_{1 \leq |\alpha|\leq s} \overset{\quad =0}{\overbrace{\langle D^\alpha w, C_1^i\partial_i D^\alpha w\rangle}}
-2\sum_{1 \leq |\alpha|\leq s}\langle D^\alpha w, A_1^0[ D^\alpha,(A_1^0)^{-1}A_1^i]\partial_i w\rangle\nnb\\
&-2\sum_{1 \leq |\alpha|\leq s}\langle D^\alpha w, [\tilde{A}_1^0, D^\alpha](A_1^0)^{-1}C_1^i\partial_i w\rangle
+\frac{2}{t}\sum_{1 \leq |\alpha|\leq s}\langle D^\alpha w,\mathfrak{A}_1 D^\alpha \mathbb{P}_1 w\rangle  \nnb\\
&+\frac{2}{t}\sum_{1 \leq |\alpha|\leq s}\langle D^\alpha w, A_1^0[(A_1^0)^{-1}\mathfrak{A}_1, D^\alpha] \mathbb{P}_1 w\rangle+2\sum_{1\leq |\alpha|\leq s}\langle D^\alpha w, A_1^0 D^\alpha[(A_1^0)^{-1} (H_1+ F_1)]\rangle \nnb \\
\leq & C(K_1)(\vertiii{w}^2_{1,R^s}+\vertiii{\Pbb_4w}_{L^2}^2+\mathcal{C}^*\vertiii{w}_{1,R^s}+\mathcal{C}_* \vertiii{Dw}_{1,\Hs})
+\frac{2}{t}\bigl[ \hat{\kappa} -C_1(K_1) \|w\|_{R^s} \bigr] \vertiii{\mathbb{P}_1 w}^2_{1,R^s}
\nnb \\
\leq & C(K_1)(\vertiii{Dw}^2_{1,H^{s-1}}+\vertiii{\Pbb_4w}_{L^2}^2+\sigma\vertiii{Dw}_{1,\Hs}+\mathcal{C}_*^2 )
+\frac{2}{t}\bigl[\hat{\kappa} -C_1(K_1) \|w\|_{R^s} \bigr] \vertiii{\mathbb{P}_1 w}^2_{1,R^s}
\end{align}
for $t\in [T_0, T_*)$, where we note that the last inequality follows Theorem \ref{Sobolev}.\eqref{sob1}.
By similar calculation, we obtain, from differentiating $\langle D^\alpha u, A^0_2 D^\alpha u\rangle$ with respect to $t$, the estimate
\begin{align}
\partial_t\vertiii{Du}^2_{2,H^{s-2}} =&\sum_{1\leq |\alpha|\leq s-1}\langle D^\alpha u, (\partial_t A_2^0) D^\alpha u \rangle+2\sum_{1\leq |\alpha|\leq s-1}\langle D^\alpha u, A_2^0 D^\alpha \partial_t u\rangle   \nnb  \\
\leq &  C(K_1 ) \|u\|^2_{\Qs}-\frac{1}{t}C(K_1,K_2) (\|u\|_{\Qs}+\|w\|_{R^s})
(\vertiii{\mathbb{P}_2u}^2_{2,\Rs}+\vertiii{\mathbb{P}_1w}^2_{1,R^s})   \nnb  \\
& \sum_{1\leq |\alpha|\leq s-1}\langle D^\alpha u, (\del{i} A_2^i) D^\alpha u\rangle -\frac{2}{\epsilon}\sum_{1\leq |\alpha|\leq s-1}
\overset{\quad = 0}{ \overbrace{\langle D^\alpha u, C_2^i\partial_i D^\alpha u\rangle } }  \nnb  \\
&-2\sum_{1\leq |\alpha|\leq s-1}\langle D^\alpha u, A_2^0[ D^\alpha,(A_2^0)^{-1}A_2^i]\partial_i u\rangle
-2\sum_{1\leq |\alpha|\leq s-1}\langle D^\alpha u, [\tilde{A}_2^0, D^\alpha](A_2^0)^{-1}C_2^i\partial_i u\rangle
\nnb   \\
&+\frac{2}{t}\sum_{1\leq |\alpha|\leq s-1}\langle D^\alpha u,\mathfrak{A}_2 D^\alpha \mathbb{P}_2 u\rangle-\frac{2}{t}\sum_{1\leq |\alpha|\leq s-1}\langle D^\alpha u, A_2^0[(A_2^0)^{-1}\mathfrak{A}_2, D^\alpha]\mathbb{P}_2u\rangle \nnb  \\
&+2\sum_{1\leq |\alpha|\leq s-1}\left\langle D^\alpha u, A_2^0 D^\alpha[(A_2^0)^{-1} \bigr(H_2+ \frac{1}{t} M_2 \Pbb_3 U +F_2  \bigr)]\right\rangle \nnb  \\
\leq &   C(K_1,K_2, K_3)(\vertiii{Du}^2_{2,\Hsss}+ \vertiii{Dw}_{1,H^{s-1}}^2+\vertiii{\Pbb_4w}_{L^2}^2+\sigma\vertiii{Du}_{2,\Hsss}+\mathcal{C}_*^2 ) \nnb  \\
&-\frac{1}{4t}C_2(K_1,K_2)(
\|u\|_{\Qs}+\|w\|_{R^s})\vertiii{\mathbb{P}_1w}_{1,R^s}^2  - C(K_1)\frac{1}{t} (\vertiii{Du}^2_{2,\Hsss}+\vertiii{Dw}^2_{1,\Hs})\vertiii{\Pbb_3 U}_{\Rs} \nnb  \\
& +\frac{2}{t}\bigl[\hat{\kappa}-C_2(K_1,K_2)(\|u\|_{\Qs}+\|w\|_{R^s})\bigr] \vertiii{\mathbb{P}_2u}^2_{2,\Rs} \label{E:DTU1}
\end{align}
for $t\in [T_0, T_*)$.

Next, we estimate $\vertiii{\Pbb_4 w}_{L^2}$. Acting on both sides of \eqref{E:MODELEQ1a} with $\Pbb_4$, we deduce from Assumption
\ref{ASS1}.\eqref{A:P4}
that
\begin{align}
A_1^0  \partial_t \Pbb_4 w+A_1^i  \partial_i \Pbb_4 w+\frac{1}{\epsilon}C_1^i\partial_i \Pbb_4 w =\frac{1}{t}\mathfrak{A}_1 \mathbb{P}_1 \Pbb_4 w +\Pbb_4H_1+\Pbb_4F_1.   \label{e:P4eq}
\end{align}
Then using \eqref{E:PTB0}, \eqref{e:normcp0},  \eqref{e:P4eq} and similar energy estimate to derive \eqref{E:DTW1} and \eqref{E:DTU1},
we find, with the help of the estimate $\|\Pbb_4 H_1(\epsilon,t,x, w)\|_{L^2} \leq C(K_1)\|\Pbb_4w\|_{L^2}$,
which follows from $\Pbb_4 H_1(\epsilon,t,x,\Pbb_4^\perp w)=0 $ (see Assumption \ref{ASS1}.\eqref{A:GH}), that
\begin{align}
\del{t} &\vertiii{\Pbb_4 w}_{1,L^2}^2 =  2\la\Pbb_4w, A^0_1\del{t}\Pbb_4 w\ra+\la\Pbb_4w, (\del{t}A^0_1)\Pbb_4 w\ra  \nnb  \\
&\leq  C(K_1, K_3)(\vertiii{\Pbb_4w}^2_{1,L^2}+\vertiii{Dw}^2_{1,\Hs}+\sigma\vertiii{\Pbb_4w}_{1,L^2}+C^2_*) -\frac{1}{4t}C_4(K_1, K_3)\bigl(\|\Pbb_4 w\|_{L^2}+\|w\|_{R^s}\bigr) \vertiii{\Pbb_1 w}^2_{1,R^s}   \nnb  \\
&\hspace{1cm} +\frac{4}{t}\Bigl( \hat{\kappa} - C_4(K_1, K_3)\bigl(\|\Pbb_4 w\|_{L^2}+\|w\|_{R^s}\bigr)\Bigr) \vertiii{\Pbb_1\Pbb_4 w}^2_{1,L^2}. \label{e:p4wL2}
\end{align}
Applying the operator $A^0 D^\alpha \Pbb^3(A^0)^{-1}$ to \eqref{E:MODELEQ2a}, we conclude,
with the help of \eqref{E:P32a}-\eqref{E:P32c}, that
\al{PROJEQ}{
	A^0\partial_t D^\alpha \Pbb_3 U+\Pbb_3 A^i\Pbb_3 \partial_i D^\alpha\Pbb_3  U+  \frac{1}{\epsilon}\Pbb_3 C^i\Pbb_3 \partial_i D^\alpha \Pbb_3 U
	=&-A ^0[ D^\alpha,(A ^0)^{-1}\Pbb_3 A^i\Pbb_3 ]\partial_i \Pbb_3 U \notag \\
	-[\tilde{A}^0, D^\alpha](A^0)^{-1}\Pbb_3 C^i\Pbb_3 \partial_i \Pbb_3 U
	+\frac{1}{t}\Pbb_3 \mathfrak{A} \Pbb_3 D^\alpha  \Pbb_3 U &+\frac{1}{t}A ^0[ D^\alpha, (A ^0)^{-1}
	\Pbb_3 \mathfrak{A}\Pbb_3 ]\Pbb_3  U
	+A ^0 D^\alpha[(A ^0)^{-1}\Pbb_3  H].  }
By similar arguments that were used to derive \eqref{E:DTW1} and \eqref{E:DTU1}, we obtain from \eqref{E:PROJEQ} the
estimate
\begin{align}
\partial_t\vertiii{D\Pbb_3 U}^2_{H^{s-2}}   =&\sum_{1\leq |\alpha|\leq s-1}\langle D^\alpha \Pbb_3 U, (\partial_t A^0) D^\alpha \Pbb_3 U \rangle+2\sum_{1\leq |\alpha|\leq s-1}\langle D^\alpha \Pbb_3 U, \Pbb_3 A^0 \Pbb_3 D^\alpha \partial_t \Pbb_3 U\rangle    \nnb  \\
\leq &  -\frac{1}{t} C(K_1)\|\Pbb_1 w\|_{R^s}\vertiii{\Pbb_3 U}^2_{\Rs}+C(K_1) \|\Pbb_3 U\|^2_{\Qs}  \nnb   \\
& +\sum_{1\leq |\alpha|\leq s-1}\langle D^\alpha \Pbb_3 U, (\del{i}A^i) D^\alpha \Pbb_3 U\rangle
-\frac{2}{\epsilon}\sum_{1\leq |\alpha|\leq s-1} \overset{\quad =0}{\overbrace{\langle D^\alpha \Pbb_3 U, C ^i\partial_i D^\alpha \Pbb_3 U\rangle }}
\nnb  \\
&-2\sum_{1\leq |\alpha|\leq s-1}\langle D^\alpha \Pbb_3 U, A^0[ D^\alpha,(A ^0)^{-1}A ^i]\partial_i \Pbb_3 U+[\tilde{A}^0, D^\alpha](A^0)^{-1}C^i\partial_i \Pbb_3 U\rangle
\nnb  \\
&+\frac{2}{t}\sum_{1\leq |\alpha|\leq s-1}\langle D^\alpha \Pbb_3 U,\mathfrak{A} D^\alpha \Pbb_3 U\rangle+\frac{2}{t}\sum_{1\leq |\alpha|\leq s-1}\langle D^\alpha \Pbb_3 U, A^0[(A^0)^{-1}\mathfrak{A}, D^\alpha]\Pbb_3 U\rangle   \nnb  \\
&\hspace{5.5cm} +2\sum_{1\leq |\alpha|\leq s-1}\left\langle D^\alpha \Pbb_3 U, A^0 D^\alpha[(A^0)^{-1} \Pbb_3 H]\right\rangle   \nnb  \\
\leq &  C(K_1 ) \|\Pbb_3 U\|^2_{\Qs}  +C(K_1)\|\Pbb_3 U\|_{\Qs}\Bigl(\| H_1\|_{\Qs}+\|H_2\|_{\Qs}+\| F_1\|_{\Qs} \nnb \\
&\hspace{1.5cm}+\|F_2\|_{\Qs}\Bigr)
+\frac{1}{t}\Bigl(2 \hat{\kappa}-C (K_1,K_2 ) \bigl(\|w\|_{R^s}+\|u\|_{\Qs}\bigr)\Bigr)\vertiii{\Pbb_3 U}^2_{\Rs} \nnb\\
\leq & C(K_1 ) \vertiii{\Pbb_3 U}^2_{\Rs}  + C(K_1,K_2)\bigl(\vertiii{w}_{1,R^s}+\vertiii{u}_{2,\Rs}+\sigma+\mathcal{C}_*+\vertiii{\Pbb_4w}_{1,L^2}\bigr)\vertiii{\Pbb_3 U}_{\Rs} \nnb \\
&\hspace{4.1cm}+\frac{1}{t}\Bigl(2 \hat{\kappa}-C(K_1,K_2 ) \bigl(\|w\|_{R^s}+\|u\|_{\Qs}\bigr)\Bigr)\vertiii{\Pbb_3 U}^2_{\Rs}.\label{e:DPUT}
\end{align}
We also find, by a similar derivation used to establish \eqref{e:normcp}, that
\begin{align}\label{e:punorm}
\|\Pbb_3 U\|_{L^6}\leq C_{\text{S}} \sqrt{\gamma_1}
\vertiii{D\Pbb_3 U}_{\Hsss},
\end{align}
from which it follows that
\begin{align}\label{e:puinq}
\vertiii{D\Pbb_3 U}_{\Hsss} \leq \vertiii{\Pbb_3U}_{R^{s-1}}\leq (1+C_{\text{S}} \sqrt{\gamma_1} )
\vertiii{D\Pbb_3 U}_{\Hsss}
\end{align}
by adding $\vertiii{D\Pbb_3 U}_{\Hsss}$ to both sides of \eqref{e:punorm}.
Furthermore, using \eqref{e:DPUT} and \eqref{e:puinq}, we see that $\partial_t\vertiii{D\Pbb_3 U}^2_{H^{s-2}}$ is dominated by
\begin{align*}
&\partial_t\vertiii{D\Pbb_3 U}^2_{H^{s-2}} 		\leq   C(K_1,K_2)\bigl(\vertiii{w}_{1,R^s}+\vertiii{u}_{2,\Rs}+\sigma+\mathcal{C}_*  +\vertiii{\Pbb_4w}_{1,L^2}\bigr)
\vertiii{D\Pbb_3 U}_{\Hsss} \nnb \\
+C(K_1 )& \vertiii{\Pbb_3 U}_{\Rs}
\vertiii{D\Pbb_3 U}_{\Hsss}  +\frac{1}{t}\Bigl(2 \hat{\kappa}-C (K_1,K_2 ) \bigl(\|w\|_{R^s}+\|u\|_{\Qs}\bigr)\Bigr)\vertiii{\Pbb_3 U}_{\Rs}
\vertiii{D\Pbb_3 U}_{\Hsss}.   
\end{align*}
Using this in conjunction with \eqref{e:normcp0} yields the estimate
\begin{align}
\partial_t\vertiii{D\Pbb_3 U}_{\Hsss}  &\leq  C(K_1 ) \vertiii{\Pbb_3 U}_{\Rs}  + C(K_1,K_2)\bigl(\vertiii{Dw}_{1,H^{s-1}}+
\vertiii{Du}_{2,H^{s-2}}+\sigma+\mathcal{C}_* +\vertiii{\Pbb_4w}_{1,L^2}\bigr) \notag \\
&\hspace{2cm}+\frac{1}{ t}\biggl( \hat{\kappa}- C_2(K_1,K_2 )  \bigl(\|w\|_{R^s}+\|u\|_{\Qs}\bigr)\biggr) \vertiii{\Pbb_3 U}_{\Rs}.
\label{E:DTP3U}
\end{align} 

To proceed, we choose $\sigma>0$ small enough so that the inequality
\begin{equation*}
\Bigl(C_1(\hat{R}) +  2C_2(\hat{R},\hat{R})+    2C_4(\hat{R},\hat{R}) \Bigr)
\sigma < \frac{\hat{\kappa} }{4}
\end{equation*}
holds in addition to \eqref{sigmaC1}.
Then
\begin{align*}
\hat{\kappa} - \Bigl(C_1(K_1(T_0))\norm{w(T_0)}_{R^s} +   C_2(K_1(T_0),&K_2(T_0))\bigl(\norm{w(T_0)}_{R^s}
+ \norm{u(T_0)}_{R^{s-1}}\bigr)  \nnb  \\
&+ C_4(K_1(T_0), K_3(T_0))\bigl(\|\Pbb_4 w(T_0)\|_{L^2}+\|w(T_0)\|_{R^s}\bigr)\Bigr)  > \frac{\hat{\kappa} }{2},
\end{align*}
and we see by continuity that either
\begin{align*}
\hat{\kappa}  - \Bigl(C_1(K_1(t))\norm{w(t)}_{R^s} +   C_2(K_1(t),&K_2(t))\bigl(\norm{w(t)}_{R^s}  + \norm{u(t)}_{R^{s-1}}\bigr)\nnb  \\
& + C_4(K_1(t), K_3(t))\bigl(\|\Pbb_4 w(t)\|_{L^2}+\|w(t)\|_{R^s}\bigr)\Bigr)  > \frac{\hat{\kappa} }{2}, \quad 0\leq t < T_*,
\end{align*}
or else there exists a first time $T^* \in (0,T_*)$ such that
\begin{align*}
\hat{\kappa}  - \Bigl(C_1(K_1(T^*))\norm{w(T^*)}_{R^s} +   C_2(K_1(T^*),&K_2(T^*))\bigl(\norm{w(T^*)}_{R^s}  + \norm{u(T^*)}_{R^{s-1}}\bigr)\nnb  \\
& + C_4(K_1(T^*), K_3(T^*))\bigl(\|\Pbb_4 w(T^*)\|_{L^2}+\|w(T^*)\|_{R^s}\bigr)\Bigr) = \frac{\hat{\kappa} }{2}.
\end{align*}
Letting $T^*=T_*$ if the first case holds, we then have that
\begin{align} \label{E:qq}	
\hat{\kappa}  - \Bigl(C_1(K_1(t))\norm{w(t)}_{R^s} +   C_2&(K_1(t),K_2(t))\bigl(\norm{w(t)}_{R^s}  + \norm{u(t)}_{R^{s-1}}\bigr)\nnb  \\
& + C_4(K_1(t), K_3(t))\bigl(\|\Pbb_4 w(t)\|_{L^2}+\|w(t)\|_{R^s}\bigr)\Bigr)  > \frac{\hat{\kappa} }{2}, \quad 0\leq t < T^*\leq T_*.
\end{align}
Taken together, the estimates \eqref{K1ineq},  \eqref{E:DTW1}, \eqref{E:DTU1}, \eqref{e:p4wL2}, \eqref{E:DTP3U} and \eqref{E:qq}, with the help of \eqref{e:puinq} and Young's inequality, 
imply that
\begin{align}
\partial_t\vertiii{D w}^2_{1,H^{s-1}} \leq & C(\hat{R})\Bigl(\vertiii{Dw}^2_{1,H^{s-1}}+\vertiii{\Pbb_4w}_{1,L^2}^2+\sigma^2+\mathcal{C}_*^2 \Bigr)
+\frac{\hat{\kappa} }{ t} \vertiii{\mathbb{P}_1 w}^2_{1,R^s}, \label{E:ENERGEST1}\\
\partial_t\vertiii{Du}^2_{2,\Hsss} \leq & C(\hat{R} )\Bigl(\vertiii{Du}^2_{2,\Hsss}+\vertiii{Dw}_{1,\Hs}^2+\vertiii{\Pbb_4w}_{1,L^2}^2+\sigma^2+\mathcal{C}_*^2 \Bigr)-\frac{\hat{\kappa} }{8t}\vertiii{\mathbb{P}_1w}_{1,R^s}^2\nnb \\
& \hspace{2cm}+\frac{\hat{\kappa} }{ t} \vertiii{\mathbb{P}_2u}^2_{2,\Rs}  - \frac{1}{t} C_3(\hat{R}) \bigl(\vertiii{Du}^2_{2,\Hsss}+\vertiii{Dw}^2_{1,\Hs}\bigr)
\vertiii{D\Pbb_3 U}_{\Hsss},
\label{E:ENERGEST2}  \\
\del{t} \vertiii{\Pbb_4 w}_{1,L^2}^2
\leq & C(\hat{R})\Bigl(\vertiii{\Pbb_4w}^2_{1,L^2}+\vertiii{Dw}^2_{1,\Hs}+\sigma^2+\mathcal{C}_*^2 \Bigr)+\frac{2\hat{\kappa} }{t}\vertiii{\Pbb_1\Pbb_4 w}^2_{1,L^2} -\frac{\hat{\kappa} }{8t}  \vertiii{\Pbb_1 w}^2_{1,R^s}  \label{E:ENERGEST4}\\
\intertext{and}
\partial_t\vertiii{D \Pbb_3 U} _{\Hsss} \leq &   C(\hat{R})\Bigl(
\vertiii{D\Pbb_3 U}_{\Hsss}  + \vertiii{Dw}_{1,H^{s-1}}+
\vertiii{Du}_{2,H^{s-2}}+\sigma+\mathcal{C}_*  +\vertiii{\Pbb_4w}_{1,L^2}\Bigr)+
\frac{\hat{\kappa} }{2t}
\vertiii{D\Pbb_3 U}_{\Hsss} \label{E:ENERGEST3}
\end{align}
for $0\leq t < T^* \leq T_*$.

Next, we set
\begin{gather*}
X = \vertiii{Dw}^2_{1,H^{s-1}} + \vertiii{Du}^2_{2,\Hsss}+\vertiii{\Pbb_4 w}_{1,L^2}^2,  \\
Y = \vertiii{\Pbb_1 w}^2_{1,R^s} + \vertiii{\Pbb_2 u}^2_{2,\Rs}+\vertiii{\Pbb_1\Pbb_4w}_{1,L^2}^2,
\AND
Z = \vertiii{D\Pbb_3U}_{\Hsss}.
\end{gather*}
Since $C_3(\hat{R})X(T_0)/\sigma \leq C(\hat{R})\sigma$, we can choose $\sigma$ small enough so that
$C_3(\hat{R})X(T_0)/\sigma < \hat{\kappa} /8$. Then by continuity, either $ C_3(\hat{R})X(t)/\sigma \leq \hat{\kappa} /8$ for $t\in [T_0,T^*)$,
or else there exists a first time $T\in (T_0,T^*)$ such that $C_3(\hat{R})X(T)/\sigma = \hat{\kappa} /8$. Setting
$T=T^*$ if the first case holds, we then have that
\begin{equation} \label{Tdef}
C_3(\hat{R})\frac{X(t)}{\sigma} < \hat{\kappa}  /8, \quad T_0 \leq t < T\leq T^* \leq T_*.
\end{equation}
Adding the inequalities \eqref{E:ENERGEST1}, \eqref{E:ENERGEST2} and \eqref{E:ENERGEST4} and dividing the results by $\sigma$, we
obtain, with the aid of \eqref{Tdef}, the inequality
\begin{equation} \label{Xest1}
\del{t}\biggl(\frac{X}{\sigma}\biggr) \leq C(\hat{R})\left(\frac{X+\mathcal{C}_*^2}{\sigma}+\sigma\right) -\frac{\hat{\kappa} }{8t}Z +\frac{\hat{\kappa} }{2t}\frac{Y}{\sigma}, \quad T_0 \leq t < T\leq T^*\leq T_*,
\end{equation}
while the inequality
\begin{equation} \label{Zest1}
\del{t}Z \leq C(\hat{R})\biggl( Z + \sigma + \frac{X+\mathcal{C}_*^2}{\sigma}\biggr) + \frac{\hat{\kappa} }{4t} Z, \quad T_0 \leq t < T^*\leq T_*
\end{equation}
follows from \eqref{E:ENERGEST3} and Young's inequality.
Adding \eqref{Xest1} and \eqref{Zest1} gives
\begin{equation} \label{XZest1}
\del{t}\biggl(\frac{X }{\sigma} +Z -\frac{\hat{\kappa} }{8}\int_{T_0}^t \frac{1}{\tau}\biggl(\frac{Y}{\sigma}+Z\biggr)\, d\tau  + \sigma \biggr)
\leq C(\hat{R})  \biggl(\frac{X+\mathcal{C}_*^2}{\sigma} +Z -\frac{\hat{\kappa} }{8}\int_{T_0}^t \frac{1}{\tau}\biggl(\frac{Y}{\sigma}+Z\biggr)
\, d\tau  + \sigma \biggr).
\end{equation}
Noting that
$\mathcal{C}_*(t)^2 \lesssim \int_{T_0}^t X(\tau) d \tau$ by the H\"{o}lder and Sobolev inequalities, see Theorems \ref{T:HOLDER}.\eqref{T:H1} and \ref{Sobolev}.\eqref{sob1},
we see, after adding $X(t)/\sigma$ to both sides of \eqref{XZest1}, that the inequality
\begin{equation} \label{XZest2}
\del{t}\biggl(\frac{X+\int_{T_0}^t X(\tau) d \tau }{\sigma} +Z -\frac{\hat{\kappa} }{8}\int_{T_0}^t \frac{1}{\tau}\biggl(\frac{Y}{\sigma}+Z\biggr)\, d\tau  + \sigma \biggr)
\leq C(\hat{R})  \biggl(\frac{X+\int_{T_0}^t X(\tau) d \tau}{\sigma} +Z -\frac{\hat{\kappa} }{8}\int_{T_0}^t \frac{1}{\tau}\biggl(\frac{Y}{\sigma}+Z\biggr)
\, d\tau  + \sigma \biggr)
\end{equation}
holds for $T_0 \leq t < T\leq T^*\leq T_*$. Since $X(T_0)\leq C(\hat{R})\sigma^2$ and $Z(T_0) \lesssim \sigma$, it follows
directly from  \eqref{XZest2} and Gr\"{o}nwall's inequality that
\begin{equation*}
\frac{1}{\sigma}\Bigl(X(t)+\int_{T_0}^t X(\tau) d \tau \Bigr) +Z(t) -\frac{\hat{\kappa} }{8}\int_{T_0}^t \frac{1}{\tau}\biggl(\frac{Y(\tau)}{\sigma}+Z(\tau)\biggr)\, d\tau  + \sigma
\leq e^{C(\hat{R})(t-T_0)}C(\hat{R})\sigma,  \quad T_0 \leq t < T\leq T^*\leq T_*,
\end{equation*}
from which it follows that
\begin{equation*} 
\|\Pbb_4w(t)\|_{L^2} +
\left(-\int_{T_0}^{t} \frac{1}{\tau} \|\Pbb_1\Pbb_4 w\|^2_{L^2}\, d\tau\right)^{\frac{1}{2}}+\|w\|_{M^\infty_{\Pbb_1, s}([T_0,t)\times \Rbb^3)}+\|u\|_{M^\infty_{\Pbb_2, s-1}([T_0,t)\times \Rbb^3)} -
\int_{T_0}^{t} \frac{1}{\tau} \|\Pbb_3 U(\tau)\|_{\Qs}\, d\tau  \leq C(\hat{R})\sigma,
\end{equation*}
for $T_0 \leq t < T\leq T^*\leq T_*$,
where we stress that the constant $C(\hat{R})$ is independent of $\epsilon$ and the times $T$, $T^*$, $T_*$, and $T_1$. Choosing
$\sigma$ small enough, it is then clear from the estimate \eqref{XZest2} and the definition of the times $T$, $T^*$, and $T_1$
that $T=T^*=T_*=T_1$, which completes the proof.
\end{proof}

\subsection{Error estimates\label{S:MODELerr}}
In this section, we consider solutions of the singular initial value problem:
\begin{align}
A_1^0 (\epsilon,t,x,w)\partial_t w+A_1^i (\epsilon,t,x,w)\partial_i w+\frac{1}{\epsilon}C_1^i\partial_i w&=\frac{1}{t}\mathfrak{A}_1(\epsilon,t,x,w)\mathbb{P}_1  w+H_1+F_1  &&\mbox{in} \quad[T_0, T_1)\times\Rbb^3, \label{E:MODELEQ3a}\\
w(T_0,x) &= \mrw^0(x) +\epsilon s^0(\epsilon,x) &&\text{in} \quad \{T_0\}\times\Rbb^3, \label{E:MODELEQ3b}
\end{align}
where the matrices $A_1^0$, $A_1^i$, $i=1,\ldots,n$, and $\mathfrak{A}_1$ and the source terms $H_1$ and $F_1$ satisfy the conditions from
Assumption \ref{ASS1}. Our aim is to use the uniform a priori estimates from Theorem \ref{L:BASICMODEL} to establish
uniform a priori estimates for solutions of \eqref{E:MODELEQ3a}-\eqref{E:MODELEQ3b} and the corresponding \textit{limit equation} defined by
\begin{align}
\mathring{A}_1^0\partial_t \mathring{w}+\mathring{A}_1^i\partial_i \mathring{w}&=\frac{1}{t}\mathring{\mathfrak{A}}_1\mathbb{P}_1\mathring{w}-C_1^i\partial_i v+\mathring{H}_1+\mathring{F}_1
&& \mbox{in} \quad[T_0, T_1)\times\Rbb^3,  \label{E:LIMITINGEQa}\\
C_1^i\partial_i\mathring{w}&=0 && \mbox{in} \quad[T_0, T_1)\times\Rbb^3, \label{E:LIMITINGEQb}\\
\mathring{w}(T_0,x)&=\mathring{w}^0(x) &&\text{in} \quad \{T_0\}\times\Rbb^3, \label{E:LIMITINGEQc}
\end{align}
and to establish an error estimate
between solutions of  \eqref{E:MODELEQ3a}-\eqref{E:MODELEQ3b} and \eqref{E:LIMITINGEQa}-\eqref{E:LIMITINGEQc}.

In the limit equation, $\mathring{A}^0_1$ and $\mathring{\mathfrak{A}}_1$ are defined by \eqref{E:DECOMPOSITIONOFA01} and
\eqref{E:DECOMPOSITIONOFCALB} with $a=1$, respectively,
while $\mathring{A}_1^i$ and $\mathring{H}_1$ are defined by the limits
\begin{equation} \label{E:HRIN}
\mathring{A}_1^i(t,x, \mathring{w})=\lim_{\epsilon\searrow 0}A_1^i(\epsilon,t,x, \mathring{w}) \AND
\mathring{H}_1(t,x, \mathring{w})=\lim_{\epsilon\searrow 0} H_1(\epsilon,t,x,\mathring{w}),
\end{equation}
respectively. We further assume that the following conditions hold for fixed constants
$R >0$, $T_0 < T_1 <0$ and $s\in \Zbb_{>n/2+1}$:

\begin{ass}\label{ASS3}$\;$

\begin{enumerate}
	\item \label{A3a} The source terms\footnote{The source term $\mathring{F}_1$ should be thought of as the $\epsilon \searrow 0$ limit of
		$F_1$. This is made precise by the hypothesis \eqref{HFLip} of Theorem  \ref{T:MAINMODELTHEOREM}.}
	$\mathring{F}_1$ and $v$ satisfy  $\mathring{F}_1 \in C^0\bigl([T_0,T_1),H^s(\Rbb^3,\Rbb^{N_1})\bigr)$
	and $v\in \bigcap_{\ell=0}^1 C^\ell \bigl([T_0,T_1),R^{s+1-\ell}(\Rbb^3,\Rbb^{N_1})\bigr)$, respectively.
	\item \label{A3b} The matrices $\mathring{A}_1^i$, $i=1,\ldots, n$ and the source term $\mathring{H}_1$ satisfy\footnote{From
		the assumptions \ref{ASS1}.\eqref{A:GH}-\eqref{A:Bi} on $A_1^i$ and $H_1$, it follows directly from
		the \eqref{E:HRIN} that
		$\mathring{A}_1^i \in  E^0\big((2T_0,0)\times \Rbb^3 \times B_R\bigl(\Rbb^{N_1}\bigr), \mathbb{S}_{N_1}\bigr)$
		and $\mathring{H}_1 \in  E^0\big((2T_0,0)\times \Rbb^3 \times B_R\bigl(\Rbb^{N_1}\bigr), \mathbb{R}^{N_1}\bigr)$.
	}
	$t\mathring{A}_1^i \in E^1\big((2T_0,0)\times \Rbb^3 \times B_R\bigl(\Rbb^{N_1}\bigr), \mathbb{S}_{N_1}\bigr)$,
	$t\mathring{H}_1  \in E^1\big((2T_0,0)\times \Rbb^3 \times B_R\bigl(\Rbb^{N_1}\bigr), \mathbb{R}^{N_1}\bigr)$,
	and $D_{t}\bigl(t\mathring{H}_1(t,x,0\bigr)\bigr)=0$.
	\end{enumerate}
\end{ass}

We are now ready to state and establish uniform a priori estimates for solutions of the singular initial value problem \eqref{E:MODELEQ3a}-\eqref{E:MODELEQ3b} and the associated limit equation defined by
\eqref{E:LIMITINGEQa}-\eqref{E:LIMITINGEQc}.
\begin{theorem}\label{T:MAINMODELTHEOREM}
Suppose $R>0$, $s\in\mathbb{Z}_{\geq 3}$,  $T_0< T_1 \leq 0$, $\epsilon_0>0$, $\mrw^0\in 
H^s(\Rbb^3, \Rbb^{N_1})$,
$s^0\in   L^\infty\bigl((0,\epsilon_0),R^s(\Rbb^3, \Rbb^{N_1})\bigr)$, Assumptions \ref{ASS1} and \ref{ASS3}
hold,  the maps
\als
{
	(w, \mrw) \in \bigcap_{\ell=0}^1 C^\ell\bigl([T_0,T_1),R^{s-\ell}\bigl(\Rbb^3, \Rbb^{N_1}\bigr)\bigr) \times
	\bigcap_{\ell=0}^1 C^\ell\bigl([T_0,T_1),H^{s-\ell}\bigl(\Rbb^3, \Rbb^{N_1}\bigr)\bigr)
}
define solutions to the initial value problems
\eqref{E:MODELEQ3a}-\eqref{E:MODELEQ3b} and \eqref{E:LIMITINGEQa}-\eqref{E:LIMITINGEQc}, and,
for $t\in [T_0,T_1)$, the following estimates hold:
\begin{align}
&\|v(t)\|_{R^{s+1}}-\frac{1}{t}\|\Pbb_1 v(t)\|_{R^{s+1}}+\|\del{t} v(t)\|_{R^s}+\|\mathring{F}_1(t)\|_{H^s}+\|t \partial_t \mathring{F}_1(t)\|_{\Rs}+\|F_1(\epsilon,t)\|_{R^s}  \nnb  \\
\leq & C(\|w\|_{\Li([T_0,t),R^s)},\|\mrw\|_{\Li([T_0,t),H^s)})\Bigl(
\mathcal{C}^*
+\|w(t)\|_{R^s}+\|\mrw(t)\|_{H^s}+\int_{T_0}^t (\|\mathring{w}(\tau)\|_{H^s}+\|w(\tau)\|_{R^s})d\tau\Bigr), \label{F1est}
\end{align}
\begin{align}
\|A^i_1(\epsilon,t,\cdot,\mathring{w}(t))-\mathring{A}^i_1(t,\cdot,\mathring{w}(t))\|_{\Rs} \leq \epsilon
C\bigl(\|\mrw(t)\|_{\Li([T_0,t),R^s)}\bigr), \label{AiLip}
\end{align}
and
\begin{align}
\|H_1(\epsilon,&t,\cdot,\mathring{w}(t))-\mathring{H}_1(t,\cdot,\mathring{w}(t))\|_{\Rs}+\|F_1(\epsilon,t)-\mathring{F}_1(t)\|_{\Rs} \notag \\
\leq  & \epsilon
C\bigl(\|w\|_{\Li([T_0,t),R^s)},\|\mrw\|_{\Li([T_0,t),H^s)}\bigr)\Bigl(
\mathcal{C}^* +\|w(t)\|_{R^s}+\|z(t)\|_{\Qs}+\|\mrw(t)\|_{H^s} \nnb  \\
&+\int_{T_0}^t(\|w(\tau)\|_{R^s}+\|z(\tau)\|_{\Qs}+\|\mrw(\tau)\|_{H^s})d\tau\Bigr), \label{HFLip}
\end{align}
where 
\begin{equation*}\label{E:Z}
z=\frac{1}{\epsilon}\left(w-\mathring{w}-\epsilon v\right),
\end{equation*}
and the constants $\mathcal{C}^*$, $ C\bigl(\|w\|_{\Li([T_0,t),R^s)}\bigr)$, $C\bigl(\|\mrw\|_{\Li([T_0,t),R^s)}\bigr)$ and
$C\bigl(\|w\|_{\Li([T_0,T_t),R^s)},\|\mrw\|_{\Li([T_0,t),R^s)}\bigr)$
are independent of $\epsilon \in (0,\epsilon_0)$ and $T_1 \in (T_0,0)$.

Then there exists a small constant $\sigma>0$, independent of $\epsilon \in (0,\epsilon_0)$  and $T_1 \in (T_0,0)$, such that if 
\begin{align}\label{E:INITIALDATA3}
\|\mathring{w}^0\|_{H^s}+\|s^0\|_{R^s} +\mathcal{C}^*
\leq \sigma, 
\end{align}
then
\begin{equation} \label{wrmwsupest}
\max\{\norm{w}_{L^\infty([T_0,T_1)\times \Rbb^3)} ,\norm{\mrw}_{L^\infty([T_0,T_1)\times \Rbb^3)} \} \leq \frac{R}{2}
\end{equation}
and there exists a constant $C>0$, independent of $\epsilon\in (0,\epsilon_0)$ and $T_1 \in (T_0,0)$, such that
\begin{align}
&\|\mrw\|_{L^\infty([T_0,t),L^2)}+\|w\|_{M^\infty_{\mathbb{P}_1,s}([T_0,t)\times\Rbb^3)}  +  \|\mrw\|_{M^\infty_{\mathbb{P}_1,s}([T_0,t)\times\Rbb^3)}
+\|t\partial_t \mrw\|_{M^\infty_{\mathds{1},s-1}([T_0,t)\times\Rbb^3)}  \notag \\
&\hspace{3.0cm} \left(-\int_{T_0}^{t} \frac{1}{\tau} \|\Pbb_1\mrw\|^2_{L^2}\, d\tau\right)^{\frac{1}{2}}+\int_{T_0}^{t}\|\del{t}\mrw\|_{\Qs}d\tau-\int_{T_0}^{t}\frac{1}{\tau}\|\Pbb_1 \mrw\|_{\Qs} d\tau  \leq C \sigma, \label{E:FINALEST1a}
\end{align}
\begin{gather}
\|w-\mathring{w}\|_{L^\infty([T_0, t),R^{s-1})}  \leq \epsilon C\sigma \label{E:FINALEST1b}
\intertext{and}
-\int_{T_0}^{t}\frac{1}{\tau}\|\mathbb{P}_1(w-\mrw)\|^2_{R^{s-1}}d\tau  \leq \epsilon^2 C\sigma^2 \label{E:FINALEST1c}
\end{gather}
for $T_0 \leq t < T_1$.
\end{theorem}
\begin{proof}
First, we observe, by \eqref{E:DECOMPOSITIONOFA01} and \eqref{E:PAP}, that $\mathring{A}^0_1$ satisfies
\begin{equation} \label{Aringproj}
\Pbb_1^\perp \mathring{A}^0_1 \Pbb_1 = \Pbb_1 \mathring{A}^0_1 \Pbb_1^\perp.
\end{equation}
Using this, we find, after applying $\Pbb_1$ to the limit equation \eqref{E:LIMITINGEQa}, that
\al{d}{
	b=\Pbb_1 \mrw
}
satisfies the equation
\al{PW2}{
	\Pbb_1 \mrA^0_1\Pbb_1 \del{t} b +\Pbb_1 \mrA^i_1\Pbb_1 \del{i} b =\frac{1}{t}\Pbb_1 \mathfrak{\mrA}_1 \Pbb_1 b  +
	\Pbb_1 \mrH_1+ \Pbb_1\bar{F}_2,
}
where $\bar{F}_2 = -\Pbb_1 \mrA^i_1\Pbb_1^\perp \del{i}\mrw + \Pbb_1 \mathring{F}_1-\Pbb_1 C^i_1 \del{i} v$.
Clearly, $\bar{F}_2$ satisfies
\begin{equation} \label{Fb2est}
\norm{\bar{F}_2(t)}_{R^{s-1}} \leq C\bigl(\|\mrw\|_{\Li([T_0,t),H^s)},\|w\|_{\Li([T_0,t),R^s)}\bigr)\Bigl(\mathcal{C}^*
+\|w(t)\|_{R^s}+\|\mrw(t)\|_{H^s}+\int_{T_0}^t (\|\mathring{w}(\tau)\|_{H^s}+\|w(\tau)\|_{R^s})d\tau\Bigr)
\end{equation}
for $0\leq t < T_1$ by \eqref{E:HQR}, \eqref{F1est} and the calculus inequalities from Appendix \ref{A:INEQUALITIES}, while
\begin{equation} \label{bidata}
\norm{b(T_0)}_{R^{s-1}} \leq \norm{ \mrw^0}_{R^s}\lesssim \norm{ \mrw^0}_{H^s}  \leq \sigma,
\end{equation}
follows from  the assumption \eqref{E:INITIALDATA3} on the initial data, and we note that  $\Pbb_1\mathring{H}_1(t,x,\mrw)$ satisfies
\begin{equation} \label{P1H1zerp}
\Pbb_1\mathring{H}_1(t,x,0) = 0
\end{equation}
by Assumption \ref{ASS1}.(3).

Next, we set
\als
{y=t\del{t} \mrw.}
In order to derive an evolution equation for $y$, we apply $t\del{t}$ to \eqref{E:LIMITINGEQa} and use
the identity
\begin{equation*}
t\partial_t f= t D_t f+ [D_{\mrw} f\cdot t\partial_t \mrw]
=D_t (t f)- f+ [D_{\mrw} f\cdot t\partial_t \mrw], \quad f=f(t,x,\mrw(t,x)),
\end{equation*}
to obtain
\begin{equation}\label{E:MODELEQ4}
\mrA_1^0\partial_t y+\mrA_1^i\partial_i y
=\frac{1}{t}\bigr(\mathbb{P}_1\mathfrak{\mrA}_1\mathbb{P}_1 +\mrA_1^0\bigr ) y-\frac{1}{t}\mathfrak{\mrA}_1 b  + \tilde{R}_2  +\tilde{H}_2+\tilde{F}_2,
\end{equation}
where
\begin{align*}
\tilde{H}_2&=   D_t (t \mathring{H}_1)-\mrH_1 + [D_{\mrw} \mrH_1 \cdot y]+(D_t \mathring{\mathfrak{A}}_1) b-(D_t \mrA^0_1)y
\intertext{and}
\tilde{F}_2 &=  -[D_{\mrw} \mrA_1^i \cdot y]\partial_i \mrw - D_t (t\mrA^i_1) \del{i} \mrw + \mrA^i_1 \del{i} \mrw + t \partial_t \mathring{F}_1 +tC^i_1\del{i}\del{t} v. 
\end{align*}
Note that in deriving the above equation, we have used the identity
\begin{equation} \label{Afrcom}
\mathring{\mathfrak{A}}_1 \Pbb_1 =  \Pbb_1 \mathring{\mathfrak{A}}_1  =\Pbb_1 \mathring{\mathfrak{A}}_1\Pbb_1,
\end{equation}
which follows directly from \eqref{E:DECOMPOSITIONOFCALB} and \eqref{E:COMMUTEPANFB}.
We further note by \eqref{F1est}, Assumption \ref{ASS1}.(4) and Assumption \ref{ASS3}.(2) that $\tilde{F}_2$
and $\tilde{H}_2 = \tilde{H}_2(t,x,\mrw,b,y)$ satisfy
\begin{align} \label{Ft2est}
\norm{\tilde{F}_2(t)}_{\Qs} \leq C\bigl(\|\mrw\|_{\Li([T_0,t),H^s)},\|w\|_{\Li([T_0,t),R^s)}\bigr)\Bigl(&  \mathcal{C}^* +\norm{y(t)}_{\Qs}
+\|w(t)\|_{R^s}+\|\mrw(t)\|_{H^s} \nnb  \\
& + \int_{T_0}^t (\|\mathring{w}(\tau)\|_{H^s}+\|w(\tau)\|_{R^s})d\tau\Bigr)
\end{align}
for $T_0\leq t < T_1$ and
\begin{equation} \label{Ht2zero}
\tilde{H}_2(t,x,0,0,0)=0,
\end{equation}
respectively. Using \eqref{E:LIMITINGEQa} and \eqref{E:INITIALDATA3}, we deduce that
\begin{equation*}
y|_{\Sigma}= \Bigl[(\mrA_1^0)^{-1} \mathfrak{\mrA}_1\mathbb{P}_1 \mrw- t(\mrA_1^0)^{-1} \mrA_1^i \partial_i \mrw- t(\mrA_1^0)^{-1} C_1^i\partial_i v + t(\mrA_1^0)^{-1} \mrH_1+ t(\mrA_1^0)^{-1} \mathring{F}_1\Bigr]\Bigl|_{\Sigma},
\end{equation*}
which in turn, implies, via \eqref{F1est},  \eqref{E:INITIALDATA3}, and the calculus
inequalities from Appendix \ref{A:INEQUALITIES}, that
\begin{equation*}\label{E:P_UT0}
\|y(T_0)\|_{\Qs} \leq C(\sigma) \sigma.
\end{equation*}

Next, a short computation using \eqref{E:MODELEQ3a}, \eqref{E:LIMITINGEQa} and \eqref{E:LIMITINGEQb}  shows that
\begin{align}\label{E:DIFFEQ2}
A_1^0\partial_ t z+A_1^i\partial_iz+\frac{1}{\epsilon}C_1^i\partial_iz
=\frac{1}{t} \mathfrak{A}_1\mathbb{P}_1z +  \hat{R}_2 +\hat{F}_2,
\end{align}
where
\begin{gather*}
\hat{F}_2= \frac{1}{\epsilon} (H_1-\mathring{H}_1)+\frac{1}{\epsilon} (F_1-\mathring{F}_1)
-\frac{1}{\epsilon}(A_1^i-\mathring{A}^i_1)\partial_i \mrw
-  A_1^i\partial_i v -  A_1^0\partial_t v +\frac{1}{t} \Pbb_1  \mathfrak{A}_1 \mathbb{P}_1v 
\intertext{and}
\hat{R}_2=- \frac{1}{t}\tilde{A}_1^0 y+ \frac{1}{t}\tilde{\mathfrak{A}}_1 b,  
\end{gather*}
and we recall that $\tilde{A}_1^0$ and $\tilde{\mathfrak{A}}_1$ are defined by the expansions
\eqref{E:DECOMPOSITIONOFA01}-\eqref{E:DECOMPOSITIONOFCALB}.
To proceed, we estimate
\al{FHATC}{
	\frac{1}{\epsilon} \|H_1&(\epsilon,t,\cdot,w(t))-\mathring{H}_1(t,\cdot, \mathring{w}(t))\|_{\Qs}   \notag \\
	&\leq  \frac{1}{\epsilon} \|H_1(\epsilon,t,\cdot,w(t))-H_1(\epsilon,t,\cdot, \mathring{w}(t))\|_{\Qs} +
	\frac{1}{\epsilon} \|H_1(\epsilon,t,\cdot,\mathring{w}(t))-\mathring{H}_1(t,\cdot, \mathring{w}(t))\|_{\Qs}   \nnb  \\
	&\leq C\bigl(\|w\|_{\Li([T_0,t),R^s)},\|\mrw\|_{\Li([T_0,t),H^s)}\bigr)\Bigl(\mathcal{C}^* +\|w(t)\|_{R^s}+\|z(t)\|_{\Qs}+\|\mrw(t)\|_{H^s} \nnb  \\
	&\hspace{1cm}+\int_{T_0}^t(\|w(\tau)\|_{R^s}+\|z(\tau)\|_{\Qs}+\|\mrw(\tau)\|_{H^s})d\tau \Bigr),
}
for $T_0\leq t < T_1$, where in deriving the second inequality, we used \eqref{HFLip}, Taylor's Theorem (in the last variable),
and the calculus inequalities.
By similar arguments, we also
see that the inequality
\begin{align}
&\frac{1}{\epsilon}\|(A_1^i(\epsilon,t,\cdot,w(t))-\mathring{A}^i_1(t,\cdot,\mrw(t)) \|_{\Qs} \notag \\
\leq & C(\|w\|_{\Li([T_0,t),R^s)},\|\mrw\|_{\Li([T_0,t),H^s)})\Bigl(\mathcal{C}^* +\|w(t)\|_{R^s}+\|z(t)\|_{\Qs}+\|\mrw(t)\|_{H^s} \nnb  \\
&+\int_{T_0}^t(\|w(\tau)\|_{R^s}+\|z(\tau)\|_{\Qs}+\|\mrw(\tau)\|_{H^s})d\tau+1\Bigr) ,
\label{FHATD}
\end{align}
holds for $T_0\leq t < T_1$. The estimates \eqref{F1est}, \eqref{HFLip}, and \eqref{E:FHATC}-\eqref{FHATD}
together with the calculus inequalities then imply that
\begin{align} \label{Fh2est}
\|\hat{F}_2(\epsilon,t)\|_{\Qs}\leq C\bigl(\|w\|_{\Li([T_0,t),R^s)},\|\mrw\|_{\Li([T_0,t),H^s)}, &\|z\|_{\Li([T_0,t),R^{s-1})}\bigr)\Bigl(\mathcal{C}^* +\|w(t)\|_{R^s}+\|z(t)\|_{\Qs}+\|\mrw(t)\|_{H^s} \nnb  \\
&+\int_{T_0}^t(\|w(\tau)\|_{R^s}+\|z(\tau)\|_{\Qs}+\|\mrw(\tau)\|_{H^s})d\tau\Bigr)
\end{align}
for $T_0\leq t < T_1$. Furthermore, we see from \eqref{F1est} and \eqref{E:INITIALDATA3} that we can estimate $z$ at $t=T_0$ by
\al{INITIALZEST}{
	\|z(T_0)\|_{\Qs} \leq C(\sigma)\sigma.
}

We can combine the two equations \eqref{E:MODELEQ3a} and \eqref{E:LIMITINGEQa} into
the single system
\al{EQ1COR}{
	&\p{A_1^0 & 0 \\ 0 & \mathring{A}_1^0}\partial_t \p{w \\ \mathring{w}}+\p{A_1^i & 0 \\ 0 & \mathring{A}_1^i} \partial_i \p{ w \\ \mathring{w}} +\frac{1}{\epsilon} \p{C_1^i & 0 \\ 0 & 0 }\del{i} \p{w \\ \mrw}  \nnb \\
	&\hspace{3.0cm}=\frac{1}{t}\p{\mathfrak{A}_1 & 0\\0 & \mathring{\mathfrak{A}}_1} \p{\Pbb_1 & 0 \\ 0 & \mathbb{P}_1}\p{w \\ \mrw}+\p{H_1\\ \mathring{H}_1 }+\p{F_1\\ \mathring{F}_1- C_1^i\partial_i v},
}
and collect the three equations \eqref{E:PW2}, \eqref{E:MODELEQ4} and \eqref{E:DIFFEQ2}
together to get
\al{WHOLESYS1}{
	A_2^0\del{t} \begin{pmatrix} b\\ y\\ z\end{pmatrix}  +A_2^i\del{i} \begin{pmatrix} b\\ y\\ z\end{pmatrix} +\frac{1}{\epsilon}C_2^i\del{i} \begin{pmatrix} b\\ y\\ z\end{pmatrix}=\frac{1}{t}\mathfrak{A}_2 \Pbb_2 \begin{pmatrix} b\\ y\\ z\end{pmatrix} +H_2+ R_2+F_2,
}
where
\al{WHOLESYS2}{
	A_2^0:=\p{\Pbb_1 \mrA^0_1 \Pbb_1   & 0 & 0 \\   0  & \mrA^0_1 & 0\\  0 &   0 & A^0_1 }, \quad A_2^i:=\p{\Pbb_1 \mrA^i_1 \Pbb_1   & 0 & 0  \\  0   & \mrA^i_1 & 0\\ 0   & 0 & A^i_1 },
}
\al{}{
	C_2^i:=\p{0 & 0 & 0 \\   0  & 0 & 0\\  0   & 0 & C^i_1 }, \quad  \Pbb_2:=\p{ \Pbb_1 & 0 & 0  \\ 0 &   \mathds{1} & 0\\ 0  & 0 & \Pbb_1 }, \quad \mathfrak{A}_2=\p{\Pbb_1\mathfrak{\mrA}_1\Pbb_1 & 0  & 0 \\  -\Pbb_1\mathring{\mathfrak{A}}_1\Pbb_1  & \Pbb_1 \mathring{\mathfrak{A}}_1 \Pbb_1+\mrA^0_1 & 0 \\  0 & 0 & \mathfrak{A}_1},
}
\al{}{
	H_2:=\p{\Pbb_1 \mathring{H}_1 \\ \tilde{H}_2  \\ 0 },  \quad R_2:=\p{0 \\ 0  \\  \hat{R}_2} \AND F_2:=  \p{\Pbb_1\bar{F}_2 \\ \tilde{F}_2  \\ \hat{F}_2}.
}
We remark that due to the projection operator $\Pbb_1$ that appears in the definition \eqref{E:d} of $b$
and in the top row of \eqref{E:WHOLESYS2}, the vector $(b,y,z)^{\mathrm{T}}$ takes values
in the vector space $\Pbb_1\Rbb^{N_1} \times \Rbb^{N_1} \times \Rbb^{N_1}$ and  \eqref{E:WHOLESYS2}
defines a symmetric hyperbolic system, i.e. $A^0_2$ and $A_2^i$ define symmetric linear operators on  $\Pbb_1\Rbb^{N_1} \times \Rbb^{N_1} \times \Rbb^{N_1}$ and $A^0_2$ is non-degenerate.

Setting
\begin{equation*}
\Pbb_3:=\p{0 & 0 & 0 & 0 & 0\\ 0 & 0 & 0 & 0 & 0\\ 0 & 0 & \mathbb{P}_1 & 0 & 0 \\ 0 & 0 & 0 & \mathds{1} & 0\\ 0 & 0 & 0 & 0 & 0 } \AND \Pbb_4=\p{0 & 0 \\ 0 & \mathds{1}},
\end{equation*}
it is then not difficult to verify from the estimates  \eqref{F1est}, \eqref{Fb2est},  \eqref{Ft2est}
and \eqref{Fh2est}, the initial bounds \eqref{E:INITIALDATA3}, \eqref{bidata} and \eqref{E:INITIALZEST},
the relations \eqref{Aringproj}, \eqref{P1H1zerp}, \eqref{Afrcom} and \eqref{Ht2zero}, and the assumptions
on the coefficients $\{$$A^0_1$, $A^i_1$, $\mathring{A}^0_1$, $\mathring{A}^i_1$, $\mathfrak{A}_1$,
$\mathring{\mathfrak{A}}_1$, $H$, $F$$\}$ (see Assumptions \ref{ASS1} and \ref{ASS3}) that the system
consisting of \eqref{E:EQ1COR} and \eqref{E:WHOLESYS1} and the solution $U=(w,\mrw,b,y,z)^{\mathrm{T}}$
satisfy the hypotheses of Theorem \ref{L:BASICMODEL}, and thus, for $\sigma >0$ chosen small enough,
there exists a constant $C>0$, independent of $\epsilon \in (0,\epsilon_0)$ and $T_1 \in (T_0,0)$, such that
\begin{equation} \label{wmrwLinfty}
\norm{(w,\mrw)}_{L^\infty([T_0,T_1)\times \Rbb^3)} \leq \frac{R}{2}
\end{equation}
and
\begin{align} \label{Uproofest}
\|\mrw\|_{L^\infty([T_0,t),L^2)}+&\left(-\int_{T_0}^{t} \frac{1}{\tau} \|\Pbb_1\mrw\|^2_{L^2}\, d\tau\right)^{\frac{1}{2}}+\|(w,\mrw)\|_{M^\infty_{\Pbb_1, s}([T_0,t)\times \Rbb^3)}  \nnb  \\
&+\|(b,y,z)\|_{M^\infty_{\Pbb_2, s-1}([T_0,t)\times \Rbb^3)} -\int_{T_0}^{t}
\frac{1}{\tau} \|\Pbb_3 U\|_{\Qs}\, d\tau  \leq C\sigma
\end{align}
for $T_0 \leq t < T_1$.
This completes the proof since the estimates \eqref{wrmwsupest}-\eqref{E:FINALEST1c} follow immediately from \eqref{wmrwLinfty} and \eqref{Uproofest}.
\end{proof}

\section{Proof of the Main Theorem \ref{T:MAINTHEOREM}}\label{S:MAINPROOF}

\subsection{Transforming the conformal Einstein-Euler equations\label{proof:EEeqn}}
The first step of the proof is to observe that the non-local formulation of the conformal Einstein-Euler equations
given by \eqref{E:REALEQ} can be transformed into the form \eqref{E:MODELEQ3a} analyzed in \S \ref{S:MODEL}
by making the simple change of time coordinate
\al{TIMECHANGE}{
t\mapsto \hat{t}:=-t
}
and the substitutions
\begin{gather}
w(\hat{t},x)=\mathbf{U}(-\hat{t},x), \quad A_1^0(\epsilon,-\hat{t},w)=\mathbf{B}^0(\epsilon,-\hat{t},\mathbf{U}),  \quad A_1^i(\epsilon,\hat{t},w) =-\mathbf{B}^i(\epsilon,-\hat{t},\mathbf{U}),  \quad \mathfrak{A}_1(\epsilon,\hat{t},w) =\mathbf{B}(\epsilon,-\hat{t},\mathbf{U}), \label{singdefA} \\  C_1^i=-\mathbf{C}^i, \quad
\mathbb{P}_1=\mathbf{P}, \quad     H_1(\epsilon,\hat{t},w)=-\mathbf{H}(\epsilon,-\hat{t},\mathbf{U}) \AND F_1(\epsilon,\hat{t},x)=-\mathbf{F}(\epsilon,-\hat{t},x,\mathbf{U},\del{k}\Phi,\del{t}\del{k}\Phi,\del{k}\del{l}\Phi). \label{singdefB}
\end{gather}
With these choices, we  can use the same arguments as in \cite[\S$7$]{Liu2017} to verify that all the structural assumptions listed in \S\ref{S:MODEL} hold. We omit the details.

\subsection{Limit equations\label{proof:limiteq}}
Setting
\begin{equation} \label{Uringdef}
\mathring{\mathbf{U}}=(\mathring{u}^{0\mu}_0, \mathring{w}^{0\mu}_k, \mathring{u}^{0\mu}, \mathring{u}^{ij}_0, \mathring{u}^{ij}_k, \mathring{u}^{ij}, \mathring{u}_0, \mathring{u}_k, \mathring{u}, \delta\mathring{\zeta}, \mathring{z}_i )^\textrm{T},
\end{equation}
the limit equation, see \S \ref{S:MODELerr}, associated to \eqref{E:REALEQ}  on the spacetime region $(T_2,1]\times \Rbb^3$,
$0<T_2 < 1$, is given by
\begin{align}
\mathring{\mathbf{B}}^0\partial_t \mathring{\mathbf{U}}+\mathring{\mathbf{B}}^i\partial_i \mathring{\mathbf{U}} +\mathbf{C}^i\partial_i
\mathbf{V}= & \frac{1}{t}
\mathring{\mathbf{B}} \mathbf{P}
\mathring{\mathbf{U}} + \mathring{\mathbf{H}} + \mathring{\mathbf{F}} \label{E:REALLIMITINGEQUATIONa} &&
\text{in $(T_2,1]\times \Rbb^3$},\\
\mathbf{C}^i\partial_i\mathring{\mathbf{U}} = & 0   \label{E:REALLIMITINGEQUATIONb} &&
\text{in $(T_2,1]\times \Rbb^3$},
\end{align}
where
\al{BLIM}{
\mathbf{\mathring{B}}{}^\mu(t,\mathbf{\mathring{U}}):=\lim_{\epsilon\searrow 0}\mathbf{B}^{\mu}(\epsilon,t,\mathbf{\mathring{U}}), \quad \mathbf{\mathring{B}}(t,\mathbf{\mathring{U}}):=\lim_{\epsilon\searrow 0}\mathbf{B}(\epsilon,t,\mathbf{\mathring{U}}), \quad  \mathbf{\mathring{H}}(t,\mathbf{\mathring{U}}):=\lim_{\epsilon\searrow 0}\mathbf{H}(\epsilon,t,\mathbf{\mathring{U}}),
}
and
\begin{align}
\mathbf{\mathring{F}} := \biggl( & 2  \mathring{E}^{-1}\mathring{\Omega}\delta^{kj}  \delta^\mu_j \mathring{\Phi}_k-2 \sqrt{\frac{\Lambda}{3}} t \mathring{E}^{-1}  \delta^\mu_i \varpi^i + t\delta^\mu_0 \mathring{E}^{-1} \mathring{\Upsilon}
, \Bigl(\frac{1}{2}+\mathring{\Omega}\Bigr)\mathring{E}^{-1}\delta^{kl}\delta^\mu_0\mathring{\Phi}_k
-\delta^{kl} t \mathring{E}^{-1}\delta^\mu_0 \del{0}\mathring{\Phi}_k,\nnb  \\ &
0,  0,  0, 0,
0,0,0,
0,-K^{-1}\frac{1}{2}\left(\frac{3}{\Lambda}\right)^{\frac{3}{2}}\mathring{E}^{-3}\delta^{lk}t \mathring{\Phi}_k  \biggr)^\textrm{T}.\label{E:BH2}
\end{align}
In $\mathbf{\mathring{F}}$, 
$\mathring{\Phi}$ is the Newtonian potential, see \eqref{CPeqn3}, and $\mathring{E}$, $\mathring{\Omega}$ and $\varpi^j$ are defined previously by \eqref{Eringform}, \eqref{Oringdef} and \eqref{e:varpi}, respectively.

We then observe that under the change of time coordinate \eqref{E:TIMECHANGE} and the substitutions
\begin{gather}
\mrw(\hat{t},x)=\mathbf{\mathring{U}}(-\hat{t},x), \quad \mathring{A}_1^0( \hat{t},w)=\mathbf{\mathring{B}}{}^0(-\hat{t},\mathbf{\mathring{U}}), \quad \mathring{A}_1^i( \hat{t},w)=-\mathbf{\mathring{B}}{}^i(-\hat{t},\mathbf{\mathring{U}}),
\quad  \mathfrak{\mathring{A}}_1( \hat{t},w) =\mathbf{\mathring{B}}(-\hat{t},\mathbf{\mathring{U}}), \quad  C_1^i=-\mathbf{C}^i,
\label{singdefc} \\
v(\hat{t},x)=\mathbf{V}(-\hat{t},x), \quad \mathbb{P}_1=\mathbf{P}, \quad \mathring{H}_1( \hat{t},w) =-\mathring{\mathbf{H}}(-\hat{t},\mathbf{\mathring{U}}) \quad \text{and} \quad \mathring{F}_1(\hat{t},x) =-\mathbf{\mathring{F}}(-\hat{t},x), \label{singdefd}
\end{gather}
the limit equation \eqref{E:REALLIMITINGEQUATIONa}-\eqref{E:REALLIMITINGEQUATIONb} transforms into
\begin{align*}
\mathring{A}_1^0\partial_{\hat{t}} \mathring{w}+\mathring{A}_1^i\partial_i \mathring{w}&=\frac{1}{\hat{t}}\mathring{\mathfrak{A}}_1\mathbb{P}_1\mathring{w}-C_1^i\partial_i v+\mathring{H}_1+\mathring{F}_1
&& \mbox{in} \quad[-1, -T_2)\times\mathbb{R}^3, \\
C_1^i\partial_i\mathring{w}&=0 && \mbox{in} \quad[-1, -T_2)\times\mathbb{R}^3,
\end{align*}
which is of the form analyzed in \S \ref{S:MODELerr}, see \eqref{E:LIMITINGEQa}-\eqref{E:LIMITINGEQb} and \eqref{E:HRIN}.
It is also not difficult to verify that the matrices $\mathring{A}_1^i$ and the source term $\mathring{H}_1$
satisfy the Assumption \ref{ASS3}.\eqref{A3a} from \S \ref{S:MODELerr}.

\subsection{Local existence and continuation\label{proof:loccont}}
For fixed $\epsilon \in (0,\epsilon_0)$,  we know from Proposition \ref{t:ulst}, Corollary \ref{T:poes} and Proposition \ref{T:locfull} that for $T_1 \in (0,1)$ chosen close enough to $1$ there exists a unique solution\footnote{Recall that $\mathbb{K}$ is defined in Proposition \ref{t:ulst}. } 
\begin{equation*}
\mathbf{U} \in \bigcap_{\ell=0}^1 C^\ell\bigl( (T_1,1],R^{s-\ell}(\Rbb^3,\mathbb{K})\bigr)
\end{equation*}
to \eqref{E:REALEQ} satisfying the initial condition
\begin{equation*}
\mathbf{U}|_{\Sigma}=
\bigl(u^{0\mu}_0|_{\Sigma}, w^{0\mu}_k|_{\Sigma}, u^{0\mu}|_{\Sigma}, u^{ij}_0|_{\Sigma}, u^{ij}_k|_{\Sigma}, u^{ij}|_{\Sigma}, u_0|_{\Sigma}, u_k|_{\Sigma}, u|_{\Sigma},\delta\zeta|_{\Sigma}, z _i|_{\Sigma} \bigr)^{\textrm{T}},
\end{equation*}
where the initial data is determined from Lemma \ref{L:INITIALTRANSFER} and Proposition \ref{T:locfull}.\eqref{T:locfull2}.
Moreover, we know that this solution can be continued beyond $T_1$ provided that
\begin{equation*}
\sup_{ t\in (T_1,1]} \norm{\mathbf{U}(t)}_{R^s} < \infty.
\end{equation*}

Next, by  Proposition \ref{PEexist},  there exists, for some $T_2 \in (0,1]$, a unique solution $( \mathring{\zeta},\mathring{z}^i,\mathring{\Phi})$ which verifies 
\begin{equation}\label{PEproofsol}
(\delta \mathring{\zeta},\mathring{z}^i,\mathring{\Phi})\in \bigcap_{\ell=0}^1 C^{\ell}\bigl((T_2,1],H^{s-\ell}(\mathbb{R}^3)\bigr)
\times \bigcap_{\ell=0}^1  C^{\ell}((T_2,1],H^{s-\ell}(\mathbb{R}^3,\mathbb{R}^3)) \times \bigcap_{\ell=0}^1  C^{\ell}((T_2,1],R^{s+2-\ell}(\mathbb{R}^3)),
\end{equation}
to the conformal cosmological conformal Poisson-Euler equations, given by \eqref{CPeqn1}-\eqref{CPeqn3}, satisfying
the initial condition
\begin{equation*} \label{PEproofsolid}
(\delta\mathring{\zeta},\mathring{z}_i)|_{\Sigma} = \biggl( 
\ln\Bigl(1+\frac{\delta\breve{\rho}_{\epsilon,\yve}}{\mathring{\mu}(1)}\Bigr),
\mathring{E}^2 \delta_{ij}\breve{z}^j_{\epsilon,\yve}\biggr).
\end{equation*}
Setting
\begin{equation*}\label{E:V}
\mathbf{V}=\left(V^{0 \mu}_0, V^{0\mu}_k, V^{0 \mu}, 0, 0, 0, 0, 0, 0,0,0 \right)^{\textrm{T}},
\end{equation*}
where
\begin{gather}
V^{0 \mu}_0
=-\biggl(\frac{1}{2}+\mathring{\Omega}\biggr) \delta^\mu_0 \mathring{E}^{-1}   \mathring{\Phi} +  t \mathring{E}^{-1}  \del{0} \mathring{\Phi} \delta^\mu_0
,\label{E:V1}\\
V^{0\mu}_k
=-2\mathring{E}^{-1}\mathring{\Omega}\delta^\mu_k \mathring{\Phi} +2 \delta_j^\mu \sqrt{\frac{\Lambda}{3}}t\mathring{E}^{-1}(-\Delta)^{-\frac{1}{2}}\mathfrak{R}_k \varpi^j
,  \label{E:V2}
\\
V^{0 \mu}
=\biggl(\frac{1}{2}+\Omega\biggr)\delta^\mu_0 \mathring{E}^{-1} \mathring{\Phi} 
,   \label{E:V3}
\end{gather}
it follows from Proposition \ref{PEexist} and \eqref{PEproofsol} that $\mathbf{V}$ is well-defined and lies in the
space 
\begin{equation*}
\mathbf{V} \in \bigcap_{\ell=0}^1 C^\ell\bigl( (T_2,1],R^{s+1-\ell}(\Rbb^3,\mathbb{K})\bigr).
\end{equation*}
It can be verified by a direct calculation that the pair $(\mathbf{V}, \mathring{\mathbf{U}})$, where
\begin{align}\label{E:URINGVALUE}
\mathring{\mathbf{U}}=(0,0,0,0,0,0,0,0,0,\delta\mathring{\zeta},\mathring{z}_i),
\end{align}
determines
a solution of the limit equation  \eqref{E:REALLIMITINGEQUATIONa}-\eqref{E:REALLIMITINGEQUATIONb}. Moreover,
by Proposition  \ref{PEexist}, it is clear that this solution can be continued past $T_2$ provided that
\begin{equation*}
\sup_{ t\in (T_2,1]} \norm{\mathring{\textbf{U}}(t)}_{R^s} < \infty.
\end{equation*}

\subsection{Global existence and error estimates}
To complete the proof, we use the a priori estimates from Theorem \ref{T:MAINMODELTHEOREM}
to show that the solutions $\mathbf{U}$ and $(\mathbf{V},\mathring{\mathbf{U}})$ to the reduced conformal Einstein-Euler equations and the corresponding limit equation, respectively, can be continued all the way to $t=0$, i.e. $T_1=T_2=0$, with uniform bounds and an error estimate. In order to apply Theorem \ref{T:MAINMODELTHEOREM}, we need to verify that the estimates
\eqref{F1est}-\eqref{HFLip} hold for the solutions  $\mathbf{U}$ and $(\mathbf{V},\mathring{\mathbf{U}})$. We begin by observing, via
routine calculations, that the components of $\del{t} \mathbf{V}$ are given by
\begin{align}
\partial_t V^{0\mu}_0= &  \left(1-2\mathring{\Omega}-\frac{1}{2}\right)\mathring{E}^{-1}\delta^\mu_0\partial_t  \mathring{\Phi} + t \mathring{E}^{-1} \delta^\mu_0\partial^2_t \mathring{\Phi} - \biggl( \del{t}\mathring{\Omega}-\bigl(\frac{1}{2}+\mathring{\Omega}\bigr)\frac{\mathring{\Omega}}{t}\biggr)  \mathring{E}^{-1} \delta^\mu_0\mathring{\Phi}, \label{E:DTV1}  \\
\partial_t V^{0\mu}= &  \biggl( \del{t}\mathring{\Omega} -\bigl(\frac{1}{2}+\mathring{\Omega}\bigr)\frac{\Omega}{t} \biggr) \delta^\mu_0\mathring{E}^{-1}\mathring{\Phi}+\biggl(\frac{1}{2}+ \mathring{\Omega}\biggr)\delta^\mu_0\mathring{E}^{-1}\del{t}\mathring{\Phi}  \label{E:DTV2} \\
\intertext{and}
\partial_t V^{0\mu}_k=  &2 \mathring{E}^2\sqrt{\frac{\Lambda}{3}}\delta^\mu_j   
(-\Delta)^{-\frac{1}{2}}\mathfrak{R}_k \mathring{E}^{-3}\bigl((1- \mathring{\Omega})\varpi^j+t\partial_t \varpi^j\bigr)   +  2\mathring{E}^{-1}\delta^\mu_k\biggl(\frac{\mathring{\Omega}^2}{t}-\del{t}\mathring{\Omega}\biggr) \mathring{\Phi} -2\mathring{E}^{-1}\mathring{\Omega}\delta^\mu_k\del{t}\mathring{\Phi} \label{E:DTV5}.
\end{align}
We further compute
\begin{equation*} \label{PVt}
\frac{1}{t}\mathbf{P}\mathbf{V}=\left(\frac{1}{2t}(V^{0 \mu}_0+V^{0 \mu}), \frac{1}{t}V^{0\mu}_i, \frac{1}{2t}(V^{0 \mu}_0+V^{0 \mu}), 0, 0, 0, 0, 0, 0,0,0 \right)^\mathrm{T},
\end{equation*}
where the components are given by
\begin{gather}
\frac{1}{2t} (V^{0 \mu}_0+V^{0 \mu})=
\frac{1}{2} \delta^\mu_0 \mathring{E}^{-1}  \del{0} \mathring{\Phi} , \label{E:TIV1}\\
\frac{1}{t}V_k^{0 \mu}= 
-2\mathring{E}^{-1}\frac{\mathring{\Omega}}{t}\delta^\mu_k \mathring{\Phi} +2\delta_j^\mu \mathring{E}^{-1}\sqrt{\frac{\Lambda}{3}}(-\Delta)^{-\frac{1}{2}}\mathfrak{R}_k \varpi^j \label{E:TIV2}.
\end{gather}
Then by \eqref{E:V1}-\eqref{E:V3} and \eqref{E:DTV1}-\eqref{E:TIV2}, it is clear that the estimate
\al{NORMVEST}{
&\|\mathbf{V}(t)\|_{R^{s+1}}+\|t^{-1}\mathbf{P}\mathbf{V}(t)\|_{R^{s+1}}+\|\del{t}\mathbf{V}(t)\|_{R^s}  \nnb  \\
&\hspace{1.5cm} \leq 	\|\del{t}\mathring{\Phi}(t)\|_{R^{s+1}}+\|t\partial_t^2\mathring{\Phi}(t)\|_{R^{s}}+\|\mathring{\Phi}(t)\|_{R^{s+1}}+\|(-\Delta)^{-\frac{1}{2}}\mathfrak{R}_k\varpi^j(t)\|_{R^{s+1}}+\|t\del{t}(-\Delta)^{-\frac{1}{2}}\mathfrak{R}_k\varpi^j(t)\|_{R^{s}} \nnb  \\
&\hspace{1.5cm} \leq
C\bigl(\|\delta\mathring{\zeta}\|_{L^\infty((t,1],H^s)},\|\mathring{z}^j\|_{L^\infty((t,1],H^s)}\bigr)\Bigl( \|\delta\mathring{\zeta}(t)\|_{R^s}+\|\mathring{z}^j(t)\|_{H^s} +\|\delta\breve{\rho}\|_{L^{\frac{6}{5}}\cap H^s}+\int_t^1\|\mathring{z}^k(\tau)\|_{H^s} d\tau \Bigr) \nnb  \\
&\hspace{1.5cm} \leq
C\bigl(K_4(t)\bigr)\Bigl( \|\mathring{\mathbf{U}}(t)\|_{H^s} +\|\breve{\xi}_\epsilon\|_{s} +\int_t^1 \|\mathring{\mathbf{U}}(\tau)\|_{H^s} d\tau \Bigr),
}
where 
\begin{equation*}
K_4(t)=\|\mathbf{\mathring{U}}\|_{\Li((t,1],H^s)}+\|\mathbf{U}\|_{\Li(((t,1],R^s))},
\end{equation*}
follows from the estimates \eqref{e:phine1}-\eqref{e:tdt2ph}. From similar reasoning and the embedding $H^s \hookrightarrow R^s$, it is also not difficult,
using \eqref{e:Ups1} and \eqref{e:Ups2} to estimate $\mathring{\Upsilon}$ and $\del{t}\mathring{\Upsilon}$,
to verify that $\mathring{\mathbf{F}}$, defined by \eqref{E:BH2},  satisfies the estimate
\begin{align}\label{E:TDTFCHE}
\|\mathring{\mathbf{F}}(t)\|_{H^s}+\|t\del{t}\mathring{\mathbf{F}}(t)\|_{\Rs}
\leq &
C\bigl(\|\delta\mathring{\zeta}\|_{L^\infty((t,1],H^s)},\|\mathring{z}^j\|_{L^\infty((t,1],H^s)}\bigr)\Bigl( \|\delta\mathring{\zeta}(t)\|_{R^s}+\|\mathring{z}^j(t)\|_{H^s} +\|\delta\breve{\rho} \|_{L^{\frac{6}{5}}\cap H^s} \nnb  \\
& +\int_t^1 \|\mathring{z}^k (\tau) \|_{H^s} d\tau  \Bigr)
\end{align}
for $T_2 < t \leq 1$.
Furthermore, we see from the definition of $\mathbf{F}$, see \eqref{E:REALEQc}, the estimates  \eqref{e:ddphi3}-\eqref{e:dtphi3}, and
the calculus inequalities that $\mathbf{F}$ is bounded by
\begin{align} \label{bfFest1}
\|\mathbf{F}(t)\|_{R^s}\leq & C\bigl(\|\mathbf{U}\|_{\Li((t,1],R^s)}, \| \Phi^\mu_k\|_{\Li((t,1],R^s)}\bigr)(\|\mathbf{U}(t)\|_{R^s}+\|D \Phi^\mu_k(t)\|_{R^s}+\|  \Phi^\mu_k(t)\|_{R^s}+\|t\del{t}\Phi^\mu_k(t)\|_{R^s}) \nnb  \\
\leq & C\bigl(\|\mathbf{U}\|_{\Li((t,1],R^s)}\bigr)\Bigl(\|\breve{\xi}_\epsilon \|_{s} +\|\mathbf{U}(t)\|_{R^s} +\int_t^1 (\| u^{0i}_{\epsilon,\yve}(\tau)\|_{ R^s }+\| z_{l,\epsilon,\yve} (\tau)\|_{ R^s } ) d\tau \Bigr)
\end{align}
for $T_1 < t < 1$.
Together, \eqref{E:NORMVEST}, \eqref{E:TDTFCHE} and \eqref{bfFest1} show that source terms
$\{F_1,\mathring{F}_1,v\}$, as defined by \eqref{singdefB} and \eqref{singdefd},  satisfy the
estimates \eqref{F1est}
from Theorem \ref{T:MAINMODELTHEOREM} for times $-1 \leq  \hat{t} < -T_3$, where
\begin{equation*}
T_3 = \max\{T_1,T_2\}.
\end{equation*}

Next, we verify that the Lipschitz estimates \eqref{AiLip}-\eqref{HFLip} from Theorem \ref{T:MAINMODELTHEOREM} are satisfied. We start by noticing,
with the help of \eqref{E:EINBk},  \eqref{E:BkREMAINDER} and \eqref{E:URINGVALUE}, that
\begin{align*}
\tilde{B}{}^i (\epsilon,t,\mathring{\mathbf{U}})&= 0,\\
B^i(\epsilon,t,\mathring{\mathbf{U}}) &=\sqrt{\frac{3}{\Lambda}}\p{\mathring{z}^i & E^{-2} \delta^{im} \\ E^{-2}\delta^{il} & K^{-1} E^{-2} \delta^{lm}\mathring{z}^i
}+\epsilon^2\mathscr{\bar{S}}^i(\epsilon,t,\mathring{\mathbf{U}})
\intertext{and}
B^i(0,t,\mathring{\mathbf{U}})&=\sqrt{\frac{3}{\Lambda}}\p{\mathring{z}^i & \mathring{E}^{-2} \delta^{im} \\ \mathring{E}^{-2}\delta^{il} & K^{-1} \mathring{E}^{-2} \delta^{lm}\mathring{z}^i }.
\end{align*}
From the above expressions, the expansion \eqref{E:EOEXP}, and the calculus inequalities, we then obtain the estimate
\begin{equation} \label{Blip}
\|\mathbf{B}^i(\epsilon,t,\mathring{\mathbf{U}}) -\mathring{\mathbf{B}}^i(t,\mathring{\mathbf{U}})\|_{\Rs}
\leq \epsilon C(\|\mathring{{\mathbf{U}}}\|_{\Li((t,1],R^s)}), \quad T_3 < t \leq 1.
\end{equation}
Next, using \eqref{E:EING1}, \eqref{E:EING3},  \eqref{E:EING2}, \eqref{E:G}  and \eqref{E:URINGVALUE}, we can express the components of $\mathbf{H}(\epsilon,t,\mathring{\mathbf{U}})$, see \eqref{E:REALEQc},  as follows:
\begin{gather*}
\tilde{G}_1(\epsilon,t,\mathring{\mathbf{U}})=\bigl( \epsilon\mathscr{S}^\mu(\epsilon, t, \mathring{\mathbf{U}})  ,0,0 \bigr)^\mathrm{T}, \quad
\tilde{G}_2(\epsilon,t,\mathring{\mathbf{U}})= \bigl( \epsilon\mathscr{S}^{ij}(\epsilon, t, \mathring{\mathbf{U}})
,0,0\bigr)^\mathrm{T},  \\
\tilde{G}_3(\epsilon,t,\mathring{\mathbf{U}})= \bigl(\epsilon \mathscr{S} (\epsilon, t, \mathring{\mathbf{U}})
,0,0\bigr)^\mathrm{T}, \AND
G(\epsilon,t,\mathring{\mathbf{U}})=(0,0)^\mathrm{T}  ,  
\end{gather*}
where $\mathscr{S}^\mu$, $\mathscr{S}^{ij}$ and $\mathscr{S}$ all vanish for $\mathring{\mathbf{U}}=0$.
It follows immediately from these expressions and the definitions \eqref{E:REALEQc} and \eqref{E:BLIM} that
\als
{
\mathring{\mathbf{H}}(t,\mathring{\mathbf{U}})=\biggl(0,0,0,0,0,0,0,0,0,0,0\biggr)^T,
}
and,  with the help of the calculus inequalities and \eqref{E:CHI1}-\eqref{E:EOEXP}, that
\al{HSUBH}{
\|\mathbf{H}(\epsilon,t,\mathring{\mathbf{U}})-\mathring{\mathbf{H}}(t,\mathring{\mathbf{U}})\|_{\Rs} \leq \epsilon C(\|\mathring{\mathbf{U}}\|_{\Li((t,1],R^s)}) \|\mathring{\mathbf{U}}\|_{\Rs}, \quad T_3 < t \leq 1.
}

To proceed, we define
\al{e:zdef0}{
\mathbf{Z}=\frac{1}{\epsilon}(\mathbf{U}-\mathring{\mathbf{U}}-\epsilon \mathbf{V}),
}
and set
\begin{equation*} \label{E:zreal}
z(\hat{t},x) = \mathbf{Z}(-\hat{t},x).
\end{equation*}
In view of the definitions \eqref{E:REALEQc} and \eqref{E:BH2}, it is not difficult to verify that the
inequality
\begin{align}
\|\mathbf{F}&(\epsilon,t,\cdot)-\mathring{\mathbf{F}}(t,\cdot)\|_{\Rs}
\leq  C\bigl(\|\mathbf{U}\|_{\Li((t,1],R^s)}\bigr)\Bigl(\epsilon\|\mathbf{U}(t)\|_{\Rs}+\epsilon\| \Phi_k^\mu(t)\|_{\Rs}+\epsilon \| D\Phi_k^\mu(t)\|_{\Rs}+\epsilon\|\del{0} \Phi_k^\mu(t)\|_{\Rs} \nnb \\
& +\|E^{-1}\Upsilon(t)-\mathring{E}^{-1} \mathring{\Upsilon}(t)\|_{R^{s-1}} +\epsilon \|\mathbf{Z}(t)\|_{\Rs}+\epsilon \|\mathbf{V}(t)\|_{\Rs}+\| \Phi_k^0(t)- \mathring{\Phi}_k(t)\|_{\Rs}+\|\del{0} ( \Phi_k^0(t)- \mathring{\Phi}_k(t))\|_{\Rs}\nnb  \\
&+\|\rho z^j(t)-\mathring{\rho}\mathring{z}^j(t)\|_{\Rs}\Bigr)  \nnb \\
\leq & C\bigl(\|\mathbf{U}\|_{\Li((t,1],R^s)}\bigr)\Bigl( \epsilon\|\breve{\xi}_\epsilon\|_{s}+ \epsilon
\int_t^1 (\| u^{0i}_{\epsilon,\yve}(\tau)\|_{ R^s }+\| z_{l,\epsilon,\yve} (\tau)\|_{ R^s } ) d\tau +\epsilon\|\mathbf{U}(t)\|_{R^s}  +\epsilon \|\mathbf{Z}(t)\|_{\Rs}+\epsilon \|\mathbf{\mathring{U}}(t)\|_{H^s} \nnb  \\
& +\epsilon \int_t^1 \|\mathring{\mathbf{U}}(\tau)\|_{H^s} d\tau  +\| \Phi_k^0(t)- \mathring{\Phi}_k(t)\|_{\Rs} +
\| \Upsilon(t)-  \mathring{\Upsilon}(t)\|_{R^{s-1}}
 +\|\del{0}\Phi_k^0(t)- \del{0}\mathring{\Phi}_k(t)\|_{\Rs} \, d\tau\Bigr), \label{FlipA}
\end{align}
follows from  \eqref{E:CHI1}-\eqref{E:EOEXP}, the estimates \eqref{e:ddphi3}-\eqref{e:dtphi3} and
\eqref{E:NORMVEST}, and the calculus inequalities, and holds for $T_3 < t \leq 1$. To complete the Lipschitz estimate for $\mathbf{F}$, we require the
estimates from the following lemma, which are an extension of the estimates from Propositions \ref{PEexist} and \ref{T:phiex}.

\begin{lemma}\label{T:diffphi}
The estimates
\begin{gather*}
\|\Phi^0_k-\mathring{\Phi}_k\|_{\Rs} +  \|\Upsilon-\mathring{\Upsilon}\|_{R^{s-1}}	\leq  \epsilon C(K_4) \Bigl( \|\breve{\xi}_\epsilon\|_s  + \int_t^1(\|\mathbf{Z}(\tau)\|_{\Rs}+\|\mathbf{U}(\tau)\|_{R^s}+\|\mathbf{\mathring{U}}(\tau)\|_{H^s}) d\tau \Bigr) \\
\intertext{and}
\|\del{t}(\Phi^0_k- \mathring{\Phi}_k)\|_{\Rs}\leq   \epsilon C(K_4)\Bigl(  \|\breve{\xi}_\epsilon\|_s+ \|\mathbf{Z}\|_{\Rs}+\|\mathbf{\mathring{U}}\|_{\Hs}+\|\mathbf{U}\|_{\Rs}+\int_t^1\|\mathbf{\mathring{U}}(\tau)\|_{H^s} d\tau \Bigr)	
\end{gather*}
hold for $t\in (T_3,1]$.
\end{lemma}
\begin{proof}
Noting that 
\begin{align*}
\sqrt{|\underline{\bar{g}}|} e^\zeta z^l -\sqrt{\frac{3}{\Lambda}}\mathring{E}^3 e^{\mathring{\zeta}} \mathring{z}^l=\sqrt{\frac{3}{\Lambda}}\mathring{E}^3\bigl[e^\zeta(z^l-\mathring{z}^l)+e^{\mathring{\mu}}\mathring{z}^le^{\delta\mathring{\zeta}}\bigl(e^{\delta\zeta-\delta\mathring{\zeta}}-1\bigr)\bigr]+\epsilon\mathscr{S}^l(\epsilon,t,u^{\mu\nu},u,\delta\zeta,z_j ),
\end{align*}
where $\mathscr{S}^l(\epsilon,t,0,0,\delta\zeta,0)=0$,
we have by \eqref{E:PTPHI1b}-\eqref{e:Npoev} and \eqref{E:POTEQ}-\eqref{E:UPses1} that
\begin{align}
\del{t} \bigl(\Phi_k^0-  \mathring{\Phi}_k\bigr)
= & \frac{\Lambda}{3} \partial_k\del{l}(\Delta-\epsilon^2 \beta)^{-1}    \left[  \Delta^{-1}(\Delta-\epsilon^2 \beta)
\left(\sqrt{\frac{3}{\Lambda}}\mathring{E}^3 e^{\mathring{\zeta}} \mathring{z}^l\right)-    \bigl( \sqrt{|\underline{\bar{g}}|} e^\zeta z^l\bigr)  \right]  \nnb  \\
= & -\sqrt{\frac{\Lambda}{3}}\mathring{E}^3 \partial_k\del{l}(\Delta-\epsilon^2 \beta)^{-1} \left[e^\zeta(z^l-\mathring{z}^l)+e^{\mathring{\mu}}\mathring{z}^le^{\delta\mathring{\zeta}}\bigl(e^{\delta\zeta-\delta\mathring{\zeta}}-1\bigr) +\epsilon\mathscr{S}^l(\epsilon,t,u^{\mu\nu},u,\delta\zeta,z^j) \right] \nnb  \\
& +\sqrt{\frac{\Lambda}{3}}\epsilon^2\beta (\Delta-\epsilon^2 \beta)^{-1}  \mathfrak{R}_k\mathfrak{R}_l \mathring{E}^3 e^{\mathring{\zeta}} \mathring{z}^l   \label{e:dtdiffph}
\end{align}
and
\begin{align}
\del{t} \bigl(\Upsilon-  \mathring{\Upsilon}\bigr)
= & \frac{\Lambda}{3} \epsilon \beta\del{l}(\Delta-\epsilon^2 \beta)^{-1}    \left[
\left(\sqrt{\frac{3}{\Lambda}}\mathring{E}^3 e^{\mathring{\zeta}} \mathring{z}^l\right)-    \bigl( \sqrt{|\underline{\bar{g}}|} e^\zeta z^l\bigr)  \right]  \nnb  \\
= & -\sqrt{\frac{\Lambda}{3}}\mathring{E}^3 \partial_k\del{l}(\Delta-\epsilon^2 \beta)^{-1} \left[e^\zeta(z^l-\mathring{z}^l)+e^{\mathring{\mu}}\mathring{z}^le^{\delta\mathring{\zeta}}\bigl(e^{\delta\zeta-\delta\mathring{\zeta}}-1\bigr) +\epsilon\mathscr{S}^l(\epsilon,t,u^{\mu\nu},u,\delta\zeta,z^j) \right].
\label{e:dtdiffup}
\end{align}
We also observe that
\begin{align}\label{e:1}
\|\epsilon^2\beta(\Delta-\epsilon^2\beta)^{-1} \bigl(  e^{\mathring{\zeta}} \mathring{z}^l\bigr) \|_{\Rs}\leq \epsilon \sqrt{\beta} \|(\Delta-\epsilon^2\beta)^{-\frac{1}{2}} \bigl(  e^{\mathring{\zeta}} \mathring{z}^l\bigr) \|_{\Rs}  \leq \epsilon C \| e^{\mathring{\zeta}} \mathring{z}^l\|_{H^{s-2}}
\end{align}
follows from the inequality \eqref{E:HQR} and an application of Proposition  \ref{T:genYu}.
Then applying the $R^{s-1}$ norm on both sides of \eqref{e:dtdiffph}, we find, with the help of Theorem \ref{T:rietran}, Propositions \ref{T:genYu} and  \ref{t:ddel},  and \eqref{e:sYupq}, \eqref{E:NORMVEST}  and \eqref{e:1}, that the inequality 
\begin{align*}
\|\del{t} \bigl(\Phi_k^0-  \mathring{\Phi}_k\bigr)\|_{\Rs}
\leq & \epsilon C(K_4)\Bigl(  \|\breve{\xi}_\epsilon\|_s+ \|\mathbf{Z}\|_{\Rs}+\|\mathbf{\mathring{U}}\|_{\Hs}+\|\mathbf{U}\|_{\Rs}+\int_t^1\|\mathbf{\mathring{U}}(\tau)\|_{H^s} d\tau \Bigr). 		
\end{align*}
holds for $t\in (T_3,1]$.
This concludes the proof of the second estimate in the statement of the lemma.

Turning to the first estimate, we start by estimating the initial values
 $\|(\Phi^0_k -\mathring{\Phi}_k)|_{\Sigma}\|_{\Rs}$ and $\|(\Upsilon-\mathring{\Upsilon})|_{\Sigma}\|_{R^{s-1}}$. From there, the 
desired estimates for $\|\Phi^0_k-\mathring{\Phi}_k\|_{\Rs} $ and $\|\Upsilon-\mathring{\Upsilon}\|_{R^{s-1}}$ follow
from integrating \eqref{e:dtdiffph} and \eqref{e:dtdiffup} in time and then applying the $R^{s-1}$ norm. 

Using the expansion
\begin{align}
\sqrt{|\underline{\bar{g}}|}e^\zeta \underline{\bar{v}^0}=E^3 e^\zeta+ \epsilon \mathscr{T}_1(\epsilon, t,u^{\mu\nu},u,\delta\zeta)+ \epsilon^2 \mathscr{T}_2(\epsilon, t,u^{\mu\nu},u,\delta\zeta,z_j), \label{e:expsqgev}
\end{align}
where $\mathscr{T}_1(\epsilon, t,0,0,\delta\zeta)=0$ and $\mathscr{T}_2(\epsilon, t,0,0,0,0)=0$, along with
the identity
\begin{align*}
	 \Delta^{-1} ( \mathring{E}^3 e^{\mathring{\zeta}}-\mathring{E}^3 e^{\mathring{\zeta}_H})=(\Delta-\epsilon^2 \beta)^{-1}  ( \mathring{E}^3 e^{\mathring{\zeta}}-\mathring{E}^3 e^{\mathring{\zeta}_H})-\epsilon^2 \beta\Delta^{-1}(\Delta-\epsilon^2\beta)^{-1}   ( \mathring{E}^3 e^{\mathring{\zeta}}-\mathring{E}^3 e^{\mathring{\zeta}_H}),
\end{align*}
and \eqref{E:CHI1}, \eqref{E:EOEXP}, \eqref{E:DELRHO} to expand $\delta\rho$, and \eqref{e:rhodiff},
we derive the following expansion for $\Phi^0_k-\mathring{\Phi}_k$:
\begin{align}
\Phi^0_k-\mathring{\Phi}_k
=&\frac{\Lambda}{3} \del{k} \bigl[ ( \Delta-\epsilon^2\beta)^{-1}   \bigl( \sqrt{|\underline{\bar{g}}|}e^{\zeta} \underline{\bar{v}^0} -E^3 e^{\zeta_H} \bigr)-  \Delta^{-1} ( \mathring{E}^3 e^{\mathring{\zeta}}-\mathring{E}^3 e^{\mathring{\zeta}_H}) \bigr] \nnb  \\
=&\frac{\Lambda}{3} \del{k}  ( \Delta-\epsilon^2\beta)^{-1} \bigl(E^3 (e^\zeta-e^{\zeta_H})-   \mathring{E}^3 (e^{\mathring{\zeta}}-e^{\mathring{\zeta}_H})\bigr)+ \frac{\Lambda}{3}  \del{k}  ( \Delta-\epsilon^2\beta)^{-1} \epsilon \mathscr{T}_1(\epsilon, t,u^{\mu\nu},u,\delta\zeta) \nnb  \\
& + \frac{\Lambda}{3}  \del{k}  ( \Delta-\epsilon^2\beta)^{-1} \epsilon^2 \mathscr{T}_2(\epsilon, t,u^{\mu\nu},u,\delta\zeta,z_j) +\frac{\Lambda}{3} \epsilon^2\beta( \Delta-\epsilon^2\beta)^{-1}\mathfrak{R}_k (-\Delta)^{-\frac{1}{2}}(\mathring{E}^3 e^{\mathring{\zeta}}-\mathring{E}^3 e^{\mathring{\zeta}_H})  \nnb  \\
=&\frac{\Lambda E^3}{3t^3} \del{k}  ( \Delta-\epsilon^2\beta)^{-1}  \bigl( \delta\rho- \delta\mathring{\rho}\bigr)+ \frac{\Lambda}{3} \del{k}  ( \Delta-\epsilon^2\beta)^{-1} \epsilon [ \mathscr{T}_1(\epsilon, t,u^{\mu\nu},u,\delta\zeta)+\mathscr{T}_3(\epsilon,t,\delta\mathring{\zeta})+\epsilon \mathscr{T}_4(\epsilon, t,\delta\zeta)] \nnb  \\
& + \frac{\Lambda}{3} \del{k}  ( \Delta-\epsilon^2\beta)^{-1} \epsilon^2 \mathscr{T}_2(\epsilon, t,u^{\mu\nu},u,\delta\zeta,z_j) +\frac{\Lambda \mathring{E}^3}{3} \epsilon^2\beta( \Delta-\epsilon^2\beta)^{-1}\mathfrak{R}_k (-\Delta)^{-\frac{1}{2}}(  e^{\mathring{\zeta}}- e^{\mathring{\zeta}_H}). \label{e:diffph}
\end{align}
By similar arguments, we also see that
\begin{align}
\Upsilon-  \mathring{\Upsilon}
= &\epsilon \beta \frac{\Lambda}{3 t^3} E^3 (\Delta-\epsilon^2 \beta)^{-1}  \bigl(\delta\rho-\delta\mathring{\rho} \bigr)+ \epsilon^2 \beta \frac{\Lambda}{3 }  (\Delta-\epsilon^2 \beta)^{-1} \bigl[\mathscr{T}_1(\epsilon, t,u^{\mu\nu},u,\delta\zeta) +\mathscr{T}_3(\epsilon,t,\delta\mathring{\zeta})+ \epsilon \mathscr{T}_4(\epsilon,t,\delta \zeta)\bigr] \nnb  \\
& + \epsilon^3 \beta \frac{\Lambda}{3 }  (\Delta-\epsilon^2 \beta)^{-1}  \mathscr{T}_2(\epsilon, t,u^{\mu\nu},u,\delta\zeta,z_j), \label{e:diffup}
\end{align}
where $\mathscr{T}_3(\epsilon,t,0)=\mathscr{T}_4(\epsilon,t,0)=0$.
Since $\delta\rho|_{\Sigma}=\delta\mathring{\rho}|_{\Sigma}=\delta\breve{\rho}$, we have, initially,
\begin{align*}
	\bigl[\frac{\Lambda E^3}{3t^3} \del{k}  ( \Delta-\epsilon^2\beta)^{-1}  \bigl( \delta\rho- \delta\mathring{\rho}\bigr)\bigr]\Bigr|_{\Sigma}=\bigl[\frac{\Lambda E^3}{3t^3} \epsilon \beta ( \Delta-\epsilon^2\beta)^{-1}  \bigl( \delta\rho- \delta\mathring{\rho}\bigr)\bigr]\Bigr|_{\Sigma}=0.
\end{align*}
Substituting this into \eqref{e:diffph}, we see that the estimate
\begin{align}\label{e:idifp}
\|(\Phi^0_k -\mathring{\Phi}_k)|_{\Sigma}\|_{\Rs}
\lesssim   & \|  \mathscr{T}_1(\epsilon, t,u^{\mu\nu},u,\delta\zeta)|_{\Sigma}\|_{\Rs} +\epsilon\|  \mathscr{T}_3(\epsilon, t,\delta\mathring{\zeta})|_{\Sigma}\|_{H^{s-2}}  + \| \epsilon  \mathscr{T}_2(\epsilon,t,u^{\mu\nu},u,\delta\zeta,z_j)|_{\Sigma}\|_{\Rs} \nnb  \\
&  + \epsilon\|\mathscr{T}_4(\epsilon, t, \delta\zeta)|_{\Sigma}\|_{\Rs}+ \epsilon \|\delta\mathring{\zeta}|_{\Sigma} \|_{L^\frac{6}{5}\cap H^{s-2}}  \nnb  \\
\lesssim & \|u^{\mu\nu}|_{\Sigma}\|_{R^{s-1}}+\|u|_{\Sigma}\|_{R^{s-1}}+
\epsilon\|z_j|_{\Sigma}\|_{R^{s-1}}+\epsilon\|\delta\zeta|_{\Sigma}\|_{R^{s-1}} + \epsilon \|\delta\mathring{\zeta}|_{\Sigma} \|_{L^\frac{6}{5}\cap H^{s-2}}   \nnb  \\
\lesssim &\epsilon \|\breve{\xi}_\epsilon\|_s
\end{align}
follows from  \eqref{E:yuest0} and an application of Theorems \ref{e:inithm}, \ref{T:riepot}  and \ref{T:rietran}, and Propositions \ref{T:genYu} and \ref{t:ddel}. 
Using similar arguments, it is also not difficult to verify  the inequality
\begin{align}\label{e:idifu}
\|(\Upsilon-\mathring{\Upsilon})|_{\Sigma}\|_{R^{s-1}} \lesssim \epsilon \| \breve{\xi}_\epsilon \|_s.
\end{align}

Integrating \eqref{e:dtdiffph} in time and applying the $\|\cdot\|_{\Rs}$ to the result, we see,
after recalling the definitions  \eqref{e:phij}, \eqref{PEexist1},  \eqref{E:POTENTIAL} and  \eqref{E:e:zdef0}, and using Proposition \ref{T:genYu},
\eqref{e:halfYu} and  \eqref{e:sYupq} to estimate terms involving the Yukawa potential operators,
that inequality
\begin{align*}
\|\Phi^0_k-\mathring{\Phi}_k\|_{\Rs} \leq & \|(\Phi^0_k-\mathring{\Phi}_k)|_{\Sigma}\|_{\Rs} +\int_t^1 \Bigl(\| e^\zeta(z^l-\mathring{z}^l)(\tau)\|_{\Rs}+\|e^{\mathring{\mu}}\mathring{z}^le^{\delta\mathring{\zeta}}\bigl(e^{\delta\zeta-\delta\mathring{\zeta}}-1\bigr)(\tau) \|_{\Rs}\nnb  \\ &+\|\epsilon\mathscr{S}^l(\epsilon,\tau,u^{\mu\nu},u,\delta\zeta,z_j )\|_{\Rs} +\|\epsilon^2\beta(\Delta-\epsilon^2\beta)^{-1} \bigl( e^{\mathring{\zeta}} \mathring{z}^l\bigr)(\tau) \|_{\Rs}  \Bigr) d\tau \nnb  \\
\leq & \epsilon C \|\breve{\xi}_\epsilon\|_s  +\epsilon C(K_4 )\int_t^1(\|\mathbf{Z}(\tau)\|_{\Rs}+\|\mathbf{V}(\tau)\|_{\Rs}+\|\mathbf{U}(\tau)\|_{\Rs}+\|\mathbf{\mathring{U}}(\tau)\|_{\Hs}) d\tau \nnb  \\
\leq & \epsilon C(K_4) \Bigl( \|\breve{\xi}_\epsilon\|_s  + \int_t^1(\|\mathbf{Z}(\tau)\|_{\Rs}+\|\mathbf{U}(\tau)\|_{R^s}+\|\mathbf{\mathring{U}}(\tau)\|_{H^s}) d\tau \Bigr)
\end{align*}
is a direct consequence of the expansion \eqref{e:expsqgev} and the estimates \eqref{E:NORMVEST} and  \eqref{e:idifp}. Similar arguments can also
be used to verify the estimate
\begin{align*}
\|\Upsilon-\mathring{\Upsilon}\|_{R^{s-1}}	\leq & \epsilon C(K_4) \Bigl( \|\breve{\xi}_\epsilon\|_s  + \int_t^1(\|\mathbf{Z}(\tau)\|_{\Rs}+\|\mathbf{U}(\tau)\|_{R^s}+\|\mathbf{\mathring{U}}(\tau)\|_{H^s}) d\tau \Bigr).
\end{align*}
Combining the above two estimates together yields the first estimate of the lemma and thus completes the proof.
\end{proof}

From the estimate \eqref{FlipA} and Lemma \ref{T:diffphi}, it is clear that
\al{FSUBF2}{
\|\mathbf{F}(\epsilon,t,\cdot)-\mathring{\mathbf{F}}(t,\cdot)\|_{\Rs}
\leq  & \epsilon C(K_4) \Bigl( \|\breve{\xi}_\epsilon\|_s +\|\mathbf{U}(t)\|_{R^s}+\|\mathbf{Z}(t)\|_{\Rs}+\|\mathring{\mathbf{U}}(t)\|_{H^s} \nnb  \\
& + \int_t^1(\|\mathbf{Z}(\tau)\|_{\Rs}+\|\mathbf{U}(\tau)\|_{\Rs}+\|\mathbf{\mathring{U}}(\tau)\|_{\Hs}) d\tau \Bigr), \quad T_3 < t \leq 1.
}
Together, \eqref{Blip}, \eqref{E:HSUBH} and \eqref{E:FSUBF2} show
that the Lipschitz estimates \eqref{AiLip}-\eqref{HFLip} are satisfied.

The final condition that we need to verify in order to use Theorem \ref{T:MAINMODELTHEOREM} is the bound \eqref{E:INITIALDATA3} on
the initial data. To see that this holds, we observe that the estimate
$\|(\mathbf{U}-\mathbf{\mathring{U}})|_{\Sigma}\|_{R^s}\lesssim \epsilon  \|\breve{\xi}_\epsilon\|_s$
follows directly from \eqref{E:URINGVALUE}, Theorem \ref{e:inithm},  Proposition \ref{T:locfull}.\eqref{T:locfull2}, and the
expansion $\delta\zeta|_{\Sigma}=\delta\mathring{\zeta}|_{\Sigma}+\epsilon^2 \mathscr{S}(\epsilon, \delta\breve{\rho})$,
which follows by direct calculation.

Having verified that all of the hypotheses of Theorem \ref{T:MAINMODELTHEOREM} are satisfied, we conclude that
there exists a  constant $\sigma>0$, independent of $\epsilon \in (0,\epsilon_0)$, such that if the free initial data is
chosen so that  $\|\breve{\xi}_\epsilon\|_s < \sigma$,
then the estimates
\begin{equation}
\|\mathbf{U}\|_{L^\infty((T_3,1], R^s)}\leq C\sigma, \quad
\|\mathring{\mathbf{U}}\|_{L^\infty((T_3,1], H^s)}\leq C\sigma  \AND \|\mathbf{U}-\mathring{\mathbf{U}}\|_{L^\infty((T_3,1], R^{s-1})}\leq \epsilon C \sigma \label{UUringbounds}
\end{equation}
hold
for some constant $C>0$, independent of $T_3\in (0,1)$ and  $\epsilon \in (0,\epsilon_0)$.
Furthermore, from the continuation criterion, it is clear that
the bounds \eqref{UUringbounds} imply that the solutions $\mathbf{U}$ and $\mathring{\mathbf{U}}$ exist globally on $M=(0,1]\times \mathbb{R}^3$ and satisfy the estimates
\eqref{UUringbounds} with $T_3=0$ uniformly for $\epsilon\in (0,\epsilon_0)$.
In particular, this implies, via the definition of $\mathbf{U}$ and $\mathring{\mathbf{U}}$, see \eqref{E:REALVAR} and
\eqref{Uringdef}, that
\begin{gather*}
\|\delta\zeta(t)-\delta\mathring{\zeta}(t)\|_{R^{s-1}}\leq \epsilon C \sigma, \quad \| z_j(t)-\mathring{z}_j(t)\|_{R^{s-1}}\leq \epsilon C \sigma, \\
\|u^{\mu\nu}_0(t)\|_{R^{s-1}}\leq \epsilon C \sigma, \quad \|u^{\mu\nu}_k(t)-
\delta^\mu_0\delta^\nu_0\partial_k\mathring{\Phi}(t)\|_{R^{s-1}}\leq \epsilon C\sigma, \quad \|u^{\mu\nu}(t)\|_{R^{s-1}}\leq \epsilon C \sigma, \\
\|u_0(t)\|_{R^{s-1}}\leq \epsilon C \sigma, \quad \|u_k(t)\|_{R^{s-1}}\leq \epsilon C \sigma \quad \text{and} \quad \|u(t)\|_{R^{s-1}}\leq \epsilon C \sigma
\end{gather*}
for $0 < t \leq 1$, while, from \eqref{E:V^0}, we see that
\begin{equation*}
\left\|\underline{\bar{v}^0}(t)-\sqrt{\frac{\Lambda}{3}}\right\|_{R^{s-1}}\leq C\epsilon\sigma, \quad 0 < t \leq 1.
\end{equation*}
 This concludes the proof of Theorem \ref{T:MAINTHEOREM}.

\section*{Acknowledgement}
\noindent This work was partially supported by the ARC grant FT120100045 and a Senior Scholarship from the Australia-American Fulbright Commission. Part of this work was completed during
visits by the first author to the Erwin Schrödinger International Institute for Mathematics and Physics (ESI), the Max Planck Institute for Gravitational Physics (AEI) and the School of Mathematical Sciences at Xiamen University, and  by the second author to the Mathematics Department at Princeton University. They are grateful for the support and hospitality during their visits.


\appendix

\section{Potential operators} \label{S:Potop}
In this section, we state the basic mapping properties of the Riesz and Bessel potentials that will be used throughout this article. We omit the proofs, which can be found in \cite{Grafakos2014a}.

\subsection{Riesz potentials}\label{S:Riesz}
In the following, we use
\begin{equation*}
	\Delta = \delta^{ij}\del{i}\del{j} \qquad (i,j=1,\cdots,n)
\end{equation*}
to denote the flat Laplacian on $\Rbb^n$.
For $0<s<\infty$, the \textit{Riesz potential operator of order $s$}, denoted $(-\Delta)^{-\frac{s}{2}}$, is
defined by
\begin{equation*}
(-\Delta)^{-\frac{s}{2}}(f)=(\widehat{\mathcal{K}}_s \widehat{f})^{\vee}=\mathcal{K}_s\ast f,
\end{equation*}
where
\begin{equation*}
\mathcal{K}_s(x)=\bigl(( 4\pi^2|\xi|^2)^{-\frac{s}{2}}\bigr)^{\vee}(x),
\end{equation*}
or equivalently, by
\begin{equation*}
	(-\Delta)^{-\frac{s}{2}}(f)(x)=2^{-s}\pi^{-\frac{n}{2}}\frac{\Gamma(\frac{n-s}{2})}{\Gamma(\frac{s}{2})}\int_{\Rbb^n} f(x-y) |y|^{-n+s} d^n y.
\end{equation*}

\begin{theorem}\cite[Theorem  $1.2.3$]{Grafakos2014a}\label{T:riepot}
	Suppose $0<s<n$ and $1<p<q<\infty$ satisfy $\frac{1}{p}-\frac{1}{q}=\frac{s}{n}$.
	Then
	\begin{equation*}
	\|(-\Delta)^{-\frac{s}{2}}(f)\|_{L^q}\lesssim \|f\|_{L^p }
	\end{equation*}	
for all  $f\in L^p(\Rbb^n) $.
\end{theorem}
From the point of view of applications considered in this article, the following specific cases of the above estimate
on $\Rbb^3$  will be of most interest:
\begin{gather*}
\|(-\Delta)^{-1}f\|_{L^6}\lesssim \|f\|_{L^\frac{6}{5}}  
\AND
\|(-\Delta)^{-\frac{1}{2}}f\|_{L^6}\lesssim \|f\|_{L^2}. 
\end{gather*}

Next, we recall that the \textit{Riesz transform} $\mathfrak{R}_j$ is defined by
\begin{equation*}
\mathfrak{R}_j=-\del{j}(-\Delta)^{-\frac{1}{2}}.
\end{equation*}

\begin{theorem}\cite[Corollary $5.2.8$]{Grafakos2014}\label{T:rietran}
	Suppose $1<p<\infty$ and $s\in \Zbb_{\geq 0}$. Then
	\begin{align*} 
	\|\mathfrak{R}_j(f)\|_{W^{s,p}} \lesssim\|f\|_{W^{s,p}}
	\end{align*}
	for all $f\in W^{s,p}(\Rbb^n)$.
\end{theorem}
Combining Theorems \ref{T:riepot} and \ref{T:rietran} gives:
\begin{proposition}$\;$ \label{E:ddelin1}

	\begin{enumerate}
	\item If $1<p<q<\infty$ satisfy  $\frac{1}{p}-\frac{1}{q}=\frac{1}{n}$,
	then
	\begin{align*}
	\|\del{j}\Delta^{-1} f\|_{L^q}=&\|\mathfrak{R}_j(-\Delta)^{-\frac{1}{2}} f\|_{L^q}\lesssim \|(-\Delta)^{-\frac{1}{2}} f\|_{L^q}
\lesssim \|f\|_{L^p}  
	\end{align*}
        for all $f\in L^p(\Rbb^n)$.
	\item If $1<p<\infty$ and $s\in \Zbb_{\geq 0}$, then
	\begin{align*}
	\|\del{j}\del{k}\Delta^{-1} f\|_{W^{s,p}}=&\|\mathfrak{R}_j\mathfrak{R}_k f\|_{W^{s,p}} \lesssim \|f\|_{W^{s,p}} 
	\end{align*}
for all $f\in W^{s,p}(\Rbb^n)$.
\end{enumerate}
\end{proposition}
For applications, the following case on $\Rbb^3$ will be of particular importance:
\begin{align*}
\|\del{j}\Delta^{-1} f\|_{L^6}\lesssim \|f\|_{L^2}.
\end{align*}

\subsection{Bessel potentials}\label{S:Bessel}
For $0<s<\infty$, the \textit{Bessel potential operator of order $s$}, denoted $(\mathds{1}-\Delta)^{-\frac{s}{2}}$, is defined by:
\begin{equation*}
	(\mathds{1}-\Delta)^{-\frac{s}{2}}(f)=(\widehat{\mathcal{G}}_s\widehat{f})^{\vee}=\mathcal{G}_s\ast f,
\end{equation*}
where
\begin{equation*}
	\mathcal{G}_s(x)=\bigl((1+4\pi^2|\xi|^2)^{-\frac{s}{2}}\bigr)^{\vee}(x).
\end{equation*}

\begin{theorem}\cite[Corollary $1.2.6$]{Grafakos2014a}\label{T:Bessel}
\begin{enumerate}
		\item \label{T:B1} Suppose $0<s<\infty$ and  $1\leq r\leq \infty$. Then
		\begin{equation*}
		\|(\mathds{1}-\Delta)^{-\frac{s}{2}}(f)\|_{L^r } \leq \|f\|_{L^r }
		\end{equation*}
 for all $f\in L^r(\Rbb^n)$, and
\begin{equation*}
			\|(\mathds{1}-\Delta)^{-\frac{s}{2}}\|_{\emph{op}}=\|\mathcal{G}_s\|_{L^1}=1.
\end{equation*}
		\item \label{T:B2} Suppose $0<s<n$ and $1 < p<q<\infty$ satisfy $\frac{1}{p}-\frac{1}{q}=\frac{s}{n}$.
		Then
		\begin{equation*}
		\|(\mathds{1}-\Delta)^{-\frac{s}{2}}(f)\|_{L^q}\lesssim \|f\|_{L^p }
		\end{equation*}
for all $f\in L^p (\Rbb^n)$.
	\end{enumerate}
\end{theorem}
For applications, the following cases on $\Rbb^3$ will be of particular interest:
		\begin{gather*}
		\|(\mathds{1}-\Delta)^{-1}(f)\|_{L^6} \leq \|f\|_{L^6} \AND
		\|(\mathds{1}-\Delta)^{-1}(f)\|_{L^6}\lesssim \|f\|_{L^\frac{6}{5}}.
		\end{gather*}

\section{Calculus Inequalities} \label{A:INEQUALITIES}

In this appendix, we list important calculus inequalities that will be used throughout this article. Most of the proofs
can be found in the references \cite{Adams2003b,Koch1990,Taylor2010}.
Proofs will be given for statements
that cannot be found in those references.

\subsection{Sobolev-Gagliardo-Nirenberg inequalities}
\begin{theorem}{\emph{[H\"older's inequality]}}\label{T:HOLDER}
	\begin{enumerate}
		\item \label{T:H1} If $0<p,q,r \leq \infty$ satisfy $1/p+1/q=1/r$, then
		\begin{align*}
		\|uv\|_{L^r}\leq \|u\|_{L^p}\|v\|_{L^q}
		\end{align*}
		for all $u \in L^p(\mathbb{R}^n)$ and $v\in L^q(\mathbb{R}^n)$.
		\item  \label{T:H2} If $0<p,q,r \leq \infty$, $0\leq \theta \leq 1$ and
		$\frac{1}{r}=\frac{\theta}{p}+\frac{1-\theta}{q}$,
		then
		\begin{align*}
		\|u\|_{L^r}\leq \|u\|^\theta_{L^p}\|u\|^{1-\theta}_{L^q}\lesssim \|u\|_{L^p}+\|u\|_{L^q}
		\end{align*}
		for all $u \in L^p \bigcap L^q(\Rbb^n)$.
	\end{enumerate}
\end{theorem}

\begin{theorem}{\emph{[Gagliardo-Nirenberg-Sobolev inequalities]}} \label{Sobolev}
	\begin{enumerate}
		\item \label{sob1} If $1\leq p < \infty$, then
		\begin{align*}
		\norm{u}_{L^{p*}} \leq C_{\text{Sob}} \norm{Du}_{L^p}, \qquad p^* = \frac{np}{n-p},
		\end{align*}
		for all $u\in \{ v \in L^{p*}(\Rbb^n)\cap W^{1,p}_{\emph{loc}}(\Rbb^n) \: | \: \norm{Dv}_{L^p} < \infty\}$.
		\item  \label{sob2} If $s\in \Zbb_{\geq 1}$,
		$1\leq p < \infty$ and $sp<n$, then
		\begin{align*} 
		\norm{u}_{L^q} \lesssim \norm{u}_{W^{s,p}}, \qquad p\leq q \leq \frac{np}{n-s p}
		\end{align*}
		for all $u\in W^{s,p}(\Rbb^n)$.
		\item \label{sob3} 
		If $s\in \Zbb_{\geq 1}$,
		$1\leq p < \infty$ and $sp > n$, then
		\begin{align*}
		\norm{u}_{L^\infty} \lesssim \norm{u}_{W^{s,p}} 
		\end{align*}
		for all $u\in W^{s,p}(\Rbb^n)$.
	\end{enumerate}
\end{theorem}

\subsection{Product and commutator inequalities}

\begin{theorem} \label{calcpropB} $\;$
	\begin{enumerate}
		\item \label{T:I}
		Suppose $1\leq p_1,p_2,q_1,q_2\leq \infty$, $|\alpha|=s \in \Zbb_{\geq 0}$, 
		and $\frac{1}{p_1}+\frac{1}{q_1}= \frac{1}{p_2} + \frac{1}{q_2} = \frac{1}{r}$.
		Then
		\begin{align*}
		\norm{D^\alpha(uv)}_{L^r } \lesssim & \norm{D^s u}_{L^{p_1} }\norm{v}_{L^{q_1} } + \norm{u}_{L^{p_2} }\norm{D^s v}_{L^{q_2} } 
		\intertext{and}
		\norm{[D^\alpha,u]v}_{L^r } \lesssim & \norm{Du}_{L^{p_1} }\norm{D^{s-1} v}_{L^{q_1} } + \norm{D^s u}_{
			L^{p_2} }\norm{v}_{L^{q_2} } 
		\end{align*}
		for all $u,v \in C^\infty_0(\Rbb^n)$.
		\item  \label{T:II} If $s_1,s_2\geq s_3\geq 0$, $1\leq p \leq \infty$, and $s_1+s_2-s_3 > n/p$, then
		\begin{align*}
		\norm{uv}_{W^{s_3,p} } \lesssim \norm{u}_{W^{s_1,p} }\norm{v}_{W^{s_2,p} }
		\end{align*}
for all $u\in W^{s_1,p}(\Rbb^n)$ and $v\in W^{s_2,p}(\Rbb^n)$.
	\end{enumerate}
\end{theorem}

\begin{proposition}\label{T:DUV}
If $s\in \Zbb_{\geq 1}$, then
\begin{equation*}
	\|D(uv)\|_{R^s}\lesssim \|Du\|_{R^s}\|v\|_{\Li}+\|u\|_{\Li}\|Dv\|_{R^s}  
\end{equation*}
for all $u,v \in C^\infty_0(\Rbb^3)$.
Furthermore, if $s\in \Zbb_{\geq 2}$, then
	\begin{equation*}
		\|uv\|_{R^s}\lesssim \|u\|_{R^s}\|v\|_{R^s}
    \end{equation*}
for all $u,v \in C^\infty_0(\Rbb^3)$.
\end{proposition}
\begin{proof}
 By Theorem \ref{calcpropB}.\eqref{T:I} and the definition of $R^s$ norm, we have
 \begin{align*}
 	\|D(uv)\|_{R^s}= & \|D(uv)\|_{L^6}+\|D^2(uv)\|_{H^{s-1}} \nnb  \\
 	\lesssim & \bigl( \|Du\|_{L^6}+\|D^2u\|_{H^{s-1}}\bigr)\|v\|_{\Li}+\|u\|_{\Li} \bigl(\|Dv\|_{L^6} + \|D^2v\|_{H^{s-1}}\bigr)  \nnb  \\
 	\lesssim & \|Du\|_{R^s}\|v\|_{\Li}+\|u\|_{\Li}\|Dv\|_{R^s}
 \end{align*}
for $s\geq 1$, and similarly, with the help of \eqref{E:NORMEQ1},
  \begin{align*}
 \|uv\|_{R^s}  = & \|uv\|_{L^6}+\|D(uv)\|_{H^{s-1}} \nnb  \\
\lesssim & \bigl( \| u\|_{L^6}+\|D u\|_{H^{s-1}}\bigr)\|v\|_{\Li}+\|u\|_{\Li} \bigl(\| v\|_{L^6} + \|D v\|_{H^{s-1}}\bigr)  \nnb  \\
 \lesssim & \|u\|_{R^s}\|v\|_{\Li}+\|u\|_{\Li}\|v\|_{R^s} \nnb \\
  \lesssim &\|u\|_{R^s}\|v\|_{R^s}
 \end{align*}
for $s\geq 2$.
\end{proof}

\begin{proposition}\label{T:commu}
	If $s\geq 3$, $0\leq |\alpha| \leq s$ and
	 $0\leq |\beta| \leq s-1$, then
\begin{align}
		\|[D^\beta, u] v\|_{L^2 }   \lesssim & \|Du\|_{R^{s-2}\cap L^2 } \|v\|_{R^{s-1} }\lesssim
		 \|u\|_{R^{s-1} }  \|v\|_{R^{s-1} } ,   \label{E:COMMUTATOR3} \\
		\|[D^\alpha, u] v\|_{L^2 }
		\lesssim &
		 \|Du\|_{R^{s-1}\cap L^2 } \|v\|_{R^{s-1} }\lesssim\|u\|_{R^s }\|v\|_{R^{s-1} },   \label{E:COMMUTATOR6} \\
		\|[D^\beta, u] D v\|_{L^2 } \lesssim & 
		\|Du\|_{R^{s-1}\cap L^2 } \|v\|_{R^{s-1}}\lesssim \|u\|_{R^s }\|v\|_{R^{s-1} }   \label{E:COMMUTATOR4}  \\
		\intertext{and}
		\|[D^\beta, u] w D v\|_{L^2 } \lesssim & \|Du\|_{R^{s-1}\cap L^2 } \|w\|_{R^s }\|v\|_{R^{s-1} } \lesssim  
		 \|u\|_{R^s }\|w\|_{R^s }\|v\|_{R^{s-1} }  \label{E:COMMUTATOR9}
\end{align}
	for all $u,v,w \in C^\infty_0(\Rbb^3)$.
\end{proposition}
\begin{proof}  
We first consider the inequalities \eqref{E:COMMUTATOR3}-\eqref{E:COMMUTATOR6}. Since they are trivial for $|\alpha|=|\beta|=0$, we start by assuming that $|\alpha|=|\beta|=1$.
Then
\begin{align*}
	\|[D, u]v\|_{L^2}=\|(Du)v\|_{L^2} \leq \|v\|_{L^\infty}\|Du\|_{L^2}\overset{\eqref{E:NORMEQ1}}{\lesssim} \|Du\|_{R^{s-2}\cap L^2 } \|v\|_{R^{s-1} }.
\end{align*}
Next, assuming $|\beta|=l\geq 2$, we see from Theorem \ref{T:HOLDER}.\eqref{T:H2} and Theorem \ref{calcpropB}.(1) that
\begin{align*}
	\|[D^\beta, u] v\|_{L^2 }
\lesssim &  \|D^\ell u\|_{L^2 } \|v\|_{\Li }+\|D u\|_{L^6}\|D^{l-1}v\|_{L^3} \nnb \\
\lesssim &  \|D^\ell u\|_{L^2 } \|v\|_{\Li }+\|D u\|_{L^6}(\|D^{l-1}v\|_{L^2}+\|D^{l-1}v\|_{L^6}) \nnb \\
\lesssim & \|D u\|_{H^{l-1} } \|v\|_{\Li }+\|D u\|_{L^6}(\|v\|_{W^{l-1,6}}+\|Dv\|_{H^{l-2}}) \nnb  \\
    \overset{\eqref{E:NORMEQ1}}{\lesssim} &(\|D u\|_{R^{s-2} }+\|D u\|_{L^2 }) \|v\|_{R^{s-1} }\lesssim \|u\|_{R^{s-1} }  \|v\|_{R^{s-1} }
\intertext{while if $|\alpha|=s$, then }
	\|[D^\alpha, u] v\|_{L^2 } \lesssim & \|Du\|_{\Li }\|Dv\|_{H^{s-2} }+\|v\|_{\Li }\|Du\|_{\Hss }  \nnb  \\
	\overset{\eqref{E:NORMEQ1}}{\lesssim} & (\|D u\|_{\Rs }+\|D u\|_{L^2 }) \|v\|_{R^{s-1} } \lesssim \|u\|_{R^s }\|v\|_{R^{s-1} }.
\end{align*}
Together the above three inequalities verify the validity of \eqref{E:COMMUTATOR3}-\eqref{E:COMMUTATOR6}.

Using the identity $[D^\beta, u]Dv=[D^\beta D, u]v -D^\beta \bigl((Du) v\bigr)$, it is then clear that
	\begin{align*}
		\|[D^\beta, u]Dv\|_{L^2}\leq & \|[D^\beta D, u]v\|_{L^2}+\|D^\beta \bigl((Du) v\bigr)\|_{L^2}\lesssim
		\|Du\|_{R^{s-1}\cap L^2 } \|v\|_{R^{s-1}}\lesssim \|u\|_{R^s }\|v\|_{R^{s-1} }.
	\end{align*}	
follows from Theorem \ref{calcpropB}.\eqref{T:I} and the inequality \eqref{E:COMMUTATOR6}. This establishes the
inequality  \eqref{E:COMMUTATOR4}. For the final inequality, we find, using the identity $ [D^\beta, u] w D v=[D^\beta, u] D(wv)- [D^\beta, u] \bigl((D w) v\bigr)$, the inequality \eqref{E:COMMUTATOR3} and \eqref{E:COMMUTATOR4}, Proposition \ref{T:DUV} and
the obvious inequality $\|Dw\|_{R^{s-1}}\lesssim \|w\|_{R^s}$, that
\begin{align*}
\|[D^\beta, u] w D v\|_{L^2 } \leq & \|[D^\beta, u] D(wv)\|_{L^2} +\|[D^\beta, u] \bigl((D w) v\bigr)\|_{L^2}  \nnb  \\
\lesssim & \|Du\|_{R^{s-1}\cap L^2 }\bigl(\|wv\|_{R^{s-1}}+\|(Dw)v\|_{R^{s-1}}\bigr)\nnb  \\
\lesssim & \|Du\|_{R^{s-1}\cap L^2 } \bigl(\|w\|_{R^{s-1}}\|v\|_{R^{s-1}}+\|Dw\|_{R^{s-1}}\|v\|_{R^{s-1}}\bigr) \nnb  \\
\lesssim & \|Du\|_{R^{s-1}\cap L^2 } \|w\|_{R^s }\|v\|_{R^{s-1} },
\end{align*}
which completes the proof.
\end{proof}

\subsection{Moser estimates}

\begin{theorem} \label{calcpropC}
	Suppose $s\in \Zbb_{\geq 1}$, $1\leq p \leq \infty$, $|\alpha| = s$, $f\in C^{s}(\Rbb)$ with $f(0) = 0$, and $U$ is open and bounded in $\Rbb$. Then
	\begin{align*}
	\norm{D^\alpha f(u)}_{L^{p} } &\leq C\bigl(\norm{Df}_{C^{s-1}(\overline{U})}\bigr)\norm{u}^{s-1}_{\Li } \norm{D^s u}_{L^p }
\intertext{and}
\norm{f(u)}_{L^{p}} & \leq  C\bigl(\norm{Df}_{C^{0}(\overline{U})}\bigr)\|u\|_{L^{p}}
	\end{align*}
	for all $u \in C^\infty_0(\Rbb^n)$ with
	$u(\xv) \in U$ for all $\xv\in \Rbb^n$.
\end{theorem}

\begin{theorem}  \label{moserlemB}
	If $s\in \mathbb{Z}_{\geq 1}$ and  $l \in \mathbb{Z}_{\geq 2}$, then
	\begin{align*}
	\|f_1\ldots f_l\|_{H^s}\lesssim \biggl({\|f_l\|_{H^s}\prod_{i=1}^{l-1}\|f_i\|_{L^{\infty}}+\sum_{i=1}^{l-1}\| D f_i\|_{H^{s-1}}\prod_{i\neq j}\|f_j\|_{L^{\infty}}}\biggr)
	\end{align*}
for all $f_i \in C^\infty_0(\Rbb^n)$,  $1\leq i \leq l$.
\end{theorem}

\begin{proposition}\label{T:Moser2}
Suppose $s\in \Zbb_{\geq 2}$, $f\in C^s(\Rbb)$, $f(0) = 0$ and $U$ is open and bounded in $\Rbb$. Then
	\begin{align*}  
		\|D f(u)\|_{\Qs } &\leq
C\bigl(\norm{Df}_{C^{s-1}(\overline{U})}\bigr)(1+\norm{u}^{s-1}_{\Li }) \|Du\|_{R^{s-1} }
\intertext{and}
\| f(u)\|_{R^s } &\leq C\bigl(\norm{Df}_{C^{s-1}(\overline{U})}\bigr)(1+\norm{u}^{s-1}_{\Li })
\|u\|_{R^s }
\end{align*}
for all $u \in C^\infty_0(\Rbb^3)$ with
	$u(\xv) \in U$ for all $\xv\in \Rbb^3$.
\end{proposition}
\begin{proof}
From a direct application of Theorem \ref{calcpropC} and the definition of $R^s$ norm, we see that
\begin{align*}
	\|Df(u)\|_{R^{s-1}}=  \|Df(u)\|_{L^6}+\|D^2f(u)\|_{H^{s-2}}
	\leq  & C\bigl(\norm{Df}_{C^{s-1}(\overline{U})}\bigr)(1+\norm{u}^{s-1}_{\Li })(\|Du\|_{L^6}+\|D^2u\|_{H^{s-2}}) \nnb  \\
	\leq  & C\bigl(\norm{Df}_{C^{s-1}(\overline{U})}\bigr)(1+\norm{u}^{s-1}_{\Li })\|Du\|_{R^{s-1}},
\end{align*}
which establishes the first inequality. The second inequality follows from a similar argument.
\end{proof}

\subsection{Young's inequalities}

\begin{proposition}\emph{(Young's inequality for convolutions)}\label{T:Youngineq}
	If $1/p+1/q=1+1/r$ with $1\leq p,q,r\leq +\infty$, then $u\ast v \in L^r(\Rbb^n)$ and
	\begin{equation*} 
	\|u\ast v\|_{L^r}\leq \|u\|_{L^p}\|v\|_{L^q}
	\end{equation*}
 for all $u\in L^p (\Rbb^n)$ and $v\in L^q(\Rbb^n)$, where $u\ast v$ denotes convolution.
\end{proposition}

\begin{proposition}\emph{(Young's inequality for conjugate H\"older exponents)}  \label{T:young}
	If $a,b\geq 0$ and $0<p,q < \infty$ satisfy $\frac{1}{p}+\frac{1}{q}=1$,
	then
	\begin{equation*}
		ab\leq \frac{a^p}{p}+\frac{b^q}{q}
	\end{equation*}
	with equality holding if and only if $a^p=b^q$.
\end{proposition}

\section{Matrix relations}\label{s:matr}

\begin{lemma}\label{A:INVERSEOFA}
	Suppose
	\begin{align*}
	A=\begin{pmatrix}
	a & b\\
	b^\mathrm{T} &c
	\end{pmatrix}
	\end{align*}
	is an $(n+1)\times (n+1)$ symmetric matrix, where $a$ is an $1\times 1$ matrix, $b$ is an $1\times n$ matrix and c is an $n\times n$ symmetric matrix. Then
	\begin{align*}
	A^{-1}=\begin{pmatrix}
	a & b\\
	b^\mathrm{T} &c
	\end{pmatrix}^{-1}=\begin{pmatrix}
	\frac{1}{a}\Bigr[1+b\Bigl(c-\frac{1}{a}b^\mathrm{T}b\Bigr)^{-1}b^\mathrm{T}\Bigr] & -\frac{1}{a}b\Bigl(c-\frac{1}{a}b^\mathrm{T}b\Bigr)^{-1}\\
	-\Bigl(c-\frac{1}{a}b^\mathrm{T}b\Bigr)^{-1}\frac{1}{a}b^\mathrm{T} & \Bigl(c-\frac{1}{a}b^\mathrm{T}b\Bigr)^{-1}
	\end{pmatrix}
	\end{align*}
\end{lemma}
\begin{proof}
	Follows from direct computation.
\end{proof}

We also recall the well-known Neumann series expansion.
\begin{lemma} \label{t:expinv}
	If $A$ and $B$ are $n\times n$ matrices with $A$ invertible,  then there exists an $\epsilon_0>0$ such that the map
	\begin{equation*}
	(-\epsilon_0,\epsilon_0) \ni \epsilon \longmapsto (A+\epsilon B)^{-1} \in \mathbb{M}_{n\times n}
	\end{equation*}
	is analytic and can be expanded as
	\begin{align*}
	(A+\epsilon B)^{-1}=A^{-1}+\sum_{n=1}^\infty (-1)^n\epsilon^n (A^{-1}B)^nA^{-1}, \quad  |\epsilon|< \epsilon_0.
	\end{align*}
\end{lemma}

\section{Index of notation} \label{index}

\bigskip

\begin{longtable}{ll}
	$\tilde{g}_{\mu\nu}$ & physical spacetime metric; \S \ref{S:INTRO} \\
	$\tilde{v}^\mu$ & physical fluid four-velocity; \S \ref{S:INTRO} \\
	$\bar{\rho}$ & fluid proper energy density; \S \ref{S:INTRO} \\
	$\bar{p} = \epsilon^2 K\bar{\rho}$ & fluid pressure; \S \ref{S:INTRO} \\
	$\epsilon = \frac{v_T}{c}$ & Newtonian limit parameter; \S \ref{S:INTRO} \\
	$M =(0,1]\times \mathbb{R}^3$ & spacetime manifold; \S \ref{S:INTRO}\\
	$a(t)$ & FLRW scale factor; \S \ref{S:INTRO}, eqns. \eqref{FLRW.a} and \eqref{E:TPTA}\\
	$\tilde{v}_H(t)$ & FLRW fluid four-velocity; \S \ref{S:INTRO}, eqn. \eqref{FLRW.b} \\
	$\mu(t)$ & FLRW proper energy density; \S \ref{S:INTRO}, eqn. \eqref{FLRW.c} (see also \eqref{E:OMEGAREP} and \eqref{E:ZETAH2})  \\
	$(\bar{x}^\mu) = (t,\bar{x}^i)$ & relativistic coordinates; \S \ref{S:INTRO}\\
	$(x^\mu)=(t,x^i)$ & Newtonian coordinates; \S \ref{S:INTRO}, eqn. \eqref{e:NRcoor}\\
	$\mathring{a}(t)$ & Newtonian limit of $a(t)$; \S \ref{S:INTRO}, eqn. \eqref{arhoringdef}\\
	$\mathring{\mu}(t)$ & Newtonian limit of $\mu(t)$; \S \ref{S:INTRO}, eqn. \eqref{arhoringdef}  \\
	$\underline{f}(t,x^i)$ & evaluation in Newtonian coordinates; \S \ref{iandc}, eqn. \eqref{Neval}\\
	$\|\cdot\|_{R^s}$, $\|\cdot\|_{K^s}$  &  norms; \S \ref{S:ulss}, eqn. \eqref{e:ennorm}\\
	$X^s_{\epsilon_0}(\mathbb{R}^3)$, $X^s (\mathbb{R}^3)$  & initial data function space; \S \ref{S:ulss}\\
	$\mathscr{S}(\epsilon,t,\xi)$, $\mathscr{T}(\epsilon,t,\xi)$, $\ldots$ & remainder terms that are elements of
	$E^0$, \S \ref{remainder}   \\
	$\mathscr{\hat{S}}(\epsilon,t,\xi)$, $\mathscr{\hat{T}}(\epsilon,t,\xi)$, $\ldots$ & remainder terms that are elements of
	$E^1$, \S \ref{remainder}   \\
	$\bar{g}^{\mu\nu}$ & conformal metric; \S \ref{conformalEinsteinEuler}, eqn.
	\eqref{E:CONFORMALTRANSF} \\
	$\bar{v}^\mu$ & conformal four-velocity; \S \ref{conformalEinsteinEuler}, eqn.
	\eqref{E:CONFORMALVELOCITY}\\
	$\Psi$ & conformal factor; \S \ref{conformalEinsteinEuler}, eqn. \eqref{E:CONFORMALFACTOR}\\
	$\bar{h}$ & conformal FLRW metric; \S \ref{conformalEinsteinEuler}, eqn. \eqref{E:CONFORMALFLRW}\\
	$E(t)$ & modified scale factor; \S \ref{conformalEinsteinEuler}, eqn. \eqref{E:DEFE} (see also \eqref{E:EREP})\\
	$\Omega(t)$ & modified density;  \S \ref{conformalEinsteinEuler}, eqn. \eqref{e:ome} (see also \eqref{E:OMEGAREP})\\
	$\bar{\gamma}^0_{ij}$, $\bar{\gamma}^i_{j0}$ & non-vanishing Christoffel symbols of $\bar{h}$;
	\S \ref{conformalEinsteinEuler}, eqn. \eqref{E:HOMCHRIS}\\
	$\bar{\gamma}^\sigma$ & contracted Christoffel symbols of $\bar{h}$;
	\S \ref{conformalEinsteinEuler}, eqn. \eqref{E:HOMCHRIS}\\
	$\tensor{\bar{ \mathcal{R}}}{_{\mu\nu\sigma}^\lambda}$   &   contracted Riemannian tensor; \S \ref{conformalEinsteinEuler}, eqn. \eqref{e:Hriem1}-\eqref{e:Hriem}  \\
	$\tensor{\bar{ \mathcal{R}}}{_{\mu\nu}}$, $\tensor{\bar{ \mathcal{R}}}{^{\mu\nu}}$, $ \bar{ \mathcal{R}}  $  & (inverse of) Ricci tensors and Ricci scalar of $\bar{h}$; \S \ref{conformalEinsteinEuler}, eqn. \eqref{e:R1}-\eqref{e:R2}  \\
	$\udn{\mu}$, $\bar{\nabla}_\mu$ & covariant derivative with respect to metrics $\bar{h}$ and $\bar{g}$, respectively; \S \ref{conformalEinsteinEuler}  \\
	$\bb$, $\bar{\Box}$ & d'Alembertian operator with respect to metrics $\bar{h}$ and $\bar{g}$, respectively; \S \ref{conformalEinsteinEuler}  \\ 
	$\bar{Z}^\mu$ & wave gauge vector field;  \S \ref{Wavegauge}, eqn. \eqref{E:WAVEGAUGE}\\
	$\bar{X}^\mu$ & contracted Christoffel symbols;  \S \ref{Wavegauge}, eqn. \eqref{e:X}\\
	$\bar{Y}^\mu$ & gauge source vector field;  \S \ref{Wavegauge}, eqn. \eqref{e:Y}\\
	$u^{\mu\nu}$, $u$ & modified conformal metric variables;  \S \ref{vardefs}, eqns. \eqref{E:u.a}, \eqref{E:u.d} and \eqref{E:u.f}\\
	$u^{\mu\nu}_\gamma$ & first order metric field variables;  \S \ref{vardefs}, eqns. \eqref{E:u.b}, \eqref{E:u.c}, \eqref{E:u.e}
	and \eqref{E:u.g}\\
	$z_i$ &  modified lower conformal fluid 3-velocity;  \S \ref{vardefs}, eqn. \eqref{E:z.b}\\
	$\zeta$ & modified density;  \S \ref{vardefs}, eqn. \eqref{E:ZETA}\\
	$\delta \zeta$ & difference between $\zeta$ and $\zeta_H$; \S \ref{vardefs}, eqn. \eqref{E:DELZETA}\\
	$\bar{\mathfrak{g}}^{ij}$ & densitized conformal 3-metric; \S \ref{vardefs}, eqn. \eqref{E:GAMMA} \\
	$\alpha$ & cube root of conformal 3-metric determinant;  \S \ref{vardefs}, eqn. \eqref{E:GAMMA} \\
	$\check{g}_{ij}$ & inverse of the conformal 3-metric $\bar{g}^{ij}$;  \S \ref{vardefs}, eqn. \eqref{E:GAMMA} \\
	$\bar{\mathfrak{q}}$ & modified conformal 3-metric determinant; \S \ref{vardefs}, eqn. \eqref{E:q} \\ 
	$\zeta_H(t)$ & FLRW modified density; \S \ref{vardefs}, eqns. \eqref{E:ZETAH1} and \eqref{E:ZETAH3}\\
	$C_0$ & FLRW constant; \S \ref{vardefs}, eqn. \eqref{C0def}\\
	$\mathring{\zeta}_{H}(t)$ & Newtonian limit of $\zeta_H(t)$, \S \ref{vardefs}, eqns. \eqref{zetaHringdefa}
	and \eqref{zetaHringform}  (see also \eqref{deltazetaringH})\\
	$z^i$ &  modified upper conformal fluid 3-velocity;  \S \ref{vardefs}, eqn. \eqref{E:z.a}\\
	$\mathring{\rho}$ & Newtonian fluid density; \S \ref{CPEequations}\\
	$\mathring{z}^j$ & Newtonian fluid 3-velocity; \S \ref{CPEequations}\\
	$\mathring{\Phi}$ & Newtonian potential; \S \ref{CPEequations}\\
	$\mathring{E}(t)$ & Newtonian limit of $E(t)$; \S \ref{CPEequations}, eqn. \eqref{Eringform}\\
	$\mathring{\Omega}(t)$ & Newtonian limit of $\Omega(t)$; \S \ref{CPEequations}, eqn. \eqref{Oringdef}\\
	$\mathring{\zeta}$ & modified Newtonian fluid density; \S \ref{CPEequations}\\
	$\rho$ & fluid proper energy density in Newtonian coordinates; \S \ref{epexpansions}, eqn. \eqref{E:ZETA2}\\
	$\delta\rho$ & difference between $\rho$ and $\rho_H$; \S \ref{epexpansions}, eqn. \eqref{E:DELRHO}\\
	$\mathbf{\hat{U}}$ &combined gravitational and matter field vector; \S \ref{rcEEeqns}, eqn. \eqref{E:REALVAR0} \\
	$\mathbf{\hat{U}}_1$ & gravitational field vector; \S \ref{rcEEeqns}, eqn. \eqref{E:REALVAR1}\\
	$\mathbf{U}_2$ & matter field vector; \S \ref{rcEEeqns}, eqn. \eqref{E:REALVAR1}\\
	$\hat{g}^{\mu\nu}$ &  rescaled conformal metric by $\theta$; \S \ref{S:INITIALIZATION}, eqn. \eqref{E:DEFHATG}\\	$\hmfu^{\mu\nu}$ &  modified gravitational variables of $\hat{g}^{\mu\nu}$; \S \ref{S:INITIALIZATION}, eqn. \eqref{E:DEFHATG}\\	$\hmfu^{\mu\nu}_\sigma$ &  first order derivative of modified gravitational variables of $\hat{g}^{\mu\nu}$; \S \ref{S:INITIALIZATION}, eqn. \eqref{E:DEFHATG}\\
	$\theta$ & ratio of $\sqrt{|\bar{g}|}$ and $\sqrt{|\bar{h}|}$; \S \ref{S:INITIALIZATION}, eqn. \eqref{E:SMALLTHETA}  \\
	$\smfu^{\mu\nu}$, $\smfu^{\mu\nu}_0$ & initial data of gravitational variables; \S \ref{S:INITIALIZATION}, eqn. \eqref{initvarsA}  \\
	$\delta\breve{\rho}$,  $\breve{{z}}^j$ & initial data of matter field variables; \S \ref{S:INITIALIZATION}, eqn. \eqref{initvarsB}  \\
	$\mathscr{\breve{Q}}(\xi)$, $\mathscr{\breve{R}}(\xi)$, $\mathscr{\breve{S}}(\xi)$, $\ldots$ & analytic remainder terms; \S \ref{S:INITIALIZATION} \\
	$\breve{\xi}$ & the set of free data; \S \ref{S:INITIALIZATION}, eqn. \eqref{e:freedata}  \\
	$\|\xi\|_s$   &  free initial data norm; \S \ref{S:INITIALIZATION}, eqn. \eqref{e:normfreedata} \\
	$\mathbb{K}$ & vector space; \S \ref{E:ulex}, Proposition \ref{t:ulst}\\
	$\varpi^j$ & quantity related to $\mathring{\rho}\mathring{z}^j$; \S \ref{S:locPE},  eqn. \eqref{e:varpi}  \\
	$\delta\mathring{\zeta}$ & difference between $\mathring{\zeta}$ and $\mathring{\zeta}_H$; \S \ref{S:locPE}, eqn. \eqref{deltazetaringdef}  \\
	$\mathring{z}_j$ & modified lower Newtonian fluid 3-velocity; \S \ref{S:locPE}, eqn. \eqref{deltazetaringdef}  \\$\mathring{\Upsilon}$ &  modified Newtonian potential in terms of Yukawa potential; \S \ref{S:locPE}, eqn. \eqref{Upsi}\\
	$\mathring{\Phi}_i$ & first derivative of Newtonian potential; \S \ref{S:locPE}, eqn. \eqref{e:phij}  \\
	$\Phi^\mu_k$ & first derivative of modified gravitational potential; \S \ref{S:NloEE}, eqn. \eqref{E:POTENTIAL}  \\
	$\Upsilon$ &  modified gravitational potential in terms of Yukawa potential; \S \ref{S:NloEE}, eqn.  \eqref{e:Upsi2}  \\
	$w^{0\mu}_k$ & shifted first order gravitational variable; \S \ref{nonlocaleq}, eqn. \eqref{E:WPHI}\\
	$\mathbf{U}_1$ & shifted gravitational field vector; \S \ref{completeevolution}, eqn. \eqref{E:U1}\\
	$\mathbf{U}$ & combined gravitational and matter field vector; \S \ref{completeevolution}, eqn. \eqref{E:REALVAR}\\
	$\vertiii{\cdot}_{a,H^k}$, $\vertiii{\cdot}_{a,R^k}$ & energy norms; \S \ref{S:MODELuni}, Definition \ref{energynorms}\\
	$\vertiii{\cdot}_{H^k}$, $\vertiii{\cdot}_{R^k}$ & energy norms; \S \ref{S:MODELuni}, Definition \ref{energynorms}\\
    $\|\cdot\|_{M^\infty_{\mathbb{P}_a,k}([T_0,T)\times \mathbb{R}^3)}$  & the spacetime norm; \S \ref{S:MODELuni}, Definition \ref{energynorms}
\end{longtable}

\bibliographystyle{amsplain}
\bibliography{cntlimr3}

\end{document}